\documentclass{dis}
\usepackage{amsmath}
\usepackage{times}
\usepackage{amssymb}
\usepackage{amsfonts}
\usepackage{amscd}
\usepackage{amsthm}	
\usepackage{epsfig}
\usepackage{graphicx, psfrag}
\input xy
\xyoption{all}
\usepackage{hyperref}

\newtheorem{theorem}{Theorem}[chapter]
\newtheorem{proposition}[theorem]{Proposition}
\newtheorem{lemma}[theorem]{Lemma}
\newtheorem{corollary}[theorem]{Corollary}

\newtheorem{maintheorem}[theorem]{Main Theorem}

\theoremstyle{definition}
\newtheorem{definition}[theorem]{Definition}

\newtheorem{note}[theorem]{Note}
 \mathchardef\za="710B  %\alpha
 \mathchardef\zb="710C  %\beta
 \mathchardef\zg="710D  %\gamma
 \mathchardef\zw="7121  %\omega

\def\tx#1{{\fam0\relax#1}}
\def\xd{\tx{d}}
\def\Ad{\rm Ad}
\def\matriz#1#2{\left( \begin{array}{#1} #2 \end{array}\right) }
\def\GR{\mathcal{G}}
 \def\ch{{\rm cosh}}
 \def\sh{{\rm sinh}}

 \def\pd#1#2{\frac{\partial#1}{\partial#2}}
\def\LG{\mathfrak g}
 \def\LH{\mathfrak h}
 \def\di{\bigstar}
 \def\R{\mathbb{R}}
\def\exi#1{\setbox0=\hbox{\lower1.5pt\hbox{\bsymb \char"39}}
       \setbox1=\hbox{$_{#1}$} \box0\lower2pt\box1\;}

\def\Map{\mathop{\rm Map}\nolimits}
\begin{document}

\mathclass{Primary 34A26; Secondary 34A05, 34A34, 17B66, 22E70.}
\thanks{The authors would like to thank Cristina Sard\'on and Agata Przywa\l a for critical reading of the manuscript. Partial financial support
by research projects MTM2009-11154, MTM2009-08166-E and E24/1 (DGA)
is also acknowledged.}
\abbrevauthors{J.F. Cari\~nena and J. de Lucas}
\abbrevtitle{Lie systems: theory, generalisations, and applications}

\title{Lie systems: theory, generalisations, and applications}

\author{J.F. Cari\~nena}
\address{Departamento de F\'isica Te\'orica, Facultad de Ciencias\\ Universidad
de Zaragoza\\
c. Pedro Cerbuna, 12\\ 50.009 Zaragoza, Spain\\
E-mail: jfc@unizar.es}

\author{J. de Lucas}
\address{Institute of Mathematics\\ Polish Academy of Sciences\\
ul. \'Sniadeckich, 8\\ 00-956 Warszawa, Poland\\
E-mail: delucas@impan.gov.pl}

\pagenumbering{alph}    % a, b, c, ...

\maketitledis

\pagenumbering{roman}   % i, ii, iii, iv, ...
\setcounter{page}{1}

\tableofcontents
\begin{abstract}
Lie systems form a class of systems of first-order ordinary differential
equations whose general solutions can be described in terms of certain finite
families of particular solutions and a set of constants, by means of a
particular type of mapping: the so-called superposition rule. Apart from this
fundamental property, Lie systems enjoy many other geometrical features and they
appear in multiple branches of Mathematics and Physics, which strongly motivates
their study. These facts, together with the authors' recent findings in the
theory of Lie systems, led to the redaction of this essay, which aims to
describe such new achievements within a self-contained guide to the whole theory
of Lie systems, their generalisations, and applications.

\end{abstract}
\makeabstract

\pagenumbering{arabic}  % 1, 2, 3, 4, ...
\setcounter{page}{1}

\chapter{The theory of Lie systems}
\section{Motivation and general scheme of the work}
It is a little bit surprising that the theory of {\it Lie systems}
\cite{LSPartial,LSII,LS,Ve94}, which studies a very specific class of systems of
first-order ordinary differential equations, can be employed to investigate a
large variety of topics
\cite{AndHarWin81,BecGagHusWin87,CLR07a,CRL07e,CarRam03,FLV10,JP09,WintSecond,
SW85}. Indeed, although being a Lie system is rather more an exception than a
rule \cite{In72}, these equations frequently turn up in multiple branches of
Mathematics and Physics. For instance, linear systems of first-order
differential equations, Riccati equations \cite{Davis}, and matrix Riccati
equations \cite{Gopal,Her79,Her80,Kalman60} are Lie systems that very frequently
appear in the literature \cite{CarRamcinc,FLV10,HarWinAnd83,LW96,Sc08,SW85,PW}.
This obviously motivates the study of the theory of Lie systems as a means to
investigate the properties of various remarkable differential equations and
their corresponding applications. 

The research on Lie systems involves the analysis of multiple interesting
geometric and algebraic problems. For example, the determination of the Lie
systems defined in a fixed manifold is related to the existence of
finite-dimensional Lie algebras of vector fields over such a manifold
\cite{LS,SW84b}. Furthermore, the study of Lie systems leads to the
investigation of foliations \cite{CGL09}, generalised distributions
\cite{CGM07}, Lie group actions \cite{LW96}, finite-dimensional Lie algebras
\cite{CarRamGra,LS,SW84b}, etc. As a result of the analysis of the former
themes, Lie systems provide methods to study the integrability of systems of
first-order differential equations \cite{CarRamGra}, Control Theory
\cite{CCR03,CarRam02,Clem06,Ram06}, geometric phases \cite{FLV10}, certain
problems in Quantum Mechanics \cite{CLInt09,CLR08WN}, and other topics. Finally,
it is remarkable that the theory of Lie systems has been investigated by means
of different techniques and approaches, like Galois theory \cite{DB08,BM09I} or
Differential Geometry \cite{CGM07,CarRam05b,Ram04,Ue72}.
  
When applying Lie systems to study more general systems of differential
equations than merely first-order ones (see for instance
\cite{CGL08,CGL09,CLR08,CC87,WintSecond}), the interest of their analysis
becomes even more evident. For example, in the research on systems of
second-order differential equations, which very frequently appear in Classical
Mechanics, various relevant differential equations can be studied by means of
Lie systems. Dissipative Milne--Pinney equations \cite{CL08Diss}, Milne--Pinney
equations \cite{CLR08}, Caldirola--Kanai oscillators \cite{CLRan08},
$t$-dependent frequency harmonic oscillators \cite{CRL07e},  or second-order
Riccati equations \cite{CL09SRicc,Ve95}, are just some examples of such systems
of second-order differential equations that have already been analysed
successfully through Lie systems. 

The relevance of the above studies, along with the determination of new
applications of Lie systems, is twofold. On one hand, they allow us to obtain
novel results about interesting differential equations. On the other hand, such
examples may show us new features or generalisations of the notions appearing in
the theory of Lie systems that were not previously determined. Let us briefly
provide a case in point. While studying second-order differential equations by
means of Lie systems \cite{CLR08,CLR07a,WintSecond}, a new type of
`superposition-like' expression describing the general solution of certain
systems of second-order differential equations appeared. These essays led to the
definition of a possible superposition rule notion for such systems whose main
properties are still under analysis \cite{CL09SRicc}. In addition, these works
carried out different approaches to analyse second-order differential equations:
by means of the SODE Lie system notion \cite{CLR08} and through regular
Lagrangians \cite{CLRan08}. The relations between these approaches or even the
existence of new approaches is still an open question that must be investigated
in detail \cite{CL09SRicc}.

Apart from the investigation of the above open problems, perhaps the most active
field of research into Lie systems is concerned with the development of new
generalisations of the Lie system and superposition rule notions. Quasi-Lie
systems \cite{CGL08,CGL09,CLL09Emd}, $t$-dependent superposition rules
\cite{CGL08}, PDE Lie systems \cite{CGM07,OG00}, SODE Lie systems \cite{CLR08}, 
partial superposition rules \cite{CGM07,LSPartial},  quantum Lie systems
\cite{CarRam05b}, or stochastic Lie--Scheffers systems \cite{JP09}  are just a
few generalisations of such concepts that have been carried out in order to
analyse non-Lie systems with techniques similar to those ones developed for
analysing Lie systems. Indeed, the list of generalisations is much larger and
even sometimes the superposition rule term has been used with different,
non-equivalent, meanings \cite{ReStr83,ReStr82}. 

In view of the above and many other reasons, the theory of Lie systems, along
with its multiple generalisations, can be regarded as a multidisciplinary active
field of research which involves the use of techniques from diverse branches of
Mathematics and Physics as well as their applications to Control Theory
\cite{Br72,Br73a,CCR03,CarRam03,CarRam02,Clem06,IRS10,Ram06,SW85}, Physics
\cite{CGM01,CLRan08,CarNas,PW}, and many other fields \cite{CAL}.	

Our work starts by surveying briefly the historical development of the theory of
Lie systems and several of their generalisations. In this way, we aim to provide
a general overview of the subject, the main authors, trends, and the principal
works dedicated to describing most of the results about this theme. Special
attention has been paid to provide a complete bibliography, which contains
numerous references that cannot be easily found  elsewhere. Furthermore, we have
detailed a full report containing the works published by the main contributors
to the theory of Lie systems: Lie \cite{LSPartial}-\cite{LS}, Vessiot
\cite{Ve93}-\cite{Ve97}, Winternitz
\cite{AndHarWin81,AndHarWin82,BecGagHusWin90,HarWinAnd83,GEW98,OlmRodWin86,
OlmRodWin87,PWIV,PW,PWIII,PWII}, Ibragimov \cite{NI1}-\cite{Ib08}, etc.
Additionally, we presented the main contents of some works which have been
written in other languages than English, e.g.
\cite{LSPartial,Ve93,Ve93II,Ve95}. 

After our brief approach to the history of Lie systems, the fundamental notions
of this theory and other related topics are presented. More specifically, along
with a recently developed differential geometric approach to the investigation
of Lie systems \cite{CGM07}, results about the application of Lie systems to
investigate Quantum Mechanics, partial differential equations (PDEs), systems of
second- and higher-order differential equations are discussed. This, together
with the previous historical introduction, furnishes a self-contained
presentation of the topic which can be used both as an introduction to the
subject and as a reference guide to Lie systems. 

Later on, in Chapter 2, our survey focuses on detailing the achievements
obtained by the authors who described a method to analyse second-order
differential equations. Chapter 3 is concerned with various applications of Lie
systems in Quantum Mechanics. Subsequently, we describe a theory of
integrability of Lie systems in Chapter 4. This theory is employed to
investigate some systems of differential equations appearing in Classical
Mechanics in Chapter 5 and various Schr\"odinger equations in Chapter 6.
Finally, Chapters 7 and 8  describe the theory and applications of a new
powerful technique, the {\it quasi-Lie schemes}, developed to apply the methods
for studying Lie systems to a much larger set of systems of differential
equations. In the same way as Lie systems, this method can straightforwardly be
applied to the setting of second- and higher-order differential equations and
Quantum Mechanics. Finally, diverse applications of this technique are performed
in Chapter 8.

\section{Historical introduction}\label{ISOA}
It seems that Abel dealt with the superposition rule concept for the first time,
while analysing the linearisation of nonlinear operators \cite{In72}. Apart from
this very early treatment of one of the notions studied within the theory of Lie
systems, the fundamentals of this theory were laid down during the end of the
XIX century by the Norwegian mathematician Sophus Lie
\cite{LSPartial,LSII,LSIII,LS} and the French one Ernest Vessiot
\cite{Ve93}-\cite{Ve10}. Indeed, Lie systems are also frequently referred to as
{\it Lie--Vessiot systems} in honour to their contributions. 

The first study focused on analysing differential equations admitting a
superposition rule was carried out by K\"onigsberger \cite{Ko83} in 1883. In his
work, he proved that the only first-order ordinary differential equations on the
real line admitting a superposition rule that depends algebraically on the
particular solutions are (up to a diffeomorphism) Riccati equations, linear and
homogeneous linear differential equations. Later on, in 1885, Lie proposed a
special class of systems of first-order ordinary differential equations
\cite[pg. 128]{LSPartial} whose general solutions can be worked out of certain
finite families of particular solutions and sets of constants
\cite{BM09II,Ue72}. 

Despite the above mentioned achievements, these pioneering works did not draw
too much attention. Nevertheless, the situation changed from 1893. At that time,
Vessiot and Guldberg proved, separately, a slightly more general form of
K\"onigsberger's main result. They demonstrated that (up to a diffeomorphism)
Riccati equations and linear differential equations are the only differential
equations over the real line admitting a superposition rule
\cite{Gu93,Ib00,In72,Ve93}. This result attracted Lie's attention \cite{LSII},
who claimed that their contribution is a simple consequence of his previous work
\cite{LSPartial}. More specifically, he stated that the systems which admit a
superposition rule are those ones that he had defined in 1885 \cite{LSIII}. In
view of these criticisms, Lie did not recognise the value of Vessiot and
Guldberg's discovery \cite{In72}. Nevertheless, some credit to them must be
given, as the theory of Lie does not easily lead to the case provided by Vessiot
and Guldberg \cite{In72}. 

Lie's remarks gave rise to one of the most important results about the theory of
Lie systems: the today called {\it Lie Theorem} \cite[Theorem 44]{LS}. This
theorem characterises systems of first-order ordinary differential equations
admitting a superposition rule. In addition, it provides some information on the
form of such a superposition rule. In \cite{LS}, Lie and Scheffers presented the
first detailed discussion about Lie systems. In recognition of this work, some
authors also call Lie--Scheffers systems to Lie systems.

In spite of this important success, Lie Theorem, as stated by Lie, contains some
small gaps in its proof as well as a slight lack of rigour about the definition
of superposition rule. This was noticed and fixed at the beginning of the XXI
century by Cari\~nena, Grabowski, Marmo, Bl\'azquez, and Morales
\cite{BM09II,CGM07}.

After Lie's reply, Vessiot recognised the importance of Lie's work and proposed
to call {\it Lie systems} those systems of first-order ordinary differential
equations admitting a superposition rule \cite{Ve94}. Apart from this first
`trivial result', Vessiot furnished many new contributions to the theory of Lie
systems \cite{Ve93II,Ve94,Ve96,Ve10} and he proposed various generalisations
\cite{Ve95,Ve97,Ve10}. For instance, he showed that a superposition-like
expression can be used to analyse particular types of second-order Riccati
equations \cite{Ve95}. More specifically, he proved that some of these equations
admit their general solutions to be worked out of families of four particular
solutions, their derivatives, and two real constants. As far as we know, this
constitutes the first result concerning the study of superposition rules for
nonlinear second-order differential equations. 

After a beginning in which a deep study of superposition rules and Lie systems
was carried out \cite{Gu93,LSPartial,LSII,LSIII,Ve93,Ve94,Ve95,Ve96,Ve97,Ve10},
the topic was almost forgotten for nearly a century. Just few works were devoted
to the study of superposition rules \cite{Ch48,Co55,Co61,Co65,Le59,Re46}. During
the seventies, nevertheless, the interest on the topic revived and many authors
focused again on investigating Lie systems, their generalisations, and
applications to Mathematics, Physics and Control Theory \cite{In65,AJ67,Op65}.
Among the reasons that motivated that rebirth of the theory of Lie systems, we
can emphasise the importance of the works of Winternitz and Brocket. On one
hand, Brocket analysed the interest of Lie systems in Control Theory
\cite{Br72,Br73a}, what initiated a research field that continues until the
present \cite{CCR03,CarRam03,CarRam02,Clem06,IRS10,Ram02Th,Ram06,Rod09,SW85}. On
the other hand, Winternitz and his collaborators made a huge contribution to the
theory of Lie systems and their applications to Physics, Mathematics and Control
Theory
\cite{AndHarWin81,AndHarWin82,BecGagHusWin90,BecHusWin86,BecHusWin86b,BPW86,
HarWinAnd83,HavPosWin99,LW96,PW,PWII}.

In view of its important contributions, let us discuss in slight detail some of
Winternitz's results. Using diverse results derived by Lie \cite{FunLS,LS},
Winternitz and his collaborators developed and applied a method to derive
superposition rules \cite{WintSecond,SW84,SW84b}. They also studied the problem
of classification of Lie systems through transitive primitive Lie algebras
\cite{SW84b}, a concept that also appeared in some of his works about the
integrability of Lie systems \cite{BPW86,BPW86II}. Winternitz also paid
attention to the analysis of discrete problems and numerical approximations of
solutions by means of superposition rules \cite{PeWi04,RW84,WintSecond,TW99}
and, finally, he, and his collaborators, developed a new generalisation of the
superposition rule notion, the so-called {\it super-superposition rule}, in
order to study the general solutions of various types of superequations
\cite{BecGagHusWin87,BecGagHusWin90}. 

Besides their theoretical achievements, Winternitz {\it et al.} applied their
methods to the analysis of multiple discrete and differential equations with
applications to Mathematics, Physics and Control Theory. For instance, many
superposition rules were derived for Matrix Riccati equations
\cite{AndHarWin81,HarWinAnd83,LW96,OlmRodWin87,RW84,SW85}, which play an
important role in Control Theory, as well as for diverse Lie systems, like
projective Riccati equations \cite{BPW86}, various superequations
\cite{BecGagHusWin87,BecGagHusWin90}, or others
\cite{AndHarWin82,BecHusWin86,BecHusWin86b,GHW88,HavPosWin99}. Finally, it is
also worth mentioning Winternitz's research on Milne--Pinney equations
\cite{WintSecond}, which represents one of the first papers devoted to analysing
second-order differential equations through Lie systems.

Currently, many researches investigate the theory of Lie systems and other
closely related topics. Let us merely point out here some of them along with
some of their works: Bl\'azquez and Morales \cite{DB08,BM09II,BM09I}, Cari\~nena
\cite{CGL08,CGM00,CGM07}, Clemente \cite{CCR03}, Grabowski
\cite{CGM00,CGM07,CGM01}, Ibragimov \cite{NI1,Ib93,Ib95,Ib00}, de Lucas
\cite{CGL08,CGL09,CLR08}, L\'azaro-Cam\'i and Ortega \cite{JP09}, Marmo
\cite{CGM00,CGM07,CGM01}, Odzijewicz and Grundland \cite{OG00}, Ramos
\cite{CarRamGra,CarRam03,CarRamcinc}, Ra\~nada
\cite{CLuc08b,CLR08,CLR07a,CRL07e} and Nasarre \cite{CarMarNas,CarNas}. As
a result of their contributions, multiple interesting results about the
fundamentals, applications, and generalisations of the theory of Lie systems
were furnished. 

Among the above works, it is interesting to describe briefly the content of
\cite{CGL08,CGM00,CGM07}. The book \cite{CGM00} presents an instructive
geometric introduction to the basic topics of the theory of Lie systems. The
second one \cite{CGM07} provides multiple relevant contributions to the
comprehension of the theory of Lie systems. First, it fixes a remarkable gap in
the proof of Lie Theorem. Additionally, this work establishes that the
superposition rule concept amounts to a certain type of flat connection, what
substantially clarifies its properties. The furnished demonstration of Lie
Theorem shows that the Lie system notion can be naturally extended to the case
of PDEs. Finally, this work led, more or less indirectly, to the
characterisation of families of systems of first-order differential equations
admitting a $t$-dependent superposition rule \cite{CGL09} and the definition of
the mixed and partial superposition rule notions \cite{CGM07,CLR08}. Finally, we
can mention the usefulness of the {\it Lie scheme} concept provided in
\cite{CGL08}, which enables us to generalise the Lie system notion and leads to
the discovery of new properties for multiple systems of differential equations,
including non-Lie systems, appearing in Physics and Mathematics
\cite{CGL08,CLL09Emd,CL08Diss,CL09SRicc,CLRAbel}.

Let us now turn to discuss some of the authors' contributions that gave rise to
the redaction of this work. On one hand, Cari\~nena and his collaborators
investigated the integrability of Lie systems 
\cite{CarRamGra,CLuc08b,CRL07d,CLR07b,CLRan08,CarRam}, a generalisation of the
Wei--Norman method devoted to the study of Lie systems \cite{CarMarNas}, the
application of Lie systems techniques to analyse systems of second-order
differential equations \cite{CL09SRicc,CL09SecSup,CLR08,CLR07a}, and other
topics like the analysis of certain Schr\"odinger equations
\cite{CLInt09,CLR08WN,CarRam03}. In this way, they provided a continuation of
diverse previous articles dedicated to some of these themes
\cite{CC87,OG00,WintSecond,Ve95} and they opened several research lines
\cite{CarRam03}. 

Besides the above contributions, Cari\~nena and his collaborators also developed
numerous applications of Lie systems to Classical Physics
\cite{CGM01,CLuc08b,CL08b,CL08Diss,CLR08,CLRan08,CRL07e,CarNas,CarRamcinc},
Quantum Mechanics \cite{CLInt09,CLR08WN,CarRam03,CarRam05b}, Financial
Mathematics \cite{CAL}, and Control Theory \cite{CarRam05b,CarRam02}. 

Apart from the aforementioned generalisations of the Lie system notion that are
related to other works appearing in the literature
\cite{AFV09,OG00,WintSecond,Ve95}, a new approach to the generalisation of the
Lie system and superposition rule notions was carried out by Cari\~nena,
Grabowski and de Lucas: the theory of quasi-Lie schemes \cite{CGL08}. One one
hand, this approach provides us with a method to transform differential
equations of a certain type into equations of the same type, e.g. Abel equations
into Abel equations \cite{CLRAbel}. This can also be used to transform
differential equations into Lie systems \cite{CGL08}, what leads to the {\it
quasi-Lie system} notion. Such systems inherit some properties from Lie systems
and, for instance, they admit superposition rules showing an explicit dependence
on the independent variable of the system \cite{CGL08,CL09SRicc}. 

Quasi-Lie schemes admit multiple applications. they can be used not only to
analyse the properties of Lie and quasi-Lie systems but also to investigate many
other systems, e.g. nonlinear oscillators \cite{CGL08}, Emden-Fowler equations
\cite{CLL09Emd}, Mathews-Lakshmanan oscillators \cite{CGL08}, dissipative and
non-dissipative Milne--Pinney equations \cite{CL08Diss}, and Abel equations
\cite{CLRAbel} among others. As a consequence, various results about the
integrability properties of such equations have been obtained and many others
are being analysed at present. Furthermore, the appearance of $t$-dependent
superposition rules led to the examination of the so-called {\it Lie families},
which cover, as particular cases, Lie systems and quasi-Lie schemes.
Additionally, they can be used to analyse the exact solutions of very general
families of differential equations \cite{CGL09}.

As a result of all the above mentioned achievements, there exists today a vast
collection of methods and procedures to analyse Lie systems from different
points of view. All these tools can be used to provide interesting results in
Mathematics, Physics, Control Theory, and other fields. At the same time, these
applications motivate the development of new techniques, generalisations, and
applications of this theory, that presents multiple and interesting topics to be
further investigated.

\section{Fundamental notions about Lie systems and superposition
rules}\label{FNLS}
Our main purpose in this section is to review the basic notions and the
fundamental results concerning the theory of Lie systems to be employed and
analysed throughout our essay. Here, as well as in major part of our essay, we
mostly restrict ourselves to analysing differential equations on vector spaces
and we assume mathematical objects, e.g. flows of vector fields, to be smooth,
real, and globally defined. This will allow us to highlight the key points of
our exposition by omitting several irrelevant technical aspects that can be
detailed easily from our presentation. Despite this, numerous differential
equations over manifolds and diverse technical points will be presented when
relevant. 

\begin{definition}
Given the projections $\pi:(x,v)\in{\rm T}\mathbb{R}^n\mapsto x\in\mathbb{R}^n$
and $\pi_2:(t,x)\in\mathbb{R}\times\mathbb{R}^n\mapsto x\in\mathbb{R}^n$, a
$t$-dependent vector field $X$ on $\mathbb{R}^n$ is a map $X:(t,x)\in
\mathbb{R}\times\mathbb{R}^n\mapsto X(t,x)\in{\rm T}\mathbb{R}^n$ such that the
diagram
\vskip 0.1cm
\centerline{
\xymatrix{&{\rm T}\mathbb{R}^n\ar[d]^{\pi}\\
\mathbb{R}\times\mathbb{R}^n\ar[ur]^X\ar[r]^{\pi_2}&\mathbb{R}^n}}
\noindent is commutative, i.e. $\pi\circ X=\pi_2$.  
\end{definition}
In view of the above definition, it follows that $X(t,x)\in \pi^{-1}(x)={\rm
T}_x\mathbb{R}^n$ and hence $X_t:x\in \mathbb{R}^n\mapsto X_t(x)\equiv
X(t,x)\in{\rm T}\mathbb{R}^n$ is a vector field over $\mathbb{R}^n$ for every
$t\in\mathbb{R}$. From here, it is immediate that each $t$-dependent vector
field $X$ is equivalent to a family $\{X_t\}_{t\in\mathbb{R}}$ of vector fields
over $\mathbb{R}^n$. 

The $t$-dependent vector field concept includes, as a particular instance, the
standard vector field notion. Indeed, every vector field $Y$ over $\mathbb{R}^n$
can be naturally regarded as a $t$-dependent vector field $X$ of the form 
$X_t=Y$ for every $t\in\mathbb{R}$. Conversely, a `constant' $t$-dependent
vector field $X$ over $\mathbb{R}^n$, i.e. $X_t=X_{t'}$ for every
$t,t'\in\mathbb{R}$, can be considered as a vector field $Y=X_0$ over this
space.

As vector fields, $t$-dependent vector fields also admit local integral curves,
see \cite{Car96}. For each $t$-dependent vector field $X$ over $\mathbb{R}^n$,
this gives rise to defining its {\it generalised flow} $g^X$, i.e. the map
$g^X:\mathbb{R}\times\mathbb{R}^n\rightarrow\mathbb{R}^n$ such that
$g^X(t,x)\equiv g^X_t(x)=\gamma_x(t)$ with $\gamma_x(t)$ being the unique
integral curve of $X$ such that $\gamma_x(0)=x$.

\begin{definition}
A $t$-dependent vector field $X$ over $\mathbb{R}^n$ is said to be {\it
projectable} under a projection $p:\mathbb{R}^n\rightarrow \mathbb{R}^{n'}$ if
every $X_t$ is projectable, as a usual vector field, under such a map.
\end{definition}

The usage of $t$-dependent vector fields is fundamental in the theory of Lie
systems. They provide us with a geometrical object which contains all necessary
information to study systems of first-order differential equations. Let us start
by showing how systems of first-order differential equations are described by
means of $t$-dependent vector fields.

\begin{definition}
Given a $t$-dependent vector field 
\begin{equation}\label{AssocVec}
X(t,x)=\sum_{i=1}^nX^i(t,x)\frac{\partial}{\partial x^i},
\end{equation}
over $\mathbb{R}^n$, its {\it associated system} is the system of first-order
differential equations determining its integral curves, namely,
\begin{equation}\label{Assoc}
\frac{dx^i}{dt}=X^i(t,x),\qquad i=1,\ldots,n.
\end{equation}
\end{definition}

Note that there exists a one-to-one correspondence between $t$-dependent vector
fields and systems of first-order differential equations of the form
(\ref{Assoc}). That is, every $t$-dependent vector field has an associated
system of first-order differential equations and each system of this type, in
turn, determines the integral curves of a unique $t$-dependent vector field.
Taking this into account, we can hereby use $X$ to refer to both a $t$-dependent
vector field and the system of equations describing its integral curves. This
simplifies our exposition and it does not lead to confusion as the difference of
meaning is clearly noticed from the context.

The following definition and lemma, whose proof is straightforward and it shall
not be detailed, notably simplify the statements and proofs of various results
about the theory of Lie systems.

\begin{definition}\label{LieSpan} Given a (possibly infinite) family
$\mathcal{A}$ of vector fields on $\mathbb{R}^n$, we denote by ${\rm
Lie}(\mathcal{A})$ the smallest Lie algebra $V$ of vector fields on
$\mathbb{R}^n$ containing $\mathcal{A}$.
\end{definition}

\begin{lemma}\label{LieFam} Given a family of vector fields $\mathcal{A}$, the
linear 	space ${\rm Lie}(\mathcal{A})$ is spanned by the vector fields 
$$\mathcal{A},\,\,[\mathcal{A},\mathcal{A}],\,\,[\mathcal{A},[\mathcal{A},
\mathcal{A}]],\,\,[\mathcal{A},[\mathcal{A},[\mathcal{A},\mathcal{A}]]],
\ldots$$ 
where $[\mathcal{A},\mathcal{B}]$ denotes the set of vector fields
obtained through the Lie brackets between elements of the families of vector fields $\mathcal{A}$ and $\mathcal{B}$.
\end{lemma}

Throughout this work two different notions of linear independence are used
frequently. In order to state a clear meaning of each, we provide the following
definition.

\begin{definition}Let us denote by  $\mathfrak{X}(\mathbb{R}^n)$ the space of
vector fields over $\mathbb{R}^n$. We say that the vector fields,
$X_1,\ldots,X_r,$ on $\mathbb{R}^n$ are {\it linearly independent over
$\mathbb{R}$} if they are linearly independent as elements of
$\mathfrak{X}(\mathbb{R}^n)$ when considered as a $\mathbb{R}-$vector space,
i.e. whenever 
$$
\sum_{\alpha=1}^r\lambda_\alpha X_\alpha=0
$$
for certain constants, $\lambda_1,\ldots,\lambda_r$, then
$\lambda_1=\ldots=\lambda_r=0$. On the other hand, the vector fields,
$X_1,\ldots,X_r,$ are said to be {\it linearly independent at a generic point}
if they are linearly independent as elements of $\mathfrak{X}(\mathbb{R}^n)$
when regarded as a $C^{\infty}(\mathbb{R}^n)-$module.  That is, if one has 
$$
\sum_{\alpha=1}^rf_\alpha X_\alpha=0
$$
over any open set of $\mathbb{R}^n$ for certain functions $f_1,\ldots,f_r\in
C^{\infty}(\mathbb{R}^n)$, then $f_1=\ldots=f_r=0$.
\end{definition}

In this essay, we frequently deal with linear spaces of the form
$\mathbb{R}^{n(m+1)}$. Such spaces are always considered as a product
$\mathbb{R}^n\times\stackrel{m+1-{\rm times}}{\ldots}\times\mathbb{R}^n$. Each
point of $\mathbb{R}^{n(m+1)}$ is denoted by $(x_{(0)},\ldots,x_{(m)})$, where
$x_{(j)}$ stands for a point of the $j$-th copy of the manifold $\mathbb{R}^n$
within  $\mathbb{R}^{n(m+1)}$.

Associated with $\mathbb{R}^{n(m+1)}$, there exists a group of permutations
$S_{m+1}$ whose elements, $S_{ij}$, with $i\leq j=0,1,\ldots,m,$ act on
$\mathbb{R}^{n(m+1)}$ by permutating the variables $x_{(i)}$ and $x_{(j)}$.
Finally, let us define the projections
\begin{equation}\label{proj}
{\rm
pr}:(x_{(0)},\ldots,x_{(m)})\in\mathbb{R}^{n(m+1)}\mapsto(x_{(1)},\ldots,x_{(m)}
)\in\mathbb{R}^{nm}
\end{equation}
and
\begin{equation}\label{proj2}
{\rm pr}_0:(x_{(0)},\ldots,x_{(m)})\in\mathbb{R}^{n(m+1)}\mapsto
x_{(0)}\in\mathbb{R}^{n},
\end{equation}
to be employed in various parts of our work.

Once the fundamental definitions and assumptions to be used hereafter have been
established, we proceed to introduce the notion of {\it superposition rule},
which plays a central role in the study of Lie systems.

For each system of first-order ordinary homogeneous linear differential
equations on $\mathbb{R}^n$ of the form
\begin{equation}\label{lin}
\frac{dy^i}{dt}=\sum_{j=1}^nA^i_{j}(t)y^j,\qquad i=1,\ldots,n,
\end{equation}
where $A^i_j(t)$, with $i,j=1,\ldots,n$, is a family of $t$-dependent functions,
its general solution, $y(t)$, can be written as a linear combination of the form
\begin{equation}\label{Super_Lineal}
y(t)=\sum_{j=1}^n k_jy_{(j)}(t),
\end{equation}
with, $y_{(1)}(t),\ldots,y_{(n)}(t),$ being a family of $n$ generic (linearly
independent) particular solutions, and, $k_1,\ldots,k_n,$ being a set of
constants. The above expression is called {\it linear superposition rule} for
system (\ref{lin}). 

Linear superposition rules allow us to reduce the search for the general
solution of a linear system to the determination of a finite set of particular
solutions. This property is not exclusive for homogeneous linear systems.
Indeed, for each linear system
\begin{equation}\label{afin}
\frac{dy^i}{dt}=\sum_{j=1}^nA^i_{j}(t)y^j+B^i(t),\qquad i=1,\ldots,n,
\end{equation}
where $A^i_j(t), B^i(t)$, with $i,j=1,\ldots,n$, are a family of $t$-dependent
functions, its general solution, $y(t)$, can be written as a linear combination
of the form
\begin{equation}\label{Super_Afin}
y(t)=\sum_{j=1}^{n}k_j(y_{(j)}(t)-y_{(0)}(t))+y_{(0)}(t),
\end{equation}
with, $y_{(0)}(t),\ldots,y_{(n)}(t),$ being a family of $n+1$ particular
solutions such that $y_{(j)}(t)-y_{(0)}(t)$, with $j=1,\ldots,n$, are linearly
independent solutions of the homogeneous problem associated with (\ref{afin}),
and, $k_1,\ldots,k_n,$ being a set of constants. 

In a more general way, system (\ref{lin}) becomes (generally) a nonlinear
system 
\begin{equation}\label{nlin}
\frac{dx^i}{dt}=X^i(t,x),\qquad i=1,\ldots,n,
\end{equation}
through a diffeomorphism $\varphi:\mathbb{R}^n\ni y\mapsto
x=\varphi(y)\in\mathbb{R}^n$. In view of the linear superposition rule
(\ref{Super_Lineal}), the above system admits its general solution, $x(t)$, to
be described in terms of a family of certain particular solutions,
$x_{(1)}(t),\ldots,x_{(m)}(t)$, as
$$
x(t)=\varphi\left(\sum_{j=1}^nk_j\varphi^{-1}(x_{(j)}(t))\right).
$$
This clearly shows that there exist many systems of first-order differential
equations whose general solutions can be described, nonlinearly, in terms of
certain families of particular solutions and sets of constants. A relevant
family  of different equations admitting such a property are Riccati equations
\cite{AS64,CarRamdos,Na05,HarWinAnd83,Na00,Ra62,SW85} of the form 
\begin{equation}\label{riccequation}
\frac{dx}{dt}=b_1(t)+b_2(t)x+b_3(t)x^2,
\end{equation}
with $x\in\bar{\mathbb{R}}\equiv \mathbb{R}\cup\{\infty\}$. More specifically,
for each of such Riccati equations, its general solution, $x(t),$ can be cast
into the form 
\begin{equation}\label{SupRiccat}
x(t)=\frac{x_1(t)(x_3(t)-x_2(t))-kx_2(t)(x_3(t)-x_1(t))}{
(x_3(t)-x_2(t))-k(x_3(t)-x_1(t))},
\end{equation}
where, $x_1(t)$, $x_2(t)$, $x_3(t),$ are three particular solutions of the
equation and $k\in\bar{\mathbb{R}}$. 

It is worth noting that, given a fixed family of three different particular
solutions with initial conditions within $\mathbb{R}$, if we only choose $k$ in
$\mathbb{R}$, the above expression does not recover the whole general solution
of the Riccati equation, as $x_2(t)$ cannot be recovered.

The above examples show the existence of a certain type of expression, the
so-called {\it global superposition rule}, which enables us to express the
general solution of certain systems of first-order ordinary differential
equations in terms of certain families of particular solutions and a set of
constants. Let us state a rigorous definition of this notion for systems of
differential equations in $\mathbb{R}^n$.

\begin{definition}\label{SRI} The system of first-order ordinary differential
equations 
\begin{equation}\label{LieSystem}
\frac{dx^i}{dt}=X^i(t,x),\qquad i=1,\ldots,n,
\end{equation}
is said to admit a {\it global superposition rule} if there exists a
$t$-independent map
$\Phi:(\mathbb{R}^{n})^{m}\times\mathbb{R}^n\rightarrow\mathbb{R}^n$ of the form
\begin{equation}\label{FSup}
x=\Phi(x_{(1)},\ldots,x_{(m)};k_1,\ldots,k_n),
\end{equation}
such that its general solution, $x(t)$, can be written as
\begin{equation}\label{FirstSup}
x(t)=\Phi(x_{(1)}(t),\ldots,x_{(m)}(t);k_1,\ldots,k_n),
\end{equation}
with, $x_{(1)}(t),\ldots,x_{(m)}(t),$ being any generic  family of particular
solutions of system (\ref{LieSystem}) and, $k_1,\ldots,k_n,$ being a set of $n$
constants to be related to initial conditions. 
\end{definition}

In order to grasp the meaning of the above definition, it is necessary to
understand the sense in which the term `generic' is used in the above statement.
Precisely speaking, it is said that expression (\ref{FirstSup}) is valid for any
generic family of $m$ particular solutions if there exists an open dense subset
$U\subset (\mathbb{R}^n)^m$ such that expression (\ref{FirstSup}) is satisfied
for every set of particular solutions $x_{1}(t),\ldots,x_{m}(t),$ such that
$(x_{1}(0),\ldots,x_{m}(0))$ lies in $U$. 

Let us now show that the aforementioned examples admit a global superposition
rule. Consider the function
$\Phi:(\mathbb{R}^n)^n\times\mathbb{R}^n\rightarrow\mathbb{R}^n$ of the form
\begin{equation}\label{linsup}
\Phi(x_{(1)},\ldots,x_{(n)};k_1,\ldots,k_n)=\sum_{j=1}^nk_jx_{(j)}.
\end{equation}
This mapping is a superposition rule for system (\ref{lin}). Indeed, note that
for each set of particular solutions $x_{(1)}(t),\ldots,x_{(m)}(t),$ of
(\ref{lin}) such that the point $(x_{(1)}(0),\ldots,x_{(m)}(0))$ belongs to the
open dense subset 
$$
U=\left\{(x_{(1)},\ldots,x_{(n)})\in (\mathbb{R}^n)^n\,\bigg|\,
{\rm det}\left(\begin{array}{ccc}
x^1_{(1)}&\ldots&x^1_{(n)}\\
\ldots&\ldots&\ldots\\
x^n_{(1)}&\ldots&x^n_{(n)}
\end{array}\right)\neq 0
\right\},
$$
of $(\mathbb{R}^n)^n$, the general solution $x(t)$ of (\ref{lin}) can be written
in the form (\ref{Super_Lineal}). Likewise, a superposition rule can be now
proved to exist for the systems (\ref{nlin}) obtained from (\ref{lin}) by means
of a diffeomorphism.

The function
$\Phi:(\mathbb{R}^n)^{n+1}\times\mathbb{R}^n\rightarrow\mathbb{R}^n$ of the form
\begin{equation}\label{SuperAfin}
\Phi(x_{(0)},\ldots,x_{(n)};k_1,\ldots,k_n)=\sum_{j=1}^nk_j(x_{(j)}-x_{(0)})+x_{
(0)},
\end{equation}
is a superposition function for the system (\ref{afin}). In fact, note that for
each set of particular solutions, $x_{(0)}(t),\ldots,x_{(n)}(t),$ of
(\ref{afin}) such that the point $(x_{(0)}(0),\ldots,x_{(n)}(0))$ belongs to the
open dense subset 
$$
U=\left\{(x_{(0)},\ldots,x_{(n)})\in (\mathbb{R}^n)^{n+1}\,\bigg|\,
{\rm det}\left(\begin{array}{ccc}
x^1_{(1)}-x^1_{(0)}&\ldots&x^1_{(n)}-x^1_{(0)}\\
\ldots&\ldots&\ldots\\
x^n_{(1)}-x^n_{(0)}&\ldots&x^n_{(n)}-x^n_{(0)}
\end{array}\right)\neq 0
\right\},
$$
of $(\mathbb{R}^n)^{n+1}$, the general solution $x(t)$ of (\ref{afin}) can be
put in the form (\ref{Super_Afin}). 

Finally, let us analyse the case of Riccati equations in $\bar{\mathbb{R}}$.
This example differs a little from previous ones, as it concerns a differential
equation defined in the manifold $\bar{\mathbb{R}}\simeq S^1$. Nevertheless, the
generalisation of Definition \ref{SRI} to manifolds is obvious. It is only
necessary to replace $\mathbb{R}^n$ by a manifold $N$. In view of this, the map
$\Phi:\bar{\mathbb{R}}^3\times\bar{\mathbb{R}}\rightarrow\bar{\mathbb{R}}$ of
the form
\begin{equation}\label{GlobRicc}
\Phi(x_{(1)},x_{(2)},x_{(3)};k)=\frac{x_{(1)}(x_{(3)}-x_{(2)})-kx_{(2)}(x_{(3)}
-x_{(1)})}{(x_{(3)}-x_{(2)})-k(x_{(3)}-x_{(1)})}
\end{equation}
is a global superposition rule for Riccati equations in $\bar{\mathbb{R}}$. To
verify this, it is sufficient to note that given one of these equations with
three particular solutions, $x_{(1)}(t), x_{(2)}(t), x_{(3)}(t),$ such that
$(x_{(1)}(0), x_{(2)}(0), x_{(3)}(0))\in U$, where
$$
U=\left\{(x_{(1)},x_{(2)},x_{(3)})\in \mathbb{R}^{3}\,\bigg|\,
x_{(1)}\neq x_{(2)}, x_{(1)}\neq x_{(3)}\,\, {\rm and}\,\,x_{(2)}\neq
x_{(3)}\right\},
$$
its general solution can be cast into the form (\ref{SupRiccat}). 

The aforementioned superposition rules illustrate that for each permutation of
their arguments, $x_{(1)},\ldots,x_{(m)}$, e.g. an interchange of the arguments
$x_{(i)}$ and $x_{(j)}$, one has, in general, that
$$
\Phi(x_{(1)},\ldots,x_{(i)},\ldots,x_{(j)},\ldots,x_{(m)};k)\neq
\Phi(x_{(1)},\ldots,x_{(j)},\ldots,x_{(i)},\ldots,x_{(m)};k).
$$
Nevertheless, it can be proved (cf. \cite{CGM07}) that there exists a map
$\varphi:k\in\mathbb{R}^n\rightarrow \varphi(k)\in\mathbb{R}^n$ such that
$$
\Phi(x_{(1)},\ldots,x_{(i)},\ldots,x_{(j)},\ldots,x_{(m)};k)=\Phi(x_{(1)},\ldots
,x_{(j)},\ldots,x_{(i)},\ldots,x_{(m)};\varphi(k)).
$$

It is interesting to note that, if we consider Riccati equations to be defined
on the real line, a global superposition rule for such equations would be a map
of the form $\Phi:\mathbb{R}^m\times\mathbb{R}\rightarrow\mathbb{R}$. Obviously,
expression (\ref{GlobRicc}) does not give rise to a global superposition of this
form. Indeed, if we restrict (\ref{GlobRicc}) to $\mathbb{R}^3\times\mathbb{R}$,
we will not be able to recover $x_{(2)}(t)$ from a set of different particular
solutions, $x_{(1)}(t),x_{(2)}(t),x_{(3)}(t),$ for any $k\in\mathbb{R}$. Even
more, the function (\ref{GlobRicc}) is not globally defined over
$\mathbb{R}^3\times\mathbb{R}$. Nevertheless, such a function is what in the
literature is known as a {\it superposition rule} for Riccati equations over the
real line \cite{Gu93,LS,Ve93}.

In the literature, the superposition rule notion appears as a `milder' version
of aforementioned global superposition rule concept. In other words,
superposition rules admit almost the same properties as global superposition
rules but, for instance, they may fail to recover certain particular solutions.
Although it is enough to bear in mind the above example for Riccati equations to
understand fully the main difference between both notions, the precise
definition of a local superposition rule is very technical (see \cite{BM09II})
and it does not provide, in practice, a much deeper knowledge about Lie systems.
That is why, as everywhere else in the literature
\cite{CGM00,Gu93,Ib00,Ib08,LSPartial,LS,Ve93,Ve93II,PW}, we will assume
hereafter superposition rules to recover general solutions and to be globally
defined. This simplifies considerably our theoretical presentation and it
highlights the main features of superposition rules and Lie systems. Despite
these assumptions, a fully rigorous treatment can be easily carried out and some
technical remarks will be discussed when relevant. 

A relevant question now arises: which systems of first-order ordinary
differential equations admit a superposition rule? Several works have been
devoted to investigating this question. Its analysis was accomplished by
K\"onigsberger \cite{Ko83}, Vessiot \cite{Ve93}, and Guldberg \cite{Gu93}. They
proved that every system of first-order differential equations defined over the
real line admitting a superposition rule is, up to a diffeomorphism, a Riccati
equation or a first-order linear differential equation.

Apart from these preliminary results, it was Lie \cite{LSPartial,LSII,LS} who
established the conditions ensuring that a system of first-order differential
equations of the form (\ref{LieSystem}) admits a superposition rule. His result,
the today named {\it Lie Theorem}, reads in modern geometric terms as follows.

\begin{theorem}{\bf (Lie Theorem)} A system of first-order ordinary differential
equations (\ref{LieSystem}) admits a superposition rule (\ref{FSup}) if and only
if its corresponding $t$-dependent vector field (\ref{AssocVec}) can be cast
into the form
\begin{equation}\label{Decompo}
X(t,x)=\sum_{\alpha=1}^r b_\alpha(t)\, X_\alpha(x),
\end{equation}
with, $X_1,\ldots, X_r,$ being a family of vector fields over $\mathbb{R}^n$
spanning a $r$-dimensional real Lie algebra of vector fields $V$.
\end{theorem}

Within the proof to his theorem \cite[Theorem 44]{LS}, Lie also claimed that the
dimension of the decomposition (\ref{Decompo}) and the number $m$ of particular
solutions for the superposition rule are related. More specifically, he proved
that the existence of a superposition rule depending on $m$ particular solutions
for a system (\ref{LieSystem}) in $\mathbb{R}^n$ implies that there exists a
decomposition (\ref{Decompo}) associated with a Lie algebra $V$ obeying the
inequality $\dim\,V\leq m\cdot n$, the referred to as {\it Lie's condition}.
Conversely, given a decomposition of the form (\ref{Decompo}), we can ensure the
existence of a superposition rule for system (\ref{LieSystem}) whose number of
particular solutions obeys the same condition.

Although Lie Theorem solves theoretically the problem of determining whether a
system (\ref{LieSystem}) admits a superposition rule, it does not provide a
solution for many other questions concerning the study of superposition rules.
Let us briefly comment on some of these queries.

\begin{itemize}
\item From a practical point of view, it is not straightforward, solely in view
of Lie Theorem, to prove that a system of first-order differential equations
does not admit a superposition rule. Later on in this section, we will sketch a
procedure to do so. 
\item Lie Theorem says nothing about the possible existence of multiple
superposition rules for the same system. What is more, it does not explain
explicitly how to determine any of such superposition rules (although its proof
\cite[Theorem 4]{LS} furnishes some key hints). These questions are addressed
later in this Chapter, where we review a recent geometrical approach to Lie
systems developed in \cite{CGM07}.
\item A system $X(t,x)$ admitting a superposition rule may be written in the
form (\ref{Decompo}) in one or, sometimes, several different ways. Each one of
these decompositions is related to a different finite-dimensional Lie algebra of
vector fields $V$. Such Lie algebras are generally called the {\it
Vessiot--Guldberg Lie algebras} associated with a system. Lie Theorem does not
explain the possible relations amongst all possible Vessiot--Guldberg Lie
algebras of a system (\ref{LieSystem}). In fact, only Lie's condition suggests
that each different Vessiot--Guldberg Lie algebra may be related to different
superposition rules. We will discuss these questions, in a more extensive way,
later in this section and next.
\item Finally, it is worth noting that Lie Theorem cannot be used to
characterise straightforwardly systems of first-order differential equations of
the form $F^i(t,x,\dot x)=0$, with $i=1,\ldots,n$. Indeed, this is an open
question of the research on Lie systems.
\end{itemize}

The discovery of Lie Theorem \cite{LS} in 1893 established definitively the Lie
system notion, which, on the other hand, had already been suggested long time
ago by Lie \cite{LSPartial}, and whose name was coined by Vessiot in \cite{Ve94}
as a recognition to Lie's success in characterising systems admitting a
superposition rule. The definition of this relevant notion goes as follows.

\begin{definition} A system of the form (\ref{LieSystem}) is a {\it Lie system}
if and only if its corresponding $t$-dependent vector field, namely
(\ref{AssocVec}), admits a decomposition of the form (\ref{Decompo}).
\end{definition}

In view of Lie Theorem, the above definition of Lie system can be rephrased by
saying that a system (\ref{LieSystem}) is a Lie system if and only if it admits
a superposition rule. From here, it is obvious that the systems of first-order
differential equations (\ref{lin}), (\ref{afin}) and (\ref{riccequation}), which
admit the global superposition rules (\ref{linsup}), (\ref{SuperAfin}) and
(\ref{GlobRicc}), respectively, are Lie systems. Let us analyse in detail such
examples. This brings us the opportunity to illustrate diverse characteristics
of Lie systems and the Lie Theorem here and in forthcoming sections.

Consider again the homogeneous linear system (\ref{lin}). This system describes
the integral curves of the $t$-dependent vector field
\begin{equation}\label{lab1}
 X(t,x)= \sum_{i,j=1}^nA^i\,_j(t)\, x^j\,\pd{}{x^i}\ ,
\end{equation}
which is a linear combination of vector fields of the form
 \begin{equation}\label{DecomLin}
X(t,x)= \sum_{i,j=1}^nA^i\,_j(t)\,X_{ij}(x),
 \end{equation}
of the $n^2$ vector fields
 \begin{equation}
 X_{ij}=x^j\,\pd{}{x^i}\ , \qquad i,j=1,\ldots,n\,\label{xij}.
\end{equation}
Furthermore, one has that 
\begin{equation*}
[X_{ij},X_{lm}]=\delta_m^{i}\,X_{lj}-\delta_j^{l}\,X_{im},
\end{equation*}
where $\delta^i_m$ is the Kronecker delta function, i.e. the vector fields
(\ref{xij}) close on a $n^2$-dimensional Vessiot--Guldberg Lie algebra
isomorphic to the Lie algebra $\mathfrak{gl}(n,{\mathbb{R}})$, see
\cite{CarRamcinc}. 

In view of decomposition (\ref{DecomLin}), each system (\ref{lin}) is a Lie
system. This is not a surprise, as each system (\ref{lin}) admits the
superposition rule (\ref{linsup}) and Lie Theorem states that every system
admitting a superposition rule must be a Lie system. Moreover, in view of Lie's
condition, since homogeneous linear systems in $\mathbb{R}^n$ admit a
superposition rule depending on $n$ particular solutions, their associated
$t$-dependent vector fields must take values in {\it some} Lie algebra of
dimension lower or equal to $n^2$. Indeed, note that decomposition
(\ref{DecomLin}) shows that $X(t,x)$ takes values in a Lie algebra isomorphic to
$\mathfrak{gl}(n,\mathbb{R})$, what clearly obeys the Lie's condition
corresponding to the superposition rule (\ref{linsup}). 

Note that we have italicised the last `some' in the paragraph above. We did it
because we wanted to stress that a Lie system can take values in different Lie
algebras, some of which do not need to satisfy the same Lie's condition. This
will become more clear with the next example.

Let us now turn to analyse an inhomogeneous system of the form (\ref{afin}).
This system describes the integral curves of the $t$-dependent vector field
 \begin{equation}\label{AfinDecom}
X(t,x)= \sum_{i=1}^n\left( \sum_{j=1}^nA^i\ _j(t)\, x^j+B^i(t)\right)\pd{}{x^i}\
,
\end{equation}
which is a linear combination with $t$-dependent coefficients,
 \begin{equation}\label{DecomAfin}
X_t= \sum_{i,j=1}^nA^i\,_j(t)\,X_{ij}+\sum_{i=1}^n B^i(t)\, X_i\ ,
\end{equation}
of the vector fields (\ref{xij}) and 
 \begin{equation}\label{AfinVec}
X_i=\pd{}{x^i}\ , \qquad i=1,\ldots,n\,.
\end{equation}
The above vector fields satisfy the commutation relations
$$
[X_i,X_j]=0\ ,\qquad  i,j=1,\ldots,n\,, \qquad [X_{ij},X_l]=-\delta^{lj}\, X_i\
,\qquad  i,j,l=1,\ldots,n\,.
$$
This shows that the vector fields (\ref{xij}) and (\ref{AfinVec}) span a Lie
algebra of vector fields isomorphic to the $(n^2+n)$-dimensional Lie algebra of
the affine group \cite{CarRamcinc}. Then, in view of decomposition
(\ref{DecomAfin}), systems (\ref{afin}) are Lie systems.

As systems (\ref{afin}) admit a superposition rule (\ref{SuperAfin}) depending
on $n+1$ particular solutions, Lie's condition implies that their $t$-dependent
vector fields must take values in some Lie algebra of dimension lower or equal
to $n(n+1)$.  In fact, the above results easily show that this is the case. 

The previous example allows us to exemplify that a Lie system may admit multiple
Vessiot-Guldberg Lie algebras. Recall that every homogeneous linear system
(\ref{lin}) is related to a $t$-dependent vector field taking values in a Lie
algebra isomorphic to $\mathfrak{gl}(n,\mathbb{R})$. Additionally, as a
particular instance of system (\ref{afin}), its $t$-dependent vector field also
takes in the above defined $n^2+n$-dimensional Lie algebra of vector fields. In
other words, linear systems admit, at least, two non-isomorphic
Vessiot--Guldberg Lie algebras. 

Now, we can illustrate how different superposition rules for the same system may
be associated with multiple, non-isomorphic, Vessiot--Guldberg Lie algebras and
lead to distinct Lie's conditions. We showed that linear systems admit a linear
superposition rule which leads, in view of Lie's condition, to the existence of
an associated Vessiot--Guldberg Lie algebra of dimension lower or equal to
$n^2$, which was determined. Nevertheless, the abovementioned second
Vessiot--Guldberg Lie algebra for linear systems does not hold this condition.
On the contrary, this second Vessiot-Guldberg Lie algebra shows that there must
exist a second superposition rule, namely (\ref{Super_Afin}), which, along with
this Vessiot--Guldberg Lie algebra, satisfies a new Lie's condition. 

To sum up, Lie Theorem implies that a system admitting a superposition rule is
related to the existence of, at least, one Vessiot--Guldberg Lie algebra
satisfying the Lie's condition relative to this superposition. Nevertheless, the
system can possess more Vessiot--Guldberg Lie algebras, some of which do not
need to obey the Lie's condition for the assumed superposition rule. In that
case, the other Vessiot--Guldberg Lie algebras are related to other
superposition rules for which, a new Lie's condition is satisfied.

In order to detail the last of the most usual examples of Lie systems admitting
a superposition rule, we now consider Riccati equations (\ref{riccequation}).
These differential equations determine the integral curves of the $t$-dependent
vector field on $\bar{\mathbb{R}}$ of the form
\begin{equation}
X(t,x)=(b_1(t)+b_2(t)x+b_3(t)x^2)\frac{\partial}{\partial x} \,.\label{vfRic2}
\end{equation}
As Riccati equations admit a global superposition rule, they must satisfy the
assumptions detailed in Lie Theorem. Indeed, note that $X$ is a linear
combination with
$t$-dependent coefficients of the three  vector fields
\begin{equation}\label{ric}
X_1 =\frac{\partial}{\partial x} ,  \quad X_2 =x \frac{\partial}{\partial x}  ,
 \quad X_3 = x^2
\frac{\partial}{\partial x},
\end{equation}
which close on a three-dimensional Lie  algebra with defining relations
\begin{equation}\label{conmutL5}
[X_1,X_2] = X_1 ,      \quad [X_1,X_3] = 2X_2 ,
  \quad [X_2,X_3] = X_3.
\end{equation}
Thus, as it was expected, Riccati equations obey the conditions given by Lie to
admit a superposition rule. Moreover, Riccati equations are associated with a
Vessiot--Guldberg Lie algebra isomorphic to $\mathfrak{sl}(2,\mathbb{R})$. Since
this Lie algebra is three dimensional and Riccati equations admit a
superposition rule depending on three particular solutions, it is immediate that
the equations (\ref{riccequation}) satisfy the corresponding Lie's condition.

The existence of different Vessiot--Guldberg Lie algebras for a system of
first-order ordinary differential equations is an important question because
their characteristics determine, among other features, the integrability by
quadratures of Lie systems \cite{CAL}.

Let us now turn our attention to determine when a system (\ref{LieSystem}) is
{\it not} a Lie system. In order to analyse this question, it becomes useful to
rewrite Lie Theorem in the following, abbreviated, form.

\begin{proposition}{\bf (Abbreviated Lie Theorem)}\label{AbbLieTheorem} A system
$X$ on $\mathbb{R}^n$ is a Lie system if and only if ${\rm
Lie}(\{X_t\}_{t\in\mathbb{R}})$ is finite-dimensional.  
\end{proposition}

In view of the above result, determining that (\ref{LieSystem}) is not a Lie
system reduces to showing that ${\rm Lie}(\{X_t\}_{t\in\mathbb{R}})$ is
infinite-dimensional. The standard procedure to prove this consists in
demonstrating that there exists an infinite chain, $\{Z_j\}_{j\in\mathbb{N}}$ of
linearly independent vector fields over $\mathbb{R}$ obtained through successive
Lie brackets of elements in $\{X_t\}_{t\in\mathbb{R}}$. In order to illustrate
how this is usually made, consider the particular example based on the study of
the Abel equation of the first-type
$$
\frac{dx}{dt}=x^2+b(t)x^3, \qquad b(t)\neq 0,
$$
where $b(t)$ is additionally a non-constant function. These equations describe
the integral curves of the $t$-dependent vector field
$$
X_t=(x^2+b(t)x^3)\frac{\partial}{\partial x}.
$$
Consider the chain of vector fields
$$
Z_1=x^2\frac{\partial}{\partial x},\qquad Z_2=x^3\frac{\partial}{\partial
x},\qquad Z_j=[X_1,X_{j-1}],\qquad j=3,4,5,\ldots
$$
Since $Z_j=x^{j+1}\partial/\partial x$, it turns out that ${\rm
Lie}(\{X_t\}_{t\in\mathbb{R}})$ admits the infinite chain of linearly
independent vector fields $\{Z_j\}_{j\in\mathbb{R}}$ and, in consequence, in
view of the abbreviated Lie Theorem, Abel equations of the above type are not
Lie systems.

There are many other relevant Lie systems associated with important systems of
differential equations appearing in the physical and mathematical literature.
For instance, a non exhaustive brief list of these Lie systems includes 

\begin{enumerate}
\item Linear first-order systems and, more specifically, Euler-systems
\cite{CarRamcinc,FLV10}.
\item Riccati equations \cite{CRL07d,Ve93,PW} and coupled Riccati equations of
projective type \cite{And80}.
\item Matrix Riccati equations \cite{HarWinAnd83,LW96,OlmRodWin87,RW84,SW85,PW}.
\item Bernoulli equations, several equations appearing in supermechanics
\cite{BecGagHusWin90}, etc.
\end{enumerate} 

Apart from the above instances, there are other important systems of
differential equations which can be studied through other Lie systems. Several
of such Lie systems will be detailed throughout next sections.

The determination of the general solution of any Lie system reduces to deriving
a particular solution of a particular type of Lie system defined in a Lie group.
Let us analyse in detail this claim. 

Consider a Lie system related to a $t$-dependent vector field (\ref{Decompo})
over $\mathbb{R}^n$ and associated, for simplicity, with a Vessiot--Guldberg Lie
algebra $V$ made up of complete vector fields. This gives rise to a Lie group
action $\Phi:G\times \mathbb{R}^n\rightarrow\mathbb{R}^n$ whose fundamental
vector fields are exactly those of $V$. Obviously, this implies that the Lie
algebra $\mathfrak{g}\simeq {\rm T}_eG$ is isomorphic to $V$. Choose now a basis
$\{{\rm a}_1,\ldots,{\rm a}_r\}$ of $\mathfrak{g}$ such that $\Phi:G\times
\mathbb{R}^n\rightarrow \mathbb{R}^n$ and 
\begin{equation}\label{convention}
\Phi(\exp(-s{\rm a}_\alpha),x)=g^{(\alpha)}_s(x),\qquad\alpha=1,\ldots,r,\qquad
s\in\mathbb{R},
\end{equation}
where $g^{(\alpha)}:(s,x)\in\mathbb{R}\times\mathbb{R}^n\mapsto
g^{(\alpha)}(s,x)=g^{(\alpha)}_s(x)\in\mathbb{R}^n$ is the flow of the vector
field $X_\alpha$. In this way, each vector field $X_\alpha$ becomes the
fundamental vector
field corresponding to ${\rm a}_\alpha$ and the map
$\phi:\mathfrak{g}\rightarrow V$ such that $\phi({\rm a}_\alpha)=X_\alpha$ for
$\alpha=1,\ldots,r,$ is a Lie algebra isomorphism.

Let $X^{\rm R}_\alpha$ be the right-invariant vector field on $G$ with $(X^{\rm
R}_\alpha)_e={\rm a}_\alpha$, i.e. $(X^{\tt R}_\alpha)_g=R_{g*e}{\rm a}_\alpha$,
where $R_g:g'\in G\mapsto g'g\in G$ is the right action of $G$ on itself. Then,
the
 $t$-dependent right-invariant vector field
\begin{equation}\label{eqLG1}
X^G(t,g)=-\sum_{\alpha=1}^r b_\alpha(t)X^{\rm R}_\alpha(g),
\end{equation}
defines a Lie system on $G$ whose integral curves are the solutions of the
system on $G$ given by
\begin{equation}\label{PseudoLine}
\frac{dg}{dt}=-\sum_{\alpha=1}^rb_\alpha(t)\,X^{\rm R}_\alpha(g).
\end{equation}
Applying $R_{g^{-1}*g}$ to both sides of the equation, we see that its general
solution, $g(t)$, satisfies that 
\begin{equation}\label{eqLG2}
R_{g^{-1}(t)*g(t)}\dot g(t)\,=-\sum_{\alpha=1}^r b_\alpha(t){\rm a}_\alpha\in
{\rm T}_eG \ .
\end{equation}
Note that right-invariance implies that the knowledge of one particular solution
of the above equation, e.g. the particular one $g_0(t)$, with $g_0(0)=g_0$, is
enough to obtain the general solution of the equation (\ref{eqLG2}). Indeed,
consider $g'(t)=R_{\bar g}g_0(t)$ for a given $\bar g\in G$. Such a curve obeys
that
$$
\frac{dg'}{dt}(t)=R_{\bar
g*g_0(t)}\left(\frac{dg_0}{dt}(t)\right)\Longleftrightarrow
\frac{dg'}{dt}(t)=R_{\bar g*g_0(t)}\left(-\sum_{\alpha=1}^r b_\alpha(t)X^{\rm
R}_\alpha(g_0(t))\right).
$$
Taking into account that $R_{\bar g*g_0}X^{\rm R}_\alpha(g_0)=X^{\rm
R}_\alpha(g_0\bar g)$, one has that
$$
\frac{dg'}{dt}(t)=-\sum_{\alpha=1}^r b_\alpha(t)X^{\rm R}_\alpha(R_{\bar
g}g_0(t))=-\sum_{\alpha=1}^r b_\alpha(t)X^{\rm R}_\alpha(g'(t))
$$
and $g'(t)$ is another particular solution of (\ref{eqLG1}) with initial
condition $g'(0)=R_{\bar g} g_0$. In consequence, the general solution $g(t)$
for equation (\ref{eqLG2}) can be written as
$$g(t)=R_{\bar g}g_0(t), \quad\qquad\qquad\quad \bar g\in G.$$
That is, system (\ref{eqLG1}) admits a superposition rule and, according to Lie
Theorem, it must be a Lie system. This is not surprising, as the vector fields
$X_\alpha^{\rm R}$ span a Lie algebra of vector fields isomorphic to $V$ and, in
consequence, system (\ref{PseudoLine}) describes the integral curves of a
$t$-dependent vector field taking values in a finite-dimensional Lie algebra of
vector fields.

The relevance of the Lie system (\ref{eqLG2}) relies on the fact that the
integral curves of the $t$-dependent vector field $X(t,x)$
can be obtained from one particular solution of equation (\ref{eqLG2}). More
explicitly, the general solution $x(t)$ of the Lie system $X(t,x)$ reads
$x(t)=\Phi(g_e(t),x_0)$, where $x_0$ is the initial condition of the particular
solution and $g_e(t)$ is the particular solution of  equation (\ref{eqLG2}) with
$g_e(0)=e$. 

Note that, in view of Ado's Theorem \cite{Ad62}, every finite-dimensional Lie
algebra, e.g. the above Vessiot--Guldberg Lie algebra $V$, admits an isomorphic
matrix Lie algebra. Related to this matrix Lie algebra, there exists a matrix
Lie group $\bar G$. In this way, the system describing the $t$-dependent vector
field (\ref{Decompo}) reduces to solving an equation of the form
$$
\dot A(t)A^{-1}(t)=-\sum_{\alpha=1}^rb_\alpha(t)M_\alpha\Longrightarrow \dot
A=-\sum_{\alpha=1}^rb_\alpha(t)M_\alpha A,
$$
with $A(t)$ being a curve taking values in the matrix Lie group $\bar G$ and,
$M_1,\ldots,M_r,$ being a basis closing the same structure constants as the
elements, $X_1,\ldots,X_r$. Obviously, the above equation becomes a homogeneous
linear differential equation in the coefficients of the matrix $A$.
Consequently, determining the general solution of a Lie system reduces to
solving a linear problem.

Although the above process was described for Lie systems associated with
Vessiot-Guldberg Lie algebras of complete vector fields, it can be proved that a
similar process, with almost identical final results, can be applied to any Lie
system $X(t,x)$. Indeed, this can be done by taking the compactification of
$\mathbb{R}^n$ in order to make all vector fields complete (as in the case of
the Riccati equation) or just by considering that the induced action is just a
local one.

A generalisation of the method \cite{CarMarNas} used by Wei and Norman for
linear systems \cite{WN1,WN2} is very useful for solving equations
(\ref{eqLG2}). Furthermore, there exist reduction techniques that can also be
used \cite{CarRamGra}. Such techniques show, for instance, that Lie systems
related to solvable Vessiot--Guldberg Lie algebras are integrable by quadratures
(\cite{CarRamGra}, Section 8). Finally, as right-invariant vector fields $X^{\tt
R}$ project onto the fundamental vector fields in each homogeneous space for
$G$, the solution of equation (\ref{eqLG2})
enables us to find the general solution for the corresponding Lie system in each
homogeneous space. Conversely,
the knowledge of particular solutions of the associated system in a homogeneous
space gives us a method for
reducing the problem to the corresponding isotopy group \cite{CarRamGra}. 

\section{Geometric approach to superposition rules}\label{GeoAppr}
Let us now turn to review the modern geometrical approach to the theory of Lie
systems carried out in \cite{CGM07}. Although we here basically point out the
results given in that work, several slight improvements have been included in
our presentation.

A fundamental notion in the geometrical description of Lie systems is the
so-called {\it diagonal prolongation} of a $t$-dependent vector field. Its
definition and most important properties are described below.

\begin{definition}
Given a $t$-dependent vector field over $\mathbb{R}^n$ of the form
$$X(t,x_{(0)})=\sum_{i=1}^nX^i(t,x_{(0)})\frac{\partial}{\partial x_{(0)}^i},$$
its {\it diagonal prolongation} to $\mathbb{R}^{n(m+1)}$ is the $t$-dependent
vector field over this latter space given by
$$
\widehat
X(t,x_{(0)},\ldots,x_{(m)})=\sum_{a=0}^m\sum_{i=1}^nX^i(t,x_{(a)})\frac{\partial
}{\partial x^i_{(a)}}.
$$
\end{definition}

Recall that every vector field $X$ over $\mathbb{R}^n$ can be regarded as a
$t$-dependent vector field in a natural way. Evidently, it is immediate that the
above definition can also be applied to define diagonal prolongations for vector
fields over $\mathbb{R}^n$. Obviously, such diagonal prolongations turn out to
be vector fields over $\mathbb{R}^{n(m+1)}$ as well.

Note that diagonal prolongations can be redefined in an intrinsic, and
equivalent, way as follows.

\begin{definition}
Given a $t$-dependent vector field $X$ over $\mathbb{R}^n$, its {\it diagonal
prolongation} to $\mathbb{R}^{n(m+1)}$ is the unique $t$-dependent vector field
$\widehat X$ over $\mathbb{R}^{n(m+1)}$ such that:
\begin{itemize}
\item The $t$-dependent vector field $\widehat X$ is invariant under the action
of the symmetry group $S_{m+1}$ over $\mathbb{R}^{n(m+1)}$.
\item The vector fields $\widehat X_t$ are projectable under the projection
$pr_0$ given by (\ref{proj2}) and ${\rm pr_0}_*\widehat X_t=X_t$. 
\end{itemize}
\end{definition}

\begin{lemma}\label{CommProl}
For every two vector fields $X, Y\in \mathfrak{X}(\mathbb{R}^n)$, it is
immediate that $[\widehat X,\widehat Y]=\widehat{[X,Y]}$. In consequence, given
a Lie algebra of vector fields $V\subset\mathfrak{X}(\mathbb{R}^n)$, the
prolongations of its elements to $\mathbb{R}^{n(m+1)}$ span an isomorphic Lie
algebra of vector fields.
\end{lemma}
\begin{proof} It is straightforward and it is left to the reader. 
\end{proof}

 \begin{lemma}\label{CommProl2} Consider a family, $X_1,\ldots,X_r,$ of vector
fields over $\mathbb{R}^n$ satisfying that their diagonal prolongations to
$\mathbb{R}^{nm}$ are linearly independent at a generic point. Given the
diagonal prolongations, $\widehat X_1,\ldots,\widehat X_r,$ to
$\mathbb{R}^{n(m+1)}$, the vector field $\sum_{\alpha=1}^rb_\alpha\widehat
X_\alpha$, with $b_\alpha\in C^{\infty}(\mathbb{R}^{n(m+1)})$, is also a
diagonal prolongation if and only if the coefficients, $b_1,\ldots,b_r,$ are
constant.
 \end{lemma}

\begin{proof}
Let us write in local coordinates
$$
X_\alpha=\sum_{i=1}^nA^i_\alpha(x)\pd{}{x^i},\qquad \alpha=1,\ldots,r,
$$
what implies that
$$
\widehat
X_\alpha=\sum_{i=1}^n\sum_{a=0}^mA_\alpha^i(x_{(a)})\pd{}{x^i_{(a)}},\qquad
\alpha=1,\ldots,r.$$
Then,
\begin{eqnarray*}
\sum_{\alpha=1}^rb_\alpha(x_{(0)},\ldots,x_{(m)})\widehat
X_\alpha=\sum_{\alpha=1}^r\sum_{i=1}^n\sum_{a=0}^mb_\alpha(x_{(0)},\ldots,x_{(m)
})A^i_\alpha(x_{(a)})\pd{}{x^i_{(a)}},
\end{eqnarray*}
which is a diagonal prolongation if and only if there exist functions $B^j:x\in
\mathbb{R}^n\mapsto B^j(x)\in\mathbb{R}$, with $j=1,\ldots,n$, such that for
each pair of indexes $j$ and $a$,
\begin{equation*}
\begin{aligned}
&\sum_{\alpha=1}^rb_\alpha(x_{(0)},\ldots,x_{(m)})A^i_\alpha(t,x_{(a)})=B^i(x_{
(a)}),\quad a=0,\ldots,m,\quad i=1,\ldots,n.\\
\end{aligned}
\end{equation*}
In particular, the functions $b_\alpha(x_{(0)},\ldots,x_{(m)})$, with
$\alpha=1,\ldots,r$, solve the subsystem of linear equations in the variables,
$u_1,\ldots,u_r,$ given by
\begin{equation*}
\begin{aligned}
&\sum_{\alpha=1}^ru_\alpha A^i_\alpha(x_{(a)})=B^i(x_{(a)}),\quad
a=1,\ldots,m,\quad i=1,\ldots,n.\\
\end{aligned}
\end{equation*}
The coefficient matrix of the above system of $m\cdot n$ equations with $r$
unknowns has rank $r\leq m\cdot n$ since the ${\rm pr}_*(\widehat X_\alpha)$ are
linearly independent. Hence, the solutions, $u_1,\ldots,u_r,$ are completely
determined in terms of the functions $B^i(x_{(a)})$, with $a=1,\ldots,m,$ and
$i=1,\ldots,n,$ and do not depend on $x_{(0)}$. But since the prolongations are
invariant under the action of the symmetry group $S_{m+1}$, functions
$u_\alpha=b_\alpha(x_{(0)},\ldots,x_{(m)})$, with $\alpha=1,\ldots,r$, must
satisfy this symmetry. Consequently, they cannot depend on the variables
$x_{(1)},\ldots,x_{(m)},$ and therefore they must be constant.
\end{proof}

\begin{lemma}\label{ProlThe} For every family of vector fields,
$X_1,\ldots,X_r\in\mathfrak{X}(\mathbb{R}^n)$ linearly independent over
$\mathbb{R}$, there exists an integer $m$ such that their prolongations to
$\mathbb{R}^{nm}$ are linearly independent at a generic point.
\end{lemma}
\begin{proof}
Denote by $\widehat X^q_\alpha$ the diagonal prolongation to $\mathbb{R}^{nq}$
of $X_\alpha$ and define $\sigma(q)$ to be the maximum number of vector fields,
among the family $\widehat X^q_\alpha$, linearly independent at a generic point
of $\mathbb{R}^{nq}$. 

By reduction to the absurd, we assume that each family, $\widehat
X^q_1,\ldots,\widehat X^q_r,$ of diagonal prolongations are linearly dependent
at a generic point of $\mathbb{R}^{qn}$, in other words, $1\leq \sigma(q)<r$ for
every $q$. Therefore, the function $\sigma(q)$ must admit a maximum $p<r$ for a
certain integer $\bar m$, i.e. $p=\sigma(\bar m)$. We can assume, without loss
of generality, that, $\widehat X^{\bar m}_1,\ldots,\widehat X^{\bar m}_p,$ are
linearly independent at generic point of $\mathbb{R}^{n\bar m}$. Moreover, the
vector fields, $\widehat X^{\bar m+1}_1, \ldots, \widehat X^{\bar m+1}_p,$ are
also linearly independent at a generic point of $\mathbb{R}^{n(\bar m+1)}$ and,
as $\sigma(\bar m)$ is a maximum, it must be $\sigma(\bar m+1)=\sigma(\bar m)$.
In consequence, there exist $p$ uniquely defined functions $\bar f_1,\ldots,\bar
f_p\in C^{\infty}(\mathbb{R}^{n(\bar m+1)})$ obeying the equation
\begin{equation}\label{exp2}
\bar f_1\widehat X^{\bar m+1}_1+\ldots+\bar f_p\widehat X^{\bar m+1}_p=\widehat
X^{\bar m+1}_{p+1}.
\end{equation}
This forces the left-hand side to be a diagonal prolongation. Additionally,
since $\widehat X^{\bar m}_1, \ldots, \widehat X^{\bar m}_p,$ are linearly
independent in a generic point, Lemma (\ref{CommProl2}) applies and it turns out
that, $\bar f_1,\ldots,\bar f_p,$ must be constant. Then, projecting the above
expression by ${\rm pr}_0$, it follows that, $X_1,\ldots,X_{p+1},$ are linearly
dependent over $\mathbb{R}$. This violates our initial assumption and thereby we
conclude that our initial premise, i.e. $\sigma(q)<r$ for every $q$, must be
false and there must exist an integer $m$ such that the diagonal prolongations
of, $X_1\ldots,X_r,$ to $\mathbb{R}^{nm}$ become linearly independent at a
generic point, what proves our lemma.
\end{proof}

The above lemma already contains the key point to prove the following result.

\begin{lemma}\label{ProlThe2} If $\sigma(q)<r$, then $\sigma(q)<\sigma(q+1)$.
\end{lemma}
\begin{proof}
It is immediate that $\sigma(q)\leq\sigma(q+1)$. Now, by reduction to absurd, if
we assume $p=\sigma(q)<r$ and $\sigma(q)=\sigma(q+1)$, one can pick up, among
the $\widehat X^q_\alpha$, a family of $p$ vector fields linearly independent at
a generic point of $\mathbb{R}^{nq}$. We can assume, with no loss of generality,
that they are $\widehat X^q_1,\ldots,\widehat X^q_p$. Consequently, as in the
above lemma, we can write
\begin{equation*}
\bar{f}_1\widehat X^{q+1}_1+\ldots+\bar{f}_p\widehat X^{q+1}_p=\widehat
X^{q+1}_{p+1},
\end{equation*}
for certain uniquely defined functions $\bar f_1,\ldots,\bar f_r\in
C^{\infty}(\mathbb{R}^{n(m+1)})$. In a similar way to the proof of the former
lemma, this yields that, $X_1,\ldots,X_{p+1},$ are linearly dependent over
$\mathbb{R}$. This is in contradiction with our initial assumption. In
consequence, if $p<r$, the vector field $\widehat X^{q+1}_{p+1}$ is linearly
independent at a generic point with respect to the previous vector fields and
$\sigma(q+1)>\sigma(q)$.
\end{proof}
Taking into account the above two lemmas, it follows trivially that $\sigma(q)$
grows monotonically until it reaches the maximum $r$. This gives rise to the
following proposition.

\begin{proposition}\label{PropFin} For every family of vector fields
$X_1,\ldots,X_r\!\in\!\mathfrak{X}(\mathbb{R}^n)$ linearly independent over
$\mathbb{R}$, there exists an integer $m\leq r$ such that their prolongations to
$\mathbb{R}^{nm}$ are linearly independent at a generic point. 
\end{proposition}

The above proposition constitutes an explicit proof for vector fields over
$\mathbb{R}^n$ of the analog result for vector fields over manifolds pointed out
in \cite{CGM07}. Let us now turn to describe a geometric interpretation of the
superposition rule notion.

Consider a $t$-dependent vector field (\ref{AssocVec}) associated with the
system 
\begin{equation}\label{equation1}
\frac{dx^i}{dt}=X^i(t,x),\qquad i=1,\ldots,n,
\end{equation}
describing its integral curves. Recall that the above system admits a
superposition rule if there exists a map  $\Phi:\mathbb{R}^{n(m+1)}\rightarrow
\mathbb{R}^{n}$ of the form $x=\Phi(x_{(1)},\ldots,x_{(m)};k_1,\ldots,k_n)$ such
that its general solution, $x(t)$, can be written as
$$
x(t)=\Phi(x_{(1)}(t),\ldots,x_{(m)}(t);k_1,\ldots,k_n),
$$
with, $x_{(1)}(t),\ldots,x_{(m)}(t),$ being a generic family of particular
solutions and $k_1,\ldots,k_n,$ a set of constants associated with each
particular solution.

The map $\Phi(x_{(1)},\ldots,x_{(m)};\cdot):\mathbb{R}^n\rightarrow\mathbb{R}^n$
can be inverted, at least locally around points of an open dense subset of
$\mathbb{R}^{nm}$, to give rise to a map $\Psi:\mathbb{R}^{n(m+1)}\rightarrow
\mathbb{R}^n$,
$$
k=\Psi(x_{(0)},\ldots, x_{(m)}),
$$
where we write $x_{(0)}$ instead of $x$ and $k=(k_1,\ldots,k_n)$ in order to
simplify the notation. Note that the map $\Psi$ is defined so that
$$
k=\Psi(\Phi(x_{(1)},\ldots,x_{(m)};k),x_{(1)},\ldots,x_{(m)}).
$$
Hence, the map $\Psi$ defines an $n$-codimensional foliation on the manifold
$\mathbb{R}^{n(m+1)}$. 

As the fundamental property of the map $\Psi$ states that
\begin{equation}\label{constant2}
k=\Psi(x_{(0)}(t),\ldots, x_{(m)}(t)),
\end{equation}
for any $(m+1)$-tuple of generic particular solutions of system
(\ref{equation1}), the
foliation determined by $\Psi$ is invariant under permutations of its $(m+1)$
arguments, $x_{(0)},\ldots,x_{(m)}$.
Moreover, when differentiating expression (\ref{constant2}) with respect to the
variable $t$, we get
$$
\sum_{a=0}^{m}\sum_{j=1}^nX^j(t,x_{(a)}(t))\pd{\Psi^k}{x^j_{(a)}}(\bar
p(t))=\widehat X_t \Psi^k(\bar p(t))=0, \qquad k=1,\ldots,n,
$$
where $(\Psi^1,\ldots,\Psi^n)=\Psi$ and $\bar p(t)=(x_{(0)}(t),\ldots,
x_{(m)}(t))$. Thus, the functions $\Psi^1,\ldots,\Psi^n$ are first-integrals for
the vector fields $\{\widehat X_t\}_{t\in\mathbb{R}}$ defining an
$n$-codimensional foliation $\mathfrak{F}$ over $\mathbb{R}^{n(m+1)}$ such that
the vector fields $\{\widehat X_t\}_{t\in\mathbb{R}}$ are tangent to its leaves.

The foliation $\mathfrak{F}$ has another important property. Given a leaf
$\mathfrak{F}_k$ corresponding to the level set of $\Psi$ determined by
$k=(k_1,\ldots,k_n)\in\mathbb{R}^n$ and a point
$(x_{(1)},\ldots,x_{(m)})\in\mathbb{R}^{mn}$, there exists a unique point
$(x_{(0)},x_{(1)},\ldots,x_{(m)})\in\mathfrak{F}_k$, namely,
$$(\Phi(x_{(1)},\ldots,x_{(m)};k),x_{(1)},\ldots,x_{(m)})\in\mathfrak{F}_k.$$
Consequently, the projection onto the last $m\cdot n$ factors, i.e. the map
${\rm pr}$  given by (\ref{proj}), induces diffeomorphisms between
$\mathbb{R}^{nm}$ and each one of the leaves $\mathfrak{F}_k$. In other words,
the foliation $\mathfrak{F}$ is horizontal with respect to the projection ${\rm
pr}$.

The foliation $\mathfrak{F}$ corresponds to a connection $\nabla$ on the bundle
${\rm pr}:\mathbb{R}^{n(m+1)}\rightarrow\mathbb{R}^{nm}$ with zero curvature.
Indeed, the restriction of the projection ${\rm pr}$ to a leaf gives a
one-to-one map that gives rise to a linear map among vector fields on
$\mathbb{R}^{nm}$ and `horizontal' vector fields tangent to a leaf.

Note that the knowledge of this connection (foliation) gives us the
superposition rule without referring to the map $\Psi$. If we fix a point
$x_{(0)}(0)$ and $m$ particular solutions, $x_{(1)}(t),\ldots,x_{(m)}(t),$ then
$x_{(0)}(t)$ is the unique point in $\mathbb{R}^n$ such that the point
$(x_{(0)}(t),x_{(1)}(t),\ldots, x_{(m)}(t))$ belongs to the same leaf as
$(x_{(0)}(0),x_{(1)}(0),\ldots,x_{(m)}(0))$. Thus, it is only $\mathfrak{F}$
that really matters when the superposition rule is concerned.

On the other hand, if we have a connection $\nabla$ on the bundle 
$${\rm pr}:\mathbb{R}^{n(m+1)}\rightarrow\mathbb{R}^{nm},$$
with zero curvature, i.e. a horizontal distribution $\nabla$ on
$\mathbb{R}^{n(m+1)}$ that it is involutive and can be integrated to give a
foliation on $\mathbb{R}^{n(m+1)}$ such that the vector fields $\widehat X_t$
belong to $\nabla$, then the procedure described above determines a
superposition rule for system (\ref{equation1}). Indeed, let $k\in\mathbb{R}^n$
enumerates smoothly the leaves $\mathfrak{F}_k$ of the foliation $\mathfrak{F}$,
then we can define $\Phi(x_{(1)},\ldots,x_{(m)};k)\in\mathbb{R}^n$ to be the
unique point $x_{(0)}$ of $\mathbb{R}^n$ such that
$$(x_{(0)},x_{(1)},\ldots,x_{(m)})\in\mathfrak{F}_k.$$
This gives rise to a superposition rule
$\Phi:\mathbb{R}^{nm}\times\mathbb{R}^n\rightarrow\mathbb{R}^n$ for the system
of first-order differential equations (\ref{equation1}). To see this, let us
observe the inverse relation
$$\Psi(x_{(0)},\ldots,x_{(m)})=k,$$
which is equivalent to $(x_{(0)},\ldots,x_{(m)})\in \mathfrak{F}_k$. If we fix
$k$ and take a generic family of particular solutions,
$x_{(1)}(t),\ldots,x_{(m)}(t),$ of equation (\ref{equation1}), then
$x_{(0)}(t)$, defined with the aid of the condition
$\Psi(x_{(0)}(t),\ldots,x_{(m)}(t))=k$, satisfies (\ref{equation1}). In fact,
let $x'_{(0)}(t)$ be the solution of (\ref{equation1}) with initial value
$x'_{(0)}=x_{(0)}$. Since the $t$-dependent vector fields $\widehat X(t,x)$ are
tangent to $\mathfrak{F}$, the curve $(x_{(0)}(t),x_{(1)}(t),\ldots,x_{(m)}(t))$
lies entirely within a leaf of $\mathfrak{F}$, so in $\mathfrak{F}_k$. But a
point of a leaf is entirely determined by its projection by ${\rm pr}$, then
$x'_{(0)}(t)=x_{(0)}(t)$ and $x_{(0)}(t)$ is a solution.
\begin{proposition}\label{LieConnect}
Giving a superposition rule depending on $m$ generic particular solutions for a
Lie system described by a $t$-dependent vector field $X$ is equivalent to giving
a zero curvature connection $\nabla$ on the bundle ${\rm
pr}:\mathbb{R}^{(m+1)n}\rightarrow\mathbb{R}^{nm}$ for which the vector fields
$\{\widehat X_t\}_{t\in\mathbb{R}}$ are horizontal vector fields with respect to
this connection.
\end{proposition}

Although we rejected to investigate in full detail the difference between global
superposition rules and superposition rules, it is interesting to comment
briefly this theme here. Note that a rigorous analysis of the above discussion
shows that a global or `simple' superposition rule gives rise to a zero
curvature connection. Nevertheless, on the contrary, a zero curvature connection
{\it only} ensures the existence of a superposition rule. This is due to the
connection, which only guarantees the existence of a series of {\it local}
first-integrals that give rise to a superposition rule.  In order to ensure the
existence of a global superposition rule, some extra conditions on the
connection must be required as well (see \cite{BM09II}).

\section{Geometric Lie Theorem}\label{GLT}
Let us now prove the classical Lie theorem \cite[Theorem 44]{LS} from a modern
geometric perspective by using the previous results. The following theorem
constitutes a review of the geometric version of the Lie Theorem given in
\cite[Theorem 1]{CGM07}. Our aim in doing so is to include in our exposition one
of the main results of the theory of Lie systems and, at the same time, to
furnish a slightly more detailed proof of this theorem. 

\begin{maintheorem}{\bf (Geometric Lie Theorem)}\label{GLieTheorem}
A system (\ref{equation1}) admits a superposition rule depending on $m$ generic
particular solutions if and only if the $t$-dependent vector field $X$ can be
written as
\begin{equation}\label{field2}
X_t=\sum_{\alpha=1}^rb_\alpha(t)X_\alpha,
\end{equation}
where the vector fields, $X_1,\ldots,X_r,$ form a basis for an $r$-dimensional
real Lie algebra.
\end{maintheorem}
\begin{proof}
Suppose that system (\ref{equation1}) admits a superposition rule
(\ref{FirstSup}) and let $\mathfrak{F}$ be its associated foliation over
$\mathbb{R}^{n(m+1)}$. As the vector fields $\{\widehat X_t\}_{t\in\mathbb{R}}$
are tangent to the leaves of $\mathfrak{F}$, the vector fields of ${\rm
Lie}(\{\widehat X_t\}_{t\in\mathbb{R}})$ span a generalised involutive
distribution 
$$\mathcal{D}_p=\left\{\widehat Y(t,p)|Y\in {\rm Lie}(\{\widehat
X_t\}_{t\in\mathbb{R}})\right\}\in {\rm T}_p\mathbb{R}^{n(m+1)},$$ 
whose elements are also tangent to the leaves of $\mathfrak{F}$. Since the Lie
bracket of two prolongations is a prolongation, we can choose, among the
elements of ${\rm Lie}(\{\widehat X_t\}_{t\in\mathbb{R}})$, a finite family,
$\widehat X_1,\ldots,\widehat X_r$, that gives rise to a local basis of diagonal
prolongations for the distribution $\mathcal{D}$. As the map ${\rm pr}$ projects
each leaf of the foliation $\mathfrak{F}$ into $\mathbb{R}^{nm}$
diffeomorphically, we get that the vector fields ${\rm pr}_*(\widehat
X_\alpha)$, with $\alpha=1,\ldots,r,$ are linearly independent at a generic
point of $\mathbb{R}^{nm}$. These vector fields close on the commutation
relations
$$
[\widehat X_\alpha,\widehat
X_\beta]=\sum_{\gamma=1}^rf_{\alpha\beta\gamma}\widehat X_\gamma, \qquad
\alpha,\beta=1,\ldots,r,$$
for certain functions $f_{\alpha\beta\gamma}\in
C^{\infty}(\mathbb{R}^{n(m+1)})$. In view of Lemma \ref{CommProl2}, these
functions must be constant, let us say
$f_{\alpha\beta\gamma}=c_{\alpha\beta\gamma}$, and, taking into account the
properties of diagonal prolongations, one has that, $X_1,\ldots,X_r,$ are
linearly independent vector fields obeying the relations
$$
[X_\alpha, X_\beta]=\sum_{\gamma=1}^rc_{\alpha\beta\gamma}X_\gamma,\qquad
\alpha,\beta=1,\ldots,r.$$
Since, at each time, $\widehat X_t$ is spanned by the vector fields, $\widehat
X_1,\ldots,\widehat X_r$, there are $t$-dependent functions $b_\alpha\in
C^{\infty}(\mathbb{R}\times\mathbb{R}^{n(m+1)})$, with $\alpha=1,\ldots,r$, such
that
$$
\widehat X_t=\sum_{\alpha=1}^rb_\alpha\widehat X_\alpha.
$$
But each $\widehat X_t$ is a diagonal prolongation, so, using Lemma
\ref{CommProl2}, one gets that the functions, $b_1,\ldots,b_r,$ depend only on
the time and thus
\begin{equation}\label{field}
\widehat X_t=\sum_{\alpha=1}^rb_\alpha(t)\widehat X_\alpha.
\end{equation}
From here, it is immediate that (\ref{field2}).

To prove the converse property,  assume that the $t$-dependent vector field $X$
can be put in the form (\ref{field2}), where the
vector fields, $X_1,\ldots,X_r,$ are linearly independent over $\mathbb{R}$ and
span a $r$-dimensional Lie algebra.

As the vector fields, $X_1,\ldots,X_r,$ are linearly independent over
$\mathbb{R}$, there exists, in view of Proposition \ref{PropFin}, a minimal
number $m\leq r$, such that their diagonal prolongations to $\mathbb{R}^{nm}$
are linearly independent at a generic point (what yields that $r\leq n\cdot m$).
Moreover, the diagonal prolongations, $\widehat X_1,
\ldots,\widehat X_r,$ to $\mathbb{R}^{n(m+1)}$ are linearly independent and they
form a basis for an involutive distribution $\mathcal{D}$. This distribution
leads to a $(n(m+1)-r)$-codimensional foliation $\mathfrak{F}_0$ on
$\mathbb{R}^{n(m+1)}$. As the codimension of $\mathfrak{F}_0$ is at least $n$, we can consider an $n$-codimensional foliation $\mathfrak{F}$
whose leaves include those of $\mathfrak{F}_0$. The leaves of this foliation 
project onto the last $m\cdot n$ factors
diffeomorphically and they are at least $n$-codimensional. Hence, according to
Proposition \ref{LieConnect}, foliation $\mathfrak{F}$ defines
a superposition rule depending on $m$ particular solutions. 
\end{proof}

Note that the converse part of the previous proof shows that all systems
described by $t$-dependent vector fields of the form (\ref{field}) share a
common superposition rule. More specifically, all such $t$-dependent vector
fields give rise to the same distribution $\mathcal{D}$ over the same space
$\mathbb{R}^{n(m+1)}$, and this straightforwardly ensures the existence of  a
common superposition rule for all of them. This fact will be analysed more
extensively in the second part of our work, where certain families of systems of
differential equations that admit a $t$-dependent common superposition rule, the
referred to as {\it Lie families}, are investigated.

\section{Determination of superposition rules}\label{nSR}

Note that the previous geometric demonstration of Lie Theorem also contains
information about the superposition rules associated with a Lie system. Let us
analyse this fact more carefully.

Consider a Lie system in $\mathbb{R}^n$ associated with a $t$-dependent vector
field $X$. In view of Lie Theorem, such a $t$-dependent vector field can be
written in the form
$$
X(t,x)=\sum_{i=1}^n\sum_{\alpha=1}^r b_\alpha(t)\,
X^i_\alpha(x)\frac{\partial}{\partial x^i},
$$
where the vector fields $X_\alpha(x)=\sum_{i=1}^nX^i_\alpha(x)\partial/\partial
x^i$ span a $r$-dimensional Lie algebra of vector fields. Now, the geometric
proof of Lie Theorem shows that the above decomposition gives rise to a
superposition rule depending on $m$ generic particular solutions with $r\leq
m\cdot n$. More exactly, the number $m$ coincides with the minimal integer that
makes the diagonal prolongations of $X_1,\dots,X_r,$ to $\mathbb{R}^{mn}$ to
become linearly independent at a generic point. In different words, the only
functions $f_1,\ldots,f_r \in C^{\infty}(\mathbb{R}^{nm})$ such that
\begin{equation}\label{NumSup}
\sum_{\alpha=1}^r f_\alpha\, X^i_\alpha(x_{(a)})=0\,,\qquad a=1,\ldots ,m,\quad
i=1,\ldots,n,
\end{equation}
at a generic point $(x_{(1)},\dots,x_{(k)})$ are $f_1=\ldots=f_r=0$.

Let us illustrate our above comments by means of a simple example. Consider the
Riccati equation
$$ \dot x=b_1(t)+b_2(t)\, x+b_3(t)x^2,
$$which describes the integral curves of the $t$-dependent vector field
$$
X_t=b_1(t)\frac{\partial}{\partial x}+b_2(t)x\frac{\partial}{\partial
x}+b_3(t)x^2\frac{\partial }{\partial x}.
$$
Recall that the vector fields $\{X_t\}_{t\in\mathbb{R}}$ take values in the
three-dimensional Lie algebra $V$ spanned by the vector fields
$$
X_1=\frac{\partial}{\partial x},\quad X_2=x\frac{\partial}{\partial x}, \quad
X_3=x^2\frac{\partial}{\partial x}.
$$
Consequently, we can determine the number of particular solutions for a
superposition rule for Riccati equations by considering the minimal $m$ such
that corresponding system (\ref{NumSup}) admits only the trivial solution. For
$m=2$, this system reads
$$f_1+f_2x_{(1)}+f_3x_{(1)}^2=0\,,\qquad f_1+f_2x_{(2)}+f_3x_{(2)}^2=0,
$$
and it has non-trivial solutions. Nevertheless, the system for the prolongations
to $\mathbb{R}^3$, that is, 
$$f_1+f_2x_{(1)}+f_3x_{(1)}^2=0\,,\qquad f_1+f_2x_{(2)}+f_3x_{(2)}^2=0\,,\qquad
f_1+f_2x_{(3)}+f_3x_{(3)}^2=0\,,
$$
does not admit any non-trivial solution because the determinant of the
coefficients, i.e.
$$
\left|\left(\begin{array}{ccc}
1 & x_{(1)} & x_{(1)}^2\\
1 & x_{(2)} &x_{(2)}^2\\
1 & x_{(3)} & x_{(3)}^2
\end{array}\right)\right|=(x_{(2)}-x_{(1)})(x_{(2)}-x_{(3)})(x_{(1)}-x_{(3)}),
$$
is different from zero when the three points $x_{(1)}$, $x_{(2)},$ and $x_{(3)}$
are different. Thus, we get that $m=3$ and the superposition rule for the
Riccati equation depends on three particular solutions. Obviously, the relations
$m\leq \dim V\leq m\cdot n$ are valid in this case.

Once the number $m$ of particular solutions has been determined, the
superposition rule can be worked out in terms of first-integrals
for the diagonal prolongations, $\widehat X_1,\ldots,\widehat X_r,$ over
$\mathbb{R}^{n(m+1)}$. Finally, it is worth noting that when the vector fields,
$\widehat X_1,\ldots,\widehat X_r,$ over $\mathbb{R}^{n(m+1)}$ admit more than
$n$ common first-integrals, the system $X$ admits more than one superposition
rule (see \cite{CGM07}). 

\section{Mixed superposition rules and constants of the motion}\label{SR3}
Roughly speaking, a {\it mixed superposition rule} is a $t$-independent map
describing the general solution of a system of first-order differential
equations in terms of a generic family of particular solutions of various
systems (generically different ones) of first-order differential equations and a
set of constants. Obviously, mixed superposition rules include, as particular
instances, the standard superposition rules related to Lie systems. 

\begin{definition} A {\it mixed superposition rule} for a system of first-order
differential equations determined by a $t$-dependent vector field $X$ over
$\mathbb{R}^{n_0}$ is a $t$-independent map
$\Phi:\mathbb{R}^{n_1}\times\ldots\times \mathbb{R}^{n_m}\times
\mathbb{R}^{n_0}\rightarrow \mathbb{R}^{n_0}$ of the form
$$
x=\Phi(x_{(1)},\ldots,x_{(m)};k_1,\ldots,k_{n_0}),
$$
such that the general solution, $x(t)$, of system $X$ can be written as
$$
x(t)=\Phi(x_{(1)}(t),\ldots,x_{(m)}(t);k_1,\ldots,k_{n_0}),
$$
with, $x_{(1)}(t),\ldots, x_{(m)}(t),$ being a generic family of curves
satisfying that each $x_{(a)}(t)$ is a particular solution of the system
determining the integral curves for a $t$-dependent vector field $X^{(a)}$ over
$\mathbb{R}^{n_a}$, with $a=1,\ldots,m$.
\end{definition}

As a particular example of mixed superposition rule, consider the linear system
of differential equations
\begin{equation}\label{Inhomo}
\frac{dx^i}{dt}=\sum_{j=1}^nA^i_j(t)x^j+B^i(t),\qquad i=1,\ldots,n,
\end{equation}
whose general solution, $x(t)$, can be written as
$$
x(t)=y_{(1)}(t)+\sum_{j=1}^nk_jz_{(j)}(t),
$$
in terms of one particular solution $y_{(1)}(t)$ of (\ref{Inhomo}), any family
of $n$ linearly independent particular solutions,
$z_{(1)}(t),\ldots,z_{(n)}(t),$ of the homogeneous linear system
$$
\frac{dz^i}{dt}=\sum_{j=1}^nA^i_j(t)z^j,\qquad i=1,\ldots,n,
$$
and a set of $n$ constants, $k_1,\ldots,k_n$. 

We here aim to give a method to obtain a particular type of mixed superposition
rule for a Lie system in terms of particular solutions of another Lie system.
Additionally, we relate our results to the commentary given in \cite[Remark
5]{CGM07}, where it was briefly discussed that the solutions of a certain
first-order differential equation on a manifold may be obtained in terms of
solutions of other first-order systems by constructing a certain foliation. 

Consider the system on $\mathbb{R}^{n_0}$ given by
\begin{equation}\label{system}
\frac{dx^i}{dt}=\sum_{\alpha=1}^rb_\alpha(t)X^i_\alpha(x),\qquad i=1,\ldots,n_0,
\end{equation}
determining the integral curves of the $t$-dependent vector field
\begin{equation}\label{MixSup}
X(t,x)=\sum_{\alpha=1}^r\sum_{i=1}^{n_0} b_{\alpha}(t)X^i_\alpha(x)\pd{}{x^i},
\end{equation}
where the vector fields
$X_\alpha(x)=\sum_{i=1}^{n_0}X_\alpha^i(x)\partial/\partial x^i$, close on a
$r$-dimensional Lie algebra $V$, i.e. there exist $r^3$ constants
$c_{\alpha\beta\gamma}$ such that
\[
[X_\alpha,X_\beta]=\sum_{\gamma=1}^ rc_{\alpha\beta\gamma}X_\gamma,\qquad
\alpha,\beta=1,\ldots,r.
\]
We here aim to derive a particular type of mixed superposition rule of the form
$\Phi:(\mathbb{R}^{n_1})^m\times\mathbb{R}^{n_0}\rightarrow \mathbb{R}^{n_0}$
for the above Lie system in such a way that its general solution, $x(t)$, can be
expressed as
\[
x(t)=\Phi(x_{(1)}(t),\ldots,x_{(m)}(t);k_1,\ldots,k_n),
\]
where, $x_{(1)}(t),\ldots, x_{(m)}(t),$ are a generic family of particular
solutions of a Lie system determined by a $t$-dependent vector field $X^{(1)}$
on $\mathbb{R}^{n_1}$. Let us assume that system $X^{(1)}$ takes the particular
form
\begin{equation}\label{dec}
X^{(1)}_t=\sum_{\alpha=1}^rb_\alpha(t)X^{(1)}_\alpha,
\end{equation}
where the vector fields
$X^{(1)}_\alpha\!\!\in\!\!\mathfrak{X}(\mathbb{R}^{n_1}\!)$ obey the same
commutation
relations as the vector fields $X_\alpha$, that is,
\begin{equation}\label{ConMix}
[X^{(1)}_\alpha,X^{(1)}_\beta]=\sum_{\gamma=1}^rc_{\alpha\beta\gamma}X^{(1)}
_\gamma,\qquad \alpha,\beta=1,\ldots r,
\end{equation}
It is important to clarify  when such a $t$-dependent vector field $X^{(1)}$
exists. Let us prove its existence. On one hand, Ado's Theorem states that for
every finite-dimensional Lie algebra $V$, e.g. the one spanned by the vector
fields $X_\alpha$, there exists an isomorphic matrix Lie algebra $V_M$ of
$n_1\times n_1$ square matrices. Now, since the homogeneous linear system 
$$
\dot y=A(t)y,
$$
where $A(t)$ takes values in $V_M$ is a Lie system associated with a Lie algebra
of vector fields isomorphic to $V_M$ (see \cite{CAL}),  it follows immediately
that  we can always determine a family of linear vector fields on
$\mathbb{R}^{n_1}$ obeying relations (\ref{ConMix}). In terms of this family, we
can build up a $t$-dependent vector field of the form (\ref{dec}). Apart from
the $t$-dependent vector field $X^{(1)}_t$ constructed in the aforementioned
way, there might exist other ones made of through finite-dimensional Lie
algebras of vector fields admitting a basis whose elements obey relations
(\ref{ConMix}). 

Proposition \ref{PropFin} ensures the existence of a minimal $m$ such that the
diagonal prolongations of the $X^{(1)}_\alpha$ to $\mathbb{R}^{n_1m}$ are
linearly independent at a generic point. Let us denote such prolongations by
$$
\widetilde
X_\alpha=\sum_{a=1}^{m}X^{i(1)}_\alpha(x_{(a)})\frac{\partial}{\partial
x^{i}_{(a)}},\qquad \alpha=1,\ldots,r,
$$
and define the vector fields on $\widetilde
N=\mathbb{R}^{n_0}\times\mathbb{R}^{n_1 m}$ of the form
\[
Y_\alpha=X_\alpha+\sum_{a=1}^{m}X^{i(1)}_\alpha(x_{(a)})\frac{\partial}{\partial
x^{i}_{(a)}},\qquad \alpha=1,\ldots,r.
\]
where we have considered the vector fields $X_\alpha$ and $X^{(1)}_\alpha$ as
vector fields on $\widetilde N$ in the natural way. 
From the above definition, one has
\[
[Y_\alpha, Y_\beta]=\sum_{\gamma=1}^r c_{\alpha\beta\gamma}Y_\gamma,\qquad
\alpha,\beta=1,\ldots,r.
\]
Consequently, the system of differential equations that determines the integral
curves of the $t$-dependent vector field
\[
Y_t=\sum_{\alpha=1}^{r}b_\alpha(t)Y_\alpha,
\]
is a Lie system associated with a Vessiot-Guldberg Lie algebra isomorphic to
$V$. 

Define the involutive distribution $\widetilde{\mathcal{V}}$ on $\widetilde N$
of the form
\[
\widetilde{\mathcal{V}}_{\tilde x}=\langle  (Y_1)_{\tilde
x},\ldots,(Y_r)_{\tilde x}\rangle,\qquad \tilde x\in\widetilde N,
\]
whose rank is $r$, around a generic point of $\widetilde N$. Additionally, as
$r\leq m\cdot n_1$, we may choose, at least locally, $n_0$ common
first-integrals of the vector fields, $Y_1,\ldots,Y_r,$ giving rise to a
$n_0$-codimensional local foliation $\mathcal{F}$ over
$\mathbb{R}^{n_0}\times\mathbb{R}^{n_1m}$, whose leaves project
diffeomorphically onto $\mathbb{R}^{nm_1}$ through the projection
$$
p:(x,x_{(1)},\ldots,x_{(m)})\in \widetilde N\mapsto (x_{(1)},\ldots,x_{(m)})\in
\mathbb{R}^{n_1m}.
$$
Additionally, the vector fields $Y_\alpha$ are tangent to the leaves of this
foliation. 

On one hand, it is immediate that the above results lead to defining a flat
connection $\nabla$ on the bundle $p:\widetilde N\rightarrow \mathbb{R}^{n_1m}$.
On the other hand, as it happened in the case of superposition rules (see
Section \ref{GeoAppr}), for every point
$(x_{(1)},\ldots,x_{(m)})\in\mathbb{R}^{n_1m}$ and a leave $\mathcal{F}_k$, with
$k=(k_1,\ldots,k_{n_0})$, of the foliation $\mathcal{F}$, there exists a unique
point $x_{(0)}$ in $\mathbb{R}^{n_0}$ such that
$(x_{(0)},x_{(1)},\ldots,x_{(m)})\in\mathcal{F}_k$. This gives rise to the
definition of a map
$$
x_{(0)}=\Phi(x_{(1)},\ldots,x_{(m)};k_1,\ldots,k_{n_0}).
$$
{\it Mutatis mutandis}, the same arguments showed at the end of the Section
\ref{GeoAppr} apply here, and it can easily be proved that given a generic set
of $m$ particular solutions of system $X^{(1)}$, the general solution of $X$ can
be written as 
$$
x(t)=\Phi(x_{(1)}(t),\ldots,x_{(m)}(t);k_1,\ldots,k_{n_0}),
$$
what shows that $\Phi$ is a particular type of mixed superposition rule. In this
way, we have also shown that, as claimed in  \cite[Remark 5]{CGM07}, a flat
connection $\nabla$ on a bundle of the form $N_0\times N_1\times\ldots\times
N_m\rightarrow N_1\times\ldots\times N_m$ can be used to obtain the solutions of
a first-order system in $N_0$ by means of particular solutions of other
first-order systems in $N_1,\ldots,N_m$.

\section{Differential geometry on Hilbert spaces}\label{SLSQM}
In order to provide some basic knowledge to develop the main results of the
applications of the theory of Lie systems to Quantum Mechanics, we report in
this section some known concepts of the Differential Geometry on
infinite-dimensional manifolds. For further details one can consult
\cite{CLR08WN,CarRam05b,KM97}.

As far as Quantum Mechanics is concerned, the separable complex
Hilbert space of states $\cal H$ can be seen as a
(infinite-dimensional) real manifold admitting a global chart
\cite{BCG}. Infinite-dimensional manifolds do not enjoy the same geometric
properties as finite-dimensional ones, e.g. in the most general
case, and given an open $U\subset\mathcal{H}$, there is not a
one-to-one correspondence between derivations on
$C^{\infty}(U,\mathbb{R})$ and sections of the tangent bundle  $TU$. Therefore,
some explanations must be given before dealing
with such manifolds.

On one hand, given a point $\phi\in\mathcal{H}$, a {\it kinematic tangent
  vector} with foot point $\phi$ is a pair $(\phi,\psi)$ with
$\psi\in\mathcal{H}$. We call $T_\phi\mathcal{H}$ the space of all kinematic
tangent vectors with foot point $\phi$. It consists of all derivatives $\dot
c(0)$ of smooth curves $c:\mathbb{R}\rightarrow\mathcal{H}$ with $c(0)=\phi$.
This fact gives a reason for the name of kinematic.

From the concept of kinematic tangent vector we can provide the definition of
smooth kinematic vector fields as follows: A {\it smooth kinematic vector
  field} is an element $X\in \mathfrak{X}(\mathcal{H})\equiv \Gamma({\pi})$,
with $T\mathcal{H}$ the so-called {\it kinematic tangent bundle} and $\pi:{\rm
  T}\mathcal{H}\rightarrow\mathcal{H}$ the projection of this bundle. We define
a {\it kinematic vector field} $X$ as a map
$X:\mathcal{H}\rightarrow {\rm T}\mathcal{H}$ such that $\pi\circ X={\rm
Id}_\mathcal{H}$. Given a $\psi\in\mathcal{H}$, we will denote from now on
$X(\psi)=(\psi,X_{\psi})$, with $X_\psi$ being the value of $X(\psi)$ in
$T_\psi\mathcal{H}$.

Similarly to the Differential Geometry on finite-dimensional manifolds,
we say that a kinematic vector field $X$ on $\mathcal{H}$ admits a local flow on
an open subset  $U\subset\mathcal{H}$ if there exists a map
$Fl^X:\mathbb{R}\times U\rightarrow\mathcal{H}$ such that $Fl^X(0,\psi)=\psi$
for all $\psi\in U$ and
$$
X_{\psi}=\left.\frac{d}{ds}\right|_{s=0}Fl^X(s,\psi)=\left.\frac{d}{ds}\right|_{
s=0}Fl^X_s(\psi),
$$
with $Fl^X_s(\psi)=Fl^X(s,x)$.

Let us use all these mathematical concepts to study Quantum Mechanics as a
geometric theory. Note that the Abelian translation group on $\mathcal{H}$
provides an identification of  the tangent space $T_\phi\cal H$ at any point
$\phi\in
\cal H$  with  $\cal H$ itself. Furthermore, through such an identification of
$\cal H$ with $T_\phi\cal H$ at any $\phi\in \mathcal{H}$, a continuous
kinematic vector field is simply a continuous
map $X\colon \cal H\to \cal H$.

Starting with a bounded $\mathbb{C}$-linear operator $A$ on $\mathcal{H}$, we
can define
the kinematic vector field $X^A$  by $X^A_\psi=A\psi\in\mathcal{H}\simeq
T_\psi\mathcal{H}.$ In other words, we have 
$$
X^A:\psi\in\mathcal{H}\mapsto (\psi,X\psi)\in{\rm
T}\mathcal{H}\simeq\mathcal{H}\oplus\mathcal{H}.
$$ 
Usually, operators
in Quantum Mechanics are neither continuous nor defined on the whole space
$\cal H$. The most relevant case happens when $A$ is a skew-self-adjoint
operator of the form $A=-i\, H$. The reason is that $\cal H$ can be endowed with
a natural (strongly) symplectic structure, and then such skew-self-adjoint
operators are singled out as the linear vector fields that are
Hamiltonian. The integral curves of such a Hamiltonian vector field
$X^A$ are the solutions of the corresponding Schr\"odinger equation
\cite{BCG,CLR08WN}. Even when $A$ is not bounded, if $A$ is skew-self-adjoint it
must be
densely defined and, by Stone's Theorem, its integral curves are strongly
continuous and defined in all $\mathcal{H}$.

Additionally, these kinematic vector fields related to skew-self-adjoint
operators admit local flows, i.e. any skew-self-adjoint operator $A$ has a local
flow
\begin{equation}\label{Flow}
Fl^A_s(\psi)={\rm exp}(sA)(\psi)\quad {\rm as} \quad \frac{d}{ds}Fl^
A_s(\psi)=A{\rm exp}(sA)(\psi)=A(Fl^A_s(\psi)).
\end{equation}

We remark that given two constants $\lambda, \mu\in\mathbb{R}$ and two
skew-self-adjoint operators $A$ and $B$, we get that
$X^{\lambda A+\mu B}=\lambda X^A+\mu X^B$. Moreover,  skew-self-adjoint
operators considered as vector fields are
fundamental vector fields relative to the usual action of the unitary group
$U(\mathcal{H})$ on the Hilbert space $\mathcal{H}$.

Let us turn to define the Lie bracket of two kinematic vector fields $X^A$ and
$X^B$ associated with two skew-self-adjoint operators $A$ and $B$,
correspondingly. In order to simplify the notation, and as it shall be clear
from the context, we hereafter denote both the commutator of operators, i.e.
$[A,B]=AB-BA$, and the Lie bracket of vector fields $[X^A,X^B]$ in the same way.
In view of the previous remarks, we can declare the Lie bracket of vector fields
related to skew-self-adjoint operators to be 
$$
[X^A,X^B]=X^{[B,A]}.
$$
It is worth noting that the above formula is equivalent to the standard one 
\begin{equation}\label{LieBrack}
[X,Y]_\psi=\left.\frac 12 \frac{d^2}{ds^2}\right|_{t=0}(Fl^Y_{-s}\circ
Fl^X_{-s}\circ Fl^Y_{s}\circ F^X_{s}(\psi)),
\end{equation}
for finite-dimensional Differential Geometry when the right-hand side is
properly defined. Indeed, the above formula yields
\begin{equation*}\begin{aligned}
\left[X^A,X^B\right]_\psi&=\frac
12\frac{d^2}{ds^2}\bigg|_{s=0}
\exp\left(-sB\right)\exp\left(-sA\right)\exp\left(sB\right)\exp\left(sA\right)(\
psi)\\
&=\frac
12\frac{d^2}{ds^2}\bigg|_{s=0}\left(\sum_{n_1=0}^{\infty}\frac{(-sB)^{n_1}}{n_1!
}\right)
\left(\sum_{n_2=0}^{\infty}\frac{(-sA)^{n_2}}{n_2!}\right)\\&
\left(\sum_{n_3=0}^{\infty}\frac{(sB)^{n_3}}{n_3!}\right)
\left(\sum_{n_4=0}^{\infty}\frac{(sA)^{n_4}}{n_4!}\right)(\psi)\\
&=\frac 12\frac{d^2}{ds^2}\bigg|_{s=0}\left(-s^2AB+s^2BA\right)(\psi)\\
&=\frac 12\frac{d^2}{ds^2}\bigg|_{s=0}\left(s^2[B,A]\right)(\psi)=[B,A](\psi),
\end{aligned}
\end{equation*}
when the above expressions are properly defined. From where, we obtain again 
\begin{equation}\label{FR}
[X^A,X^B]=-X^{[A,B]},
\end{equation}
as we defined. 
\section{Quantum Lie systems}\label{QLS}

The theory of Lie systems can be applied to investigate a particular class of
$t$-dependent Hamiltonians satisfying a specific set of conditions, the
so-called {\it quantum Lie systems}. Let us now precisely define this notion and
sketch some of its properties.

We call a $t$-dependent Hamiltonian $H(t)$ a $t$-parametric family of
self-adjoint operators $H_t:\mathcal{H}\rightarrow\mathcal{H}$. 

\begin{definition} We say that the $t$-dependent Hamiltonian $H(t)$ is a {\it
quantum Lie system} if it can be written as
 \begin{equation}\label{LieHamiltonian}
 H(t)=\sum_{\alpha=1}^rb_\alpha(t)H_\alpha,
\end{equation}
where the operators $iH_\alpha$ are a family of skew-self-adjoint operators on
$\mathcal{H}$ giving rise to a basis of a real $r$-dimensional Lie algebra of
operators $V$ under the commutator of operators, i.e.
\begin{equation}
  [iH_\alpha,iH_\beta]=\sum_{\gamma=1}^rc_{\alpha\beta\gamma}\ iH_\gamma,\qquad
\alpha,\beta=1,\ldots,r,
\label{algebH}
\end{equation}
for certain $r^3$ real structure constants $c_{\alpha\beta\gamma}$. We call $V$
a {\it quantum Vessiot--Guldberg Lie algebra} associated with $H(t)$.
\end{definition}

Each quantum Lie system $H(t)$ leads to a Schr\"odinger equation
\begin{equation}\label{Schr}
 \frac{d\psi}{dt}=-iH(t)\psi=-\sum_{\alpha=1}^rb_{\alpha}(t)iH_\alpha\psi,
\end{equation}
describing the integral curves for the kinematic $t$-dependent vector field on
$\mathcal{H}$ given by
$$X_t=\sum_{\alpha=1}^r b_\alpha(t)X_\alpha,$$
where $X_\alpha$ is the vector field associated with the operator
${-iH_\alpha}$. In view of the relation (\ref{FR}) and the commutation relations
(\ref{algebH}), we obtain
\begin{equation}
[X_\alpha,X_\beta]=-X^{[iH_\alpha,iH_\beta]}=\sum_{\gamma=1}^rc_{
\alpha\beta\gamma}X_\gamma,\qquad \alpha,\beta=1,\ldots,n.
\end{equation}
In consequence, the vector fields $X_\alpha$ span an $r$-dimensional Lie algebra
of vector fields. In addition, the structure constants for the basis $\{X_\alpha
\mid \alpha=1,\ldots,r\}$ coincide with those of the quantum Vessiot--Guldberg
Lie algebra for the basis $\{iH_\alpha\mid \alpha=1,\ldots,r\}$. 

Given the Lie algebra $V$, consider an isomorphic Lie algebra $\mathfrak{g}$
corresponding to a connected Lie group $G$. Choose a basis $\{{\rm
a}_\alpha\,|\,\alpha=1,\ldots,r\}$ of the Lie algebra $T_eG\simeq\mathfrak{g}$
such that the Lie brackets of its elements, denoted by
$[\cdot,\cdot]$, obey the relations
\begin{equation}
[{\rm a}_\alpha,{\rm a}_\beta]=\sum_{\gamma=1}^rc_{\alpha\beta\gamma}  {\rm
a}_\gamma\,,\qquad \alpha,\beta=1,\ldots,r.\label{Liealgdef}
\end{equation}
It can be proved that there exists a unitary action
$\Phi:G\times\mathcal{H}\rightarrow\mathcal{H}$ such that each $X_\alpha$ is the
fundamental vector field associated with the element ${\rm a}_\alpha$, according
to the relation  (\ref{Liealgdef}). Indeed, note that, fixed the basis $\{{\rm
a}_\alpha\mid \alpha=1,\ldots,r\}$, each element $g$, in a sufficiently small
open $U$ containing the neutral element of $G$, can be put in a unique way as
$$
g=\exp(-\mu_1{\rm a}_1)\times\ldots\times\exp(-\mu_r{\rm a}_r).
$$
Now, we define
$$
\Phi(\exp(-\mu_\alpha{\rm a}_\alpha),\psi)=\exp(-i\mu_\alpha
H_\alpha)\psi,\qquad \alpha=1,\ldots,r.
$$
As $G$ is connected, every element can be written as a product of elements in
$U$, what, in view of the above relations, gives rise to an action
$\Phi:G\times\mathcal{H}\rightarrow\mathcal{H}$.

Similarly to the procedure carried out to show that solving a Lie system reduces
to working out a particular solution for an equation in a Lie group (see Section
\ref{FNLS}), it can be proved that solving the Schr\"odinger equation for a
quantum Lie system $H(t)$ reduces to determining the solution of the equation in
$G$ given by
$$
R_{g^{-1}*g}\dot g=-\sum_{\alpha=1}^rb_\alpha(t){\rm a}_\alpha\equiv {\rm
a}(t),\qquad g(0)=e.
$$
More specifically, the particular solution of the Schr\"odinger equation
(\ref{Schr}) with initial condition $\psi_0$ reads $\psi_t=\Phi(g(t),\psi_0)$,
where $g(t)$ is the solution of the above equation.

\section{Superposition rules for second and higher-differential
equations}\label{SODEsSystems}
Although the theory of Lie systems is mainly devoted to the study first-order
differential equations, it can also be applied to investigate various systems of
second-order differential equations, e.g. the so-called SODE Lie systems. This
allows us to derive $t$-dependent and $t$-independent constants of the motion,
exact solutions, superposition rules or mixed superposition rules for these
equations, etc. Moreover, our methods to study systems of second-order
differential equations can also be generalised to study systems of higher-order
differential equations.

Vessiot pioneered the analysis of systems of second-order differential equations
by means of the theory of Lie systems \cite{Ve95}. Additionally, this theme was
also briefly examined by Winternitz, Chisholm and Common \cite{CC87,WintSecond}.
Apart from these few works, the analysis of systems of second-order differential
equations through the theory of Lie systems was not deeply analysed until the
beginning of the XXI century, when the SODE Lie system concept was defined and
employed to investigate various systems of second-order differential equations
\cite{CGL09KS,CL08b,CL08Diss,CL09SRicc,CLR08,CLR07a}. This allowed us to recover
previous results from a new clarifying perspective as well as to obtain some new
achievements.

The description of the general solution of systems of second-order differential
equations in terms of certain families of particular solutions and sets of
constants appears in the study of some systems in Physics and Mathematics
\cite{HL02,RR80}. Nevertheless, these results are frequently obtained through
{\it ad hoc} procedures that neither explain their theoretical meaning nor the
possibility of their generalisation. This section is concerned with the
application of the theory of Lie systems to SODE Lie systems in order to review,
through a geometrical unifying approach, some achievements previously obtained
in the literature. Not only this provides a deeper theoretical understanding of
these works, but it also offers several new achievements concerning these and
other related topics.

Recall that the theory of Lie systems initially aimed to study systems of
first-order differential equations admitting its general solution to be
expressed in terms of certain families of particular solutions and a set of
constants. Nevertheless, this property is not exclusive for systems of
first-order differential equations. For instance, each second-order differential
equation of the form $\ddot x=a(t)x$, with $a(t)$ being a $t$-dependent real
function, satisfies that its general solution, $x(t)$, can be cast into the form
\begin{equation}\label{LinearSup}
x(t)=k_1x_{(1)}(t)+k_{2}x_{(2)}(t),
\end{equation}
 with, $k_1, k_{2,}$ being a set of constants and, $x_{(1)}(t),x_{(2)}(t),$
being a family of particular solutions whose initial conditions
$(x_{(1)}(0),\dot x_{(1)}(0))$ and $(x_{(2)}(0),\dot x_{(2)}(0))$ are two
linearly independent vectors of ${\rm T}\mathbb{R}$. Note also that such a
superposition rule leads to the existence of many other nonlinear superposition
rules for other systems of second-order differential equations. For instance,
the change of variables $y=1/x$ transforms the previous system into $y\ddot
y-2\dot y^2=-a(t)y^2$ admitting, in view of the above linear superposition rule
and the above change of variable, its general solution to be written as 
\begin{equation}\label{NonLinearSup}
y(t)=\left(k_1y_1^{\,-1}(t)+k_{2}y_{2}^{\,-1}(t)\right)^{-1},
\end{equation}
in terms of certain families, $y_{(1)}(t),y_{(2)}(t),$ of particular solutions
and a set of two constants. 

Consequently, in view of the previous examples and other ones that can be found,
for instance, in \cite{CGL08,CLuc08b}, it is natural to define superposition
rules for second-order differential equations as follows.

\begin{definition} We say that a second-order differential equation 
\begin{equation}\label{SODE}
\ddot x^i=F^i(t,x,\dot x), \qquad \,\, i=1,\ldots,n, 
\end{equation}
on $\mathbb{R}^n$ admits a global superposition rule if there exists a map
$\Psi:{\rm T}\mathbb{R}^{mn}\times\mathbb{R}^{2n}\rightarrow \mathbb{R}^n$ such
that its general solution $x(t)$ can be written as
\begin{equation}\label{super}
x(t)=\Psi(x_{(1)}(t),\ldots,x_{(m)}(t),\dot x_{(1)}(t),\ldots,\dot
x_{(m)}(t);k_1,\ldots,k_{2n}),
\end{equation}
in terms of a generic family, $x_{(1)}(t),\ldots,x_{(m)}(t),$ of particular
solutions, their derivatives, and a set of $2n$ constants.
\end{definition}

In order to understand the previous definition, it is necessary to establish the
precise meaning for `generic' in the above statement. Formally, it is said that
expression (\ref{super}) is valid for a generic family of particular solutions
when it holds for every family of particular solutions,
$x_{1}(t),\ldots,x_{m}(t),$ satisfying that $(x_{1}(0),\dot
x_{1}(0),\ldots,x_{m}(0),\dot x_{m}(0))\in U$, with $U$ being an open dense
subset of $({\rm T}\mathbb{R}^n)^m$. 

There exists no characterisation for systems of SODEs of the form (\ref{SODE})
admitting a superposition rule. In spite of this, there exists a special class
of such systems, the so-called {\it SODE Lie systems} \cite{CLR08}, accepting
such a property. Even though this fact has been broadly used in the literature,
it has been proved very recently \cite{CL09SRicc}. We next furnish the
definition of the SODE Lie system  along with a proof for showing that every
SODE Lie system admits a superposition rule. In addition, some remarks on the
interest of this notion and its main properties are discussed.

\begin{definition}\label{DefSODE} We say that the system of second-order
differential equations  (\ref{SODE}) is a SODE Lie system if the system of
first-order differential equations 
\begin{equation}\label{SODEFOrder}
\left\{
\begin{aligned}
\dot x^i&=v^i,\\
\dot v^i&=F^i(t,x,v),
\end{aligned}\right.\qquad i=1,\ldots,n,
\end{equation}
obtained from (\ref{SODE}) by defining the new variables $v^i=\dot x^i$, with
$i=1,\ldots,n$, is a Lie system.
\end{definition}

\begin{proposition}\label{SR}  Every SODE Lie system (\ref{SODE}) admits a
superposition rule $\Psi:({\rm
T}\mathbb{R}^n)^m\times\mathbb{R}^{2n}\rightarrow\mathbb{R}^n$ of the form
$\Psi=\pi\circ\Phi$, where $\Phi:({\rm
T}\mathbb{R}^n)^m\times\mathbb{R}^{2n}\rightarrow {\rm T}\mathbb{R}^n$ is a
superposition rule for the system (\ref{SODEFOrder}) and $\pi:{\rm
T}\mathbb{R}^n\rightarrow\mathbb{R}^n$ is the projection associated with the
tangent bundle ${\rm T}\mathbb{R}^n$. 
\end{proposition}
\begin{proof}
Each SODE Lie system (\ref{SODE}) is associated with a first-order system of
differential equations (\ref{FOrder}) admitting a superposition rule $\Phi:({\rm
T}\mathbb{R}^{n})^m\times \mathbb{R}^{2n}\rightarrow {\rm T}\mathbb{R}^n$. This
allows us to describe the general solution $(x(t),v(t))$ of system
(\ref{SODEFOrder}) in terms of a generic set $(x_{a}(t),v_{a}(t))$, with
$a=1,\ldots,m$, of particular solutions and a set of $2n$ constants, i.e.
\begin{equation}\label{SupRel2}
\begin{aligned}
(x(t), v(t))&=\Phi\left(x_{1}(t),\ldots, x_{m}(t),v_{1}(t),\ldots,
v_{m}(t);k_1,\ldots,k_{2n}\right)\\
\end{aligned}.
\end{equation}
Each solution, $x_p(t)$, of the second-order system (\ref{SODE}) corresponds to
one and only one solution $(x_p(t),v_p(t))$ of the system of first-order
differential equations (\ref{SODEFOrder}) and vice versa. Furthermore, since one
has that $(x_p(t),v_p(t))=(x_p(t),\dot x_p(t))$, it turns out that the general
solution $x(t)$ of (\ref{SODE}) can be written as
\begin{equation}\label{SupRel4}
x(t)=\pi\circ\Phi\left(x_{1}(t),\ldots, x_{m}(t),\dot x_{1}(t),\ldots, \dot
x_{m}(t);k_1,\ldots,k_{2n}\right),
\end{equation}
in terms of a generic family $x_{a}(t)$, with $a=1,\ldots,n$, of particular
solutions of (\ref{SODE}). That is, the map $\Psi=\pi\circ\Phi$ is a
superposition rule for the system of SODEs (\ref{SODE}).
\end{proof}

Since every autonomous system is related to a one-dimensional Vessiot--Guldberg
Lie algebra \cite{CGL08}, a corollary follows immediately.

\begin{corollary} Every autonomous system of second-order differential equations
of the form $\ddot x^i=F^i(x,\dot x)$, with $i=1,\ldots,n$, admits a
superposition rule. 
\end{corollary}

The above result is, in practice, almost useless. Actually, the superposition
rule ensured by Proposition \ref{SR} relies on the derivation of a superposition rule for
an autonomous first-order system of differential equations. Applying the method
sketched in Section \ref{nSR}, it is found that determining this superposition
rule implies working out all the integral curves of a vector field on $({\rm
T}\mathbb{R}^n)^2$. Although the solution of this problem is known to exist, its
explicit description can be as difficult as solving the initial system (indeed,
this is usually the case). Consequently, deriving explicitly a superposition
rule for the above autonomous system frequently depends on the search of an
alternative superposition rule for the associated first-order system.

Many superposition rules for second-order differential equations do not present
an explicit dependence on the derivatives of the particular solutions. Consider,
for instance,  either the linear superposition rule (\ref{LinearSup}) for the
equation $\ddot x=a(t)x$,  or the affine one,
$$x(t)=k_1(x_{1}(t)-x_{2}(t))+k_2(x_{2}(t)-x_{3}(t))+x_{3}(t),$$
for $\ddot x=a(t)x+b(t)$.  Such superposition rules are called {\it velocity
free superposition rules} or even {\it free superposition rules}. The conditions
ensuring the existence of such superposition rules is an interesting open
problem. Let us provide a brief analysis about the existence of such
superposition rules.

\begin{proposition}
Every system of SODEs (\ref{SODE}) admitting a free superposition rule is a SODE
Lie system.
\end{proposition}
\begin{proof}

Suppose that system (\ref{SODE}) admits a superposition rule of the special form
\begin{equation}\label{FreeSuperRule}
\begin{aligned}
x^i&=\Phi_x^i(x_{1},\ldots, x_{m};k_1,\ldots,k_{2n}),
\end{aligned}\qquad i=1,\ldots,n.
\end{equation}
In such a case, the general solution, $x(t)$, of the system could be expressed
as
\begin{equation}\label{freeSup1}
x^i(t)=\Phi^i_x(x_{1}(t),\ldots,x_{m}(t);k_1,\ldots,k_{2n}), \qquad
i=1,\ldots,n.
\end{equation}
Define $p(t)=(x_{1}(t),\ldots,x_{m}(t),\dot x_{1}(t),\ldots,\dot x_{m}(t))$ and
$v^i=\dot x^i$ for $i=1,\ldots,n$. Take the time derivative in the above
expression. This yields
\begin{equation}\label{freeSup2}
v^i(t)=\dot
x^i(t)=\sum_{a=1}^m\sum_{j=1}^n\left(v_{a}^j(t)\frac{\partial\Phi_x^i}{\partial
x^j_{a}}(p(t))\right),\qquad i=1,\ldots,n,
\end{equation}
where we have used that $\partial\Phi^i_x/\partial v^j_{a}=0$, for
$i,j=1,\ldots,n$, and $a=1,\ldots,m$. Consequently, there exists a function
$$\Phi^i_v(x_1,\ldots,x_m,v_1,\ldots,v_m)=\sum_{a=1}^m\sum_{j=1}^n\left(v_{a}
^j\frac{\partial\Phi_x^i}{\partial x^j_{a}}\right),\qquad i=1,\ldots,n,
$$
such that 
\begin{equation*}
\left\{
\begin{aligned}
x^i(t)&=\Phi_x^i(x_{1}(t),\ldots, x_{m}(t);k_1,\ldots,k_{2n}),\\
v^i(t)&=\Phi_v^{i}(x_{1}(t),\ldots, x_{m}(t),v_{1}(t),\ldots,
v_{m}(t);k_1,\ldots,k_{2n}),
\end{aligned}\right.\qquad i=1,\ldots,n.
\end{equation*}
Therefore, system (\ref{FOrder}) admits a  superposition rule and (\ref{SODE})
becomes a SODE Lie system.
\end{proof}

Apart from the SODE Lie system notion, there exists another method to study
certain second-order differential equations admitting a regular Lagrangian, like
Caldirola--Kanai oscillators or Milne--Pinney equations \cite{CLR08,Ru10}.
Although this method cannot be used for studying all systems of second-order
differential equations, it provides some additional information that cannot be
derived by means of SODE Lie systems when it applies, e.g. information on the
$t$-dependent constants of the motion of the system \cite{Ru10}.

\section{Superposition rules for PDEs}
The geometrical formulation of the theory of Lie systems enables us to extend
the Lie system notion to partial differential equations. Here, we briefly
analyse this generalisation and its properties \cite{CGM07,Ram02Th}.

Consider the system of first-order PDEs of the form
\begin{equation}
\pd{x^i}{t^a}=X^i_a(t,x),\qquad\qquad  x\in{\mathbb{R}}^n,\
t=(t^1,\ldots,t^s)\in {\mathbb{R}}^s\,,\label{lpdesys}
\end{equation}
whose solutions are maps $x(t):{\mathbb{R}}^s\to {\mathbb{R}}^n$. When $s=1$,
the above system of PDEs becomes the system of ordinary differential equations
(\ref{equation1}). The main difference between these systems is that for $s>1$
there exists, in general, no solution with a given initial condition. For a
better understanding of this problem, let us put
(\ref{lpdesys}) in a more general and geometric framework.

Let $P_{\mathbb{R}^n}^s$ be the trivial fibre  bundle 
$$P_{\mathbb{R}^n}^s={\mathbb{R}}^s\times \mathbb{R}^n\to {\mathbb{R}}^s\,.$$
A connection $\bar Y$ on this bundle is a horizontal distribution
over ${\rm T}P_{\mathbb{R}^n}^s$. i.e. a $s$-dimensional distribution
transversal to
the fibres. This distribution may be determined by the horizontal lifts of the
vector fields $\partial/\partial t^a$ on ${\mathbb{R}}^s$, i.e.
$$\overline{X}_a(t,x)=\pd{}{t^a}+X_a(t,x),$$
where
$$X_a(t,x)=\sum_{i=1}^nX^i_a(t,x)\pd{}{x^i}\,.
$$
The solutions of system (\ref{lpdesys}) can be identified  with integral
submanifolds of the distribution $\overline X$,
$$(t,X_a(t,x))\,,\qquad t\in {\mathbb{R}}^s,\,\,x\in\mathbb{R}^n\,.
$$
It is now clear that there is a (obviously unique) solution of
(\ref{lpdesys}) for every initial data if and only if the
distribution $\overline Y$ is integrable, i.e. the connection  has a
trivial curvature. This means that
$$[\overline X_a,\overline X_b]=\sum_{c=1}^r f_{abc} \
\overline X_c$$ for some functions $ f_{abc}$ in $P_{\mathbb{R}^n}^s$. But the
commutators $[\overline X_a,\overline X_b]$ are clearly vertical, while
$\overline X_c$ are linearly independent horizontal vector fields, so
$f_{abc}=0$, which yields the integrability condition in the form of the system
of equations $ [\overline X_a,\overline X_b]=0$, i.e. in
local coordinates,
\begin{equation}
\pd{X^i_b}{t^a}(t,x)-\pd{X^i_a}{t^b}(t,x)
+\sum_{j=1}^n\left(X^j_a(t,x)\pd{X^i_b}{x^j}(t,x)-X^j_b(t,x)\pd{X^i_a}{x^j}(t,x)
\right)=0\,.\label{integcond}
\end{equation}
Let us assume now that we analyse a system of  first-order PDEs of the form
(\ref{lpdesys}) that satisfies integrability
conditions (\ref{integcond}). Then, for a given initial value, there exists a
unique solution of system (\ref{lpdesys}). Furthermore, it is immediate that the
geometrical interpretation for superposition rules for first-order described in
Section (\ref{GeoAppr}) can be generalised straightforwardly to the case of
PDEs. In consequence, Proposition \ref{LieConnect} takes now the following form.

\begin{proposition}
Giving a superposition rule for system (\ref{lpdesys}) obeying integrability
condition (\ref{integcond}) is equivalent to giving a connection on the
bundle ${\rm pr}:\mathbb{R}^{n(m+1)}\to \mathbb{R}^{nm}$
with a zero curvature such that the family of vector fields $\{(X_a)_t\mid
t\in\mathbb{R}^s,a=1,\ldots,s \}$ are
horizontal.
\end{proposition}
Also the proof of Lie Theorem remains unchanged. Therefore, we get
the following analogous of Lie Theorem for PDEs.
\begin{theorem}
The system (\ref{lpdesys}) of PDEs defined on ${\mathbb{R}^n}$ and
satisfying the integrability condition (\ref{integcond}) admits a
superposition rule if and only if the vector fields $\{(X_a)_t\}$ on
${\mathbb{R}^n}$ depending on the parameter $t\in{\mathbb{R}}^s$, can be written
in the form
\begin{equation}\label{yy}(X_a)_t=\sum_{\alpha=1}^ru_a^\alpha(t)X_\alpha\,,
\qquad
a=1,\ldots s\,,
\end{equation}
where the vector fields $X_\alpha$ span a finite-dimensional real Lie
algebra.

\end{theorem}
Note that the integrability condition for $Y_a(t,x)$ of
the form (\ref{yy}) can be written as
$$\sum_{\za,\zb,\zg=1}^r\left[(u^\zg_b)'(t)-(u^\zg_a)'(t)+
u^\za_a(t)u^\zb_b(t)c^\zg_{\za\zb}\right]X_\zg=0.$$

We now turn to illustrate the above results by means of a particular example.
Consider the following system of partial differential equations on
$\mathbb{R}^2$ associated with the $SL(2,\mathbb{R})$-action on
$\bar{\mathbb{R}}$,
\begin{equation}
 \begin{aligned}
 u_x&=&a(x,y)u^2+b(x,y)u+c(x,y)\,,\\
u_y&=&d(x,y)u^2+e(x,y)u+f(x,y)\,. 
 \end{aligned}
\end{equation}
 This equation can be written in the form of a `total differential
equation'
$$(a(x,y)u^2+b(x,y)u+c(x,y))\xd x+(d(x,y)u^2+e(x,y)u+f(x,y))\xd y=\xd u\,.$$ The
integrability condition only states that the
one-form
$$\zw=(a(x,y)u^2+b(x,y)u+c(x,y))\xd x+(d(x,y)u^2+e(x,y)u+f(x,y))\xd
y$$ is closed for an arbitrary function $u=u(x,y)$. If this is the
case, there is a unique solution with the initial condition
$u(x_0,y_0)=u_0$ and there is a superposition rule giving a
general solution as a function of three independent solutions
exactly as in the case of Riccati equations:
$$u=\frac{(u_{(1)}-u_{(3)})u_{(2)}k+u_{(1)}(u_{(3)}-u_{(2)})}
{(u_{(1)}-u_{(3)})k+(u_{(3)}-u_{(2)})}\,.
$$

\chapter{SODE Lie systems}
We already pointed out that the theory of Lie systems is mainly dedicated to the
analysis of systems of first-order differential equations. In spite of this,
such a theory can also be applied to studying a variety of systems of
second-order differential equations. This can be done in several ways that rely,
as a last resort, on using some kind of transformation to convert systems of
second-order differential equations into first-order ones
\cite{CLR08,CLRan08,CC87,GGL08,WintSecond}. A class of such systems that can be
investigated by means of these techniques are the referred to as SODE Lie
systems, which were theoretically analysed in Section \ref{SODEsSystems}. In
this chapter, we focus on analysing several instances of SODE Lie systems in
order to derive $t$-independent constants of the motion, exact solutions,
superposition rules, and other properties. This allows us not only to study the
mathematical properties of such systems, but also to provide tools to analyse
the diverse physical or control systems modelled through such equations.

Among the above applications to SODEs, one must be emphasised: the use of the
referred to as {\it mixed superposition rules}. This recently described notion
enables us to express the general solution of SODE Lie systems in terms of
particular solutions of the same, or other, SODE Lie systems. In this way, this
new concept can be employed to analyse the properties of the general solutions
of certain SODEs appearing in the Physics and mathematical literature
\cite{HL02,RR80}. As a consequence of such an analysis, new results can be
obtained and other known ones will be recovered, in a systematic way, which will
enhance their understanding.

The following section is dedicated to the application of the theory of Lie
systems to SODE Lie systems in order to review, through a geometrical unifying
approach, some results previously obtained in the literature by means of {\it ad
hoc} methods and to provide new ones. The whole chapter can be divided into two
parts: The first one is devoted to the application of the geometric theory of
Lie systems for deriving superposition rules, constants of the motion and exact
solutions for various SODE Lie systems. More specifically, we study
$t$-dependent harmonic oscillators, generalised Ermakov systems and
Milne--Pinney equations, providing a new superposition rule for the latter. The
second part is concerned with the study and application of mixed superposition
rules. 

\section{The  harmonic oscillator with {\it t}-dependent frequency}
Perhaps, the one-dimensional $t$-dependent frequency  harmonic oscillator is the
most simple SODE which allows us to illustrate the application of the SODE Lie
system notion. Let us make use of this fact to show, clearly, how this notion
applies and to analyse thoroughly the properties of such a system. 

The equation of the motion for a one-dimensional harmonic oscillator with
$t$-dependent
frequency $\omega(t)$ takes the form $\ddot x=-\omega^2(t)x$. In view of
Definition \ref{DefSODE}, this equation is a SODE Lie system if and only if the
system of first-order differential equations
\begin{equation}\label{1dimho}
\left\{
\begin{aligned}
\dot x&=v,\\
\dot v&=-\omega^2(t)x,
\end{aligned}
\right.
\end{equation}
is a Lie system. This feature depends on the properties of the $t$-dependent
vector field over ${\rm T}\mathbb{R}$ given by
\[X(t,x,v)=v\frac{\partial}{\partial x} -\omega^2(t) x\,
\frac{\partial}{\partial v},
\] 
which describes the integral curves of system (\ref{1dimho}). It is immediate
that 
\begin{equation}\label{deco}
X_t=X_1+ \omega^2(t)X_3, 
\end{equation}
where $X_1$ and $X_3$ are the vector fields 
\[
X_1=v\, \frac{\partial}{\partial x},\qquad X_3= -x\frac{\partial}{\partial v}.
\]
These vector fields obey the commutation relations 
\begin{equation}\label{CR}
[X_1,X_3]=2\, X_2\,, \quad [X_2,X_3]= X_3 \,,\quad [X_1,X_2]=X_1,
\end{equation}
with $X_2$ being the vector field on ${\rm T}\mathbb{R}$ given by
\[X_2=\frac 12 \left(x\frac{\partial}{\partial x}-v\frac{\partial}{\partial
v}\right).
\]

According to the commutation relations (\ref{CR}) and decomposition
(\ref{deco}), it follows that $X_t$ defines a Lie system associated with a
Vessiot--Guldberg Lie algebra $V=\langle X_1,X_2,X_3\rangle$. Hence,
one-dimensional harmonic oscillators with a $t$-dependent frequency are SODE Lie
systems. 

Determining the general solution of every SODE Lie system reduces to working out
the solution of an equation on a Lie group. Unsurprisingly, since the general
solution of a SODE Lie system is straightforwardly related to the solution of a
Lie system whose solution can be obtained from a equation in a Lie group. Let us
illustrate our claim in detail through the example of harmonic oscillators. 

Since system (\ref{1dimho}) is a Lie system, its general solution can be worked
out by means of the solution of an equation on a certain Lie group (see Section
\ref{FNLS}). Recall that as the elements of $V$ are complete, there exists a Lie
group action $\Phi_L:G\times{\rm T}\mathbb{R}\rightarrow{\rm T}\mathbb{R}$ whose
fundamental vector fields are exactly those corresponding to $V$. It is easy to
check that this action can be chosen to be $\Phi_L:SL(2,\mathbb{R})\times {\rm
T}\mathbb{R}\rightarrow{\rm T}\mathbb{R}$, with
\[
\Phi_L\left(\left(\begin{array}{cc}
\alpha\,&\,\beta\\ \gamma\,&\delta\,
\end{array}\right),\left(\begin{array}{c}x\\v\end{array}
\right)\right)=\left(\begin{array}{cc}
\alpha\,&\,\beta\\ \gamma\,&\delta\,
\end{array}\right)\left(\begin{array}{c}x\\v\end{array}\right)=
\left(\begin{array}{c}
\alpha x+\beta v\\
\gamma x+\delta v\\
\end{array}\right).
\]  
Indeed, if we take the basis 
\begin{equation}\label{thebasis}
{\rm a}_1=\left(\begin{array}{cc}
0&-1\\
0&0
\end{array}\right),\quad
{\rm a}_2=\frac 12\left(\begin{array}{cc}
-1&0\\
0&1
\end{array}\right),\quad
{\rm a}_3=\left(\begin{array}{cc}
0&0\\
1&0
\end{array}\right),
\end{equation}
of the Lie algebra of $2\times 2$ traceless matrices (the usual representation
of the Lie algebra $\mathfrak{sl}(2,\mathbb{R})$), its elements satisfy the same
commutation relations as the vector fields, $X_1,X_2,X_3$. Furthermore, it can
be easily verified that the vector fields $X_1,X_2$ and $X_3$ are the
fundamental vector fields associated with the matrices, ${\rm a}_1,{\rm
a}_2,{\rm a}_3,$  according to our convention (\ref{convention}). 

Once the action $\Phi_L$ is determined, it enables us to write the general
solution $(x(t),v(t))$ of system (\ref{1dimho}) in the form
\begin{equation}\label{sol}
\left(\begin{aligned}
x(t)\\v(t) 
\end{aligned}\right)=\Phi_L\left(g(t),\left(\begin{aligned}
x_0\\v_0 
\end{aligned}\right)\right),\qquad {\rm with}\,\, \left(\begin{aligned}
x_0\\v_0 
\end{aligned}\right)\in {\rm T}\mathbb{R},
\end{equation}
where $g(t)$ is the solution of the Cauchy problem 
$$
R_{g^{-1}*}\dot g=-\sum_{\alpha=1}^3b_\alpha(t){\rm a}_\alpha,\qquad g(0)=e,
$$
on $SL(2,\mathbb{R})$. This immediately gives us the general solution, $x(t)$,
of the equation (\ref{1dimho}) from expression (\ref{sol}). Moreover, this
process is easily generalised to every SODE Lie system.

Apart from the above Lie group approach, the SODE Lie system notion furnishes us
with a second approach to investigate one-dimensional $t$-dependent frequency
harmonic oscillators. This is based on determining a superposition rule for the
Lie system (\ref{1dimho}).

Recall that a superposition rule for a Lie system can be worked out by means of
a set of first-integrals for certain diagonal prolongations of the vector fields
of an associated Vessiot--Guldberg Lie algebra $V$. As it was discussed in
Section \ref{nSR}, the way to obtain these first-integrals requires to determine
the minimal integer $m$ such that the prolongations to $\mathbb{R}^{nm}$ of the
elements of a basis of the Lie algebra $V$ become linearly independent at a
generic point. This yields that $\dim\,V\leq m\cdot n$. Additionally, if we
consider the diagonal prolongations of such a basis to $\mathbb{R}^{n(m+1)}$,
these elements are again linearly independent at a generic point and a family of
$m\cdot n-r$ first-integrals appears. 
These first-integrals allow us to determine a superposition rule. 

We next illustrate the above process by means of the study of harmonic
oscillators. In addition, we analyse in parallel the problem of finding
$t$-independent constants of the motion for systems made of some copies of the
initial system. This problem will be proved to be related to the above process
and, in addition, will permit us to show interesting properties about harmonic
oscillators.

Consider two copies of the same one-dimensional harmonic oscillator, i.e.  
\begin{equation}
\left\{\begin{array}{rcl}\ddot x_1&=& -\omega^2(t) x_1,\cr\ddot x_2&=&
-\omega^2(t)
  x_2.\end{array}\right.\label{2dimho}
\end{equation}
This system of SODEs, which corresponds to a two-dimensional isotropic harmonic
oscillator with a $t$-dependent frequency $\omega(t)$, is 
related to the following system of first-order differential equations 
\begin{equation}\label{2dimhob}
\left\{
\begin{array}{rcl}
\dot x_1&=&v_1,\cr 
\dot x_2&=&v_2,\cr 
\dot v_1&=&-\omega^2(t) x_1,\cr 
\dot v_2&=& -\omega^2(t) x_2.
\end{array} \right.
\end{equation}
Its solutions are the integral curves of the $t$-dependent vector field 
\[X^{2d}_t=v_1\frac{\partial}{\partial x_1}+
v_2\frac{\partial}{\partial x_2} -\omega^2(t)x_1\, \frac{\partial}{\partial v_1}
-\omega^2(t) x_2\, \frac{\partial}{\partial v_2}\ ,
\]
which is a linear combination 
\begin{equation}\label{dec2d}
X^{2d}_t=X_1^{2d}+\omega^2(t) X_3^{2d},
\end{equation}
with $X^{2d}_1$ and $X^{2d}_3$
being the vector fields 
\[
X^{2d}_1=v_1 \frac{\partial}{\partial x_1}+v_2 \frac{\partial}{\partial x_2}
\,,\qquad X^{2d}_3= -x_1\frac{\partial}{\partial v_1}-
x_2\frac{\partial}{\partial v_2}\,,
\]
satisfying the commutation relations
\begin{equation}\label{CR2} 
[X^{2d}_1,X^{2d}_3]=2\, X^{2d}_2\,, \quad [X^{2d}_2,X^{2d}_3]=X^{2d}_3 \,,\quad
[X^{2d}_1,X^{2d}_2]=X^{2d}_1\,,
\end{equation}
where $X_2$ reads
\[X^{2d}_2=\frac 12 \left(x_1\frac{\partial}{ \partial
x_1}+x_2\frac{\partial}{\partial x_2}-v_1\frac{\partial}{\partial v_1}
-v_2\,
\frac{\partial}{\partial v_2}
\right) \,.
\]
The previous decomposition of the $t$-dependent vector field $X^{2d}$ has been
obtained by considering the new vector fields, $X^{2d}_1,X^{2d}_2,X^{2d}_3,$ to
be diagonal prolongations to ${\rm T}\mathbb{R}^2$ of the vector fields,
$X_1,X_2,X_3$. In this way, we get that the commutation relations (\ref{CR2})
are the same as (\ref{CR}) and, in view of decomposition (\ref{dec2d}), this
$t$-dependent vector field defines a Lie system related to a Lie algebra of
vector fields isomorphic to $\mathfrak{sl}(2,\mathbb{R})$. 

The distribution associated with the Lie system $X^{2d}_t$, i.e. 
$$\mathcal{V}^{2d}_p= \langle (X^{2d}_1)_p,(X^{2d}_2)_p,(X^{2d}_3)_p\rangle,
\qquad p\in {\rm T}\mathbb{R}^2,$$
has rank lower or equal to the dimension of the Lie algebra $V$. More
specifically, it has rank three in an open dense of subset ${\rm
T}\mathbb{R}^2$. Hence, there exists a local non-trivial first-integral common
to all the vector fields of the above distribution. Furthermore, this
first-integral is a $t$-independent constant of the motion of system
(\ref{2dimhob}). Let us analyse this statement more carefully. Given a constant
of the motion $F:(x_1,v_1,x_2,v_2)\in{\rm T}\mathbb{R}^2\mapsto
F(x_1,v_1,x_2,v_2)\in \mathbb{R}$ of system (\ref{2dimhob}), it follows that
$$
\frac{dF}{dt}(p(t))=\sum_{j=1}^2\left(\frac{dx^i}{dt}(t)\frac{\partial
F}{\partial x^i}(p(t))+\frac{dv^i}{dt}(t)\frac{\partial F}{\partial
v^i}(p(t))\right)=X^{2d}_tI(p(t))=0,
$$
where $p(t)=(x_1(t),v_1(t),x_2(t),v_2(t))$. If $F$ is a first-integral for the
system (\ref{2dimhob}), whatever $\omega(t)$ is, then $F$ must be a
first-integral of the vector fields of $X^{2d}_1,X^{2d}_3$ and, therefore, of
$X^{2d}_2$. 

Consequently, there exists, at least locally, a function $F$ that is a constant
of the motion for every system (\ref{2dimhob}) and such that $dF$ is incident to
the distribution generated by the, $X^{2d}_1,X^{2d}_2,X^{2d}_3$, i.e. 
$dF(X^{2d}_1)=dF(X^{2d}_2)=dF(X^{2d}_3)=0$  in a certain dense open subset $U$
of ${\rm T}\mathbb{R}^2$. 

Since $X^{2d}_3F=0$, there is a function $\bar F(\xi,x_1,x_2)$ such that
$F(x_1,x_2,v_1,v_2)=\bar F(\xi,x_1,x_2)$, with
$\xi=x_1v_2-x_2v_1$. Next, in view of condition
$X^{2d}_1\bar F=0$,  we have
\[v_1\, \frac{\partial \bar F}{\partial x_1}+v_2\, \frac{\partial \bar
F}{\partial x_2}=0\,
\]
and there exists a function $\widehat F(\xi)$ such that $\bar
F(\xi,x_1,x_2)=\widehat F(\xi)$. As $2\, X^{2d}_2=[X^{2d}_1,X^{2d}_3]$, the
conditions $X^{2d}_1\widehat F=X^{2d}_3\widehat F=0$ imply $X^{2d}_2\widehat
F=0$ and
hence $F(x_1,x_2,v_1,v_2)=x_1v_2-x_2v_1$ is a first-integral which physically
corresponds to the angular momentum. Additionally, this first-integral allows us
to solve the second-order differential equation $\ddot x=-\omega^2(t)x$ by means
of a particular solution. Actually, if $x_1(t)$ is a non-vanishing solution of
this equation, every other particular solution $x_2(t)$ gives rise to a
particular solution $(x_1(t),v_1(t),x_2(t),v_2(t))$ of system (\ref{2dimhob}).
As the first-integral $F$ is constant along this particular solution, we have
that $x_2(t)$ obeys the equation
\[x_1(t)\, \frac{dx_2}{dt} =k+\dot x_1(t)x_2\,,
\]
whose solution reads
\begin{equation}
x_2(t)=k' x_1(t)+k\, x_1(t)\int^t\frac{d\zeta}{x_1^2(\zeta)}\,,\label{redunasol}
\end{equation}
what gives us the general solution to the $t$-dependent frequency harmonic
oscillator in terms of a particular solution.

In order to look for a superposition rule, we must consider a system made of
some copies of  (\ref{1dimho}) and obtain at least as many $t$-independent
constants of the motion as the dimension of the initial manifold. Also, it must
be possible to obtain the variables of the initial manifold explicitly in terms
of the other variables and such constants. Recall that the number $m$ of
particular solutions to obtain a superposition rule satisfies that the diagonal
prolongations of the vector fields $X_1, X_2$ and $X_3$ to $\mathbb{R}^{nm}$ are
linearly independent in a generic point. 

In the case of two copies of the $t$-dependent harmonic oscillator, the
condition on the prolongations of the vector fields, $X_1,X_2,X_3$, that is,
$\lambda_1\, X^{2d}_1+\lambda_2\, X^{2d}_2+\lambda_3\, X^{2d}_3=0$, implies that
$\lambda_1=\lambda_2=\lambda_3=0$. Therefore, the one-dimensional oscillator
admits a superposition rule
involving two particular solution and, in view of our previous results, we need
to study three copies of the $t$-dependent harmonic oscillator (\ref{1dimho}) so
as to obtain a superposition
rule. Consider therefore the system of first-order ordinary differential
equations
\begin{equation}
\left\{
\begin{array}{rcl}\dot x_1&=&
v_1,\cr \dot v_1&=&
-\omega^2(t) x_1,\cr 
\dot x_2&=&
v_2,\cr \dot v_2&=& -\omega^2(t) x_2,\cr
\dot x&=&
v,\cr \dot v&=& -\omega^2(t) x,\cr
\end{array} \right.\label{3dimhob}
\end{equation}
whose solutions are the integral curves for the
 $t$-dependent vector  vector field 
\[X^{3d}_t=v_1\frac{\partial}{\partial x_1}+
v_2\frac{\partial}{\partial x_2}+v\frac{\partial}{\partial x}  -\omega^2(t)
x_1\, \frac{\partial}{\partial v_1} -\omega^2(t) x_2\, \frac{\partial}{\partial
v_2}-\omega^2(t) x\, \frac{\partial}{\partial v}\ ,
\]
which is a linear combination, $X^{3d}_t=X^{3d}_1+ \omega^2(t) X^{3d}_3$, with
$X^{3d}_1$ and $X^{3d}_3$
being the vector fields 
\[
X^{3d}_1=v_1 \frac{\partial}{\partial x_1}+v_2 \frac{\partial}{\partial x_2}+v
\frac{\partial}{\partial x}\,,
\qquad X^{3d}_3= -x_1\frac{\partial}{\partial v_1}-x_2\frac{\partial}{\partial
v_2}- x\frac{\partial}{\partial v},
\]
obeying the commutation relations
\[[X^{3d}_1,X^{3d}_3]=2\, X^{3d}_2\,, \quad [X^{3d}_2,X^{3d}_3]=X^{3d}_3
\,,\quad [X^{3d}_1,X^{3d}_2]=X^{3d}_1\,,
\]
where the vector field $X^{3d}_2$ is defined by
\[X^{3d}_2=\frac 12 \left(x_1\frac{\partial}{ \partial
x_1}+x_2\frac{\partial}{\partial x_2}+x\frac{\partial}{\partial
x}-v_1\frac{\partial}{\partial v_1}
-v_2\,
\frac{\partial}{\partial v_2}
-v\,
\frac{\partial}{\partial v}
\right) \,.
\]
We can determine the first-integrals $F$ for these three vector fields as
solutions of the system of PDEs $X^{3d}_1F=X^{3d}_3F=0$, because
$2\,X^{3d}_2=[X^{3d}_1,X^{3d}_3]$ and the previous relations automatically imply
the condition $X^{3d}_2F=0$. This last condition yields that there exists a
function $\bar F:{\mathbb{R}}^5\to {\mathbb{R}}^2$ such that
$F(x_1,x_2,x,v_1,v_2,v)=\bar F (\xi_1
,\xi_2,x_1,x_2,x)$ with $\xi_1(x_1,x_2,x,v_1,v_2,v)=xv_1-x_1v$ and
$\xi_2(x_1,x_2,x,v_1,v_2,v)=xv_2-x_2v$. In view of this, the condition
$X^{3d}_1F=0$ transforms into 
\[v_1\frac{\partial \bar F}{\partial x_1}+v_2\frac{\partial \bar F}{\partial
x_2}+v\frac{\partial \bar F}{\partial x}=0\,,
\] 
i.e. the functions $\xi_1$ and $\xi_2$ are first-integrals (Of course,
$\xi=x_1v_2-x_2v_1$ is
also a first-integral). They produce a superposition
rule, because from
\[\left\{\begin{array}{crl} xv_2-x_2v&=&k_1,\cr
x_1v-v_1x&=&k_2,\end{array}\right.
\] 
 we get the expected superposition rule for two solutions
\[x=c_1\, x_1+c_2\, x_2\,,\qquad v=c_1\, v_1+c_2\, v_2\,,\qquad
c_i=\frac{k_i}{k}, \qquad  k=x_1v_2-x_2v_1\,.
\]

\section{Generalised Ermakov system}
Let us now turn to study the so-called generalised Ermakov system, i.e.
\begin{eqnarray}\label{SOrder}
\left\{\begin{aligned}
\ddot{x}=\frac{1}{x^3}f(y/x)-\omega^2(t)x,\cr
\ddot{y}=\frac{1}{y^3}g(y/x)-\omega^2(t)y,
\end{aligned}\right.
\end{eqnarray}
which has been broadly studied in \cite{GL94a,R81,RR79a,RR79b,RR80,WS81,SaCa82}.
Although this system is, in general, more complex than the standard Ermakov
system, which will be discussed later, its analysis is easier from our point of
view and it is therefore studied now. More exactly, our aim is to recover by
means of our methods its known constant of motion, which is used next to study
the Milne--Pinney equation and to obtain a superposition rule. 

For the sake of simplicity, let us consider the generalised Ermakov system on
$\mathbb{R}_+^2$. This system can be written as a system of first-order
differential
equations
\begin{equation}\label{FOrder}
\left\{\begin{aligned}
\dot x&=v_x,\\
\dot y&=v_y,\\
\dot v_x&=-\omega^2(t)x+\frac 1 {x^3} f( y/x),\\
\dot v_y&=-\omega^2(t)y+\frac 1 {y^3} g(y/x),
\end{aligned}\right.
\end{equation}
in ${\rm T}\mathbb{R}^2_+$ by introducing the new variables $v_x=\dot x$ and
$v_y=\dot y$. Therefore, we can study its solutions as the integral curves for a
$t$-dependent vector field $X_t$ on ${\rm T}\mathbb{R}_+^2$ of the form
\[X_t=v_x\,\frac{\partial}{\partial x}+v_y\,\frac{\partial}{\partial
y}+\left(-\omega^2(t)x+\frac 1 {x^3} f( y/
  x)\right)\frac{\partial}{\partial v_x}+\left(-\omega^2(t)y+\frac 1 {y^3} g(y/
    x)\right)\frac{\partial}{\partial v_y}\,,
\]
which can be written as a linear combination 
\[X_t=N_1+\omega^2(t)\, N_3,
\]
where $N_1$ and $N_3$ are the vector fields 
\[
N_1=v_x\frac{\partial}{\partial x}+v_y\frac{\partial}{\partial y}+
\frac{1}{x^3}f(y/x)\frac{\partial}{\partial v_x}+
\frac{1}{y^3}g( y/x)\frac{\partial}{\partial v_y}\,\quad N_3=-x\frac{\partial
}{\partial v_x}-y\frac{\partial }{\partial v_y}.
\]
Note that these vector fields generate a three-dimensional real Lie algebra
with the third generator 
\[N_2=\frac 12\left(x\frac{\partial}{\partial
    x}+y\frac{\partial}{\partial
    y}-v_x\frac{\partial}{\partial v_x}-v_y\frac{\partial}{\partial
v_y}\right)\,.\]
In fact, as  
\[
[N_1,N_3]=2N_2, \quad [N_1,N_2]=N_1, \quad  [N_2,N_3]=N_3\,,
\]
they generate a Lie algebra of vector fields isomorphic to  
$\mathfrak{sl}(2,\mathbb{R})$ and thus the
generalised Ermakov system is a SODE Lie system. 

As Lie system (\ref{FOrder}) is associated with an integrable distribution of
rank three
 in a generic point of a four-dimensional manifold, there exists, at least
locally,  a first-integral, $F:{\rm T}\mathbb{R}_+^2\rightarrow \mathbb{R}$,
for any $\omega^2(t)$. Such a first-integral $F$  satisfies 
$N_iF=0$ for $i=1, 2, 3$, but as $[N_1,N_3]=2N_2$ it is sufficient to impose
$N_1F=N_3F=0$ to get $N_2F=0$. Then, if $N_3F=0$ we have
\[
x\,\frac{\partial F}{\partial v_x}+y\, \frac{\partial F}{\partial v_y}=0\,,
\]
and  the associated system of 
 characteristics is
\[
\frac{dx}{0}=\frac{dy}{0}=\frac{dv_x}{x}=\frac{dv_y}{y}\,.
\]
In view of this, we conclude that there exists a function $\bar
F:\mathbb{R}^3\rightarrow \mathbb{R}$ such that
$F(x,y,v_x,v_y) =\bar F(x,y,\xi=xv_y-yv_x)$ and, taking this into account, the
condition   $N_1F=0$ reads
\[
v_x\frac{\partial\bar F}{\partial x}+v_y\frac{\partial\bar F}{\partial
  y}+\left(-\frac{y}{x^3}f({y}/{x})+
\frac{x}{y^3}g({y}/{x})\right)\frac{\partial \bar F}{\partial \xi}=0\,.
\]
We can therefore consider the associated system of characteristics
\[
\frac{dx}{v_x}=\frac{dy}{v_y}=\frac{d\xi}{-\frac{y}{x^3}f({y}/{x})+\frac{x}{y^3}
g({y}/{x})}\,,
\]
and using that 
\[
\frac{-y\,dx+x\,dy}{\xi}=\frac{dx}{v_x}=\frac{dy}{v_y}\,,
\]
we arrive to 
\[
\frac{-y\,dx+x\,dy}{\xi}=\frac{d\xi}{-\frac{y}{x^3}f(\frac{y}{x})+\frac{x}{y^3}
g(\frac{y}{x})},
\]
i.e.
\[
-\frac{y^2d\left(\frac{x}{y}\right)}{\xi}=
\frac{d\xi}{-\frac{y}{x^3}f(\frac{y}{x})+\frac{x}{y^3}g(\frac{y}{x})}\\
\]
and integrating we obtain the following first-integral
\begin{equation}\label{genErminv}
\frac 12 \xi^2+\int^u\left[-\frac 1{\zeta^3}\, f\left(\frac 1\zeta\right)+
  \zeta\,g\left(\frac 1\zeta
\right)\right]\,d\zeta=C\,,
\end{equation}
with $u=x/y$. This
first-integral allows us  to determine, by means of quadratures,  a solution of
one subsystem 
 in terms of  a solution of the other equation.

\section{Milne--Pinney equation}\label{SecMP}

We call Milne-Pinney equation the second-order ordinary nonlinear differential 
equation
\cite{Mil30,P50}
\begin{equation}
\ddot x=-\omega^2(t)x+\frac k{x^3}\,,\label{Milneeq}
\end{equation}
where $k$ is a non-zero constant. This equation describes the $t$-evolution of
an isotonic oscillator \cite{Cal69,Pe90} (also called pseudo-oscillator), i.e.
an oscillator with an inverse quadratic potential \cite{WS78}. This oscillator
shares with the harmonic one
 the property of having a period independent of the energy \cite{ChaVes05}, i.e.
they are isochronous systems and, in the quantum case, they have an equispaced
spectrum \cite{ACMP07}. The equation (\ref{Milneeq}) appears in the study of
certain Friedmann--Lema\^{\i}tre--Robertson--Walker  spaces \cite{DW10}, certain
scalar field cosmologies \cite{HL02}, and many other works in Physics and
Mathematics (see \cite{LA08} and references therein).

The Milne--Pinney equation is defined on $\mathbb{R}^*\equiv\mathbb{R}-\{0\}$
and it is invariant under
parity, i.e. if $x(t)$ is a solution, then $-x(t)$ is a solution too.
That means that it is sufficient to restrict ourselves to analysing this
equation in $\mathbb{R}_+$.

As usual, we  can relate the Milne-Pinney equation to a system of first-order
differential equations on ${\rm T}\mathbb{R}_+$ 
\[\left\{
\begin{array}{rcl}
\dot x&=&v,\cr
\dot v&=&-\omega^2(t)x+\dfrac{k}{{x^3}},
\end{array}\right.
\]
by introducing a new
auxiliary variable 
 $v\equiv \dot x$. Then, the $t$-dependent vector field on ${\rm T}\mathbb{R}_+$
describing its integral curves reads
\[X_t=v\frac{\partial}{\partial x}+\left(-\omega^2(t)x+\frac
k{x^3}\right)\frac{\partial}{\partial v}\,.
\]
This is a Lie system because $X_t$ can be written as  $X_t=L_1+\omega^2(t)L_3$,
where the vector fields $L_1$ and $L_3$ are given by
\[
L_1=v\frac{\partial}{\partial x}+\frac k {x^3}\frac{\partial}{\partial v}
, \qquad L_3=-x\frac {\partial}{\partial v},
\]
and satisfy
\[
[L_1,L_3]=2L_2,\quad [L_1,L_2]=L_1,\quad [L_2,L_3]=L_3,
\]
with 
\[
 L_2=\frac 1 2 \left(x\frac{\partial}{\partial x}-v\frac{\partial}{\partial
v}\right),
\]
i.e. they  span a 3-dimensional real  Lie algebra of vector fields isomorphic 
to $\mathfrak{sl}(2,\mathbb{R})$. 

Let us choose the basis (\ref{thebasis}) for $\mathfrak{sl}(2,\mathbb{R})$,
which satisfies the same commutation relations as the vector fields,
$L_1,L_2,L_3$. Actually, it is possible to show that each $L_\alpha$ is the
fundamental vector field corresponding to ${\rm a}_\alpha$ with respect to the 
action $\Phi: (A,(x,v))\in SL(2,\mathbb{R})\times {\rm T}\mathbb{R}_+\mapsto
(\bar x,\bar v)\in {\rm T}\mathbb{R}_+$ given by
\[
\left\{\begin{array}{rl}
\bar x&=\sqrt{\dfrac{k+\left[(\beta v+\alpha x)(\delta
      v+\gamma x)+ k({\delta\beta}/{x^2})\right]^2}{(\delta
    v+\gamma x)^2+k({\delta}/{x})^2}},\cr
\bar v&=\kappa \sqrt{\left(\delta v+\gamma
  x\right)^2+\dfrac{k\delta^2}{x^2}\left(1-\dfrac{x^2}{\delta^2\bar
    x^2}\right)},\end{array}\right.\qquad {\rm with}\quad
A\equiv\left(\begin{array}{cc}
\alpha\,&\,\beta\\ \gamma\,&\delta\,\end{array}\right),
\]
where $\kappa$ is $\pm 1$ or $0$, depending on the initial point $(x,v)$ and the
element of the group $SL(2,\mathbb{R})$ that acts on it. In order to obtain an
explicit expression for $\kappa$ in terms of $A$ and $(x,v)$, we can use the
below decomposition for every element of the group $SL(2,\mathbb{R})$
{\small 
\[
A=\exp(-\alpha_1{\rm a}_1)\exp(\alpha_3{\rm a}_3)\exp(-\alpha_2{\rm a}_2)=
\left(\begin{array}{cc}
      1&\alpha_1\\
0&1
      \end{array}\right)
\left(\begin{array}{cc}
      1&0\\
\alpha_3&1
      \end{array}\right)
\left(\begin{array}{cc}
      e^{\alpha_2/2}&0\\
0&e^{-\alpha_2/2}
      \end{array}\right),
\]}
from where we obtain that $\alpha_3=\gamma\delta$ and
$\alpha_1=\beta/\delta$. As we know that $$\Phi(\exp(-\alpha_2{\rm
a}_2),(x,v))$$
is the integral curve of the vector field $L_2$ starting from  the point $(x,
v)$
parametrised by $\alpha_2$, it is straightforward to check that
\[(x_1,v_1)\equiv\Phi(\exp(-\alpha_2{\rm a}_2),(x,v))=(\exp(\alpha_2/2)x,
\exp(-\alpha_2/2)v),\]
and in a similar way \[(x_2,v_2)\equiv\Phi(\exp(\alpha_3{\rm
a}_3),(x_1,v_1))=(x_1,\alpha_3 x_1+v_1).\]

Finally, we want to obtain $(\bar x,\bar v)=\Phi(\exp(-\alpha_1{\rm
  a}_1),(x_2,v_2))$, and taking into account that the integral curves of $L_1$
satisfy that
\begin{gather}\label{intcur}
\frac{x^3dv}{k}=\frac{dx}{v}=d\alpha_1,
\end{gather}
it turns out that when
$k>0$ we have $\bar v^2+k/\bar x^2=v_2^2+k/x_2^2\equiv\lambda$ with
$\lambda>0$. Thus,
 using  this fact and (\ref{intcur}) we obtain
\[
\frac{k^{1/2}dv}{(\lambda-v^2)^{3/2}}=d\alpha_1,
\]
and integrating $v$ between $v_2$ and $\bar v$,
\begin{gather*}
\frac{\bar v}{(\lambda-\bar
v^2)^{1/2}}=\alpha_1\frac{\lambda}{k^{1/2}}+\frac{v_2}{(\lambda-v_2^2)^{1/2}}
=\frac{1}{k^{1/2}}\left(\alpha_1\lambda+v_2|x_2|\right).
\end{gather*}
As $\kappa={\rm sign}[\bar v]$, we see that  $\kappa$ is given by
\[
\kappa={\rm sign}[\alpha_1\lambda+v_2|x_2|]={\rm
sign}\left[\frac\beta\delta(x\gamma+v\delta)^2+\frac{k\delta\beta}{x^2}+\frac{
|x|}{\delta}(v\delta+x\gamma)\right].
\]

System (\ref{Milneeq}) has no non-trivial first-integrals independent of
$\omega(t)$, i.e. there is no function $I:U\subset{\rm
T}\mathbb{R}_+\rightarrow\mathbb{R}$ such
 that $X_tI=0$ for $X$ determined by any function $\omega(t)$. This is
equivalent to $dI(L_\alpha)=0$ on an open $U$, with $\alpha=1,2,3$. Thus,
 the first-integrals we are looking for hold that $dI_p$ is incident to the
involutive
 distribution  
$\mathcal{V}_p\simeq \langle (L_1)_p,(L_2)_p,(L_3)_p\rangle$ generated by the
fundamental vector fields $L_\alpha$ in $U$. In almost any point we obtain that
$\mathcal{V}_p={\rm T}_p{\rm T}\mathbb{R}_+$. Then, as $dI_p=0$ in a generic
point  $p\in U\subset {\rm T}\mathbb{R}_+$, the only possibility is $dI=0$ and
therefore $I$ is a constant first-integral.

\section{A new superposition rule for the Milne--Pinney equation}
Our aim now is to show that there exists 
 a superposition rule for the Milne--Pinney equation (\ref{Milneeq}) for the
case $k>0$  \cite{CLR07a,Mil30,P50}
in terms of a pair of its particular solutions \cite{CL08b}. The case $k<0$ can
be analogously described.

In fact, one sees from the first-integral (\ref{genErminv}) that in the
particular case of
$f=g=k$, if  a particular solution $x_1$ is known,  there is a $t$-dependent
constant
 of motion  for the Milne--Pinney
equation
given by (see e.g.  \cite{CLR07a}):
\begin{equation}\label{C1}
I_1=(x_1\dot x-\dot
x_1x)^2+k\left[\left(\frac{x}{x_1}\right)^2+\left(\frac{x_1}{x}\right)^2\right]\
,.
\end{equation}

If another particular solution $x_2$ of the equation (\ref{Milneeq})  is given,
then
we have another $t$-dependent constant of motion 
\begin{equation}\label{C2}
I_2=(x_2\dot x-\dot
x_2x)^2+k\left[\left(\frac{x}{x_2}\right)^2+\left(\frac{x_2}{x}\right)^2\right]\
,.
\end{equation}
Moreover, the two solutions $x_1$ and   $x_2$ provide a function of $t$ which
is a constant of the motion and generalises the Wronskian $W$ of two solutions
of the equation (\ref{Milneeq})
\begin{equation}\label{C3}
I_3=(x_1\dot x_2-x_2\dot
x_1)^2+k\left[\left(\frac{x_2}{x_1}\right)^2+\left(\frac{x_1}{x_2}
\right)^2\right]\,.
\end{equation}

Remark that for any real number $\alpha$ the inequality
$(\alpha-1/\alpha)^2\geq 0$ implies 
$$\alpha^2+\frac 1{\alpha^2}\geq 2\,,$$
and the equality sign is valid if and only if $|\alpha|=1$, 
$$\alpha^2+\frac
1{\alpha^2}= 2\Longleftrightarrow \ |\alpha|=1\,.
$$
Therefore, as we have considered $k>0$, we see that $I_i\geq 2\,k$, 
for $i=1,2,3$.  
Moreover, as the solutions  $x_1(t)$ and $x_2(t)$ are different
 solutions of the Milne--Pinney equation, it turns out that $I_3>2k$.

The knowledge of the two first-integrals  $I_1$ and $I_2$, together with
 the constant value of
$I_3$ 
for a pair of solutions of equation (\ref{Milneeq}),
can be used to obtain the superposition rule for the Milne--Pinney equation. In
fact, 
given two particular solutions $x_1$ and $x_2$ , the first-integral
 (\ref{C2}) allows us to write an explicit expression for $\dot x$ in terms of 
$x, x_2$ and $I_2$
\[
\dot x=\dot x_2
\frac{x}{x_2}\pm\sqrt{-k\frac{x^2}{x_2^4}+I_2\frac{1}{x_2^2}-k\frac{1}{x^2}}\,,
\]
and  using  such an expression with the first-integral (\ref{C1}), we see, after
a careful computation,  
 that $x$ satisfies the following fourth degree equation
\begin{multline}\label{bicua}
(I_2^2-4k^2)x_1^4-2(I_1I_2-2I_3 k)x_1^2x_2^2+(I_1^2-4 k^2)x_2^4
-\\-2((I_2I_3-2 I_1 k)x_1^2+(I_1 I_3-2 I_2k)x_2^2) x^2+(I_3^2-4 k^2) x^4=0\,,
\end{multline}
where we have used that $I_3$ is constant along pairs of solutions, $x_1(t)$,
$x_2(t),$ of
the Milne--Pinney equation.

Hence, we can obtain from the condition (\ref{bicua}) the expression for the
square of
the  solutions of the Milne--Pinney equation in terms of any pair of its
particular positive solutions by means of a superposition rule 
\begin{equation}
x^2=k_1x_1^2+k_2x_2^2\pm
2\sqrt{\lambda_{12}[-k (x_1^4+x_2^4)+I_3\,x_1^2x_2^2\,]},\label{SRxdos}
\end{equation}
 where the constants $k_1$ and $k_2$ are given by
\begin{eqnarray*}
k_1=\frac{I_2 I_3-2 I_1 k}{I_3^2-4 k^2},\qquad k_2=\frac{I_1 I_3-2 I_2
k}{I_3^2-4 k^2}\,,
\end{eqnarray*}
and $\lambda_{12}$ is a constant which reads as follows
\begin{equation*}
\lambda_{12}=\lambda_{12}(k_1, k_2; I_3, k)=
\frac{k_1k_2 I_3+k(-1+k_1^2+k_2^2)}{I_3^2-4
  k^2}=\varphi( I_1, I_2; I_3, k)\,,
\end{equation*}
where the function $\varphi$ is given by
\begin{equation*}
\varphi(I_1, I_2; I_3, k)=\frac{I_1 I_2 I_3-(I_1^2+I_2^2+I_3^2)k+4k^3}{(I_3^2-4
k^2)^2}\,.
\end{equation*}

It is important to remark that  if $k_1<0$ then $k_2>0$ and if
$k_2<0$ then $k_1>0$, i.e. if $k_1<0$ then $I_2I_3<2I_1k$, and  thus
$I_2<2kI_1/I_3$. Therefore,
$\lambda_2(I_3^2-4k^2)=I_1I_3-2k I_2>I_1I_3-4k^2 I_1/I_3=I_1(I_3^2-4k^2)>0$, and
thus, as $I_3>2k$, $k_2>0$. Similarly we obtain that $k_2<0$ implies $k_1>0$. 

The parity invariance of (\ref{Milneeq}) is displayed by (\ref{SRxdos}), which
gives us the solutions 
\begin{equation}\label{SR4}
x^2=k_1x_1^2+k_2x_2^2\pm
2\sqrt{\lambda_{12}[-k (x_1^4+x_2^4)+I_3\,x_1^2x_2^2\,]}\,.
\end{equation}

In order to ensure that the right-hand term of the above formula is positive,
which gives rise to a real solution
 of the Milne--Pinney equation, the constants $k_1$ and $k_2$ in the preceding 
expression should satisfy some additional restrictions. In particular, they must
obeying $$\lambda_{12}[-k
(x_1^4(0)+x_2^4(0))+I_3\,x_1^2(0)x_2^2(0)\,]\geq 0$$ and
$$k_1x_1^2(0)+k_2x_2^2(0)\pm
2\sqrt{\lambda_{12}[-k (x_1^4(0)+x_2^4(0))+I_3\,x_1^2(0)x_2^2(0)\,]}>0.$$
If
these conditions are satisfied, then,  differentiating expression (\ref{SR4}) 
in $t=0$ for $x_1=x_1(t)$ and $x_2=x_2(t)$ solutions of the Milne--Pinney
equation (\ref{Milneeq}), it can be
checked that $\dot x(0)$  is also a real constant. As $x(t)$ is a solution
with real initial conditions, then $x(t)$ given by (\ref{SR4}) is real in an
interval of $t$ and thus all the obtained conditions are valid in an interval of
$t$. 

If we take into account that we have considered $x_2> 0$, we can simplify  the
study of such restrictions by  writing (\ref{SR4}) in terms of the variables
$x_2$ and $z=(x_1/x_2)^2$ as
\begin{equation*}
x^2=x_2^2\left(k_1z+k_2\pm
2\sqrt{\lambda_{12}[-k (z^2+1)+I_3\,z\,]}\right)\,,
\end{equation*}
and  the preceding  conditions turn out to be  $\lambda_{12}[-k
(z^2+1)+I_3\,z\,]\geq 0$ and $k_1z+k_2\pm
2\sqrt{\lambda_{12}[-k (z^2+1)+I_3\,z\,]}>0$. 

Next, in order to get  $\lambda_{12}[-k (z^2+1)+I_3\,z\,]\geq0$, we 
first notice that this expression  is not definite because its discriminant is 
$\lambda_{12}^2(I_3^2-4 k^2)\geq 0$, and this restricts the possible values of
$k_1$ and $k_2$ for a given $z$.
With this aim we define the polynomial $P(z)$ given by
\begin{equation*}
P(z)=-k(z^2+1)+I_3\,z,
\end{equation*}
with roots 
\begin{equation*}
z=z_\pm=\frac{I_3\pm\sqrt{I_3^2-4k^2}}{2k},
\end{equation*}
which  can be written in terms of the variable $\alpha_3=I_3/2k$ as
\begin{equation*}
z_\pm=\alpha_3\pm\sqrt{\alpha_3^2-1}.
\end{equation*}

As $\alpha_3>1$, then $\alpha_3>\sqrt{\alpha_3^2-1}>0$ and thus
$z_\pm>0$. The sign of the polynomial $P(z)$ is  displayed  in Fig. 1.

\begin{figure}[ht!]
\centerline{
\psfrag{a}{$t\equiv x_1=\sqrt{\alpha_3}x_2$}
\psfrag{b}{$r\equiv x_1=\sqrt{z_-}x_2$}
\psfrag{c}{$s\equiv x_1=\sqrt{z_+}x_2$}
\psfrag{e}{$x_1$}
\psfrag{f}{$x_2$}
\psfrag{g}{}
\psfrag{i}{$A(-)$}
\psfrag{j}{$B(+)$}
\psfrag{k}{$B(+)$}
\psfrag{l}{$A(-)$}
\psfrag{m}{}
\includegraphics[width=8cm]{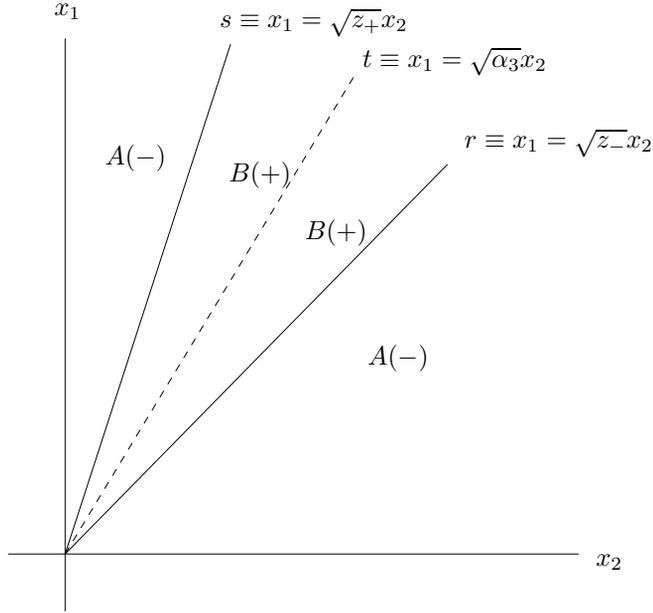}
}
\caption{Sign of the polynomial $P(x_1,x_2)$.}
\end{figure}
The region  $\mathbb{R}_+\times \mathbb{R}_+$ splits into three regions, 
$$A=\{(x_1,x_2)\in \mathbb{R}_+\times \mathbb{R}_+\mid x_1>\sqrt{z_+}\,
x_2\}\bigcup 
\{(x_1,x_2)\in \mathbb{R}_+\times \mathbb{R}_+\mid x_1<\sqrt{z_-}\, x_2\}\,,$$
$$B=\{(x_1,x_2)\in \mathbb{R}_+\times \mathbb{R}_+\mid \sqrt{z_-}x_2<
x_1<\sqrt{z_+}\, x_2\}$$
separated by the region 
 $$C=\{(x_1,x_2)\in \mathbb{R}_+\times \mathbb{R}_+\mid x_1=\sqrt{z_+}\,
x_2\}\bigcup 
\{(x_1,x_2)\in \mathbb{R}_+\times \mathbb{R}_+\mid x_1=\sqrt{z_-}\, x_2\}$$
of the straight lines $x_1=\sqrt{z_+}\,x_2$ and  $x_1=\sqrt{z_-}\,x_2$.  The
condition to make $\lambda_{12}P(z)$ non-negative
in region $A$, where the polynomial $P$ takes negative values, is 
 to choose $k_1$ and
 $k_2$ so that $\lambda_{12}(k_1,k_2, I_3, k)\leq 0$. Similarly, as $P$ is
positive in region $B$
 we have to choose $k_1$ and $k_2$ such that
 $\lambda_{12}(k_1, k_2, I_3, k)\geq 0$. Finally, as $P$ vanishes in 
region $C$, there is no restriction on the coefficients
 $k_1$ and $k_2$.

Once we have stated the conditions for $\lambda_{12}P(z)$ to be
non-negative we still have to impose the condition 
\begin{equation}\label{Con4}
k_1z+k_2\pm
2\sqrt{\lambda_{12}[-k (z^2+1)+I_3\,z\,]}> 0.
\end{equation}

In order to study these conditions, we study the sign of the polynomial
\begin{equation*}
\begin{aligned}
P_{I_3,k}(z,k_1,k_2)&=(k_1z+k_2)^2-4\lambda_{12}[-k(z^2+1)+I_3z]\\
&=\frac{4P(z)I_3}{I_3^2-4 k^2}+\left(a k_1+b k_2\right)^2\,,
\end{aligned}
\end{equation*}
where
\begin{equation*}
a=\sqrt{-\frac{4P(z)k}{I_3^2-4k^2}+z^2},\qquad
b=\sqrt{1-\frac{4P(z)k}{I_3^2-4k^2}}.
\end{equation*}

As we remarked before,  the constants $k_1,k_2$  cannot be both negative. Let
$K$
denote the set 
$$K=\mathbb{R}^2-\{(k_1,k_2)\in\mathbb{R}^2\mid k_1<0, k_2<0 \}$$
and 
consider three cases:
\begin{enumerate}
\item If $(x_1,x_2)\in A$, then as $P(z)\leq 0$,  it must be $\lambda_{12}\leq
  0$ in order to
 satisfy $\lambda_{12}P(z)\geq 0$. In this case, if $K_1$ and $K_2$ are
 the sets 
\begin{equation*}
\begin{aligned}
K_1=\left\{(k_1,k_2) \in K ;  \sqrt{-\frac{4P(z)I_3}{I_3^2-4k^2}}>|a k_1+b
k_2|\right\}\,,\\
K_2=\left\{(k_1,k_2) \in K ; \sqrt{-\frac{4P(z)I_3}{I_3^2-4k^2}}<|a k_1+b
k_2|\right\}\,.
\end{aligned}
\end{equation*}
 We find   the following particular cases
\begin{enumerate}
\item  If   $(k_1,k_2)\in K_1$, then $P_{I_3,k}(z,k_1,k_2)>0$.
\item  If   $(k_1,k_2)\in K_2$ then $P_{I_3,k}(z,k_1,k_2)<0$,
\end{enumerate}
that can be summarised  by means of Figure 2.

\begin{figure}[ht!]
\centerline{\psfrag{a}{$\lambda_1$}
\psfrag{b}{$\lambda_2$}
\psfrag{c}{$K_1$}
\psfrag{d}{$K_2$}
\includegraphics[width=8cm]{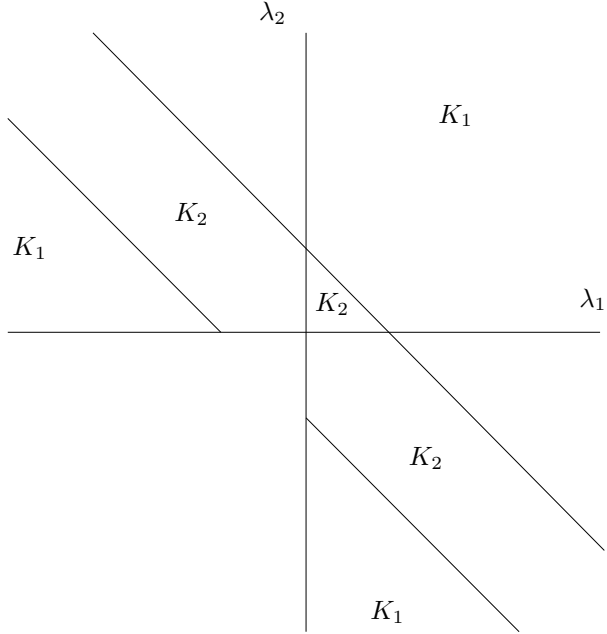}
}
\caption{Sign of the polynomial $P_{I_3,k}(z,k_1,k_2)$ in $K$.}
\end{figure}
\item If $(x_1, x_2)\in B$, as $P(z)$ is positive, then  $\lambda_{12}      $
must
  also be positive, $\lambda_{12}\geq 0$. Thus for $(k_1, k_2)\in K_1\cup K_2$,
$P_{I_3,k}(z,k_1,k_2)>0$. 
\item If $(x_1,x_2)\in C$, then for $(k_1, k_2)\in K_1\cup K_2$,
$P_{I_3,k}(z,k_1,k_2)>0$.
\end{enumerate}

In those cases in which $P_{I_3,k}(z,k_1,k_2)>0$, we can assert that 
$$|k_1z+
k_2 |>2\sqrt{\lambda_{12}[-k (z^2+1)+I_3\,z\,]}$$
but we still have to impose  that $\lambda_1z+\lambda_2>0$ for (\ref{Con4}) to
be positive. Nevertheless, this is very simple, because if the pair $(k_1,k_2)$
does not  satisfy 
 $k_1z+k_2>0$, the pair of opposite elements $(-k_1,-k_2)$ does it, while the
other conditions are invariant under the change $k_i\rightarrow -k_i$ with
$i=1,2$.

In those cases in which $P_{I_3,k}(z,k_1,k_2)<0$ we can assert that 
$$|k_1z+k_2|<2\sqrt{\lambda_{12}[-k (x_1^4+x_2^4)+I_3\,x_1^2x_2^2\,]}$$ and in
this case the unique valid superposition rule is
\begin{equation*}
x=|x_2|\left(k_1z+k_2+
2\sqrt{\lambda_{12}[-k (z^2+1)+I_3\,z\,]}\right)^{1/2}\,,
\end{equation*}
which is equivalent to
\begin{equation*}
x=\left(k_1x_1^2+k_2x_2^2+
2\sqrt{\lambda_{12}[-k (x_1^4+x_2^4)+I_3\,x_2^2x_1^2\,]}\right)^{1/2}\,.
\end{equation*}

Note that if we had considered no restriction on $k_1, k_2$, we would have
obtained real and imaginary solutions of the Milne--Pinney equation.

Expression (\ref{SR4}) provides us with a superposition rule for the positive
solutions of the Pinney equation (\ref{Milneeq})
in terms of two of its independent particular positive solutions. 
Therefore, once two particular solutions of the equation (\ref{Milneeq}) are
known, 
we can write its general solution.
Note also that, because
of the parity symmetry of 
 (\ref{Milneeq}), the superposition (\ref{SR4}) can be used with both positive
and
 negative solutions. In all these ways we obtain non-vanishing  solutions of
 (\ref{Milneeq}) when $k>0$.  {\it Mutatis mutandis}, the above procedure can
also be applied 
 to analyse Milne--Pinney equations when $k<0$.

A similar superposition rule works for negative solutions of Milne--Pinney
equation (\ref{Milneeq}): 
\begin{equation}\label{SR2n}
x=-\left(k_1x_1^2+k_2x_2^2\pm
2\sqrt{\lambda_{12}(-k (x_1^4+x_2^4)+I_3\,x_1^2x_2^2)}\right)^{1/2},
\end{equation}
where once  again 
 $x_1$ and $x_2$ are arbitrary  solutions.

\section{Painleve-Ince equations and other SODE Lie systems}

In this section we show a new relevant instance of SODE Lie systems including,
as particular instances, some Painlev\'e--Ince equations \cite{EEU07}. In the
process of analysing that this particular case of Painlev\'e--Ince is a SODE Lie
system, we find a much larger family of SODE Lie systems which frequently
occur in the mathematical and physical literature.

Consider the family of differential equations 
\begin{equation}\label{MDPIeq}
\ddot x+3x\dot x+x^3=f(t),
\end{equation}
with $f(t)$ being any $t$-dependent function. The interest in these equations is
motivated by their frequent appearance in Physics and Mathematics
\cite{CRS05,CLS05II,KL09}. The different properties of these equations have been
deeply analysed since their first analysis by Vessiot and Wallenberg
\cite{Ve94,Wall03} as a particular case of second-order Riccati equations. For
instance, these equations appear in  \cite{GL99} in the study of the Riccati
chain. There, it is stated that such equations can be used to derive solutions
for certain PDEs. In addition, equation (\ref{MDPIeq}) also appears in the book
by Davis \cite{Davis}, and the particular case with $f(t)=0$  has recently been
treated through geometric methods in \cite{CGR09,CRS05}.

The results described in previous sections can be used to study differential
equations (\ref{MDPIeq}). Let us first  show that the above differential
equations are SODE Lie systems and, in view of Proposition 1, they admit a
superposition rule that is derived.  According to  definition \ref{SODE},
equation (\ref{MDPIeq}) is a SODE Lie system if and only if the system 
\begin{equation}\label{FO}
\left\{\begin{aligned}
\dot x&=v,\\
\dot v&=-3xv-x^3+f(t),
\end{aligned}\right.
\end{equation}
determining the integral curves of the $t$-dependent
vector field of the form 
\begin{equation}\label{Dec}
X_{PI}(t,x,v)=X_1(x,v)+f(t)X_2(x,v), 
\end{equation}
with
$$
X_1=v\frac{\partial}{\partial x}-(3xv+x^3)\frac{\partial}{\partial v},\qquad
X_2=\frac{\partial}{\partial  v},
$$
is a Lie system.

In view of the decomposition (\ref{Dec}), all equations (\ref{MDPIeq}) are SODE
Lie systems if the vector fields $X_1$ and $X_2$ are included in a
finite-dimensional real  Lie algebra of vector fields $V$. This happens  if and
only if ${\rm Lie} (\{X_1,X_2\})$ span a finite-dimensional linear space. We
consider the family of vector fields on ${\rm T}\mathbb{R}$ given by {\small
\begin{equation}\label{VF}
\begin{aligned}
X_1&=v\frac{\partial}{\partial x}-(3xv+x^3)\frac{\partial}{\partial v},\,\,
&X_2&=\frac{\partial}{\partial  v},\\
X_3&=-\frac{\partial}{\partial x}+3x\frac{\partial}{\partial v},\,\,
&X_4&=x\frac{\partial}{\partial x}-2x^2\frac{\partial}{\partial v},\\
X_5&=(v+2x^2)\frac{\partial}{\partial x}-x(v+3x^2)\frac{\partial}{\partial
v},\,\, &X_6&=2x(v+x^2)\frac{\partial}{\partial 
x}+2(v^2-x^4)\frac{\partial}{\partial v},\\
X_7&=\frac{\partial}{\partial x}-x\frac{\partial}{\partial v},\,\,
&X_8&=2x\frac{\partial}{\partial x}+4v\frac{\partial}{\partial v},
\end{aligned}
\end{equation}}where $X_3=[X_1,X_2]$, $-3 X_4=[X_1,X_3]$, $X_5=[X_1,X_4]$,
$X_6=[X_1,X_5]$, $X_7=[X_2,X_5]$, $X_8=[X_2,X_6]$, and 
then the vector fields, $X_1,\ldots, X_8,$ are linearly independent  over
$\mathbb{R}$. Their commutation relations read
{\small\begin{equation}\label{Rel2}
\begin{aligned}
\left[X_1,X_2\right]&=X_3,\quad &[X_1,X_3]&=-3X_4,\quad &[X_1,X_4]&=X_5,\quad
&[X_1,X_5]&=X_6,\\
\left[X_1,X_6\right]&=0,\quad &[X_1,X_7]&=\frac 12 X_8,\quad
&[X_1,X_8]&=-2X_1,\quad &[X_2,X_3]&=0,\\
\left[X_2,X_4\right]&=0,\quad &[X_2,X_5]&=X_7,\quad &[X_2,X_6]&=X_8,\quad
&[X_2,X_7]&=0,\\
\left[X_2,X_8\right]&=4X_2,\quad &[X_3,X_4]&=-X_7,\quad &[X_3,X_5]&=-\frac
12X_8,\quad &[X_3,X_6]&=-2X_1,\\
\left[X_3,X_7\right]&=-2X_2,\quad &[X_3,X_8]&=2X_3,\quad &[X_4,X_5]&=- X_1,\quad
&[X_4,X_6]&=0,\\
\left[X_4,X_7\right]&=X_3,\quad &[X_4,X_8]&=0,\quad &[X_5,X_6]&=0,\quad
&[X_5,X_7]&=-3X_4,\\
\left[X_5,X_8\right]&=-2X_5,\quad &[X_6,X_7]&=-2X_5,\quad
&[X_6,X_8]&=-4X_6,\quad &[X_7,X_8]&=2X_7,\\
\end{aligned}
\end{equation}}In other words, the vector fields, $X_1,\ldots,X_8,$ span an
eight-dimensional Lie algebra of vector fields $V$ containing $X_1$ and $X_2$.
 Therefore, equation  (\ref{MDPIeq}) is a SODE Lie system.  Moreover,  the
elements of the following  family of traceless real $3\times 3$ matrices
$$
\begin{aligned}
M_1&=\left(
\begin{array}{ccc}
0&-1&0\\
0&0&-1\\
0&0&0.
\end{array}\right),
\quad
&M_2&=\left(
\begin{array}{ccc}
0&0&0\\
0&0&0\\
-1&0&0.
\end{array}\right),\\
M_3&=\left(
\begin{array}{ccc}
0&0&0\\
1&0&0\\
0&-1&0.
\end{array}\right),
\quad
&M_4&=-\frac{1}{3}\left(
\begin{array}{ccc}
-1&0&0\\
0&2&0\\
0&0&-1.
\end{array}\right),\\
M_5&=\left(
\begin{array}{ccc}
0&1&0\\
0&0&-1\\
0&0&0.
\end{array}\right),
\quad
&M_6&=\left(
\begin{array}{ccc}
0&0&2\\
0&0&0\\
0&0&0.
\end{array}\right),\\
M_7&=\left(
\begin{array}{ccc}
0&0&0\\
-1&0&0\\
0&-1&0.
\end{array}\right),
\quad
&M_8&=\left(
\begin{array}{ccc}
2&0&0\\
0&0&0\\
0&0&-2.
\end{array}\right),
\end{aligned}
$$
obey the same commutation relations as the
corresponding 
vector fields, $X_1,\ldots,X_8$, i.e. the linear map 
$\rho:\mathfrak{sl}(3,\mathbb{R})\rightarrow V$,
such that $\rho(M_\alpha)=X_\alpha$, with $\alpha=1,\ldots,8$, is a
Lie algebra isomorphism. Consequently, the finite-dimensional Lie algebra of
vector
fields $V$ is isomorphic to $\mathfrak{sl}(3,\mathbb{R})$ and the systems of
differential equations describing the integral curves for the $t$-dependent
vector fields 
\begin{equation}\label{family2}
X(t,x,v)=\sum_{\alpha=1}^8b_\alpha(t)X_\alpha(x,v),
\end{equation}
are Lie systems related to a Vessiot--Guldberg Lie algebra isomorphic to
$\mathfrak{sl}(3,\mathbb{R})$.

Many instances of the family of Lie systems (\ref{family2}) are associated with
interesting SODE Lie systems with applications to Physics or related to
remarkable mathematical problems. In all these cases, the theory of Lie systems
can be applied to investigate these second-order differential equations, recover
some of their known properties, and, possibly, provide new results. Let us
illustrate this assertion by means of a few examples.

Another equation appearing in the Physics literature \cite{CLS05II,CLS05,TT07}
which can be analysed by means of our methods is 
\begin{equation}\label{Exam2}
\ddot x+3x\dot x +x^3+\lambda_1x=0,
\end{equation}
which is a special kind of Li\'enard equation $\ddot x+f(x)\dot x+g(x)=0$, with
$f(x)=3x$ and $g(x)=x^3+\lambda_1x$. The above equation can also be related to a
generalised form of an Emden equation occurring in the thermodynamical study of
equilibrium configurations of spherical clouds of gas acting under the mutual
attraction of their molecules \cite{DT90}.

As in the study of equations (\ref{MDPIeq}), by considering the new variable
$v=\dot x$, equation (\ref{Exam2}) becomes the system
\begin{equation}
\left\{\begin{aligned}
\dot x&=v,\\
\dot v&=-3xv-x^3-\lambda_1x,
\end{aligned}\right.
\end{equation}
describing the integral curves of the vector field $X=X_1-\lambda_1/2(X_7+X_3)$
included in the family (\ref{family}). 

Finally, we can also treat the equation  
\begin{equation}\label{general}
\ddot x+3x\dot x +x^3+f(t)(\dot x+x^2)+g(t)x+h(t)=0,
\end{equation}
describing, as particular cases, all the previous examples \cite{KL09}.
The system of first-order differential equations associated with this equation
reads
\begin{equation}\label{FirstGeneral}
\left\{\begin{aligned}
\dot x&=v,\\
\dot v&=-3xv-x^3-f(t)(v+x^2)-g(t)x-h(t).
\end{aligned}\right.
\end{equation}
Hence, this system describes the integral curves of the $t$-dependent vector
field 
$$X_t=X_1-h(t)X_2-\frac 14 f(t)\,(X_8-2X_4)-\frac 12 g(t)\,(X_7+X_3).$$
Therefore, equation (\ref{general}) is a SODE Lie system and the theory of
Lie systems can be used to analyse its properties. 

Some particular cases of system (\ref{general}) were pointed out in
\cite{CLS05,KL09}. Additionally, the case with
$f(t)=0$, $g(t)=\omega^2(t)$ and $h(t)=0$ was studied in  \cite{CLS05II} and it
is related to harmonic oscillators. The case with $g(t)=0$ and $h(t)=0$ appears
in the catalogue of equations possessing the Painlev\'e property \cite{In86}.
Additionally, our result generalises Vessiot's contribution \cite{Ve95}
describing the existence of an expression determining the general solution of a
system  like (\ref{general}) (but with constant coefficients) in terms of four
of their particular solutions, their derivatives and two constants.

Finally, it is worth noting that the second-order differential equation
(\ref{general}) is a particular case of second-order Riccati equations
\cite{CRS05,GL99}. Such equations were analysed through Lie systems in
\cite{CC87}. The approach carried out in that paper is based on the use of
certain {\it ad hoc} changes of variables which transform second-order Riccati
equations into some Lie systems. The advantage of our approach here is that it
allows us to study equations (\ref{general}) without using, as it was performed
in \cite{CC87}, any {\it ad hoc} transformations. In addition, our presentation
along with the theory of quasi-Lie schemes can be used to perform a quite
complete study of second-order Riccati equations in a systematic way
\cite{CL09SRicc}.

\section{Mixed superposition rules and Ermakov systems}\label{ES}
Let us now turn to show how the theory developed in Section \ref{SR3} for mixed
superposition rules works. By adding some, probably different, Lie systems to an
initial one, we get new Lie systems that admit constants of motion which do not
depend on the $t$-dependent coefficients of these systems and relate the
different solutions of the constituting Lie systems. Moreover, if we add enough
copies, these constants of the motion can be used to construct a mixed
superposition rule.

We here investigate Ermakov systems. These systems are formed by a second-order
homogeneous linear differential equation and a Milne--Pinney equation, i.e.
\begin{equation*}
\left\{\begin{array}{rcl}
\ddot x&=&-\omega^2(t)x+\dfrac {k}{{x^3}},\cr
\ddot y&=&-\omega^2(t)y,
\end{array}\right.\qquad (x,y)\in \mathbb{R}_+^2.
\end{equation*}
These systems have been broadly studied in Physics and Mathematics since its
introduction until the present day. In Physics they appear in the study of
Bose-Einstein condensates and cosmological models \cite{Ha02,HL02,Li04} and in
the solution of $t$-dependent harmonic or anharmonic oscillators
\cite{DL84,FM03,Ga84,Le67,RR79a,WS78}. A lot of works have also been devoted to
the usage of Hamiltonian or Lagrangian structures in the study of such systems,
see e.g. \cite{RR80}. Here we recover a constant of the motion, the so-called
{\it Lewis-Ermakov invariant} \cite{Le67}, which  appears naturally. 

In order to use the theory of Lie systems to analyse Ermakov systems, consider
the system of ordinary first-order differential equations \cite{DL84,PGL91} 
\begin{equation}\label{Ermak}
\left\{
\begin{aligned}
\dot x&=v_x,\\
\dot y&=v_y,\\
\dot v_x&=-\omega^2(t)x+\dfrac {k}{{x^3}},\\
\dot v_y&=-\omega^2(t)y,
\end{aligned}
\right.
\end{equation}
defined over ${\rm T}\mathbb{R}^2_+$ and built by adding the new variables $\dot
x=v_x$ and $v_y=\dot y$ to the Ermakov systems and satisfying the conditions
explained in Section \ref{SR3}. Its solutions are the integral curves for the
 $t$-dependent vector field 
\[X_t=v_x\frac{\partial}{\partial x}+v_y\frac{\partial}{\partial
y}+\left(-\omega^2(t)x+\frac{k}
  {x^3}\right)\frac{\partial}{\partial v_x}-\omega^2(t)y\frac{\partial}{\partial
v_y}\,,
\]
which is a linear combination with $t$-dependent coefficients,
 $X_t=X_1+\omega^2(t)X_3$, of 
\[X_1=v_x\frac{\partial}{\partial x}+v_y\frac{\partial}{\partial y}+\frac{k}
  {x^3}\frac{\partial}{\partial v_x},\qquad X_3=-x\frac{\partial}{\partial
v_x}-y\frac{\partial}{\partial v_y}\,.
\]
Taking into account the vector field
\[X_2=\frac 12\left(x\frac{\partial}{\partial
    x}+y\frac{\partial}{\partial
    y}-v_x\frac{\partial}{\partial v_x}-v_y\frac{\partial}{\partial
v_y}\right)\,,
\]
the vector fields $X_1, X_2$ and $X_3$ span a three dimensional Lie algebra
isomorphic to $\mathfrak{sl}(2,\mathbb{R})$. In this way, this system is a SODE
Lie system related to a Lie algebra of vector fields isomorphic to
$\mathfrak{sl}(2,\mathbb{R})$. 

The vector fields, $L_1,L_2,L_3,$ associated with the Milne--Pinney equation
(see Section \ref{SecMP}) span a distribution of rank two on ${\rm
T}\mathbb{R}_+$. Consequently, there is no local first-integral $I$ such that
$(L_1+\omega(t)^2(t)L_2)I=0$ for any given $\omega(t)$. In other words,
Milne--Pinney equations do not admit a common $t$-independent constant of the
motion. 

By adding the other $\mathfrak{sl}(2,\mathbb{R})$ linear Lie system appearing in
the Ermakov
system, i.e. the harmonic oscillator with $t$-dependent angular frequency
$\omega(t)$, the distribution spanned by $X_1,X_2$ and $X_3$ has rank three over
a dense open subset of ${\rm T}\mathbb{R}_+^2$. Therefore, there is a local a
first-integral. This one can
be obtained from $X_1F=X_3F=0$. But $X_3F=0$ implies that there exists a
function $\bar
F:\mathbb{R}^3\to \mathbb{R}$ such that $F(x,y,v_x,v_y)=\bar F(x,y,\xi)$, with
$\xi=yv_x-xv_y$,
and then $X_1F=0$ is written
\[
v_x\frac{\partial \bar F}{\partial x}+v_y\frac{\partial \bar F}{\partial
y}+k\frac y{x^3}\frac{\partial \bar F}{\partial \xi}
\]
and we obtain the associated system of characteristics  
\[k\frac {y\,dx-x\, dy}{\xi}=\frac {x^3\,d\xi}{y}\Longrightarrow
\frac{d(y/x)}{\xi}+\frac{x\,d\xi}{ky}=0\,.
\]
From here, the following first-integral is found \cite{Le67}
\[
\psi(x,y,v_x,v_y)=k\left(\frac{y}{x}\right)^2+\xi^2=k\left(\frac{y}{x}\right)^2+
(yv_x-xv_y)^2\,,
\]
which is the well-known Ermakov--Lewis invariant \cite{DL84,PGL91,RR79a}. 

Once we have obtained a first-integral, we can obtain new constants by adding
new copies of any of the systems we have already used. For instance, consider
the system of first-order differential equations
\begin{equation}\label{3d}
\left\{
\begin{array}{rcl}
\dot x&=&v_x,\cr
\dot y&=&v_y,\cr
\dot z&=&v_z,\cr
\dot v_x&=&-\omega^2(t)x+\dfrac{k}{{x^3}},\cr
\dot v_y&=&-\omega^2(t)y,\cr
\dot v_z&=&-\omega^2(t)z,
\end{array}\right.
\end{equation}
which corresponds to the vector field 
\[X_t=v_x\frac{\partial}{\partial x}+v_y\frac{\partial}{\partial
y}+v_z\frac{\partial}{\partial z}+\frac{k}{x^3}
\frac{\partial}{\partial v_x}-\omega^2(t)\left(x\frac{\partial }{\partial
v_x}+y\frac{\partial }{\partial v_y}+
z\frac{\partial}{\partial v_z}\right)\,.
\]
The $t$-dependent vector field $X_t$ can be expressed as
$X_t=N_1+\omega^2(t)N_3$
where $N_1$ and $N_3$ are
\[N_1=v_x\frac{\partial}{\partial
  x}+v_y\frac{\partial}{\partial y}+v_z\frac{\partial}{\partial z}+
\frac{k}{x^3}\frac{\partial}{\partial v_x},
\quad N_3=-x\frac{\partial }{\partial v_x}-y\frac{\partial }{\partial v_y}-
z\frac{\partial}{\partial v_z}.
\] 
These vector fields generate a three-dimensional real Lie algebra with the
vector field $N_2$ given by 
\[
 N_2=\frac 12\left(x\frac{\partial}{\partial x}+y\frac{\partial}{\partial
y}+z\frac{\partial}{\partial z}-v_x\frac{\partial}{\partial
     v_x}-v_y\frac{\partial}{\partial
     v_y}-v_z\frac{\partial}{\partial
     v_z}\right)\,.
\]
In fact, they span a Lie algebra isomorphic to   $\mathfrak{sl}(2,\mathbb{R})$
because  
 \[
[N_1,N_3]=2N_2, \quad [N_1,N_2]=N_1, \quad  [N_2,N_3]=N_3\,.
\]

The distribution spanned by these fundamental vector fields has rank three in a
open dense subset of ${\rm T}\mathbb{R}_+^3$. 
Thus, there exist three local first-integrals for all the vector fields of the
latter distribution. In other words, system (\ref{3d}) admits three
$t$-independent constants of the motion which turn out to be
the Ermakov invariant  $I_1$ of the subsystem involving variables $x$ and $y$, 
the Ermakov invariant $I_2$ 
of the subsystem involving variables $x$ and $z$, i.e. 
\[
I_1=\frac 12\left((yv_{x}-xv_y)^2+k\left(\frac {y}x\right)^2\right)\,,\qquad
I_2=\frac 12\left((xv_{z}-zv_x)^2+k\left(\frac {z}x\right)^2\right),
\]
and 
the Wronskian $W=yv_{z}-zv_{y}$ of the subsystem involving variables $y$ and
$z$.
They define a  foliation  with three-dimensional leaves. 
We can use this foliation to obtain in terms of it a superposition rule. That is
reached by describing $x$ in terms of $y, z$ and the integrals $I_1, I_2, W$,
i.e.
\begin{equation}\label{OldSR}
x=\frac {\sqrt 2}{|W|}\left(I_2y^2+I_1z^2\pm\sqrt{4I_1I_2-kW^2}\
yz\right)^{1/2}\,.
\end{equation}
This can be interpreted, as pointed out by Pinney \cite{P50}, 
 as saying that there is a superposition rule allowing
us to express  the general solution of the Milne--Pinney equation in terms of
two
independent solutions of the corresponding harmonic oscillator with
the same  $t$-dependent angular frequency.

\section{Relations between the new and the known superposition rule}
We can now compare the known superposition rule for the Milne--Pinney equation
\begin{equation}\label{supMP}
x(t)=\frac {\sqrt 2}{|W|}\left(I_2y_1^2(t)+I_1y_2^2(t)\pm\sqrt{4I_1I_2-kW^2}\
y_1(t)y_2(t)\right)^{1/2}\,,
\end{equation}
where $y_1(t)$ and $y_2(t)$ are two independent solutions of
\begin{equation}\label{hos1}
\ddot y=-\omega^ 2(t)y,
\end{equation}
and (\ref{SR4}) and
check that actually the latter reduces to the former when $x_1$ and $x_2$ are
obtained from solutions $y_1$ and $y_2$ of the associated harmonic oscillator
equation.

Let $y_1$ and $y_2$ be two solutions of (\ref{hos1}) and $W$ its Wronskian. 
Consider 
the 
 two particular positive solutions of the Milne--Pinney-equation $x_1(t)$ and
 $x_2(t)$
given by 
\begin{equation}\label{Change2}
\begin{aligned}
x_1(t)&=\frac{\sqrt{2}}{|\, W\,|}\sqrt{C_1y_1^2(t)+C_2y_2^2(t)},\\
x_2(t)&=\frac{\sqrt{2}}{|\, W\,|}\sqrt{C_2y_1^2(t)+C_1y_2^2(t)},
\end{aligned}
\end{equation}
where $C_1< C_2$ and we additionally impose 
\begin{equation}
 4C_1C_2=kW^2\,.\label{addcond}
\end{equation}  

The $t$-dependent constant of the motion $I_3$ given by (\ref{C3}) for the two
particular solutions of the
Milne--Pinney equation can then be expressed  as a function of the solutions 
$y_1$ and $y_2$ of   
 the $t$-dependent harmonic oscillator and its Wronskian $W$. After a long
computation $I_3$ 
turns out to be
\begin{equation}\label{I3}
I_3=\frac{4(C_1^2+C_2^2)}{W^2}\,,
\end{equation}
and then  using the explicit form (\ref{Change2}) of the particular solutions
and
taking into account the constant (\ref{I3}) in (\ref{SR4}) we obtain that 

\begin{multline}
k_1 x_1^2+k_2x_2^2\pm 2\sqrt{\lambda_{12}(-k(x_1^4+x_2^4)+I_3 x_1^2x_2^2)}=
{\displaystyle\frac{2}{W^2}}(C_1k_1+C_2k_2)y_1^2\\ +(C_1k_2+C_2k_1)y_2^2)
\pm{\displaystyle\frac{2}{W^2}}\sqrt{4(C_1k_1+C_2k_2)(C_1k_2+C_2k_1)-kW^2}
y_1y_2.
\end{multline}
Consequently,  from the superposition rule (\ref{SR4}), we recover  expression 
(\ref{OldSR}):
\begin{equation}\label{SR2}
x=\frac{\sqrt{2}}{|\, W\,|}\sqrt{\mu_1y_1^2+\mu_2y_2^2\pm\sqrt{4\mu_1\mu_2-k
W^2}y_1 y_2},
\end{equation}
where
$$
\left\{\begin{aligned}
\mu_1&=(C_1k_1+C_2k_2),\\
\mu_2&=(C_1k_2+C_2k_1).
\end{aligned}\right.
$$
Once we have stated the superposition rule, we still have to analyse the
possible values of $\lambda_1$ and $\lambda_2$ that we can use in this case. If
we use the expression (\ref{I3}) we obtain after a short  calculation 
   the following 
values $z_\pm$ 
\begin{equation}
z_+=\frac{4C_2^2}{kW^2},\quad z_-=\frac{4C_1^2}{kW^2}.
\end{equation}
Now if we write    $y_1^2$ and $y_2^2$ 
 in terms of $x_1^2, x_2^2$ and $W$ from the system (\ref{Change2})
we obtain 
\begin{equation}
\frac 1{C_1^2-C_2^2}
\left(\begin{array}{cc}
C_1&-C_2\\-C_2&C_1
\end{array}\right)
\left(\begin{array}{c}x_1^2\\x_2^2\end{array}\right)
=\left(\begin{array}{c}y_1^2\\y_2^2 \end{array} \right).
\end{equation}

Therefore, as $C_2>C_1$ the condition of $y_1^2$ and $,y_2^2$ being positive is
\begin{equation}\left\{
\begin{aligned}
C_1x_1^2\leq C_2x_2^2\\
C_2x_1^2\geq C_1x_2^2\\
\end{aligned}\right.
\end{equation}
and it is satisfied if $x_1^2/x_2^2
\leq C_2/C_1=4C_2^2/kW^2=z_+$ and $x_1^2/x_2^2     \geq
C_1/C_2=4C_1^2/kW^2=z_-$,
because of (\ref{addcond}).
 Thus, $(x_1, x_2)\in B$ and therefore the only restrictions for $k_1, k_2$ are
$\lambda_{12}\geq 0$ and
 $k_1x_1^2+k_2x_2^2\geq 0$. Obviously, by means of the change of
 variables (\ref{Change2}) this last expression is equivalent to
 $\mu_1y_1^2+\mu_2y_2^2\geq 0$ and thus $\mu_1$ and $\mu_2$ cannot be
 simultaneously negative. Furthermore,
      $\lambda_{12}(I_3^2-4k^2)=4\mu_1\mu_2-kW^2$. As we have said that
      $\lambda_{12}\geq 0$ then $4\mu_1\mu_2\geq kW^2$, i.e. $\mu_1\mu_2$ is
      positive and thus, $\mu_1$ and $\mu_2$               are positive.
In this way we recover  the usual constants of the known superposition rule of
the Milne--Pinney equation in terms of solutions of an harmonic oscillator.

\section{A new mixed superposition rule for the Pinney equation}

In this section we derive a mixed superposition rule for the Milne--Pinney
equation in terms of a Riccati equation. Consider again the $t$-dependent
Riccati equation
\begin{equation}\label{ricceq3}
\frac{dx}{dt}=b_1(t)+b_2(t)x+b_3(t)x^2\,
\end{equation}
which has been studied in \cite{CLR07b,CarRam} from the perspective of the
theory of
Lie systems. We have already mentioned that this Riccati equation  can be
considered as the differential equation
determining the integral curves for  the $t$-dependent vector
field (\ref{vfRic2}). This vector field  is a linear combination with
$t$-dependent coefficients of the three vector fields, $X_1, X_2, X_3,$ given by
(\ref{ric}), which close on a three-dimensional real
Lie  algebra with defining relations (\ref{conmutL5}). Consequently, this Lie
algebra is isomorphic to
$\mathfrak{sl}(2,\mathbb{R})$. Note also that the commutation relations
(\ref{conmutL5}) are the same as (\ref{CR}). 

Take now the following particular case of Riccati equation
\begin{equation*}
\frac{dx}{dt}=1+\omega^2(t)x^2\,.
\end{equation*}
This Riccati equation reads in terms of the $X_i$  as the equation of the
integral curves of the $t$-dependent vector field $X_t=X_1+\omega^2(t)X_3$.
Thus, we can apply the procedure of the Section \ref{SR3} and consider the
following differential equation in $\mathbb{R}^3\times {\rm T}\mathbb{R}_+$
\begin{equation*}
\left\{\begin{aligned}
\dot x_1&=1+\omega^2(t)x_1^2,\cr
\dot x_2&=1+\omega^2(t)x_2^2,\cr
\dot x_3&=1+\omega^2(t)x_3^2,\cr
\dot x&=v,\cr
\dot v&=-\omega^2(t)x+\dfrac{k}{{x^3}},\cr
\end{aligned}\right.
\end{equation*}
where $(x_1, x_2, x_3)\in \mathbb{R}^3$, $x\in \mathbb{R}_+$ and $(x,v)\in
T_x\mathbb{R}_+$. According to our general recipe, consider 
the following vector fields
\begin{equation*}
\begin{aligned}
M_1 &= \frac{\partial}{\partial x_1}+
\frac{\partial}{\partial x_2}+
\frac{\partial}{\partial x_3}+v\frac{\partial}{\partial
x}+\frac{k}{x^3}\frac{\partial}{\partial v},\cr
M_2 &=x_1 \frac{\partial}{\partial x_1}+x_2 \frac{\partial}{\partial x_2}+x_3
\frac{\partial}{\partial x_3}+\frac{1}{2}\left(x\frac{\partial}{\partial
x}-v\frac{\partial}{\partial v}\right)\,,\cr
M_3 &=x_1^2\frac{\partial}{\partial x_1}+x_2^2\frac{\partial}{\partial
x_2}+x_3^3\frac{\partial}{\partial x_3}-x\frac{\partial}{\partial v} , \crcr
\end{aligned}
\end{equation*}
that, by construction,  satisfy same commutation relations as before, i.e.
\begin{equation*}
[M_1,M_3] = 2M_2 ,      
  \quad [M_1, M_2] = M_1,\quad [M_2, M_3] = M_3,
\end{equation*}
and the full system of  differential equations can be understood as the 
system of differential equations for the determination of the integral curves 
of the $t$-dependent vector field $M(t)=M_1+\omega^2(t)M_3$. The distribution
associated 
with this Lie system has rank three in almost any point and then there exist
locally 
 two first-integrals. As $2\,M_2=[M_1,M_3]$, it  is enough to find the
 common first-integrals for $M_1$ and $M_3$, i.e. a function 
$F:\mathbb{R}^5\rightarrow \mathbb{R}$ such that $M_1F=M_3F=0$.

We first look for first-integrals independent of $x_3$. i.e. we  suppose 
 that $F$ depends just on $x_1, x_2, x$ and $v$.
Using the method of characteristics,  the condition $M_3F=0$ implies that the
characteristics system is
\[
\frac{dx_1}{x_1^2}=\frac{dx_2}{x_2^2}=\frac{dv}{-x}=\frac{dx}{0}
\]
That means that for such a first-integral for $M_3$, which depends on
$x_1,x_2,x$
and $v$, there is a function  $\bar F:\mathbb{R}^3\rightarrow \mathbb{R}$ such
that 
 $F(x_1,x_2,x,v)=\bar F(I_1,I_2,I_3)$, with $I_1,I_2$ and $I_3$ given by
\[
I_1=\frac{1}{x_1}-\frac{1}{x_2}\,,\qquad 
I_2=\frac{1}{x_1}-\frac{v}{x}\,,\qquad 
I_3=x\,.
\]
Now, in terms of $\bar F$, the condition $M_1F=M_1\bar F=0$ implies
\begin{equation}\label{con22}
v\left( -\frac{2I_1}{I_3}\frac{\partial \bar F}{\partial
I_1}-\frac{2I_2}{I_3}\frac{\partial \bar F}{\partial I_2}+\frac{\partial\bar
F}{\partial I_3}\right)+ (I_1-2I_2)I_1\frac{\partial \bar F}{\partial
I_1}-\left(I_2^2+\frac{k}{I_3^4}\right)\frac{\partial \bar F}{\partial I_2}=0.
\end{equation}
Thus the linear term on $v$ and the other one must vanish independently. The
method of characteristics applied to the first term implies  that there exists a
map $\widehat F:\mathbb{R}^2\rightarrow\mathbb{R}$ such that $\bar
F(I_1,I_2,I_3)=\widehat F(K_1,K_2)$ where
\[ 
K_1=\frac{I_1}{I_2}\,,\qquad 
K_2=I_2I_3^2 \,.
\]
Finally, taking into account the last result in $M_1\hat F=0$, we get
\[
\left(-K_1^2-K_1+\frac{kK_1}{K_2^2}\right)\frac{\partial\widehat F}{\partial
K_1}-\left(K_2+\frac{k}{K_2}\right)\frac{\partial\widehat F}{\partial K_2}=0,
\]
and by means of the method of characteristics expression (\ref{con22}) involves
\[
\frac{dK_1}{dK_2}=\frac{K_1^2+K_1-\frac{kK_1}{K_2^2}}{K_2+\frac{k}{K_2}}
\]
which gives us the first-integral
\[
C_1=K_2+\frac{k+K_2^2}{K_1K_2}\,,
\]
that in terms of the initial variables reads
\begin{equation*}
C_1=\left(x_2-\frac{v}{x}\right)x^2+\frac{k+(x_2-\frac{v}{x})^2
x^4}{(x_1-x_2)x^2}\,.
\end{equation*}
If we repeat this procedure with the assumption that the integral does not
depend on $x_2$ we obtain the following first-integral
\begin{equation*}
C_2=\left(x_3-\frac{v}{x}\right)x^2+\frac{k+(x_3-\frac{v}{x})^2
x^4}{(x_1-x_3)x^2}.
\end{equation*}
It is a long but easy calculation to check that both are first-integrals of
$M_1, M_2$ and $M_3$.
We can obtain now the general solution $x$ of the Milne--Pinney equation in
terms of $x_1, x_2, x_3, C_1, C_2$, as
\begin{equation*}
x=\sqrt{\frac{(C_1(x_1-x_2)-C_2(x_1-x_3))^2+k(x_2-x_3)^2}{
(C_2-C_1)(x_2-x_3)(x_2-x_1)(x_1-x_3)}},
\end{equation*}
where $C_1$ and $C_2$ are constants such that, once $x_1(t)$, $x_2(t)$ and
$x_3(t)$ have been fixed, they make $x(0)$ given by the latter expression be
real.

Thus we have obtained a new mixed superposition rule which enables us to express
 the general solution of the Pinney equation in terms of three 
solutions of Riccati equations and, of course, two constants related to initial
conditions which determine each particular solution.

\chapter{Applications of quantum Lie systems}
In Sections \ref{QLS} and \ref{SLSQM}, it is proved that we can make use of the
geometric theory of Lie systems to treat a certain kind of Schr\"{o}dinger
equations, those related to the so-called quantum Lie systems. In this section
we use this point of view to investigate Quantum Mechanics.

First, we develop the geometric theory of reduction for quantum Lie systems.
Reduction techniques have already been put into practice to study Lie systems
\cite{CarRamGra,CRL07d,CLR07b,CarRam}. In these works, a variety of reduction
methods and other closely related topics are analysed. Most of these methods are
based on the properties of a special type of Lie system in a Lie group
associated with the Lie system under study. As quantum Lie systems can also be
related to such a type of Lie system in a similar way as any Lie system, we can
apply most of the methods developed in the aforementioned works to analyse
Quantum Lie systems. This is the main purpose of the present section.

In detail, we start by analysing the reduction technique for quantum Lie systems
and we complete some previous classic achievements about the topic. We next show
that the interaction picture can be explained from this geometrical point of
view in terms of this reduction technique. Furthermore, the method of unitary
transformations is analysed from our perspective to exemplify that quantum Lie
systems associated with solvable Lie algebras of linear operators, in similarity
with the classical case, can be exactly solved. On the other hand, systems
related to non-solvable Lie algebras can be solved in particular cases. Both
cases can be analysed to reproduce some results on the method of unitary
transformations in particular cases found in the literature.

\section{The reduction method in Quantum Mechanics}\label{reduction}

We here review the reduction techniques explained, for example, in
\cite{CarRamGra,CLR08WN,CarRam}. While in some previous works certain sufficient
conditions to perform a reduction process were explained
\cite{CarRamGra,CarRam}, here we show that these conditions are also as
necessary \cite{CLR08WN}. Additionally, we use the geometric reduction technique
to explain the interaction picture used in Quantum Mechanics and we review, from
a geometric point of view, the method of unitary transformations.

In Section \ref{FNLS} it was shown that the study of Lie systems can be reduced
to
that of finding the solution of the equation
\begin{equation}\label{eqal}
R_{g^{-1}*g} \dot{g}=-\sum_{\alpha=1}^rb_\alpha(t)a_\alpha\equiv {\rm a}(t)\in
{\rm T}_eG
\end{equation}
with $g(0)=e$. 

The reduction method developed in \cite{CarRamGra} shows that given a
solution $\tilde x(t)$ of a Lie system on a homogeneous space $G/H$, 
the solution of the Lie system in the group $G$, and therefore 
the general solution in the given homogeneous space, 
can be reduced to that of a Lie system in the subgroup $H$. 
More specifically,  if the curve $\tilde g(t)$ in $G$ is such that
$\tilde x(t)=\Phi(\tilde g(t),\tilde x(0))$,
 with $\Phi$ being the given action of $G$ in the homogeneous space, then
 $g(t)=\tilde g(t)g'(t)$,
 where $g'(t)$ turns out to be a curve in $H$  which 
is a solution of a Lie system in the
 Lie subgroup $H$ of $G$. 
Actually, once  the curve 
 $\tilde g(t)$ in $G$ has been fixed,  the curve $g'(t)$, that takes values in
 $H$,   
satisfies the equation \cite{CarRamGra}
\begin{equation}\label{trans24}
 R_{g'^{-1}*g'}\dot{g'}=-{\rm Ad}(\tilde g^{-1})
\left(\sum_{\alpha=1}^rb_\alpha(t){\rm a}_\alpha
+R_{\tilde g^{-1}*\tilde g}\dot{\tilde g}\right)
\equiv
{\rm a}'(t)\in {\rm T}_eH\,.
\end{equation}

This transformation law can be understood in the language of the theory of
connections. 
It has been shown in \cite{CarRamGra,CarRam05b} that Lie systems can be related
 to connections in a bundle and that the group of curves in $G$, which is
the group of automorphisms of the principal bundle $G\times\mathbb{R}$
\cite{CarRam05b}, acts on
the left on the set of Lie systems on $G$, and defines an
 induced action on the set of Lie
systems in each homogeneous space for $G$. More specifically, if $x(t)$ 
is a solution of a Lie system in a homogeneous space $N$ defined by the curve
$a(t)$ in $\mathfrak{g}$, then for each curve $\bar g(t)$ in $G$ such that $\bar
g(0)=e$ we see that $x'(t)=\Phi(\bar g(t), x(t))$ is 
a solution of the Lie system defined by the curve 
\begin{equation}
{\rm a}'(t)=R_{\bar g^{-1}*\bar g}\dot{\bar g} +{\rm Ad}(\bar g){\rm
a}(t),\label{redumet}
\end{equation}
which is the transformation law for a connection.

In conclusion, the aim of the reduction method is to find an  
automorphism $\bar g(t)$ such that the right-hand side 
in (\ref{redumet}) belongs to ${\rm T}_eH\equiv\LH$ for a certain Lie
subgroup $H$ of $G$. In this way, 
the papers \cite{CarRamGra,CarRam05b} gave a
sufficient condition
 for obtaining  this result. In this section we study the above
 geometrical development
 in Quantum Mechanics and we determine a necessary condition for the right-hand
 side in (\ref{redumet})
 to belong to $\LH$.

Quantum Lie systems are those $t$-dependent self-adjoint Hamiltonians such that
\begin{equation}
 H(t)=\sum_{\alpha=1}^rb_\alpha(t)H_\alpha,
\end{equation}
with  $iH_\alpha$ closing  under the commutator of operators on a 
finite-dimensional real Lie algebra of skew-self-adjoint operators $V$.
Therefore,
by regarding these operators as fundamental vector fields 
of a unitary action of a connected Lie group $G$ with  
Lie algebra $\LG$ isomorphic to $V$, 
we can relate the Schr\"odinger equation to a 
differential equation in $G$ determined by 
curves in ${\rm T}_eG$ given by 
${\rm a}(t)=-{\displaystyle\sum_{\alpha=1}^r}b_\alpha(t){\rm a}_\alpha$ 
by considering $-iH_\alpha$ as fundamental vector 
fields of the basis of $\LG$ given by $\{{\rm
a}_\alpha\,|\,\alpha=1,\ldots,r\}$.

Now, the preceding methods enable us to
transform the problem into a new one in  the same group $G$, for each choice
of the curve $\bar g(t)$  but with a new curve ${\rm a}'(t)$. The action of $G$
on $\mathcal{H}$  is given by a unitary representation $U$, and therefore the
$t$-dependent vector
field determined by the original $t$-dependent Hamiltonian $H(t)$  becomes a new
one with
$t$-dependent Hamiltonian $H'(t)$. Its integral curves are the solutions of 
the equation 
\begin{equation*}
\frac{d\psi'}{dt}=-iH'(t)\psi',
\end{equation*}
where
\begin{equation*}
-iH'(t)=-iU(\bar g(t))H(t)U^\dagger(\bar g(t))
+\dot U(\bar g(t))U^\dagger(\bar g(t))
\end{equation*}

 That is, from a geometric point of view, we 
have related a Lie system on the Lie group $G$ to
 certain curve ${\rm a}(t)$ in ${\rm T}_eG$ and the corresponding 
system in $\mathcal{H}$ determined by a unitary 
representation of $G$ to another one with different 
curve ${\rm a}'(t)$ in ${\rm T}_eG$ and its associated 
one in $\mathcal{H}$.

Let us choose a basis of ${\rm T}_eG$
given by $\{c_\alpha\mid \alpha=1,\ldots,r\}$ 
with $r=\dim\, \LG$, such that
$\{c_\alpha\mid \alpha=1,\ldots,s\}$ be a basis of $T_eH$, 
where $s=\dim\,\LH$, 
and denote $\{c^\alpha\mid\alpha=1,\ldots,r\}$ 
the dual basis of $\{c_\alpha\mid
\alpha=1,\ldots,r\}$. In order to find
  $\bar g$  such that the right-hand term of (\ref{redumet})
belongs to ${\rm T}_eH$ for all $t$, the
condition for $\bar g$ is
\begin{equation*}
c^\alpha\left( {\rm Ad} (\bar g){\rm a}(t)+R_{\bar g^{-1}*\bar g}\dot{\bar
g}\right) =0\,,\qquad \alpha=s+1,\ldots,r\,.
\end{equation*}
Now, if $\theta^\alpha$ is the left invariant 1-form on
$G$ induced from $c^\alpha$, the previous equation implies
\begin{equation*}
\theta^\alpha_{\bar g^{-1}}
\left(R_{\bar g^{-1}*e}{\rm a}(t)-\frac{d{\bar g}^{-1}}{dt}\right)=0\,,\qquad
\alpha=s+1,\ldots,r\,.\\
 \end{equation*}

Let be $\tilde g=\bar g^{-1}$, the latter expression implies that 
$R_{\tilde g*e}{\rm a}(t)-\dot{\tilde g}$ is generated by left invariant vector
fields on $G$ from the elements of $\LH$. Then, 
given $\pi^L:G\to G/H$, the kernel of $\pi^L_*$ is spanned 
by the left invariant vector fields on $G$ generated 
by the elements of $\LH$. Then it follows
\begin{equation}
\pi^L_{*\tilde g}(R_{\tilde g*e}{\rm a}(t)-\dot{\tilde g})=0\,.
\end{equation}

Therefore, if we use that $\pi^L_*\circ X^R_\alpha=-X^L_\alpha\circ\pi^L$,
where $X^L_\alpha$ denotes the fundamental vector field of 
the action of $G$ in $G/H$ and $X^R_\alpha$ denotes 
the right-invariant vector field in $G$ whose value in $e$ is ${\rm a}_\alpha$,
we can prove 
that $\pi^L(\tilde g)$ is a solution on $G/H$ of the equation
\begin{equation}\label{redeq}
\frac{d\pi^L(\tilde g)}{dt}
=\sum_{\alpha=1}^rb_\alpha(t)X_\alpha^L(\pi^L(\tilde g))\,.
\end{equation}

Thus, we obtain that given a certain solution $g'(t)$ in $\LH$ related 
to the initial $g(t)$ by means of $\tilde g(t)$ according to 
$g(t)=\tilde g(t)g'(t)$, then the projection to $G/H$ of $\tilde g(t)$, i.e.
$\pi^L(\tilde g(t))$, is a solution of (\ref{redeq}).
This result shows that whenever $g'(t)$ 
is a curve in $H$, then $\tilde g(t)$ satisfies equation (\ref{redeq}).
Moreover, as it has been shown in \cite{CarRamGra},
 if $\tilde g(t)$ satisfies 
(\ref{redeq}), then $g'(t)$ is a curve in $H$
satisfying (\ref{trans24}). 
The previous result shows that such a condition for 
obtaining (\ref{trans24}) is 
not only sufficient but necessary too. Thus, we provide a new result which 
completes that one found in \cite{CarRamGra}.

Finally, it is worth noting that even when this last proof has been developed
for Quantum Mechanics, it can also be applied to ordinary differential
equations, because it appears as a consequence of the group structure of Lie
systems which
is the same for both quantum and ordinary Lie systems.

\section{Interaction picture and Lie systems}

As a first application of the reduction method for Lie systems, we analyse here
how this theory can be applied to explain the interaction picture used in
Quantum Mechanics. This picture has been proved to be very effective in the
developments
of perturbation methods. It plays a r\^ole when the $t$-dependent Hamiltonian
can be written 
as a linear combination with $t$-dependent coefficients of a simpler Hamiltonian
$H_1$ and a perturbation $V(t)$. 
In the framework of Lie systems, we can analyse what happens when the
$t$-dependent Hamiltonian is 
\begin{equation*}
H(t)=H_1+V(t)=H_1+\sum_{\alpha=2}^rb_\alpha(t)H_\alpha
=\sum _{\alpha=1}^rb_\alpha(t)H_\alpha,\quad  b_1(t)=1,
\end{equation*}
where the set of skew-self-adjoint operators 
$\{-iH_\alpha\, |\, \alpha=1,\ldots,r\}$ is closed 
under commutation and generates a finite dimensional 
real Lie algebra. The situation is very similar to the case of 
control systems with a drift term (here $H_1$) that are linear in the 
control functions. The functions $b_\alpha(t)$ correspond to the control
functions. 

According to the theory of Lie systems, take 
a basis $\{{\rm a}_\alpha\, |\, \alpha=1,\ldots,r\}$ of the 
Lie algebra with corresponding associated fundamental 
vector fields $-iH_\alpha$. The equation to be studied 
in ${\rm T}_eG$ is the one (\ref{eqal})
and whether we define $g'(t)=\bar g(t)g(t)$, where $\bar g(t)$ is a 
previously chosen curve, it obeys a similar equation for $g'(t)$
given by (\ref{redumet}).  

If, in particular, we choose $\bar g(t)=\exp({\rm a}_1t)$, 
we find the new equation in ${\rm T}_eG$
\begin{equation}\label{cc}
R_{g'^{-1}*g'}\dot g'=-{\rm Ad}(\exp({{\rm a}_1t}))
\left(\sum_{\alpha=2}^rb_\alpha(t){\rm a}_\alpha\right)
=-\exp({{\rm ad}({\rm a}_1)t})\left(\sum_{\alpha=2}^rb_\alpha(t){\rm
a}_\alpha\right)\,.
\end{equation}
Correspondingly, the action of $G$ on $\mathcal{H}$ by a unitary 
representation defines a transformation on $\mathcal{H}$ in 
which the state $\psi_t$ transforms into 
$\psi'_t=\exp(iH_1t)\psi_t$ and its dynamical evolution 
is given by the vector field corresponding to the right-hand side of (\ref{cc}).
In particular, 
if $\{{\rm a}_2,\ldots,{\rm a}_r\}$ span an ideal of the 
Lie algebra $\LG$, the problem reduces to the 
corresponding normal subgroup in $G$.

\section{The method of unitary transformations}

A second application of the theory of Lie systems in Quantum Mechanics and, in
particular, of the reduction 
method is to obtain information about how to proceed to solve a quantum Lie
Hamiltonian. Let us discuss here a general procedure to accomplish this task.

Every Schr\"odinger equation of Lie type is determined by a Lie algebra $\LG$,
 a unitary representation of its connected and simply connected Lie group $G$ on
$\mathcal{H}$,
and a curve ${\rm a}(t)$ in ${\rm T}_eG$. Depending on $\mathfrak{g}$, there are
two cases. If $\LG$ is solvable, we can use the reduction method in Quantum
Mechanics to obtain the general solution. If $\LG$ is not solvable, it is not
known how to
integrate the problem in terms of quadratures in the most general case. 
Nevertheless, it
is possible to solve the problem completely for some specific curves as
for instance 
it happens for the Caldirola--Kanai Hamiltonian \cite{HW98}. A way of dealing
with such systems 
consist in trying to transform the curve ${\rm a}(t)$ into another one ${\rm
a}'(t)$, 
easier to handle, as it has been done in
the previous section for the interaction picture. 
In a more general case than the interaction picture, although any
two curves ${\rm a}(t)$ and ${\rm a}'(t)$ are always connected by an
automorphism, 
the equation
determining 
the transformation can be as difficult to solve as the initial problem. 
Because of this,
it is interesting to
 look for a curve that:
\begin{enumerate}
 \item It determines an easily solvable equation.
 \item It can be transformed through an explicitly known transformation into the
curve associated with our initial problem.
\end{enumerate}
This is the topic of next three sections, where conditions for such
Schr\"odinger equations are analysed.
 In any case, we can always express the solution of the initial 
problem in terms of a solution of the equation determining the transformation. 
In certain cases, for an appropriate choice of the curve $\bar g(t)$
 the new curve 
${\rm a}'(t)$ belongs to ${\rm T}_eH$ for all $t$, where $H$ is a solvable 
Lie subgroup of $G$. In
this case we can 
reduce the problem from $\LG$ to a certain solvable Lie subalgebra $\LH$ of
$\LG$. Of course, 
in order to do this, a solution of the equation of reduction is
needed, 
but once this is known we can solve the problem completely in terms of it. 
Other methods have alternatively been used in the
literature, like the Lewis-Riesenfeld (LR) method. 
However, this method seems to offer a complete solution 
only if $\LG$ is solvable. If $\LG$ is not solvable, the LR method 
offers a solution which depends on a solution of a system of differential
equations, 
like in the method of reduction. 

To sum up, given a Lie system an associated with a Lie
algebra $\LG$, whose Lie group $G$ acts, by unitary operators,
on $\mathcal{H}$, and determined by a curve
${\rm a}(t)$ in ${\rm T}_eG$, 
the systematic procedure to be used is the following: 
\begin{itemize}
\item If $\LG$ is solvable, we can solve 
the problem easily by quadratures as it  
appears in \cite{Fe01, Gu01}.
\item  If $\LG$ is not solvable, we can try to solve 
the problem for a given curve like in the 
Caldirola--Kanai Hamiltonian in \cite{HW98}, by
choosing
 a curve $\bar g(t)$ transforming
the curve ${\rm a}(t)$ into another one easier to solve, 
like in the interaction picture. If this does not work we can try to reduce the
problem to an integrable case like in the $t$-dependent mass and frequency
harmonic oscillator or quadratic one dimensional Hamiltonian in
\cite{CLR08,FM03,So00,YKUGP94}.
\end{itemize}

\section{{\it t}-dependent operators for quantum Lie systems} 

In this section we apply our methods to obtain the $t$-dependent
evolution operators of several problems
found in the Physics literature in an algorithmic way. 

We first provide a simple example in order to illustrate the main points of our
theory. Next, we analyse $t$-dependent quadratic Hamiltonians. These
Hamiltonians describe a very large class of physical models. 
Sometimes, one of these physical models is described by a certain 
family of quadratic Hamiltonians associated with a Lie subalgebra of operators
of the one given for general
quadratic Hamiltonians. If this Lie subalgebra  is solvable,
the differential equations related to it through the 
Wei--Norman methods are solvable too and the $t$-evolution 
operator can be explicitly obtained. In these cases, we can find the 
explicit solution of these problems in the literature using different methods
for each case.
We also describe some approaches 
to study these quantum Lie systems in the non-solvable cases. 

\section{Initial examples}
We start our investigation by studying the motion of a particle 
with a $t$-dependent mass under the action of a $t$-dependent 
linear potential term. The Hamiltonian describing 
this physical case is
\begin{equation*}
 H(t)=\frac{P^2}{2m(t)}+S(t)X\,.
\end{equation*}

The Lie algebra associated with this example is a central
 extension of the Heisenberg Lie algebra. 
A basis for the Lie algebra of vector fields related 
to this physical model is
\begin{equation*}
Z_1=i\frac{P^2}{2},\quad Z_2=iP,\quad Z_3=iX,\quad Z_4=iI,
\end{equation*}
which closes on a Lie algebra under the commutation relations
\begin{equation*}
\begin{aligned}
\left[Z_1, Z_2\right]&=0,
\quad&\left[Z_1, Z_3\right]&=2Z_2,
\quad&\left[Z_1,Z_4\right]&=0,\\
\left[Z_2, Z_3\right]&=Z_4,\quad&\left[Z_2, Z_4\right]&=0,\quad&&\\
\left[Z_3, Z_4\right]&=0. &&&&\\
\end{aligned}
\end{equation*}
This Lie algebra is solvable, and then, the related equations 
obtained through the Wei--Norman method, 
can be solved by quadratures for any pair 
of $t$-dependent coefficients $m(t)$ and $S(t)$. 
The solution of the associated Wei-Norman system
allows us to obtain the $t$-evolution operator 
and the wave function solution of the $t$-dependent 
Schr\"odinger equation. 

This $t$-dependent Hamiltonian has been studied in \cite{UYG02}  for some
particular cases 
 using {\it ad-hoc} methods
and in general in \cite{Fe01}. Here, we investigate it through  
the Wei--Norman method. Its equation in the group $G$ with ${\rm T}_eG\simeq V$,
is 
\begin{equation*}
R_{g^{-1}*g}\dot g=-\frac 1{m(t)}a_1-S(t)a_3\equiv a_{MS}(t)\,,
\end{equation*}
where th,e $a_1,\ldots,a_4,$ are a basis of $\mathfrak{g}$ closing on the same
commutation relations as the operators, $Z_1,\ldots,Z_4$.
The factorisation
\begin{equation*}
g(t)=\exp(v_2(t)a_2)\exp(-v_3(t)a_3)\exp(-v_4(t)a_4)\exp(-v_1(t)a_1),
\end{equation*}
allows us to solve the equation in $G$ by the Wei--Norman method to get 
\begin{equation*}
\begin{aligned}
\dot{v}_1&=\dfrac1{m(t)}\,,\cr
\dot{v}_2&=\dfrac{v_3}{m(t)}\,,\cr
\dot{v}_3&=S(t)\,,\cr
\dot{v}_4&=-S(t)v_2-\dfrac{v_3^2}{2m(t)}\,,\cr
\end{aligned}
\end{equation*}
with initial conditions $v_1(0)=v_2(0)=v_3(0)=v_4(0)=0$.
The solution of this system can be expressed using 
quadratures because the related group is solvable
\begin{equation}\label{nonconstmass}
\begin{array}{l}
v_1(t)={\displaystyle\int^t_0}\dfrac{du}{m(u)},\cr
v_2(t)={\displaystyle\int^t_0}\dfrac{du}{m(u)}
\left( {\displaystyle\int^u_0}S(v)dv\right), \cr
v_3(t)={\displaystyle\int^t_0}S(u)du,\cr
v_4(t)=-{\displaystyle\int^t_0}S(u)
\left({\displaystyle \int^u_0}\dfrac{dv}{m(v)}
\left({\displaystyle \int^v_0}S(w)dw\right)\right)du 
-{\displaystyle\int^t_0}\dfrac{du}{2m(u)}
\left({\displaystyle\int^u_0}S(v)dv\right)^2,
\end{array}
\end{equation}
and the $t$-evolution operator is 
\begin{equation*}
\begin{aligned}
U(g(t))&=\exp(v_2(t)Z_2)\exp(-v_3(t)Z_3)\exp(-v_4(t)Z_4)\exp(-v_1(t)Z_1)\\
 &=\exp(iv_2(t)P)\exp(-iv_3(t)X)\exp(-iv_4(t)I)\exp(-iv_1(t)\frac{P^2}2).
\end{aligned}
\end{equation*}

\section{Quadratic Hamiltonians}
After dealing with an easy example before, we can proceed now in a similar way
in order to treat the $t$-dependent quadratic Hamiltonian given by \cite{KBW}
(see \cite{CarRam03}) 
\begin{equation}
H(t)=\alpha(t)\,\frac{P^2}2+\beta(t)\,\frac{X\,P+P\,X}4+\gamma(t)\,\frac{X^2}2+
\delta(t)P+\epsilon (t)\, X+\phi(t)I ,\label{gqH}
\end{equation}
where $X$ and $P$ are the position and momentum operators satisfying the
commutation relation 
$$[X,P]=i \, I\ .
$$
It is important to solve this quantum quadratic Hamiltonian because it
frequently appears in Quantum Mechanics.

In order to prove that (\ref{gqH}) is a quantum Lie system, we must check that
this $t$-dependent Hamiltonian can be written as a sum with $t$-dependent
coefficients of some self-adjoint Hamiltonians closing on a real
finite-dimensional Lie algebra of operators. 

As we can write
$$H(t)=\alpha(t)\, H_1+ \beta(t)\, H_2+\gamma(t)\, H_3-\delta(t)\,
H_4+\epsilon(t)\, H_5\,+\phi(t)\,H_6,
$$
with the Hamiltonians
$$
H_1=\frac {P^2}2\,,\quad H_2= \frac 14 (XP+PX),\quad \
H_3=\frac {X^2}2\,,
$$
$$ H_4=-P\,,\qquad  H_5=X\,, \qquad 
H_6=I\,,
$$
satisfying the commutation relations
\begin{equation*}
\begin{aligned}
&[iH_1,iH_2]=iH_1, &[iH_2,iH_3]&=iH_3, &[iH_3,iH_4]&=iH_5, &[iH_4,iH_5]=-iH_6,\\
&[iH_1,iH_3]=2i\, H_2,&[iH_2,iH_4]&=-\frac i2\, H_4,&[iH_3,iH_5]&=0,& \\
&[iH_1,iH_4]=0,&[iH_2,iH_5]&=\frac i2\, H_5\,,\ &&&\\
&[iH_1,iH_5]=-iH_4\,,&& && & 
\end{aligned}
\end{equation*}
and $[iH_\alpha,iH_6]=0$, $\alpha=1,\dots,5$, we get that $H(t)$ is a quantum
Lie system.

This means that the skew-self-adjoint operators $i\, H_\alpha$ generate a 
six-dimensional real Lie $V$ algebra of operators. Now, we can relate them to
the basis $\{{\rm a}_1,\ldots,{\rm a}_6\}$ for an abstract real Lie algebra
isomorphic to the one spanned by the $-iH_\alpha$. This basis is chosen in such
a way that
 \begin{equation*}
\begin{aligned}
&[{\rm a}_1,{\rm a}_2]={\rm a}_1,&[{\rm a}_2,{\rm a}_3]&={\rm a}_3,&[{\rm
a}_3,{\rm a}_4]&={\rm a}_5,&[{\rm a}_4,{\rm a}_5]&=-{\rm a}_6,&[{\rm a}_5,{\rm
a}_6]&=0,\\
&[{\rm a}_1,{\rm a}_3]=2\, {\rm a}_2,&[{\rm a}_2,{\rm a}_4]&=-\frac 12\, {\rm
a}_4,&[{\rm a}_3,{\rm a}_5]&=0,&[{\rm a}_4,{\rm a}_6]&=0,&&\\
&[{\rm a}_1,{\rm a}_4]=0,&[{\rm a}_2,{\rm a}_5]&=\frac 12\, {\rm a}_5, &[{\rm
a}_3,{\rm a}_6]&=0,&& &&\\
&[{\rm a}_1,{\rm a}_5]=-{\rm a}_4,&[{\rm a}_2,{\rm a}_6]&=0,&& && &&\\
&[{\rm a}_1,{\rm a}_6]=0\,.&&&&&&&&
\end{aligned}
\end{equation*}  
This six-dimensional real Lie algebra is a semidirect sum of the Lie algebra
$\mathfrak{sl}(2,\mathbb{R})$  
spanned by $\{{\rm a}_1,{\rm a}_2,{\rm a}_3\}$ and the  Heisenberg--Weyl 
Lie algebra generated by $\{{\rm a}_4,{\rm a}_5,{\rm a}_6\}$, which is an ideal.

In order  to find the $t$-evolution provided by the $t$-dependent Hamiltonian
(\ref{gqH}) 
we should find the curve $g(t)$ in $G$, with ${\rm T}_eG\simeq V$, such that 
$$R_{g^{-1}*g}\dot g=-\sum_{\alpha=1}^6 b_\alpha(t)\, {\rm a}_\alpha\ ,
\qquad g(0)=e\,,
$$
with 
$$
b_1(t)=\alpha(t)\,,\ b_2(t)=\beta(t)\,,\ b_3(t)=\gamma(t)\,,\
b_4(t)=-\delta(t)\,,\ b_5(t)=\epsilon(t)\,,\ b_6(t)=\phi(t)\,.
$$

This can be carried out by using the generalised Wei--Norman method, i.e. 
by writing the curve $g(t)$ in $G$ in terms of a set of second class canonical
coordinates.  
For instance, 
\begin{eqnarray}\label{factorization}
g(t)&=&\exp(-v_4(t){\rm a}_4)\exp(-v_5(t){\rm a}_5)\exp(-v_6(t){\rm
a}_6)\times\cr
&&\times \exp(-v_1(t){\rm a}_1)\exp(-v_2(t){\rm a}_2)\exp(-v_3(t){\rm a}_3),
\end{eqnarray}
and a straightforward application of the above mentioned Wei--Norman
 method technique leads to the system
\begin{eqnarray}\label{WNQH}
\left\{\begin{aligned}
\dot v_1&=b_1+b_2\, v_1+b_3\,v_1^2\ ,&\quad&\dot v_4=b_4+\frac 12\, b_2\,
v_4+b_1\,v_5\ ,\cr
\cr
\dot v_2&=b_2+2\,b_3\,v_1\ ,&\quad&\dot v_5=b_5-b_3\, v_4-\frac 12\, b_2\,v_5\
,\cr
\cr
\dot v_3&=e^{v_2}\,b_3\ , &\quad&\dot v_6=b_6-b_5\, v_4+\frac 12\,
b_3\,v_4^2-\frac 12\, b_1\,v_5^2,
\cr
\end{aligned}\right.
\end{eqnarray}
with $v_1(0)=v_2(0)=v_3(0)=v_4(0)=v_5(0)=v_6(0)=0$.

If we consider the following vector fields 
\begin{equation}
\begin{aligned}
X_1&=\pd{}{v_1}+v_5\pd{}{v_4}-\frac{1}{2}v_5^2\pd{}{v_6},\\
X_2&=v_1\pd{}{v_1}+\pd{}{v_2}+\frac{1}{2}v_4\pd{}{v_4}-\frac{1}{2}v_5\pd{}{v_5},
\\
X_3&=v_1^2\pd{}{v_1}+2v_1\pd{}{v_2}+e^{v_2}\pd{}{v_3}-v_4\pd{}{v_5}+\frac{1}{2}
v_4^2\pd{}{v_6},\\
X_4&=\pd{}{v_4},\\
X_5&=\pd{}{v_5}-v_4\pd{}{v_6},\\
X_6&=\pd{}{v_6}\,,\\
\end{aligned}
\end{equation}
we can check that these vector fields satisfy 
the same commutation relations as the corresponding $\{{\rm
a}_\alpha\,|\,\alpha=1,\ldots,6\}$ and 
thus, system (\ref{WNQH}) is a Lie system related to the same 
Lie algebra as the $t$-dependent Hamiltonian (\ref{gqH}) or its corresponding 
equation in a Lie group.

Now, once the functions $v_\alpha(t)$, with $\alpha=1,\ldots,6$, have been
determined, the 
$t$-evolution of any state is given by 
\begin{multline*}
\psi_t=\exp(-v_4(t)iH_4)\exp(-v_5(t)iH_5)\exp(-v_6(t)iH_6)\times\cr
\times \exp(-v_1(t)iH_1)\exp(-v_2(t)iH_2)\exp(-v_3(t)iH_3)\psi_0,
\end{multline*}
and thus
\begin{multline}\label{solutionQH}
\psi_t=\exp(v_4(t)iP)\exp(-v_5(t)iX)\exp(-v_6(t)iI)\times\cr
\times \exp\left(-v_1(t)i\frac{P^2}{2}\right)
\exp\left(-v_2(t)i\frac{PX+XP}{4}\right)
\exp\left(-v_3(t)i\frac{X^2}{2}\right)\psi_0.
\end{multline}

\section{Particular cases}
$t$-dependent quadratic Hamiltonians describe a very large class of physical
models.
Sometimes, one of these physical models is described by a certain 
family of quadratic Hamiltonians that can be regarded as   a
quantum Lie system related to a Lie subalgebra of the one given for general
quadratic Hamiltonians. If they are associated with a Lie solvable subalgebra,
then
the system of differential equations related to it through the 
Wei--Norman method is solvable too and the $t$-evolution 
operator can be explicitly obtained. In this section 
we treat some instances of this case through a unified approach. In these
instances, we can also find the 
explicit solutions of these problems in the literature, but by different {\it ad
hoc} methods.

Once we have obtained the solution for a generic quadratic Hamiltonian $H(t)$,
we can review the solution for a system with constant mass and linear potential
given by
\begin{equation}
 H(t)=\frac{P^2}{2m}+S(t)X\,,
\end{equation}
to obtain, in view of equations (\ref{WNQH}), 
\begin{equation*}
\begin{array}{l}
v_1(t)=\dfrac{t}{m},\cr
v_2(t)=0,\cr
v_3(t)=0,\cr
v_4(t)=\dfrac{1}{m}{\displaystyle\int^t_0}
\left({\displaystyle \int^u_0}S(v)dv\right)du, \cr
v_5(t)={\displaystyle\int^t_0}S(u)du,\cr
v_6(t)=-\dfrac{1}{m}{\displaystyle\int^t_0}
\left(S(u){\displaystyle\int^u_0}
\left({\displaystyle\int^v_0}S(w)dw\right)dv\right)du
-\dfrac{1}{2m}{\displaystyle\int^t_0}
\left({\displaystyle\int^u_0}S(v)dv\right)^2du\,,
\end{array}
\end{equation*}
which give the $t$-evolution operator 
if we substitute them into the $t$-evolution operator (\ref{solutionQH}).

Now we can consider particular instances of this $t$-dependent Hamiltonian. 
For example, for the curves with constant mass 
$m$ and $S(t)=q\epsilon_0+q\,\epsilon \,\cos(\omega t)$, 
studied in \cite{Gu01}, we obtain
\begin{equation*}
\begin{gathered}
v_1(t)=\dfrac tm\,,\quad v_2(t)=0\,,\quad v_3(t)=0\,,\quad \\
v_4(t)=\dfrac{q}{2m\omega^2}
(2\epsilon+\epsilon_0 \omega^2 t^2-2\epsilon \cos(\omega t))\,,\quad
v_5(t)=\dfrac{q}{\omega}(\epsilon_0 \omega t+\epsilon \sin(\omega t))\,,\quad
\end{gathered}
\end{equation*}
and 
\begin{multline*}
v_6(t)=\dfrac{-q^2}{12 m\omega^3}
\left(4 \epsilon_0^2\omega^3 t^3
-3\epsilon(\epsilon-4\epsilon_0)\omega t
\right.+\\ 
\left. 3\epsilon(4\epsilon+2 \epsilon_0(\omega^2 t^2-2)
-3\epsilon\cos(\omega t))\sin(\omega t)\right)\,.
\end{multline*}

The procedure to obtain a solution with arbitrary non-constant mass and 
$S(t)=q\epsilon_0+q\epsilon \cos(\omega t)$ was pointed out in
\cite{Gu01} and solved in \cite{Fe01}. From our point of view, 
the most general solution comes straightforwardly 
from expression (\ref{nonconstmass}), because all cases in 
the literature are particular instances of our  
approach with general functions $m(t)$ and $S(t)$.

Now, we can obtain the wave function solution 
of this system. We know that the wave function 
solution $\psi_t$ with initial condition $\psi_0$ is
\begin{equation*}
\begin{aligned}
\psi_t(x)&=U(g(t))\psi(x,0)\\
&=\exp(iv_6(t))\exp(-v_4(t)iP)\exp(-v_5(t)iX)
\exp\left(-v_1(t)i\frac{P^2}{2}\right)\psi_0(x).
\end{aligned}
\end{equation*}
However, if we express the initial wave function $\psi_0(x)$ 
in the momentum space as $\phi_0(p)$, 
the solution will take a similar form as before but with
$U(g(t))$ in the momentum representation.
In this case the solution with initial condition $\phi_0(p)$ is
\begin{equation*}
\begin{aligned}
 \phi_t(p)&=U(g(t))\phi_0(p)\\
&=\exp(-iv_6(t))\exp(v_4(t)iP)\exp(-v_5(t)iX)
\exp\left(-iv_1(t)\frac{P^2}{2}\right)\phi_0(p)\\
&=\exp(-iv_6(t))\exp(v_4(t)iP)\exp(-v_5(t)iX)
\exp\left(-iv_1(t)\frac{p^2}{2}\right)\phi_0(p)\\
&=\exp(-iv_6(t))\exp(v_4(t)iP)
\exp\left(-iv_1(t)\frac{(p+v_5(t))^2}{2}\right)\phi_0(p+v_5(t))\\
&=\exp\left(-iv_6(t)+iv_4(t)p-iv_1(t)
\frac{(p+v_5(t))^2}{2}\right)\phi_0(p+v_5(t)).
\end{aligned}
\end{equation*}

\section{Non-solvable Hamiltonians and particular instances}

In the preceding section the differential equations associated 
with the $t$-dependent quantum Hamiltonians were Lie systems 
related to a solvable Lie algebra. 
Thus, it was proved that the differential equations obtained were integrable by
quadratures through the Wei--Norman method. 
If this does not happen, it is not easy to obtain a 
general solution. 
Now, we describe some examples of `non-solvable' $t$-dependent quadratic
Hamiltonians. In general we do not 
obtain a general solution in terms of the 
$t$-dependent functions of the quadratic 
Hamiltonians. Nevertheless, we show that 
for some instances of them, whose coefficients satisfy certain integrability
conditions \cite{CLR08, CLRan08}, the differential 
equations can be integrated.

As a first case, consider the Hamiltonian 
for a forced harmonic oscillator with 
$t$-dependent mass and frequency given by
$$
 H(t)=\frac{P^2}{2m(t)}+\frac{1}{2}m(t)\omega^2(t)X^2+f(t)X\,.
$$

This case, either with or without $t$-dependent frequency, 
has been studied in \cite{Ci1,Gu01,YKUGP94}. 
The equations describing the solutions of 
this Lie system by the method of Wei--Norman are
$$
\begin{aligned} 
 \dot{v}_1&=\dfrac{1}{m(t)}+m(t)\omega^2(t)v_1^2,\\
\dot{v}_2&=2m(t)\omega^2(t)v_1,\\
\dot{v}_3&=e^{v_2}m(t)\omega^2(t),\\
\dot{v}_4&=\dfrac{1}{m(t)}v_5,\\
\dot{v}_5&=f(t)-m(t)\omega^2(t)v_4,\\
\dot{v}_6&=\dfrac{1}{2}m(t)\omega^2(t)v_4^2-f(t)v_4-\dfrac{1}{2m(t)}v_5^2,
\end{aligned}
$$
with initial conditions $v_1(0)=v_2(0)=v_3(0)=v_4(0)=v_5(0)=v_6(0)=0$,
where the factorisation (\ref{factorization}) has been used. 
The solution of this system cannot be obtained by quadratures 
in the general case because the associated 
Lie algebra is not solvable. Nevertheless, we can consider 
a particular instance of this kind of Hamiltonian, 
the so-called Caldirola--Kanai Hamiltonian \cite{HW98}. 
In this case, for the particular $t$-dependence 
$m(t)=e^{-rt}m_0$, $\omega(t)=\omega_0$ and $f(t)=0$ 
the Hamiltonian reads
\begin{equation*}
 H(t)=\frac{P^2}{2m_0}e^{rt}+\frac{1}{2}m_0e^{-rt}\omega_0^2X^2\,.
\end{equation*}

In this case the solution is completely known and is given by
\begin{equation*}
\begin{aligned}
v_1(t)&=\dfrac{2e^{rt}}{m_0(r+\bar \omega_0\,{\rm coth}\,(\frac{t}{2}\bar
\omega_0))}
\,, \\
v_2(t)&=r t+2\log{\bar \omega_0}-2\log\left(r\sinh\left(\frac{t}{2}\bar
\omega_0\right)
+\bar \omega_0 \cosh\left(\frac{t}{2}\bar \omega_0\right)\right)\,,\\
v_3(t)&=\dfrac{2 m_0 \omega_0^2}{r+\bar \omega_0\,{\rm coth}\,(\frac{t}{2}\bar
\omega_0)}\,,
\\
v_4(t)&=0\,,\quad v_5(t)=0\,,\quad v_6(t)=0\,,
\end{aligned}
\end{equation*}
where $\bar{\omega}_0=\sqrt{r^2-4 \omega_0^2}$. 
This example shows that the problem may also be 
exactly solved for particular instances of curves 
in $\LG$ of Lie systems with non solvable Lie algebras. 
Another example is the following one
\begin{equation*}
 H(t)=\frac{P^2}{2m}+\frac{m \omega_0^2}{2(t+k)^2}X^2\,,
\end{equation*}
for which the solution of the Wei--Norman system reads
\begin{eqnarray*}
\begin{aligned}
v_1(t)&=\dfrac{2(k+t)((k+t)^{\bar \omega_0}-k^{\bar \omega_0})}
{m(k^{\bar \omega_0}(\bar \omega_0-1)+(k+t)^{\bar \omega_0}(\bar
\omega_0+1))}\,,\\
v_2(t)&=(1+\bar \omega_0)\log(k+t)-(1+\bar \omega_0)\log k
+2\log(2 k^{\bar \omega_0}\bar \omega_0)\\
 &-2\log(k^{\bar \omega_0}(\bar \omega_0-1)+(k+t)^{\bar \omega_0}(\bar
\omega_0+1))\,,\\
v_3(t)&=\dfrac{2 m \omega_0^2}{k}\dfrac{(k+t)^{\bar \omega_0}-k^{\bar \omega_0}}
{k^{\bar \omega_0}(\bar \omega_0-1)+(k+t)^{\bar \omega_0}(\bar \omega_0+1)}\\
v_4(t)&=0\,,\quad v_5(t)=0\,,\quad v_6(t)=0\,,
\end{aligned}
\end{eqnarray*}
where now $\bar \omega_0=\sqrt{1-4 \omega_0^2}$.

Other examples of 
Hamiltonians, which can be studied by our method, can
be found in \cite{HW98}. We just mention two examples 
which can be completely solved
\begin{eqnarray*}
 H_1(t)&=&\frac{P^2}{2m_0}+\frac{1}{2}m_0(U+V\cos(\omega_0 t))X^2\,,\\ 
H_2(t)&=&\frac{P^2}{2m_0}e^{rt}+\frac{1}{2}m_0e^{-rt}\omega_0^2X^2+f(t)X\,.
\end{eqnarray*}
The first one corresponds to a Paul trap which has been 
studied in \cite{FW95} and admits a solution in terms of 
Mathieu's functions. The second one is a damped Caldirola--Kanai
Hamiltonian analysed in \cite{UYG02}.

\section{Reduction in Quantum Mechanics}

Quite often, when a quantum Lie system is related to a non-solvable Lie algebra,
it is
interesting to solve it in terms of (unknown) solutions of differential
equations. Next, we study some examples of how to
proceed with the method of reduction in order to deal with 
problems in this way. So, we obtain that the reduction method can be applied not
only to analyse systems of differential equations but also enables to solve
certain quantum problems in an algorithmic way. 

Consider a harmonic oscillator with $t$-dependent 
frequency whose Hamiltonian is given by
\begin{equation*}
 H(t)=\frac{P^2}{2}+\frac{1}{2}\Omega^2(t)X^2\,.
\end{equation*}
As a particular case of the Hamiltonian described 
in  Section \ref{SLSQM}, this example is related 
to an equation in the connected Lie group associated 
with the semidirect sum of $\mathfrak{sl}(2,\mathbb{R})$, spanned by the
elements $\{{\rm a}_1,{\rm a}_2,{\rm a}_3\}$, with the 
Heisenberg Lie algebra generated by the ideal 
$\{{\rm a}_4,{\rm a}_5,{\rm a}_6\}$
\begin{equation}\label{inieqG}
R_{g^{-1}*g}\dot g=-{\rm a}_1-\Omega^2(t){\rm a}_3,\quad g(0)=e.
\end{equation}
Since the solution of this equation starts from the identity
and $\{{\rm a}_1, {\rm a}_2, {\rm a}_3\}$ close on a
$\mathfrak{sl}(2,\mathbb{R})$ Lie algebra,
then the $t$-dependent Hamiltonian $H(t)$ is related 
to the group $SL(2,\mathbb{R})$. 

As a particular application of the reduction technique 
we will perform the reduction from $G=SL(2,\mathbb{R})$ to the Lie 
group related to the Lie subalgebra 
$\LH=\langle {\rm a}_1\rangle$. To obtain such a reduction, 
we have shown in Section \ref{reduction} that we 
have to solve an equation in $G/H$, namely
\begin{equation}
\frac{d\pi^L(\tilde g)}{dt}
=\sum_{\alpha=1}^3b_\alpha(t)X_\alpha^L(\pi^L(\tilde g))
\end{equation}
where $X^L_\alpha$ are the fundamental vector fields 
of the action $\lambda$ of $G$ on $G/H$. 
Now, we are going to describe this equation in a 
set of local coordinates. 
First, in an open neighbourhood $U$ of $e\in G$ we can 
write any element of this open in a unique way  as
\begin{equation}\label{des}
g=\exp(-c_3{\rm a}_3)\exp(-c_2{\rm a}_2)\exp(-c_1{\rm a}_1),
\end{equation}
where the matrices ${\rm a}_\alpha$, with $\alpha=1, 2, 3$, are given by
(\ref{thebasis}).

This decomposition allows us to establish a local 
diffeomorphism between an open neighbourhood $V\subset G/H$ 
and the set of matrices given by $\exp(-c_3{\rm a}_3)\exp(-c_2{\rm a}_2)$. 
Now, the decomposition (\ref{des}) reads in matrix terms as
\begin{equation*}
\begin{aligned}
\left(
\begin{array}{cc}
\alpha &\beta\\
\gamma &\delta
\end{array}
\right)&=
\left(
\begin{array}{cc}
1 &0\\
-c_3 &1
\end{array}
\right)
\left(
\begin{array}{cc}
e^{c_2/2} &0\\
0 &e^{-c_2/2}
\end{array}
\right)
\left(
\begin{array}{cc}
1 &c_1\\
0 &1
\end{array}
\right)\\&=
\left(
\begin{array}{cc}
e^{c_2/2} &0\\
-c_3 e^{c_2/2} &e^{-c_2/2}
\end{array}
\right)
\left(
\begin{array}{cc}
1 &c_1\\
0 &1
\end{array}
\right)\,.
\end{aligned}
\end{equation*}
If we express $c_1 ,c_2, c_3$ 
in terms of $\alpha, \beta, \gamma$ and $\delta$, we obtain that
$c_3=-\gamma/\alpha$, $c_2=\log \alpha^2$, and $c_1=\beta/\alpha$. Consequently,
we get 
\begin{equation*}
\left(
\begin{array}{cc}
\alpha &\beta\\
\gamma &\delta
\end{array}
\right)=
\left(
\begin{array}{cc}
1 &0\\
\gamma/\alpha &1
\end{array}
\right)
\left(
\begin{array}{cc}
\alpha &0\\
0 &\alpha^{-1}
\end{array}
\right)
\left(
\begin{array}{cc}
1 &\beta/\alpha\\
0 &1
\end{array}
\right)=
\left(
\begin{array}{cc}
\alpha &0\\
\gamma &\alpha^{-1}
\end{array}
\right)
\left(
\begin{array}{cc}
1 &\beta/\alpha\\
0 &1
\end{array}
\right)\,.
\end{equation*}
Thus, we can define the projection 
$\pi^L:U\subset G\rightarrow G/H$ given by
\begin{equation}\label{decomp}
\pi^L\left(
\begin{array}{cc}
\alpha &\beta\\
\gamma &\delta
\end{array}
\right)=
\left(
\begin{array}{cc}
\alpha &0\\
\gamma &\alpha^{-1}
\end{array}
\right)H,
%\equiv (\alpha,\gamma)\in \mathbb{R}^2
\end{equation}
which allows us to represent the elements of $G/H$, locally, 
as the $2\times 2$ lower triangular matrices with 
determinant one. Now, given $\lambda_g:g'H\in G/H\mapsto gg'H\in G/H$ 
as $\lambda_g\circ\pi^L=\pi^L\circ L_g$, the 
fundamental vector fields defined in $G/H$ 
by ${\rm a}_1$ and ${\rm a}_3$ through the action $\lambda:(g,g'H)\in G\times
G/H \mapsto \lambda_g(g'H)\in G/H$ are given by
\begin{equation*}
\begin{aligned}
X_1^L(\pi^L(g))&=\frac{d}{dt}\bigg|_{t=0} \pi^L\left(\exp(-t{\rm a}_1)\left(
\begin{array}{cc}
\alpha &\beta\\
\gamma &\delta
\end{array}
\right)\right)=\left(
\begin{array}{cc}
\gamma &0\\
0&-\gamma/\alpha^2
\end{array}
\right),\\
X_3^L(\pi^L(g))&=\frac{d}{dt}\bigg|_{t=0} \pi^L\left(\exp(-t{\rm a}_3)\left(
\begin{array}{cc}
\alpha &\beta\\
\gamma &\delta
\end{array}
\right)\right)=\left(
\begin{array}{cc}
0&0\\
-\alpha &0
\end{array}
\right),
\end{aligned}
\end{equation*}
and the equation on $V\subset G/H$ is described by
\begin{equation*}
\begin{aligned}
\left(
\begin{array}{cc}
\dot\alpha &0\\
\dot\gamma &-\dot\alpha\alpha^{-2}
\end{array}
\right)=
\left(
\begin{array}{cc}
\gamma &0\\
-\Omega^2(t)\alpha&-\gamma\alpha^{-2}
\end{array}
\right).
\end{aligned}
\end{equation*}
Therefore, we need to obtain a solution of the system
\begin{eqnarray}\label{Oscill}
\left\{\begin{aligned}
\ddot \alpha&=-\Omega^2(t)\alpha,\\
\gamma&=\dot\alpha\,.
\end{aligned}\right.
\end{eqnarray}
Then, taking into account (\ref{decomp}), if $\alpha_1$ is a solution of the
system (\ref{Oscill}),
the curve $\tilde g(t)$ that satisfies $g(t)=\tilde g(t)h(t)$, 
where $h(t)$ is a solution of an equation defined 
on the Lie group with Lie algebra $\LH=\langle {\rm a}_1\rangle$, reads
\begin{equation*}
\tilde g(t)=
\left(
\begin{array}{cc}
\alpha_1 &0\\
\dot\alpha_1&\alpha_1^{-1}
\end{array}
\right)=\left(
\begin{array}{cc}
e^{c_2/2}&0\\
-c_3e^{c_2/2}&e^{-c_2/2}
\end{array}
\right)=\exp\left(\frac{\dot\alpha_1}{\alpha_1}{\rm a}_3\right)\exp(-2\log
\alpha_1 {\rm a}_2),
\end{equation*}
and the curve which acts on the initial equation 
in $SL(2,\mathbb{R})$ to transform it
into one in the mentioned Lie subalgebra
is given by $\bar g(t)=\tilde g^{-1}(t)$,
\begin{equation*}
\bar g(t)=\exp(2\log \alpha_1\,{\rm a}_2)
\exp\left(-\frac{\dot \alpha_1}{\alpha_1}\,{\rm a}_3\right).
\end{equation*}
This curve transforms the initial equation in the group given by (\ref{inieqG})
into the new one given by (\ref{redumet}), i.e.
\begin{equation*}
{\rm a}'(t)=-\frac{{\rm a}_1}{\alpha_1^2(t)}\,,
\end{equation*}
which corresponds to the $t$-dependent Hamiltonian $H'(t)=P^2/(2\alpha_1^2(t))$.
The induced transformation in the Hilbert 
space $\mathcal{H}$ that transforms 
$H(t)$ into $H'(t)$ is
\begin{equation*}
\exp\left(i\frac{\log
\alpha_1}{2}(PX+XP)\right)\exp\left(-i\frac{\dot\alpha_1}{2\alpha_1}X^2\right).
\end{equation*}
Both results  can be found in \cite{FM03}. 

There are other possibilities of choosing different  
Lie subalgebras of $\LG$ in order to perform the reduction, 
however the results are always given in terms of a solution
of a differential equation.

\chapter{Integrability conditions for Lie systems}\label{IntRi}
The main aim of this chapter is concerned with the description of the main
aspects of the integrability theory for Lie systems detailed in \cite{CRL07d}
and based on the geometrical understanding of Riccati equations.

The Riccati equation can be considered as the simplest nonlinear differential
equation \cite{CarRamGra,CLR07b}. It is, basically, the only first-order
ordinary differential
equation  admitting a
nonlinear superposition rule \cite{LS,PW}. In spite of its apparent simplicity,
its general solution cannot be described by means
of quadratures with the exception of some very particular cases
\cite{CarRam,Kamke,Mu60,Ra61,Stre,Zh99}. 

The relevance of Riccati equation becomes evident when we take into account its
frequent appearance in many fields of Mathematics and Physics
\cite{CarMarNas,LK08,CJOV03,MR08,CHS07,Sc08,TM09,PW}. This also implies the
necessity of a theory of integrability providing all those integrable cases that
might lead to solvable physical models.

\section{Integrability of Riccati equations}\label{IntRicEqu}
In order to provide a first insight into the study of integrability conditions
for Riccati equations, we review here some very well-known results about this
topic. 

Recall that Riccati equations are first-order differential equations of the form
\begin{equation}\label{ricceq}
\frac{dx}{dt}=b_1(t)+b_2(t)x+b_3(t)x^2.
\end{equation}
A first particular example of Riccati equation integrable by quadratures is the
one with 
$b_3=0$. In fact, in such a case, Riccati equation reduces to an inhomogeneous
linear equation,
which can be explicitly integrated by means of two quadratures.

Additionally, the change of variable
$w=-1/x$  transforms the above Riccati  equation
into the new one
$$
\frac{dw}{dt}=b_1(t)\, w^2-b_2(t)\,w+b_3(t).
$$
Consequently, if we suppose that $b_1=0$ in equation (\ref{ricceq}), that is, if
we consider a Bernoulli equation,
 the mentioned change of variable leads to an integrable
linear equation.

Another known property on the integrability of Riccati equations establishes
that given a particular solution $x_1(t)$ of (\ref{ricceq}),
the change $x=x_1(t)+z$ permits us to transform a Riccati equation into a new
one for which the coefficient of the
term independent of $z$ is zero, i.e.
 \begin{equation*}
\frac {dz}{dt}=(2\, b_3(t)\, x_1(t)+ b_2(t)) z+ b_3(t)\,z^2, 
\end{equation*}
and, as we pointed out previously, this equation reduces to an inhomogeneous 
linear equation with the change of variables
$z=-1/u$. Consequently, the knowledge of a particular solution of a Riccati
equation allows us to find its general
solution by means of two quadratures. It is worth recalling that this property
can be more generally understood by means of the theory of Lie systems. Indeed,
this theory states that the knowledge of a particular solution of a Lie system
enables us to reduce the initial equation into a `simpler' one, see Section
\ref{ISOA} or \cite{CarRamGra}.

 If we know two particular solutions,
$x_1(t)$ and $x_2(t)$, of equation (\ref{ricceq}), its general solution can be
determined with one quadrature. Indeed, the change of variable
$z=(x-x_1(t))/(x-x_2(t))$ transforms the original equation into a homogeneous
linear differential equation
and, hence, the general solution can be immediately found.

 Finally, giving three particular solutions, $x_1(t),x_2(t),x_3(t)$, the general
solution can be
 written, without making use of any quadrature, in terms of the superposition
rule (\ref{SupRiccat}). 

The simplest case of Riccati equation, i.e. the one with $b_1$, $b_2$ and $b_3$
being constant, has been fully studied and it is integrable by
quadratures, see in example \cite{CarRamdos}. This can be viewed as the
consequence of the existence of a constant
(maybe complex) solution, permitting us to reduce the equation into an
inhomogeneous linear one. Note also that, in a similar way, separable Riccati
equations of the form
\begin{equation*}
\frac{dx}{dt}=\varphi(t)(c_1+c_2\,x+c_3\,x^2)\,,
\end{equation*}
with $\varphi(t)$ being a non-vanishing function, are integrable, because they
admit a constant solution again, which enables us to transform the equation into
a linear inhomogeneous one again. On the other hand, the integrability of the
above equation can also be related to the existence of a $t$-reparametrisation,
reducing the problem to an autonomous one.

\section{Transformation laws of Riccati equations}\label{TL}

We here describe an important property of Lie systems, in the
particular case of Riccati equations, playing a relevant r\^ole for
establishing integrability
criteria: {\it The group $\mathcal{G}$ of curves in a Lie group $G$ associated
with a
Lie system acts on the set of the related Lie systems}.

More explicitly, consider a family $X_1, X_2, X_3,$ of vector fields on
$\overline{\R}$, e.g. the set given in (\ref{ric}), spanning the
Vessiot-Guldberg Lie algebra of vector fields associated with Riccati equations
and isomorphic to $\mathfrak{sl}(2,\mathbb{R})$. In terms of this family, each
Riccati equation (\ref{ricceq}) is related to a $t$-dependent vector field
$X_t=b_1(t)X_1+b_2(t)X_2+b_3(t)X_3$, which can be considered as a
curve $(b_1(t),b_2(t),b_3(t))$ in $\mathbb{R}^3$. Each element $\bar A$ of the
group of smooth curves in $SL(2, {\R})$, i.e. $\bar A\in
{\GR}\equiv\Map({\R},\,SL(2,{\R}))$,
transforms every curve $x(t)$ in $\overline \R$
 into a new one $x'(t)=\Phi(\bar A(t),x(t))$ by means of the action
$\Phi:(A,x)\in SL(2,{\R})\times \overline{\mathbb{R}}\mapsto
\Phi(A,x)\in\bar{\mathbb{R}}$ of the form:
\begin{equation}\label{action2}
\Phi(A,x)=\left\{
\begin{aligned}
&{\frac{\alpha\, x+\beta}{\gamma\, x+\delta}}\quad
&x&\neq-{\frac{\delta}{\gamma}},\,\,\,x\neq\infty,\\
&\alpha/\gamma\ \quad  &x&=\infty,\\
&\infty \quad&x&=-\frac{\delta}{\gamma}, \\
\end{aligned}
\right.\qquad {\rm where} \,\, A=\left(\begin{array}{cc}\alpha & \beta\\
\gamma&\delta \end{array}\right).
\end{equation}
Moreover, the above $t$-dependent change of variables transforms the Riccati
equation (\ref{ricceq}) into a new one with 
$t$-dependent coefficients $b'_1,b'_2, b'_3$ given by
\begin{equation}\label{newcoeff}
\left\{\begin{aligned}
b'_3&={\delta}^2\,b_3-\delta\gamma\,b_2+{\gamma}^2\,b_1+\gamma
{\dot{\delta}}-\delta \dot{\gamma}\ ,\\
b'_2&=-2\,\beta\delta\,b_3+(\alpha\delta+\beta\gamma)\,b_2-2\,\alpha\gamma\,b_1
       +\delta \dot{\alpha}-\alpha \dot{\delta}+\beta \dot{\gamma}-\gamma
\dot{\beta}\ ,   \\
b'_1&={\beta}^2\,b_3-\alpha\beta\,b_2+{\alpha}^2\,b_1+\alpha\dot{\beta}
-\beta\dot{\alpha} \ .
\end{aligned}\right.
\end{equation}
Indeed, the above expressions define an affine action of the group
 ${\GR}$ on the set of
Riccati equations. In other words, given the elements $A_1,A_2\in\mathcal{G}$,
transforming the coefficients of a general
Riccati equation by means of two successive transformations of the above type,
e.g. first by $A_1$ and then by $A_2$, gives exactly
the same result as doing only one transformation with the element $A_2\cdot A_1$
of $\mathcal{G}$, see \cite{CarRam,LM87}.

The group $\mathcal{G}$ also acts on the set of equations of the form
(\ref{eqLG2}) on $SL(2,\mathbb{R})$. In order to show this, note first that
$\mathcal{G}$ acts on the left on the
set of curves in $SL(2, \R)$ by left translations, i.e. given two curves $\bar
A(t), A(t)\subset SL(2,\mathbb{R})$, the curve $\bar A(t)$ transforms the curve
$A(t)$  into a new one  $A'(t)=\bar A(t) A(t)$. Moreover, if
 $A(t)$ is a solution of equation (\ref{eqLG2}), then the curve $A'(t)$
satisfies a new equation like (\ref{eqLG2}) but  with a different right
hand side ${\rm a}'(t)$. Differentiating the relation $A'(t)=\bar A(t) A(t)$ and
taking into account the form of (\ref{eqLG2}), we get that, in view of the basis
(\ref{thebasis}), the relation between the curves ${\rm a}(t)$ and ${\rm a}'(t)$
in $\mathfrak{sl}(2,\mathbb{R})$ is
\begin{equation}
{\rm a}'(t)=\bar A(t){\rm a}(t)\bar A^{-1}(t)+\dot{\bar{A}}(t)\bar A^{-1}(t)
=-\sum_{\alpha=1}^3b'_\alpha(t){\rm a}_\alpha\,, \label{newricc}
\end{equation}
which yields the expressions (\ref{newcoeff}). Conversely, if $A'(t)=\bar A(t)
A(t)$ is  the solution for the  equation corresponding to the
curve ${\rm a}'(t)$ given by the transformation rule (\ref{newricc}), then
$A(t)$ is the solution of the equation (\ref{eqLG2}) determined by the curve
${\rm a}(t)$.

Summarising, we have shown that it is possible to associate each Riccati
equation with an equation on
the Lie group $SL(2,\mathbb{R})$ and to define an infinite-dimensional group of 
transformations acting on the set of Riccati equations. Additionally, this
process 
can be easily derived in a similar way for any Lie system, see \cite{CRL07d}.

\section{Lie structure of an equation of transformation of Lie systems}
Let us construct a Lie system describing the curves in $SL(2,\mathbb{R})$
which transform the Riccati equation associated with an equation on $SL(2,\R)$
characterised by the curve ${\rm a}(t)\subset\mathfrak{sl}(2,\mathbb{R})$ into
the Riccati equation associated with the curve ${\rm a}'(t)\subset
\mathfrak{sl}(2,{\mathbb{R}})$. By means of this Lie system, we later explain
the results derived in \cite{CRL07d} in order to describe, from a unified point
of view, the developments of the works \cite{CarRamGra,CLR07b}.

Multiply equation (\ref{newricc}) on the right by $\bar A(t)$ to get
\begin{equation}\label{MatrixRicc}
\dot{\bar{A}}(t)={\rm a}'(t)\bar A(t)-\bar A(t){\rm a}(t)\,.
\end{equation}
If we consider the above equation as a system of first-order differential
equations in the
coefficients of the
 curve $\bar A(t)$ in $SL(2,\R)$, with
$$
\bar A(t)=
\matriz{cc}{
\alpha(t) &\beta(t)\cr
\gamma(t)& \delta(t)}\,,\quad \alpha(t)\delta(t)-\beta(t)\gamma(t)=1,
$$
then system (\ref{MatrixRicc}) reads
\begin{equation}\label{FS}
\matriz{c}{
\dot\alpha\\
\dot\beta\\
\dot\gamma\\
\dot\delta}
=
\matriz{cccc}{
\frac{b'_2-b_2}{2}&b_3 &b'_1&0\\
-b_1& \frac{b'_2+b_2}{2}&0 &b'_1\\
-b'_3&0 &-\frac{b'_2+b_2}{2}& b_3\\
0&-b_3' &-b_1& -\frac{b'_2-b_2}{2}}
\matriz{c}{
\alpha\\
\beta\\
\gamma\\
\delta}.
\end{equation}
The solutions $y(t)=(\alpha(t),\beta(t),
\gamma(t),\delta(t))$ of the above system relating two
given Riccati equations are associated with curves in $SL(2,\mathbb{R})$, i.e.
they are such that, at any time,
$\alpha\delta-\beta\gamma=1$. Nevertheless, we can drop such a restriction for
the time being as it can be implemented by a restraint on the initial conditions
for
 the solutions and, hence, we can treat the variables, $\alpha,\beta,\gamma,
\delta,$ in the
 system  (\ref{FS}) as being independent. In this case, this linear
 system can be regarded as a Lie system linked to a Lie algebra of vector fields
isomorphic to
 $\mathfrak{gl}(4,\mathbb{R})$. Nevertheless, it may also be understood as a Lie
system related to a Lie algebra of vector fields isomorphic to a Lie subalgebra
of $\mathfrak{gl}(4,\mathbb{R})$. Indeed, consider the vector fields
{\small
\begin{equation*}
\begin{array}{ll}
N_1=-\alpha\dfrac{\partial}{\partial\beta}-\gamma\dfrac{\partial}{\partial\delta
},
&N'_1=\gamma\dfrac{\partial}{\partial\alpha}+\delta\dfrac{\partial}{
\partial\beta},\cr
N_2=\frac
12\left(\beta\dfrac{\partial}{\partial\beta}+\delta\dfrac{\partial}{
\partial\delta}-\alpha\dfrac{\partial}{\partial\alpha}-\gamma\dfrac{\partial}{
\partial\gamma}\right),
&N'_2=\frac
12\left(\alpha\dfrac{\partial}{\partial\alpha}+\beta\dfrac{\partial}{
\partial\beta}-\gamma\dfrac{\partial}{\partial\gamma}-\delta\dfrac{\partial}{
\partial\delta}\right),\cr
N_3=\beta\dfrac{\partial}{\partial\alpha}+\delta\dfrac{\partial}{\partial\gamma}
,&
N'_3=-\alpha\dfrac{\partial}{\partial\gamma}-\beta\dfrac{\partial}{
\partial\delta},\nonumber
\end{array}
\end{equation*}
}spanning a Vessiot-Guldberg Lie algebra of vector fields isomorphic to
$\LG\equiv\mathfrak{sl}(2,\mathbb{R})\oplus\mathfrak{sl}(2,\mathbb{R}
)\subset\mathfrak{gl}(4,\mathbb{R})$. Consequently, the linear system of
differential equation (\ref{FS}) is a Lie system on $\R^4$  associated with a
Lie algebra of vector fields isomorphic to $\mathfrak{g}$, see \cite{CRL07d}. 

If we  denote $y\equiv\left(\alpha,\beta,\gamma,\delta\right)\in \R^4$,  system
(\ref{FS}) is
 a differential equation on $\mathbb{R}^4$ of the form
\begin{equation}\label{IRTRR4}
\frac{dy}{dt}=N(t,y),
\end{equation}
with $N$ being the $t$-dependent  vector field
\begin{equation*}
N_t=\sum_{\alpha=1}^3\left(b_\alpha(t)N_\alpha+b'_\alpha(t)N'_\alpha\right).
\end{equation*}

The vector fields $\{N_1,N_2,N_3,N'_1,N'_2,N'_3\}$ span a regular distribution
$\mathcal{D}$ with rank
three in almost any point of $\mathbb{R}^4$ and thus there exists, at least
locally, a
first-integral for all the vector fields in the distribution $\mathcal{D}$. The
method of characteristics allows us to determine that this first-integral can
be 
$$
I:y=(\alpha,\beta,\gamma,\delta)\in\mathbb{R}^4\mapsto \det\,y\equiv
\alpha\delta-\beta\gamma\in\mathbb{R}.
$$
Moreover, this first-integral is related to the
 determinant of a matrix $\bar A\in SL(2,\mathbb{R})$ with coefficients given by
the components
 of $y=(\alpha,\beta,\gamma,\delta)$. Therefore, if we have a solution of the
system (\ref{FS})
with initial condition such that $\det
y(0)=\alpha(0)\delta(0)-\beta(0)\gamma(0)=1$,
then $ \det y(t)=1$  at any time $t$ and the solution can be understood
as a curve in $SL(2,\mathbb{R})$. Summarising, we have proved the following
theorem.

\begin{theorem}\label{THLS} The curves in $SL(2,\mathbb{R})$  transforming
equation (\ref{eqLG2}) into a new equation of the same form characterised by a
curve ${\rm a}'(t)=-\sum_{\alpha=1}^3b'_\alpha(t){\rm a_\alpha}$ are described
through the solutions of the Lie systems
\begin{equation}\label{Sys}
\frac{dy}{dt}=N(t,y)=\sum_{\alpha=1}^3b_\alpha(t)N_\alpha(y)+\sum_{\alpha=1}
^3b'_\alpha(t)N'_\alpha(y).
\end{equation}
such that $\det y(0)=1$. Furthermore, the above Lie system is related to a
non-solvable Vessiot--Guldberg Lie algebra  isomorphic to
$\mathfrak{sl}(2,\mathbb{R})\oplus\mathfrak{sl}(2,\mathbb{R})$.
\end{theorem}

A consequence of the above Theorem is the following corollary, whose proof is
omitted and left to the reader.

\begin{corollary} \label{CorCur} Given two Riccati equations associated with
  curves ${\rm a}'(t)$ and ${\rm a}(t)$ in
$\mathfrak{sl}(2,\R)$, there always exists a curve $\bar A(t)$ in $SL(2,\R)$
transforming the Riccati equation related to ${\rm a}(t)$ into the new one
associated with ${\rm a}'(t)$. Furthermore, if $\bar A(0)=I$, this curve is
uniquely defined.
\end{corollary}

Even if we know that given two equations on the Lie group $SL(2,\mathbb{R})$
there
always exists a transformation relating both,  in order to obtain such a curve
 we need to solve  the differential
equation (\ref{IRTRR4}) which, unfortunately, is Lie system related to a
non-solvable Vessiot--Guldberg. Consequently, it is
not easy to find its solutions in general as, for instance, it is not integrable
by
quadratures.

Nevertheless, many known and new properties on integrability conditions for
Riccati equations can be determined by means of Theorem  \ref{THLS}.
Furthermore, the procedure to obtain the Lie system (\ref{IRTRR4}) can be
generalised to deal with any Lie system related to a Lie group $G$ with Lie
algebra $\mathfrak{g}$ (cf. \cite{CRL07d}). 

\section{Description of some known integrability conditions}\label{DIC}

Note that Lie systems on $G$ of the form (\ref{eqLG2}) determined by a constant
curve, ${\rm a}=-{\sum_{\alpha=1}^3} c_\alpha {\rm a}_\alpha$,
are integrable and, therefore, the same happens for
curves of the form ${\rm a}(t)=-D\left({\displaystyle \sum_{\alpha=1}^3}
c_\alpha {\rm a}_\alpha\right)$, where $D=D(t)$ is a non-vanishing function, as
a $t$-reparametrisation reduces the problem to the previous one.

Our aim now is to determine the curves $\bar A(t)$ in $SL(2,\R)$ transforming
the equation on $SL(2,\mathbb{R})$ characterised by a curve ${\rm a}(t)$ into
the
equation on $SL(2,\mathbb{R})$ characterised by ${\rm a}'(t)=-D(c_1{\rm
a}_1+c_2{\rm a}_2+c_3{\rm a}_3)$, with $D=D(t)$ a non-vanishing function and
$c_1c_3\neq 0$. As the final
equation is associated with a solvable one-dimensional Vessiot-Guldberg Lie
algebra, the transformation
establishing the relation to such a final integrable equation allows us to find
by quadratures the solution of the initial equation and, therefore, the solution
for its associated Riccati equation. In order to get the
transformation between the Riccati equations linked to the before equations on
$SL(2,\mathbb{R})$, we look for  particular curves $\bar A(t)$ in 
$SL(2,\mathbb{R})$ satisfying certain conditions in order to get an integrable
equation (\ref{FS}).
Nevertheless, under the assumed restrictions, we may obtain a
 system of differential equations which does not admit any solution. In such a
case, the conditions ensuring the existence of solutions will describe
integrability conditions. As an
 application
we  show that many known achievements about the integrability of Riccati
equations can be
recovered  and explained in this way.

We have already showed that Riccati equations (\ref{ricceq}), with $b_1b_3\equiv
0$, are reducible to linear
differential equations and therefore they are always
integrable \cite{CarMarNas}. Hence, they are not interesting in the study of
integrability conditions and we can focus our attention on reducing Riccati
equations with $b_1b_3\neq 0$ into  integrable ones by means of the action of a
curve in $SL(2,\mathbb{R})$. With this aim, consider the family of curves with 
$\beta=0$ and $\gamma=0$, i.e. take curves in $SL(2,\mathbb{R})$ of the form
$$A(t)=\matriz{cc}{\alpha(t)&0\\0&\delta(t)}\subset SL(2,\mathbb{R})\,,\qquad
\alpha(t)\delta(t)=1.$$

The curve $\bar A(t)$ in $SL(2,\mathbb{R})$ determines a $t$-dependent change of
variables in $\overline{\mathbb{R}}$ given by $x'(t)=\Phi(\bar A(t),x)$. In view
of the action (\ref{action2}), and as $\alpha\delta=1$, we get that the previous
change of variables reads
\begin{equation}\label{yprime}
x'=\alpha^2(t)x=G(t)x\,,\quad G(t)\equiv\frac{\alpha(t)}{\delta(t)}>0.
\end{equation}
In view of the relations (\ref{newcoeff}), the initial Riccati equations is
transformed, by means of the curve $\bar A(t)$, into the new Riccati equation
with $t$-dependent coefficients
$$b'_1=\alpha^2\, b_1\,,\qquad
b'_2=\alpha\,\delta\,b_2+\dot\alpha\,\delta-\alpha\,\dot\delta\,,\qquad
b'_3=\delta^2\,b_3\,.$$
Moreover, the functions $\alpha(t)$ and $\delta(t)$ are solutions of the system
(\ref{IRTRR4}), which in this case reduces to
\begin{equation}\label{RLFS}
\matriz{c}{
\dot\alpha\\
0\\
0\\
\dot\delta}
=\matriz{cccc}{
\frac{b'_2-b_2}{2}&b_3 &b'_1&0\\
-b_1& \frac{b'_2+b_2}{2}&0 &b'_1\\
-b'_3&0 &-\frac{b'_2+b_2}{2}& b_3\\
0&-b_3' &-b_1& -\frac{b'_2-b_2}{2}}
\matriz{c}{
\alpha\\
0\\
0\\
\delta}.
\end{equation}
The existence of solutions for the above system related to elements of
$SL(2,\mathbb{R})$ that satisfy the required conditions determines the
integrability of a Riccati equation by the method described. Thus, let us
analyse the existence of such solutions to get these integrability conditions.

From some of the relations of the above system, we get that
\begin{equation*}
-b_1\alpha+b'_1\delta=0,\qquad -b'_3\alpha+b_3\delta=0.
\end{equation*}
As $\alpha(t)=1$, these relations imply that $b'_1\,b'_3=b_1\,b_3$ and 
\begin{equation*}
\alpha^2=\frac{b_1'}{b_1}=\frac{b_3}{b_3'}\equiv G>0\,.
\end{equation*}
Hence, the transformation formulas (\ref{newcoeff}) reduce to
\begin{equation}\label{relat}
b'_3=\alpha^{-2}\,b_3\,,\qquad b'_2=b_2+2\frac{\dot \alpha}\alpha
\,,\qquad b'_1=\alpha^2 b_1\,.
\end{equation}
Then, in order to exist a $t$-dependent function $D$ and two real constants
$c_1$ and $c_3$, with $c_1c_3\neq 0$, such that $b'_3=Dc_3$ and $b'_1=Dc_1$, the
function $D$ must be given by 
\begin{equation*}
D^2c_1c_3=b_1b_3\Longrightarrow D=\pm\sqrt{\frac{b_1b_3}{c_1c_3}}\,,
\end{equation*}
where we have used that $b'_1b'_3=b_1b_3$. On the other hand, as
$b'_1/b_1=\alpha^2>0$, we have to fix the sign $\kappa$ of the function $D$ in
order to satisfy this relation, i.e. ${\rm sg}(c_1D)={\rm sg}(b_1)$. Therefore,
$$
\kappa={\rm sg}(D)={\rm sg}(b_1/c_1).
$$
Also, as $b_1b_3=b'_1b'_3$, we get that ${\rm sg}(b_1b_3)={\rm sg}(c_1c_3)$.
Furthermore, in view of relations (\ref{relat}), $\alpha$ is determined, up to a
sign, by 
\begin{equation}
\alpha=\sqrt{\frac{Dc_1}{b_1}}=\left(\frac{c_1}{c_3}\,\frac{b_3}{b_1}
\right)^{1/4}.\label{otroalfa}
\end{equation}
and therefore the change of variables (\ref{yprime}) reads
\begin{equation}\label{Change}
x'=\frac{D(t)c_1}{b_1(t)}x\,.
\end{equation}
Finally, as a consequence of (\ref{relat}), in order for $b'_2$ to be the
product $b'_2=c_2\, D$, we see that
\begin{equation}\label{equat}
b_2+2\, \frac{\dot \alpha}{\alpha}=\kappa c_2 \sqrt{\frac{b_1b_3}{c_1c_3}}\,. 
\end{equation}
Using (\ref{otroalfa}) and the above equality, we see that the integrability
condition is
\begin{equation*}
\sqrt{\frac{c_1c_3}{b_1b_3}}\left[b_2+\frac{1}{2}\left(\frac{\dot
b_3}{b_3}-\frac{\dot b_1}{b_1}\right)\right]=\kappa c_2\,.
\end{equation*}

Conversely, if the above integrability condition is valid and
$D^2c_1c_3=b_1b_3$, the change of variables (\ref{Change}) transforms the
Riccati equation (\ref{ricceq}) into $dx'/dt=D(t)(c_1+c_2y'+c_3y'^2)$, with
$c_1c_3\neq 0$. To sum up, we have proved the following theorem.

\begin{theorem}\label{TU} The necessary and sufficient conditions
for the existence of a  transformation
\begin{equation*}
x'=G(t)x,\quad G(t)>0,
\end{equation*}
 relating the Riccati equation
\begin{equation*}
\frac{dx}{dt}=b_1(t)+b_2(t)x+b_3(t)x^2\,,  \qquad b_1b_3\ne 0,
\end{equation*}
to an integrable one given by
\begin{equation}
\frac{dx'}{dt}=D(t)(c_1+c_2x'+c_3x'^2)\,,\qquad c_1c_3\neq 0,\qquad D(t)\neq 0,
\label{eqDcs}
\end{equation}
where $c_1,c_2,c_3$ are real numbers and $D(t)$ is non-vanishing functions,
are
\begin{equation}
D^2c_1c_3=b_1b_3,\qquad \left(b_2+\frac{1}{2}\left(\frac{\dot
b_3}{b_3}-\frac{\dot b_1}{b_1}\right)\right)\sqrt{\frac{c_1c_3}{b_1b_3}}=\kappa
c_2.\label{DinTh2}
\end{equation}
where $\kappa={\rm sg}(D)={\rm sg}(b_1/c_1)$. The  transformation is then
uniquely defined by
\begin{equation*}
x'=\sqrt{\frac{b_3(t)c_1}{b_1(t)c_3}}\,x\,.
\end{equation*}
\end{theorem}

From previous results, the following corollary follows.
\begin{corollary}\label{CTU}
A Riccati equation (\ref{eqDcs}) with  $b_1b_3\ne 0$ can be transformed into a
Riccati equation of the form (\ref{eqDcs}) by a $t$-dependent  change of
variables $y'=G(t)y$, with $g(t)>0$, if and only if
\begin{equation}
\frac{1}{\sqrt{|b_1b_3|}}\left(b_2+\frac{1}{2}\left(\frac{\dot
b_3}{b_3}-\frac{\dot
      b_1}{b_1}\right)\right)=K,
\label{resCor2}
\end{equation}
for a certain real constant $K$. In such a case, the Riccati equation
(\ref{ricceq})  is integrable by quadratures.
\end{corollary}

In view of Theorem \ref{TU}, if we start with the integrable Riccati
equation (\ref{eqDcs}), we can
obtain the set of
all Riccati equations that can be reached from it by means of a transformation
of the form (\ref{yprime}).
\begin{corollary}\label{C2TU} Given an integrable Riccati equation
\begin{equation*}
\frac{dx}{dt}=D(t)(c_1+c_2x+c_3x^2),\qquad c_1c_3\neq 0,\quad D(t)\neq 0,
\end{equation*}
with $D(t)$ a non-vanishing function, the set of Riccati equations which can be 
obtained with a
transformation $x'=G(t)x$, with $G(t)>0$, are
those of the form
\begin{equation*}
\frac{dx'}{dt}=b_1(t)+\left( \frac{\dot b_1(t)}{b_1(t)}-\frac{\dot
D(t)}{D(t)}+c_2D(t)\right) x'+\frac{D^2(t)c_1c_3}{b_1(t)}x'^2\,,
\end{equation*}
and the function $G$ is then given by
$$G=\frac{Dc_1}{\sqrt{b_1}}\,.$$
\end{corollary}

Therefore, starting with an integrable equation,  we can generate a family of
solvable Riccati equations whose coefficients are parametrised by a
non-vanishing function $b_1$.
Moreover, the integrability condition to check if a Riccati equation belongs to
this family can be easily verified.

The previous results can now be used for a better comprehension of some
integrability conditions found in the literature. Let us illustrate this claim
by reviewing some well-known integrability conditions through our methods.

\medskip

 $\bullet$  {\it The case of Allen and Stein}
\medskip

 The main achievements of the article \cite{AS64} can be recovered through our
more general approach. In
 that work, a Riccati equation (\ref{ricceq}), with $b_1b_3>0$ and $b_0$, $b_2$
being differentiable functions satisfying the condition
\begin{equation}\label{IntAS}
\frac{b_2+\frac{1}{2}\left(\frac{\dot b_3}{b_3}-\frac{\dot
b_1}{b_1}\right)}{\sqrt{b_1b_3}}=C,
\end{equation}
where $C$ is a real constant, was transformed into the integrable one
\begin{equation}\label{FREAS64}
\frac{dx'}{dt}=\sqrt{b_1(t)b_3(t)}\left(1+Cx'+x'^2\right),
\end{equation}
through a $t$-dependent linear transformation of the form
\begin{equation*}
x'=\sqrt{\frac{b_3(t)}{b_1(t)}}x\,.
\end{equation*}

If a Riccati equation obeys the integrability condition (\ref{IntAS}), such an
equation also satisfies the assumptions of Corollary \ref{CTU} and, therefore,
the integrability condition given in Theorem \ref{TU} with 
\begin{equation*}
c_1=c_3=1,\quad c_2=C,\quad D=\sqrt{b_1b_3}.
\end{equation*}
Consequently, the corresponding $t$-dependent change of variables described by
Theorem \ref{TU} reads
\begin{equation*}
x'=\sqrt{\frac{b_3(t)}{b_1(t)}}x\,,
\end{equation*}
showing that the transformation in \cite{AS64} is a particular case of our
results. This is not surprising, as Theorem \ref{TU} shows that if such a
$t$-dependent change of variables is used to transform a Riccati equation
(\ref{ricceq}) into one of the form (\ref{eqDcs}), this change of variables must
be one of the form (\ref{Change}) and the initial Riccati equation must satisfy
integrability conditions  (\ref{DinTh2}).

\medskip

$\bullet$ {\it The case of Rao and Ukidave}:

\medskip

Rao and Ukidave stated in their work \cite{RU68} that a Riccati equation
(\ref{ricceq}), with $b_1b_3>0$, can be transformed into an integrable one 
\begin{equation*}
\frac{dx'}{dt}=\sqrt{cb_1b3}\left(1+kx'+\frac{1}{c}{x'}^2\right),
\end{equation*}
through a $t$-dependent linear transformation
\begin{equation*}
x'=\frac{1}{v(t)}x,
\end{equation*}
if there exist two real constants $c$ and $k$ such that the following
integrability condition is satisfied
\begin{equation}\label{CondRU1}
 b_3=\frac{b_1}{cv^2},
\end{equation}
with $v(t)$ being a solution of the differential equation
\begin{equation}\label{CondRU2}
\frac{dv}{dt}=kb_1(t)+b_2(t)v\,.
\end{equation}

Note that, in view of (\ref{CondRU1}), necessarily $c>0$ and if the
integrability conditions (\ref{CondRU1}) and (\ref{CondRU2}) hold with constants
$c$ and $k$ and a negative solution $v(t)$, the same conditions are valid for
the constants $c$, $-k$ and a positive solution $-v(t)$. Consequently, we can
restrict ourselves to studying the integrability conditions (\ref{CondRU1}) and
(\ref{CondRU2}) for positive solutions $v(t)>0$. In such a case, the above
method uses a $t$-dependent  linear change of coordinates of the form
(\ref{yprime}) and the final Riccati equation are of the type described in our
work (\ref{eqDcs}). Therefore, the integrability conditions derived by Rao and
Ukidave have to be a particular instance of the integrable cases described by
Theorem \ref{TU}. 

Using the value of $v(t)$ in terms of the constant $c$ and the functions $b_1$
and $b_3$ obtained with the aid of the formula (\ref{CondRU1}) and equation
(\ref{CondRU2}), we get that
\begin{equation*}
\frac{1}{\sqrt{|b_1b_3|}}\left(b_2+\frac 12\left(\frac{\dot b_3}{b_3}-\frac{\dot
b_1}{b_1}\right)\right)=-k{\rm sg}(b_0)\sqrt{c}.
\end{equation*}
Hence, the Riccati equations holding conditions (\ref{CondRU1}) and
(\ref{CondRU2}) satisfy the integrability conditions of Corollary \ref{C2TU}.
Moreover, if we choose 
$$
D^2=cb_1b_3,\qquad c_1=1,\qquad c_2=-k,\qquad c_3=c^{-1},
$$
then $D=\sqrt{cb_1b_3}$ and the only possible transformation (\ref{yprime})
given by Theorem \ref{TU} reads
\begin{equation*}
x'=\alpha^2(t)x=\sqrt{\frac{cb_3(t)}{b_1(t)}}x,
\end{equation*}
and hence,
\begin{equation*}
\frac{1}{v}=\sqrt{\frac{cb_3}{b_1}}.
\end{equation*}
In this way, we recover one of the results derived by Rao and Ukidave in
\cite{RU68}.

In short, many integrability conditions found in the literature can be described
by our more general methods.

\section{Integrability and reduction}
Now we develop a similar procedure to the one derived above, but now we assume
the solutions of system (\ref{FS}) to be included within a two-parameter subset
of $SL(2,\R)$. As a result, we recover some known integrability conditions and
review, from a more general point of view, the integrability method described in
\cite{CarRamGra}.

As we did previously, let us try to relate the Riccati equation (\ref{ricceq})
to an integrable one associated, as a Lie system, with a curve
 ${\rm a}'(t)=-D(t)(c_1{\rm a}_1+c_2{\rm a}_2+c_3{\rm a}_3)$, with $c_3\neq 0$
and a non-vanishing function $D=D(t)$. Nevertheless, we consider solutions of
system (\ref{IRTRR4}) with $\gamma=0$, $\alpha>0$, and related to a curve in
$SL(2,\mathbb{R})$, i.e. we analyse transformations
$$x'=\frac{\alpha(t)}{\delta(t)}x+\frac{\beta(t)}{\delta(t)}=\alpha^2(t)\,x+
\frac{\beta(t)}{\delta(t)}\,.$$
In this case, using the expression in coordinates (\ref{FS}) of system
(\ref{Sys}), we get that
\begin{equation}\label{PC}
\matriz{c}{
\dot\alpha\\
\dot\beta\\
0\\
\dot\delta}=
\matriz{cccc}{
\frac{b'_2-b_2}{2}&b_3 &b'_1&0\\
-b_1& \frac{b'_2+b_2}{2}&0 &b'_1\\
-b'_3&0 &-\frac{b'_2+b_2}{2}& b_3\\
0&-b_3' &-b_1& -\frac{b'_2-b_2}{2}}
\matriz{c}{
\alpha\\
\beta\\
0\\
\delta}\,,
\end{equation}
where $b'_j=D\,c_j$ and $c_j\in\mathbb{R}$ for $j=1,2,3$. As we suppose
$b_3'\neq 0$, the third equation of the above system yields
\begin{equation*}
\frac{\alpha}{\delta}=\frac{b_3}{b'_3}=\frac{b_3}{D c_3}.
\end{equation*}
Since $\alpha\delta=1$ so that the solution of (\ref{Sys}) is related to an
element of $SL(2,\R)$, and $b_3'=Dc_3$, the above expression implies
\begin{equation}\label{Drelation}
\alpha^2=\frac{b_3}{Dc_3}.
\end{equation}
Therefore, $\alpha$ is determined by the values of $b_3(t)$, $D$ and $c_3$.
Additionally, the first differential equation of system (\ref{PC}) determines
$\beta$ in terms of $\alpha$ and the coefficients of the initial and final
Riccati equations, i.e. 
$$
\beta=\frac{1}{b_3}\left(\dot \alpha-\frac{b'_2-b_2}{2}\alpha\right).
$$
Taking into account the relation (\ref{Drelation}) and as $\alpha\delta=1$, we
can define $M=\beta/\alpha$ and rewrite the above expression as follows 
\begin{equation*}
\frac{dD}{dt}=\left(b_2(t)+\frac{\dot b_3(t)}{b_3(t)}\right)D-c_2D^2-2b_3(t)MD.
\end{equation*}
Considering the differential equation in $\dot \beta$ in terms of $M$, we get
the equation
\begin{equation*}
\frac{dM}{dt}=-b_1(t)+\frac{c_1c_3}{b_3(t)}D^2+b_2(t) M-b_3(t) M^2\,.
\end{equation*}
Finally, as $\delta\alpha=1$ is a first-integral of system (\ref{Sys}), if the
system for the variables $M$ and $D$ and all the abovementioned conditions are
satisfied, the value $\delta=\alpha^{-1}$ obeys its corresponding differential
equations of the system (\ref{PC}). Summarising, we have stated the following
theorem.
\begin{theorem}\label{FT2} Given a  Riccati equation (\ref{ricceq})
there exists a transformation
\begin{equation*}
x'=G(t)x+H(t)\,,\qquad G(t)>0\,,
\end{equation*}
relating it to the integrable equation
\begin{equation}\label{equation83}
\frac{dx'}{dt}=D(t)(c_1+c_2x'+c_3x'^2),
\end{equation}
with $c_3\neq 0$, and $D$ a non-vanishing function, if and only if there exist
functions $D$ and $M$ satisfying the system
\begin{eqnarray*}
\left\{\begin{aligned}
\frac{dD}{dt}=\left(b_2(t)+\frac{\dot
b_3(t)}{b_3(t)}\right)D-c_2D^2-2b_3(t)MD,\\
\frac{dM}{dt}=-b_1(t)+\frac{c_1c_3}{b_3(t)}D^2(t)+b_2(t) M-b_3(t) M^2.
\end{aligned}\right.
\end{eqnarray*}
The transformation is then given by
\begin{equation}\label{ChangeT3}
x'=\frac{b_3(t)}{D(t)c_3}(x+M(t))\,.
\end{equation}
\end{theorem}

If we consider $c_1=0$ in equation (\ref{equation83}), the system determining
the curve in $SL(2,\mathbb{R})$ which performs the transformation of Theorem
\ref{FT2} reads 
\begin{equation}
\left\{
\begin{array}{rcl}
\dfrac{dD}{dt}&=&\left(b_2(t)+\dfrac{\dot
b_3(t)}{b_3(t)}\right)D-c_2D^2(t)-2b_3(t)MD,\\
\dfrac{dM}{dt}&=&-b_1(t)+b_2(t) M-b_3(t) M^2.
\end{array}\label{RedSep}\right.
\end{equation}
Note that this system does not involve any integrability condition, since there
always exists a solution for every initial condition. Nevertheless, finding such
solutions can be as difficult as solving the initial
Riccati equation. Therefore, we need to assume some simplification in order to
find a particular solution. Let us put, for instance, $M=b_1/b_2$. In this case,
the first differential equation of the above system does not depend on $M$ and
reduces to
$$
\frac{dD}{dt}=\left(-b_2(t)+\frac{\dot b_3(t)}{b_3(t)}\right)D-c_2D^2
$$
whose solutions read
\begin{equation*}
D(t)=\frac{\exp\left(\int_0^t
A(t')dt'\right)}{C+c_2\int^t_0\exp\left(\int_0^{t''}
A(t')dt'\right)dt''}\,,\qquad A(t)=\left(-b_2(t)+\frac{\dot
b_3(t)}{b_3(t)}\right).
\end{equation*}

Meanwhile, as $M=b_2/b_3$ must satisfy the second equation in
(\ref{RedSep}), we obtain that
\begin{equation*}
\frac{d}{dt}\left(\frac{b_2}{b_3}\right)=-b_1\,,
\end{equation*}
which gives rise to an integrability condition. This summarises one of the
integrability conditions considered in \cite{Ra62}. 

Let us recover, from our point of view, the result that establishes that the
knowledge of a particular solution of the Riccati equation allows us to obtain
its general solution. In fact, under the change of variables $M=-x$, the system
(\ref{RedSep}) becomes
\begin{eqnarray}\label{eq8}
\left\{\begin{aligned}
\frac{dD}{dt}&=\left(b_2(t)+\frac{\dot
b_3(t)}{b_3(t)}\right)D-c_2D^2+2b_3(t)xD,\\
\dfrac{dx}{dt}&=b_1(t)+b_2(t) x+b_3(t) x^2.
\end{aligned}\right.
\end{eqnarray}
Each particular solution of the previous system takes the form
$(D_p(t),x_p(t))$, with $x_p(t)$ being a particular solution of the Riccati
equation (\ref{ricceq}). Therefore, given such a particular solution $x_p(t)$,
the function $D_p=D_p(t)$, corresponding to $(D_p(t),x_p(t))$, satisfies the
equation 
\begin{equation}\label{PS}
\frac{dD_p}{dt}=\left(b_2(t)+\frac{\dot
b_3(t)}{b_3(t)}+2b_3(t)x_p(t)\right)D_p-c_2D_p^2,
\end{equation}
which is is a Bernoulli equation and, therefore, is integrable by quadratures.
Consequently, the knowledge of a particular solution $x_p(t)$ of the Riccati
equation (\ref{ricceq}) allows us to determine a particular solution
$(D_p(t),x_p(t))$ of system (\ref{eq8}) and, in view of the change of variables
$x=-M$, a particular solution $(D_p(t),M_p(t))=(D_p(t),-x_p(t))$ of system
(\ref{RedSep}). Finally, the  functions $M_p(t)$ and $D(t)$ lead to the change
of variables (\ref{ChangeT3}) described in Theorem \ref{FT2} which transforms
the initial Riccati equation (\ref{ricceq}) into another one related to a
solvable Lie algebra of vector fields. 

The above process describes a reduction process similar to the one derived in
\cite{CarRamGra}, but our method allows us to obtain a direct reduction into an
integrable Riccati equation (\ref{equation83}) through a particular
solution.

There exist many ways to impose conditions on the coefficients of the second
equation in (\ref{eq8}) to obtain a particular
solution easily. For instance, if there exists a real constant $c$ such that for
the $t$-dependent functions $b_1$, $b_2$ and $b_3$ we have that $b_1+b_2 c+b_3
c^2=0$, then $c$ is a particular solution, for example:
\begin{enumerate}
 \item $b_1+b_2+b_3=0$ implies that $c=1$ is a particular solution.
\item $k_2^2b_1+k_2k_3b_2+k_3^2b_3=0$ means that $c=k_3/k_2$ is a particular
solution.
\end{enumerate}
This sketches some cases found in
\cite{CarRamGra, Stre}. 

As a first application of the above method, we can integrate the Riccati
equation
\begin{equation}\label{Hovy}
\frac{dx}{dt}=-\frac{n}{t}+\left(1+\frac{n}{t}\right)x-x^2.
\end{equation}
related to Hovy's equation \cite{Ro07}. This Riccati equation admits the
particular constant solution $x_p(t)=1$. Using such a particular solution in
equation (\ref{PS}) and taking, for instance, $c_1=0$ and $c_2=0$, we can obtain
a particular solution for equation (\ref{PS}), e.g. $D_p(t)=t^ne^{-t}$. Hence,
$(t^ne^{-t},1)$ is a particular solution of system (\ref{eq8}) related to
equation (\ref{Hovy}) and $(t^ne^{-t},-1)$ is a solution of the system
(\ref{RedSep}). In this way, Theorem \ref{FT2} states that the transformation
(\ref{ChangeT3}), determined by the $D_p(t)=t^ne^{-t}$ and $M_p(t)=-1$, of the
form
\begin{equation}\label{rel}
x'=-t^{-n}e^tc_3^{-1}(x-1),
\end{equation}
relates the solutions of equation (\ref{Hovy}) to those of the integrable
equation
$$
\frac{dx'}{dt}=e^{-t}t^nc_3x'^2.
$$
If we fix $c_3=1$, the solution of the above equation reads
\begin{equation*}
x'(t)=\frac{1}{K-\Gamma(1+n,t)},
\end{equation*}
where $K$ is an integration constant and $\Gamma(a,b)$ is the incomplete Euler's
Gamma function
\begin{equation*}
\Gamma(a,t)=\int^\infty_t t'^{a-1}e^{-t'}dt'.
\end{equation*}
In view of the change of variables (\ref{rel}), the solutions $x(t)$ of the
Riccati equation (\ref{Hovy}) 
and $x'(t)$ are related through the expression
$x'(t)=-t^{-n}e^tc_3^{-1}(x(t)-1)$. Therefore, if we
 substitute the general solution $x'(t)$ in this expression, we can derive the
 general solution 
for the Riccati equation (\ref{Hovy}), that is, 
\begin{equation*}
x(t)=1-\frac{e^{-t}t^n}{\Gamma(n+1,t)+K}.
\end{equation*}

\section{Linearisation of Riccati equations}

To finish this chapter, we shall analyse the problem of the linearisation of
Riccati equations through the linear fractional transformations ({\ref{yprime}).
As a main result, we establish various integrability conditions ensuring that a
Riccati equation can be transformed into a linear one by means of a
diffeomorphism on $\overline{\mathbb{R}}$ associated with a linear fractional
transformation of a certain class.

As a first insight in the linearisation process, notice that Corollary
\ref{CorCur} states that
 there exists a curve in $SL(2,\mathbb{R})$, and therefore a
 $t$-dependent
 linear fractional
 transformation on ${\overline{\R}}$, transforming each given Riccati equation
into any other one (and, in particular, into a linear one). This clearly implies
that Riccati equations are always linearisable by means of this class of
transformations. Nevertheless, as Lie system (\ref{IRTRR4}) describing such
transformations is related to a non-solvable Lie algebra of vector fields,
determining such a  transformation can be as difficult as solving the Riccati
equation to be linearised.

Let us try to transform a given Riccati equation into a linear differential
equation by means of a linear fractional transformation (\ref{action2})
determined by a constant vector $(\alpha,\beta,\gamma,\delta)\in \mathbb{R}^4$
with $\alpha\delta-\beta\gamma=1$. In this case, determining the conditions
ensuring the existence of solutions of system (\ref{IRTRR4}) performing such a
transformation is an easy task. Moreover, as solving system (\ref{IRTRR4}) also
becomes straightforward, we can determine some linearisability conditions and,
when these conditions hold, specify the corresponding change of variables. 

Note that as $(\alpha,\beta,\gamma,\delta)$ is a constant, we have
$\dot\alpha=\dot\beta=\dot\gamma=\dot\delta=0$ and, in view of (\ref{FS}), the
diffeomorphism on $\overline{\mathbb{R}}$ performing the transformation is
related to a vector in the kernel of the matrix
\begin{equation}\label{EM}
B=\matriz{cccc}{
\frac{b'_2-b_2}{2}&b_3 &b'_1&0\\
-b_1& \frac{b'_2+b_2}{2}&0 &b'_1\\
0&0 &-\frac{b'_2+b_2}{2}& b_3\\
0&0 &-b_1& -\frac{b'_2-b_2}{2}},
\end{equation}
where we assume $b_1\neq 0, b_3\neq 0$. We omit the study of the case
$b_1(t)b_3(t)=0$ in an open interval because, as it was shown in Section
\ref{IntRicEqu}, this case is known to be integrable.

The necessary and sufficient condition for $\ker B$ to be non-trivial  is $\det
B= 0$. Therefore, a short calculation shows that $\dim\,\, {\rm ker}\, B>0$ if
and only if
$-b_2^2+b_3'^2(t)+4 b_1 b_3)^2=0.$ Thus, $b_3'=\pm \sqrt{b_2^2-4 b_1 b_3}$ and
$b_3'$ is fixed, up to a sign, by the values of $b_1$, $b_2$ and $b_3$. Let us
study the kernel of the matrix $B$ in the positive and negative cases for
$b'_2$.

$\bullet$ {\bf Positive case}: The kernel of matrix (\ref{EM}) is given by the
vectors

\begin{equation*}
\left(\delta \frac{b_1'}{b_1}+\beta\frac{b_2+\sqrt{b_2^2-4 b_1 b_3}}{2 b_1}
,\beta,-\delta\frac{-b_2+\sqrt{b_2^2-4 b_1b_3}}{2 b_1},\delta\right),\qquad
\delta,\beta\in\mathbb{R}.
\end{equation*}
Recall that we are only considering the constant elements of $\ker B$, therefore
there should be two real constants $K_1$ and $K_2$ such that
\begin{equation}\label{eq23}
K_1=\delta\frac{b_1'}{b_1}+\beta\frac{b_2+\sqrt{b_2^2-4 b_1 b_3}}{2 b_1},\qquad
K_2=\frac{-b_2+\sqrt{b_2^2-4 b_1b_3}}{2 b_1},
\end{equation}
Moreover, in order to relate these vectors to elements in $SL(2,\mathbb{R})$, we
have to impose the condition $\det (K_1,\beta,-\delta K_2,\delta)=\delta(K_1
+\beta K_2)=1$. 

The second condition in (\ref{eq23}) imposes a restriction on the coefficients
of the
initial Riccati equation to be linearisable  by a constant linear fractional
transformation (\ref{action2}). Then, if this is satisfied, we can choose
$\beta,\gamma, K_1$ and $b_2'$ to satisfy the other conditions. Thus, the only
linearisation  condition  is the second one in (\ref{eq23}).

$\bullet$ {\bf Negative case}: In this case, $\ker\,B$ reads
\begin{equation*}
\left(\frac{\delta b_1'}{b_1}+\beta\frac{b_2-\sqrt{b_2^2-4 b_1 b_3}}{2 b_1}
,\beta,-\delta\frac{-b_2-\sqrt{b_2^2-4 b_1b_3}}{2 b_1},\delta\right),\qquad
\delta,\beta\in\mathbb{R},
\end{equation*}
and now the new conditions reduce to the existence of two real constants $K_1$
and $K_2$ such that
\begin{eqnarray*}
K_1=\frac{\delta b_1'}{b_1}+\beta\frac{b_2-\sqrt{b_2^2-4 b_1 b_3}}{2 b_1}, \quad
K_2=\frac{-b_2-\sqrt{b_2^2-4 b_1b_3}}{2 b_1},
\end{eqnarray*}
with $\delta(K_1+\beta  K_2)=1$. If the second expression of the above
conditions is satisfied, we can proceed in a similar fashion as for the positive
case to obtain the transformation that performs the linearisation of the initial
Riccati equation.

Summarising:

\begin{theorem}
The necessary and sufficient condition for the existence of a diffeomorphism on
$\bar{\mathbb{R}}$ of linear fractional type associated with a transformation
on $SL(2,\R)$ transforming the Riccati (\ref{ricceq}) into a linear differential
equation is the existence of a real constant $K$ such that
\begin{equation}\label{IntCond}
K=\frac{-b_2\pm\sqrt{b_2^2-4 b_1b_3}}{2 b_1}.
\end{equation}
\end{theorem}

As a Riccati equation (\ref{ricceq}) satisfies the above condition if and only
if $K$ is a constant particular solution, we get the following corollary:

\begin{corollary}
A Riccati equation can be linearised by means of a diffeomorphism on
$\overline{\mathbb{R}}$ of the form (\ref{action2}) if and only if it admits a
constant particular solution.
\end{corollary}

Ibragimov showed that a Riccati equation (\ref{ricceq}) is linearisable by means
of a change of variables $z=z(x)$  if and only if the Riccati equations admits a
constant solution \cite{Ib08}. Additionally, we have proved that in such a case,
the change of variables can be described by means of a transformation of the
type (\ref{action2}).

\chapter{Lie integrability in Classical Physics}

In spite of their apparent simplicity, the methods developed throughout the
previous chapter reduce the analysis of certain integrability conditions for
Riccati equations to studying integrability conditions for an equation on
$SL(2,\mathbb{R})$. Moreover, these methods can also be applied to any other Lie
system related to the same equation on $SL(2,\mathbb{R})$. For instance, we here
use the results on integrability of Riccati equations to study $t$-dependent
(frequency and/or mass) harmonic oscillators (TDHOs), which are associated with
the same kind of equations on $SL(2,\mathbb{R})$ as Riccati equations.  As a
particular application of our results, we supply 
$t$-dependent constants of the motion for certain one-dimensional TDHOs and the
solutions for a two-dimensional TDHO. 
Also, our
approach provides a unifying framework which allows us to apply our developments
to all Lie systems associated with equations in $SL(2,\mathbb{R})$ and
generalise our methods to study any Lie system.

\section{TDHO as a SODE Lie system}\label{TDFHOLS}
Let us prove that every TDHO  is a SODE Lie systems (see
\cite{CGM00,CLuc08b,CLR08}).  Each TDHO is described by a $t$-dependent
Hamiltonian of the form
$$
H(t)=\frac {p^2}{2m(t)}+\frac 12F(t)\omega^2x^2\,,
$$
whose Hamilton equations read
\begin{equation}\label{TDFHO2}
\left\{\begin{aligned}
\dot x&=\frac{\partial H}{\partial p}=\frac{p}{m(t)},\\
\dot p&=-\frac{\partial H}{\partial x}=-F(t)\omega^2x.
\end{aligned}\right.
\end{equation}
The solutions of the above system are integral curves for the $t$-dependent
vector field 
$$X_t=p\,\pd{}x-F(t)\omega^2x\pd{}p\,,
$$
over ${\rm T}^*\mathbb{R}$.  Let $X^{HO}_1, X^{HO}_2$ and $X^{HO}_3$ be  the
vector fields
\begin{equation}\label{FundamentHO}
X^{HO}_1=p\pd{}{x},\quad X^{HO}_2=\frac 12\left(x\pd{}{x}-p\pd{}{p}\right),\quad
X^{HO}_3=-x\pd{}{p}\,,
\end{equation}
which satisfy the commutation relations
\begin{equation*}
[X^{HO}_1,X^{HO}_3]=2X^{HO}_2,\quad [X^{HO}_1,X^{HO}_2]=X^{HO}_1,\quad
[X^{HO}_2,X^{HO}_3]=X^{HO}_3,
\end{equation*}
and therefore span a Lie algebra of vector fields $V^{HO}$ isomorphic to
$\mathfrak{sl}(2,\mathbb{R})$.
Then, the $t$-dependent vector field $X^{HO}$ associated with system
(\ref{TDFHO2})
can be written as
\begin{equation}\label{TDFHO4}
 X^{HO}(t)=F(t)\omega^2X^{HO}_3+\frac{1}{m(t)}X^{HO}_1\,,
\end{equation}	
i.e. it is a linear combination with $t$-dependent coefficients 
\begin{equation}\label{DecExp}
X^{HO}(t)=\sum_{\alpha=1}^3b_\alpha(t)X^{HO}_\alpha,
\end{equation}
 with $b_1(t)=1/m(t)$, $b_2(t)=0$ and $b_3(t)=F(t)\omega^2$. Hence, TDHOs are
SODE Lie systems. 

Consider the basis $\{{\rm a}_1, {\rm a}_2, {\rm a}_3\}$ for
$\mathfrak{sl}(2,\mathbb{R})$ given in (\ref{thebasis}). Its elements satisfy
the same commutation relations as the vector fields 
$X^{HO}_\alpha$. Denote by $\Phi^{HO}:SL(2,\mathbb{R})\times{\rm
T}^*\mathbb{R}\rightarrow {\rm T}^*\mathbb{R}$ the action that associates each
${\rm a}_\alpha$ with the fundamental vector field $X^{HO}_\alpha$, i.e. each
one-parameter subgroup $\exp(-t{\rm a}_\alpha)$ acts on ${\rm T}^*\mathbb{R}$
with infinitesimal generator $X^{HO}_\alpha$. It can be verified that this
action 
reads
\begin{equation*} 
\Phi^{HO}\left(\left(\begin{array}{cc}
\alpha&\beta\\
\gamma&\delta
\end{array}
\right),\left(\begin{array}{c}
x\\p\end{array}\right)\right)=
\left(\begin{array}{cc}
\alpha&\beta\\
\gamma&\delta
\end{array}
\right)\left(\begin{array}{c}
x\\p\end{array}\right).
\end{equation*}
Obviously, the linear map $\rho^{HO}:\mathfrak{sl}(2,\mathbb{R})\rightarrow
V^{HO}$ that maps each  ${\rm a}_\alpha$ to $ X_\alpha$ is
a Lie algebra isomorphism. 

The action $\Phi^{HO}$ allows us to relate (\ref{TDFHO2}) to an equation on
$SL(2,\mathbb{R})$ given by
\begin{equation}\label{IIeLA}
R_{A^{-1}*A}\dot A=-\sum_{\alpha=1}^3b_\alpha(t){\rm a}_\alpha,\qquad A(0)=I.
\end{equation}
Thus, if $A(t)$ is the solution of (\ref{IIeLA}) and we denote $\xi=(x,p)\in{\rm
T}^*\mathbb{R}$, then the solution starting from $\xi(0)$ is
$\xi(t)=\Phi^{HO}(A(t),\xi(0))$ (see e.g. \cite{CarRamGra}). In summary, system
(\ref{TDFHO2}) is a Lie system in ${\rm
  T}^*\mathbb{R}$ related
to an equation on $SL(2,\mathbb{R})$ and the solution of equation (\ref{IIeLA})
 allows us to obtain the solutions of (\ref{TDFHO2}) in terms of the
 initial condition by means of the action $\Phi^{HO}$. 

\section{Transformation laws of Lie equations on {\it SL(2,R)}}
Each $t$-dependent harmonic oscillator (\ref{TDFHO2}) can be considered as a
curve in ${\mathbb{R}}^3$ of the form $(b_1(t),b_2(t),b_3(t))$ through the
decomposition (\ref{DecExp}). Then, we can
transform each curve $\xi(t)$ in ${\rm T}^*\mathbb{R}$,
by an element $\bar A(t)$ of $\mathcal{G}$ as follows:
\begin{eqnarray} {\rm If}  \ \bar A(t)=\!\matriz{cc}
{{\bar{\alpha}(t)}&{\bar{\beta}(t)}\\{\bar{\gamma}(t)}&{\bar{\delta}(t)}}\in{
\cal G},\qquad 
\Theta(\bar A,\xi)(t)=\!
\left(\begin{aligned}
{\bar{\alpha}(t) x(t)+\bar{\beta}(t) p(t)}\\
{\bar{\gamma}(t) x(t)+\bar{\delta}(t)p(t)}
\end{aligned}\right)\!.%\label{acciondecurva}
\end{eqnarray}
The above change of variables transforms the TDHO (\ref{TDFHO2})
into an analogous 
 TDHO with new coefficients $b'_1,b'_2, b'_3$ given by
\begin{equation}\label{trans40}
\left\{\begin{aligned}
b'_3=&{\bar\delta}^2\,b_3-\bar\delta\bar\gamma\,b_2+{\bar\gamma}^2\,
b_1+\bar\gamma {\dot{\bar\delta}}-\bar\delta \dot{\bar\gamma}\ ,
\nonumber \\
b'_2=&-2\,\bar\beta\bar\delta\,b_3+(\bar\alpha\bar\delta+\bar\beta\bar\gamma)\,
b_2-2\,\bar\alpha\bar\gamma\,b_1
       +\delta \dot{\bar\alpha}-\bar\alpha \dot{\bar\delta}+\bar\beta
\dot{\bar\gamma}-\bar\gamma \dot{\bar\beta}\ ,   \\
b'_1=&{\bar\beta}^2\,b_3-\bar\alpha\bar\beta\,b_2+{\bar\alpha}^2\,
b_1+\bar\alpha\dot{\bar\beta}-\bar\beta\dot{\bar\alpha} \ .
 \nonumber
\end{aligned}\right.
\end{equation}
The solutions of the transformed TDHO are of the form $\Theta(\bar
A(t),\xi(t))$, with $\xi(t)$ being a solution of the initial TDHO. Additionally,
the above expressions define an affine action (see e.g. \cite{LM87} for the
general definition of
this concept) of the group
 ${\GR}$ on the set of
TDHOs \cite{CarRam}.
This means that in order  to transform the coefficients of a TDHO by means of
two transformations of the above type, first through $A_1$ and then by means of
$A_2$, 
it suffices to do the transformation induced by
the product $A_2\,A_1$.

The result of this action of $\mathcal{G}$ can also be studied from the point of
view
of the equations in $SL(2,\mathbb{R})$. First, $\mathcal{G}$  acts on the left
on the
set of curves in $SL(2, \mathbb{R})$ by left translations, i.e. a curve $\bar
A(t)$
transforms the curve $A(t)$  into a new one  $A'(t)=\bar A(t) A(t)$. Therefore,
if
 $A(t)$ is a solution of (\ref{IIeLA}), characterised by a curve ${\rm a}(t)\in
\mathfrak{sl}(2,\mathbb{R})$,
then the new curve satisfies a new equation like (\ref{IIeLA}) but  with a
different right-hand side,
 ${\rm a}'(t)$, and thus it corresponds to a new equation on
$SL(2,\mathbb{R})$ associated with a new TDHO. Of course, $A'(0)=\bar A(0)A(0)$,
and if we want $A'(0)={\rm Id}$, we have to impose the additional condition 
$\bar A(0)={\rm Id}$. In this way
$\mathcal{G}$
 acts on the set of curves in
$T_ISL(2,\mathbb{R})\simeq\mathfrak{sl}(2,\mathbb{R})$. It can be shown that
 the relation between both
curves ${\rm a}(t)$ and ${\rm a}'(t)$ in $\mathfrak{sl}(2,\mathbb{R})$ is given
by \cite{CarRamGra}
\begin{equation}
{\rm a}'(t)=-\sum_{\alpha=1}^3b'_\alpha(t){\rm a}_\alpha=\bar A(t){\rm a}(t)\bar
A^{-1}(t)+\dot{\bar{A}}(t)\bar A^{-1}(t)
\,. \label{newTDHO}
\end{equation}

Summarising, it has been  shown that it is possible to associate, in a
one-to-one way, any TDHO with an equation in
the Lie group $SL(2,\mathbb{R})$ and to define a group $\mathcal{G}$ of
transformations on
the set of such TDHOs induced by the natural linear action of
$SL(2,\mathbb{R})$.

Recall that, in view of Theorem \ref{THLS}, system (\ref{newTDHO}) can be
regarded as a system of first-order ordinary differential equations in the
coefficients of the
 curve in $SL(2,\mathbb{R})$ of the form
$$
\bar A(t)=
\matriz{cc}{
\alpha(t) &\beta(t)\cr
\gamma(t)& \delta(t)}\,.
$$
Moreover, we can enunciate the following results, which are a straightforward
application to TDHOs of Theorem \ref{THLS} and Corollary \ref{CorCur} formulated
for the analysis of certain Lie systems on $SL(2,\R)$ related to Riccati
equations.

\begin{theorem}\label{IITHLS} The curves in $SL(2,\mathbb{R})$  transforming a
TDHO related to an equation on this Lie group determined by a curve ${\rm
a}(t)$ into a new TDHO associated with an equation on $SL(2,\mathbb{R})$
determined by the curve ${\rm a}'(t)$, with 
$$
{\rm a}'(t)=-\sum_{\alpha=1}^3b'_\alpha(t){\rm a}_\alpha\,,
\qquad {\rm a}(t)=-\sum_{\alpha=1}^3b_\alpha(t){\rm a}_\alpha,$$
 are given by the integral curves  of the $ t$-dependent vector field 
\begin{equation}\label{TRR40}
N(t)=\sum_{\alpha=1}^3\left(b_\alpha(t)N_\alpha+b_\alpha'(t)N'_\alpha\right)\,,
\end{equation}
such that $\det \bar A(0)=1$. This system is a Lie system associated with a
non-solvable Lie algebra of vector fields isomorphic to
$\mathfrak{sl}(2,\mathbb{R})\oplus\mathfrak{sl}(2,\mathbb{R})$. Moreover, such
curves also transform the TDHO related to the curve ${\rm a}(t)$ into the new
one linked to ${\rm a}'(t)$.
\end{theorem}

\begin{corollary}  Given two TDHOs associated with the
  curves ${\rm a}(t)$ and ${\rm a}'(t)$ in
$\mathfrak{sl}(2,\mathbb{R})$, there  always exists 
a curve in $SL(2,\mathbb{R})$ transforming the first TDHO into the second one.
\end{corollary}

We must remark that even if we know that given two equations in the Lie group
$SL(2,\R)$ there
 always exists a transformation relating both, in order to find such a
curve  we need to solve  the system of
differential
equations providing the integral curves of (\ref{TRR40}).  This 
is the solution of a system of differential
equations that is a Lie system related to a non-solvable Lie
algebra in general. Hence,
 it is
not  easy to find its solutions, i.e.  it may not be integrable by
quadratures.

The result  of Theorem  \ref{IITHLS},  i.e. that the system of differential
equations describing the transformations of Lie systems on $SL(2,\mathbb{R})$
is a matrix Riccati equation associated, as a Lie system, with a Lie algebra
isomorphic to
$\mathfrak{sl}(2,\mathbb{R})\oplus\mathfrak{sl}(2,\mathbb{R})$, suggests us a
method to find some sufficiency conditions
for  integrability of the TDHOs to be explained next.

\section{Description of some known integrability conditions}%\label{DIC}

We now study some cases when it is possible to find 
curves $\bar A(t)$ in $SL(2,\mathbb{R})$ transforming a
given TDHO related to an equation on $SL(2,\mathbb{R})$ characterised by a curve
${\rm a}(t)$ into a new TDHO associated with an equation on $SL(2,\mathbb{R})$
characterised by a curve of the type ${\rm a}'(t)=-D(t)(c_1{\rm a}_1+c_2{\rm
a}_2+c_3{\rm a}_3)$. This is possible if the system determined by (\ref{TRR40})
can be solved easily. The transformation
establishing the relation to such a TDHO allows us to find
 the solution of the given equation by quadratures. We first restrict ourselves
to
studying cases in which 
the curve $\bar A(t)$ lies in a one-parameter subset of $SL(2,\mathbb{R})$. The
results 
we show next are a direct translation to the framework of TDHO
of Theorem \ref{THLS} describing certain integrability properties of Riccati
equations (see also \cite{CLR07b}).
\begin{theorem}\label{ThTDFHO} The necessary and sufficient conditions
for the existence of a  transformation
\begin{equation*}
\xi'=\Phi^{HO}(\bar A_0(t),\xi),\qquad \xi=\left(\begin{array}{c}x\cr
p\cr\end{array}\right),
\end{equation*}
with
\begin{equation}
\bar A_0(t)=\left(\begin{array}{cc}
\alpha(t)&0\\
0&\alpha^{-1}(t)
\end{array}
\right),\qquad \alpha(t)>0\,,\label{gcero}
\end{equation}
 relating the TDHO associated with the $t$-dependent vector field 
\begin{equation}
X_t=b_1(t)X_1+b_2(t)X_2+b_3(t)X_3, \label{binX}
\end{equation}
where $b_1(t)b_3(t)$ has a constant sign, i.e. $b_1(t)b_3(t)\ne 0$, to another
integrable one given by
\begin{equation}
X'(t)=D(t)(c_1X_1+c_2X_2+c_3X_3)\,,\label{integrableuno}
\end{equation}
with, $c_1,c_2,c_3,$ being real numbers such that $c_1c_3\neq 0$,
are
\begin{equation*}
D^2(t)c_1c_3=b_1(t)b_3(t),\qquad {b_2(t)+\frac{1}{2}\left(\frac{\dot
      b_3(t)}{b_3(t)}-\frac{\dot b_1(t)}{b_1(t)}\right)}=c_2\sqrt{
\frac{b_1(t)b_3(t)}{c_1c_3}}.
\end{equation*}

Then, the  transformation is uniquely defined by
\begin{equation*}
\bar
A_0(t)=\left(\begin{array}{cc}\left(\frac{b_3(t)c_1}{b_1(t)c_3}\right)^{1/4}&0\\
0&\left(\frac{b_3(t)c_1}{b_1(t)c_3}\right)^{-1/4} \end{array}\right)\,.
\end{equation*}
\end{theorem}
Note that one coefficient, either $c_1$ or $c_3$, can be reabsorbed with a
redefinition of the function $D$.
As a straightforward application of the preceding theorem, which can be found in
a
similar way as those in  \cite{CLR07b}, we obtain the following corollaries:
\begin{corollary}
A TDHO (\ref{TDFHO2}) with $b_1(t)b_3(t)\neq 0$ is integrable by a $t$-dependent
change of variables 
\begin{equation*}
\xi'=\Phi^{HO}(\bar A_0(t),\xi),
\end{equation*}
with $\bar A_0$ given by (\ref{gcero}),
if and only 
\begin{equation}
\sqrt{\frac{c_1c_3}{b_1(t)b_3(t)}}\left[b_2(t)+\frac{1}{2}\left(\frac{\dot
b_3(t)}{b_3(t)}-\frac{\dot b_1(t)}{b_1(t)}\right)\right]=c_2\,,\label{Kcond}
\end{equation}
for certain real constants $c_1, c_2,$ and $c_3$. In this case 
$$D(t)=\sqrt{\frac{b_1(t)b_3(t)}{c_1c_3}},$$ and the new system is 
\begin{equation}
\frac{d\xi'}{dt}=D(t)\matriz{cc}{c_2/2&c_1\\-c_3&-c_2/2}\,\xi'\,.\label{Kevol}
\end{equation}
\end{corollary}

\begin{corollary} Given an integrable TDHO characterised by a  $t$-dependent
vector field (\ref{integrableuno}), 
the set of TDHOs which can be  obtained through a $t$-dependent 
transformation \begin{equation*}
\xi'=\Phi^{HO}(\bar A_0(t),\xi),
\end{equation*}
with $\bar A_0$ given by (\ref{gcero}), are
those of the form
\begin{equation}
X_t=b_1(t)X_1+\left( \frac{\dot b_1(t)}{b_1(t)}-\frac{\dot
D(t)}{D(t)}+c_2D(t)\right)X_2+\frac{D^2(t)c_1c_3}{b_1(t)}X_3\,.\label{Xdetbcero}
\end{equation}
Thus, $\bar A_0(t)$ reads
\begin{equation*}
\bar
A_0(t)=\left(\begin{array}{cc}\left(\frac{b_3(t)c_1}{b_1(t)c_3}\right)^{1/4}&0\\
0&\left(\frac{b_3(t)c_1}{b_1(t)c_3}\right)^{-1/4} \end{array}\right)\,.
\end{equation*}
\end{corollary}

Therefore, starting from an integrable system we can find the family of
$t$-dependent vector fields
(\ref{Xdetbcero}) describing 
solvable TDHO systems whose coefficients are parametrised by $b_1(t)$. 
Given a 
 TDHO, it is easy to check whether it belongs to such a family 
 and can be easily integrated.

The integrability conditions we have described here arise as  
requirements on the initial $t$-dependent functions $b_\alpha$ that allow us to
solve the initial TDHO exactly by a $t$-dependent transformation of the form
\begin{equation*}
\xi'=\Phi^{HO}(\exp(\Psi(t)v),\xi),
\end{equation*}
with some $v\in\mathfrak{sl}(2,\mathbb{R})$ and  $\Psi(t)$, in such a way that
the initial TDHO system
(\ref{TDFHO2}) in the variable $\xi$ is transformed into another one in the
variable $\xi'$ associated, as  a Lie system, with a  Vessiot--Guldberg Lie
algebra
isomorphic to an appropriate 
 Lie subalgebra of $\mathfrak{sl}(2,\mathbb{R})$ in such a way that 
 the equation in $\xi'$
can be integrated by quadratures and, consequently, the equation in $\xi$ is
solvable too.

\section{Some applications of integrability conditions to TDHOs}

As a first application, we show that the usual approach to the solution of the
classical Caldirola--Kanai
Hamiltonian \cite{Cal41, Ka48} can be explained through our method (the solution
of the quantum case can be obtained in a similar way). Next, we will also apply
our approach to get integrable TDHOs. 

The Hamiltonian of a  $t$-dependent harmonic oscillator is 
\begin{equation}
H(t)=\frac 12\,\frac {p^2}{m(t)}+\frac 12\, m(t)\omega^2(t)x^2\,.\label{TDHOH}
\end{equation}
For instance, a harmonic oscillator with a damping term \cite{Cal41,Ka48} with
equation of motion
$$
\frac{d}{dt}(m_0\dot x)+m_0\mu \dot x+kx=0,\qquad k=m_0\omega^2,
$$
admits a Hamiltonian description, with a $t$-dependent Hamiltonian
$$
H(t)=\frac{p^2}{2m_0}\exp(-\mu t)+\frac 12 m_0\exp(\mu t)\omega^2x^2,
$$
i.e. $m(t)$ in (\ref{TDHOH}) corresponds to $m(t)=m_0\exp(\mu t)$. In this case
equations (\ref{TDFHO2}) are
\begin{equation}
\left\{\begin{array}{rcl}
\dot x&=&\dfrac{\partial H}{\partial p}=\frac{1}{m_0}\exp(-\mu t)p,\\
\dot p&=&-\dfrac{\partial H}{\partial x}=-m_0\exp(\mu t)\, x,
\end{array}\right.\,\label{HeqTDHO}
\end{equation}
and the $t$-dependent coefficients of the associated Lie system read
$$b_1(t)=\frac 1{m_0}\,\exp(-\mu t)\,, \quad b_2(t)=0\,,\quad 
b_3(t)=m_0\omega^2\exp(\mu t)\,.
$$
Therefore, as $b_1(t)b_3(t)=\omega^2$, $b_2=0$ and 
$$\frac{\dot b_3}{b_3}-\frac{\dot b_1}{b_1}=2\mu\,,
$$
we see that (\ref{Kcond}) holds if we set $c_1=c_3=1, c_2=\mu/\omega$ and the
function $D$ is
a constant, $D=\omega$.  Hence, this
example reduces to the system
$$\frac{d}{dt}\matriz{c}{x'\\p'}=\matriz{cc}{\mu/2&\omega\\-\omega&-\mu/2}
\matriz{c}{x'\\p'},$$
which can be easily integrated. If we put $\bar\omega^2=({\mu^2}/{4})-\omega^2$,
we get
{\small
\begin{equation*}
\left(\begin{array}{c}
x'(t)\\p'(t)
\end{array}\right)
=\left(\begin{array}{cc}
\ch(\bar\omega t)+\dfrac{\mu}{2\bar\omega}\sh(\bar\omega t)&
\dfrac{\omega}{\bar\omega}\sh(\bar\omega t)\\
-\dfrac{\omega}{\bar\omega}\sh(\bar\omega t)& \ch(\bar\omega
t)-\dfrac{\mu}{2\bar \omega}\sh(\bar\omega t)\\
\end{array}\right)
\left(\begin{array}{c}
x'(0)\\p'(0)
\end{array}\right)
\end{equation*}}
and, in terms of the initial variables, we obtain
\begin{equation*}
x(t)=\frac{e^{-\mu t/2}}{\sqrt{m_0\omega}}\left(
\left(\ch(\bar\omega t)+\frac{\mu}{2\bar\omega}\sh(\bar\omega
t)\right)\sqrt{m_0\omega}x_0+
\frac{\omega}{\bar\omega}\sh(\bar\omega t)\frac{p_0}{\sqrt{m_0\omega}} \right).
\end{equation*}

We can also study a TDHO described by the $t$-dependent Hamiltonian
\begin{equation*}
H(t)=\frac 12p^2+\frac 12F(t)\omega^2x^2\,,\quad F(t)>0,
\end{equation*}
where we assume, for simplicity, $m=1$. The $t$-dependent vector field $X$ is
$$X_t=p\,\pd{}x-F(t)\omega^2x\pd{}p\,,
$$
which is a linear combination
$$
 X_t=F(t)\omega^2X^{HO}_3+X^{HO}_1\,, 
$$
i.e. the $t$-dependent coefficients in (\ref{binX}) are
$$b_1(t)=1\,, \quad b_2(t)=0\,,\quad b_3(t)=F(t)\omega^2\,,
$$
and the condition for $F$ to satisfy (\ref{Kcond}) is
$$
\frac 12\, \frac{\dot F} F=c_2\,\omega\, \sqrt{F}.
$$
Therefore, $F$ must be of the form 
\begin{equation*}
F(t)=\frac{1}{(L-c_2\omega t)^2}\,
\end{equation*}
and the Hamiltonian, which can be exactly integrated, is
\begin{equation*}
H(t)=\frac{p^2}{2}+\frac{1}{2}\frac{\omega^2}{(L-c_2\omega t)^2}x^2\,.
\end{equation*}
The corresponding Hamilton equations are
\begin{equation*}
\left\{\begin{aligned}
\dot x&=p,\\
\dot p&=-\frac{\omega^2}{(L-c_2\omega t)^2}\,x,
\end{aligned}\right.
\end{equation*}
and  the $t$-dependent change of variables  to perform is
\begin{equation*}
\left\{\begin{aligned}
x'&=\sqrt{\frac{\omega}{L-c_2\omega t}}\,x,\\
p'&=\sqrt{\frac{L-c_2\omega t}{\omega}}\,p.
\end{aligned}\right.
\end{equation*}
In consequence,
\begin{equation}\label{Lieq5}
\left\{\begin{aligned}
\frac{dx'}{dt}&=\frac{\omega}{L-c_2\omega t}\left(\frac{c_2}{2}x'+p'\right),\\
\frac{dp'}{dt}&=\frac{\omega}{L-c_2\omega t}\left(-x'-\frac{c_2}{2}p'\right),
\end{aligned}\right.
\end{equation}
and, under the $t$-reparametrisation,
\begin{equation*}
\tau(t)=\int^{t}_0\frac{\omega dt'}{L-c_2\omega t'}=\frac{1}{c_2}{\rm ln}
\left(\frac{K'}{L-c_2\omega t}\right),
\end{equation*}
the system (\ref{Lieq5}) becomes
\begin{equation*}
\left\{\begin{aligned}
\frac{dx'}{d\tau}&=\frac{c_2}{2}x'+p',\\
\frac{dp'}{d\tau}&=-x'-\frac{c_2}{2}p',
\end{aligned}\right.
\end{equation*}
which is equivalent to a transformed Caldirola--Kanai differential equation 
through the change $\tau \mapsto \omega\, t$ and $c_2\mapsto \mu/\omega$. In any
case, the solution is
\begin{equation*}
x'(\tau)=\left(\ch(\widetilde\omega\tau)+\frac{c_2}{2\widetilde\omega}
\sh(\widetilde\omega\tau)\right)x'(0)+\frac{1}{\widetilde\omega}
\sh(\widetilde\omega \tau) p '(0),
\end{equation*}
where $\widetilde\omega=\sqrt{\frac{c_2^2}{4}-1}$. Finally,
{\footnotesize \begin{equation*}
x(\tau(t))=\sqrt{\frac{L-c_2\omega
t}{\omega}}\left[\left(\ch(\widetilde\omega\tau(t))+\frac{c_2}{2\widetilde\omega
}\sh(\widetilde\omega\tau(t))\right)x'(0)+\frac{1}{\widetilde\omega}
\sh(\widetilde\omega \tau(t)) p '(0)\right].
\end{equation*}}

Let us analyse another integrability condition that, as the
preceding one, arises as a compatibility condition for a restricted case
of the system describing the integral curves
of  (\ref{TRR40}). Nevertheless, this time, the solution is restricted to a
one-parameter set of matrices of $SL(2,\mathbb{R})$ that is not a group in
general.

In this way, we deal with a family of transformations
\begin{equation}
\bar A_0(t)=\left(\begin{array}{cc}\frac{1}{V(t)}&0\\-u_1&
    V(t)\end{array}\right)\,,\qquad V(t)>0\,,\label{newfam}  
\end{equation}
where $u_1$ is a constant, i.e.  we want to relate the $t$-dependent vector
field
\begin{equation*}
X_t=X^{HO}_1+F(t)\omega^2X^{HO}_3,
\end{equation*}
characterised by the coefficients in (\ref{binX}) 
\begin{equation*}
b_1=1,\quad b_2=0,\quad b_3=F(t)\omega^2,
\end{equation*}
to an integrable one characterised by $b'_1,b'_2$ and $b'_3$,
or more explicitly, to the $t$-dependent vector field 
\begin{equation*}
X_t=D(t)(c_1X_1+c_3X_3)\,,
\end{equation*}
i.e. $b'_1=Dc_1$, $b'_2=0$, and $b'_3=Dc_3$. Moreover, if $c_1\ne 0$, we can
reabsorb its value redefining $D$ and assuming $c_1=1$.

Under the action of  (\ref{newfam}),  the original system transforms into 
 the following system
\begin{equation*}
\left\{\begin{array}{rl}
b'_3&= V^2b_3+u_1Vb_2+u_1^2b_1-u_1\dot V,\\
b'_2&=b_2+2\dfrac {u_1}Vb_1-2\dfrac{\dot V}V,\\
b'_1&=\dfrac 1{V^2}b_1\,.\end{array}\right.
\end{equation*}

As $b_2=b'_2=0$ and $b_1=1$, the second  equation yields $\dot V=u_1$,
i.e. $V(t)=u_1t+u_0$ with $u_0\in \mathbb{R}$. Moreover, using this condition on
the
first equation together with $b_1=1$, we get $b'_3= V^2b_3$. Then, as the third
equation gives us the value of $D$ as $D=b'_1=1/V^2$,
we see that $b'_3=Dc_3=V^2F(t)\omega^2$. Therefore, $F$ has to be 
 proportional to $(u_1t+u_0)^{-4}$, 
$$F(t)=\frac k{(u_1t+u_0)^{4}}\,, \qquad k=\frac{c_3}{\omega^2}\,.
$$
Let assume $k=1$ and thus, $c_3=\omega^2$. Then,  the $t$-dependent
transformation $\bar A_0(t)$ performing  this reduction is 
\begin{equation*}
\left\{\begin{aligned}
x'&=\frac{x}{V(t)},\\
p'&=-u_1x+V(t)p.
\end{aligned}\right.
\end{equation*}
Under this transformation, the initial system becomes
\begin{equation*}
\left\{\begin{aligned}
\frac{dx'}{dt}&=\frac{p'}{V^2(t)},\\
\frac{dp'}{dt}&=-\frac{\omega^2x'}{V^2(t)}\,.
\end{aligned}\right.
\end{equation*}
Using the $t$-reparametrisation
\begin{equation*}
\tau(t)=\int^t_0\frac{ dt'}{V^2(t')}=\frac{1}{u_1}\left(\frac{1}{u_0}-\frac{1}{
V(t)}\right),
\end{equation*}
we get the following autonomous linear system
\begin{equation*}
\left\{\begin{aligned}
\frac{dx'}{d\tau}&=p',\\
\frac{dp'}{d\tau}&=-\omega^2x',
\end{aligned}\right.
\end{equation*}
whose solution is
\begin{equation*}
\left(\begin{array}{c}x'(\tau)\\ p'(\tau)\end{array}\right)=
\left(\begin{array}{cc}\cos(\omega \tau)&\frac{\sin(\omega
\tau)}{\omega}\\-\omega\sin(\omega \tau)&\cos(\omega \tau)\end{array}\right)
\left(\begin{array}{c}x'(0)\\ p'(0)\,,\end{array}\right).
\end{equation*}
Thus, we obtain that
\begin{equation*}
x(t)=V(t)\left(\cos(\omega\,
\tau(t))\frac{x_0}{u_0}+\frac{1}{\omega}\sin(\omega\,
\tau(t))(-u_1x_0+u_0p_0)\right)\,.
\end{equation*}
\section{Integrable TDHOs and {\it t}-dependent constants of the motion}
The autonomisations of the transformed integrable systems obtained above enable
us to construct $t$-dependent constants of the motion. Indeed, in previous
cases, a TDHO was transformed into a Lie system related to an equation on
$SL(2,\mathbb{R})$ 
\begin{equation*}
R_{A^{-1}*A}\dot A=-D(t)\left(c_1M_0+c_2{\rm a}_1+c_3 {\rm a}_1\right),
\end{equation*}
associated with a TDHO determined by the  $t$-dependent vector field 
\begin{equation*}
X_t=D(t)(c_1X_1+c_2X_2+c_3 X_3).
\end{equation*}
Each $t$-dependent first-integral $I(t)$ of this differential equation satisfies
\begin{equation*}
\frac{dI}{dt}=\pd{I}{t}+X_t I=0.
\end{equation*}
Thus, the function $I$ is a first-integral of the vector field on
$\mathbb{R}\times {\rm T}^*\mathbb{R}$ 
\begin{equation*}
\overline X_t =c_1X_1(t)
  +c_2X_2(t)    +c_3 X_3(t)    +\frac{1}{D(t)}   \pd{}{t}.
\end{equation*}
As $\mathbb{R}\times {\rm T}^*\mathbb{R}$ is a three-dimensional manifold and
the
 differential equation we are studying is determined by a distribution of
 dimension one, there exist (at least locally) two independent
first-integrals. Next, we will analyse some integrable cases and their
corresponding constants of the motion.
\medskip

$\bullet$ Case $F(t)=(u_1t+u_0)^{-2}$:

\medskip

In this case we obtain that, according to Theorem \ref{ThTDFHO}, the 
$t$-dependent vector field of the
initial TDHO is transformed into the following one,
\begin{equation*}
X_t=\frac{\omega}{u_1t+u_0}\left(X_1^{HO}    -\frac{u_1}{\omega}X_2^{HO}   
+X_3^{HO}    \right)
\end{equation*}
and thus, using the method of characteristics, we obtain the following
constants of the motion
for this TDFHO:
\begin{equation*}
\begin{aligned}
I_1&=-\frac{u_1}{\omega}p'\, x' +x'^2+p'^2,\qquad
I_2&=\frac{(u_1+u_0 t)^{\omega/u_1}}{\left((\frac
{u_1}{\omega}x'-2p')+2\bar\omega x'\right)^{\frac 1{\bar\omega}}},
\end{aligned}
\end{equation*}
with $\bar \omega =\pm\sqrt{\frac{u_1^2}{4\omega^2}-1}$.
\medskip

$\bullet$ Case $F(t)=(u_1t+u_0)^{-4}$:

\medskip

In this case we see that  the  $t$-dependent vector field of the 
initial TDHO is transformed into
\begin{equation*}
X_t=\frac{1}{V^2(t)}\left(X_1^{HO}+\omega^2X^{HO}_3\right),
\end{equation*}
and thus, using the method of characteristics, we get the following
$t$-dependent constants of the motion 
 for the initial TDHO
\begin{equation}\label{integralnew}
\begin{aligned}
I_1&=\left(\frac{x\,\omega}{V(t)}\right)^2+\left(V(t)p-u_1x\right)^2,\\
I_2&=\arcsin\left(\frac{x\,\omega}{V(t)\sqrt{I_1}}\right)+
\frac{\omega}{u_1V(t)}\,.
\end{aligned}
\end{equation}
As we have two $t$-dependent constants of the motion over $\mathbb{R}\times {\rm
T}^*\mathbb{R}$ and the solutions in this space are of the form $(t,x(t),p(t))$,
we can obtain the solutions for our initial system.
\section{Applications to two-dimensional TDHO's}
In this section we apply our previous geometrical methods to analyse the
following
 two-dimensional $t$-dependent harmonic oscillator
\begin{equation*}
H(t,x_1,x_2,p_1,p_2)=\frac{p_1^2}{2}+\frac{p_2^2}{2}+\frac{\omega_1^2x_1^2+
\omega_2^2x_2^2}{2V^4(t)}\,,
\end{equation*}
with $\omega_1$ and $\omega_2$ being constant and
$V(t)=u_1t+u_0$. Nevertheless, our approach is also valid for the
corresponding 
generalisation to a $n$-dimensional TDHOs. This
Hamiltonian 
is related to an uncoupled pair of TDHOs and therefore the same development of
the last section applies again. In this way, we obtain that its Hamilton
equations read
\begin{equation*}
\left\{\begin{aligned}
\dot x_i&=p_i,\\
\dot p_i&=-\frac{\omega_i^2}{V^4(t)}x_i,\\
\end{aligned}\right.\qquad i=1,2,
\end{equation*}
and can be transformed into the system
\begin{equation*}
\left\{\begin{aligned}
\frac{dx'_i}{dt}&=\frac{1}{V^2(t)}p'_i,\\
\frac{dp'_i}{dt}&=-\frac{\omega_i^2}{V^2(t)}x'_i,\\
\end{aligned}\right.\qquad i=1,2,
\end{equation*}
by means of the $t$-dependent change of variables
\begin{equation*}
\left\{\begin{aligned}
x_i'&=\frac{x_i}{V(t)},\\
p_i'&=-u_1x_i+V(t)p_i,\\
\end{aligned}\right.\qquad i=1,2.
\end{equation*}
The solutions of the latter system are integral curves of a $t$-dependent vector
field
in the distribution generated by the vector field
\begin{equation*}
X=-\omega_1^2x_1'\pd{}{p_1'}+p'_1\pd{}{x'_1}-\omega_2^2x_2'\pd{}{p_2'}+p'_2\pd{}
{x'_2}.
\end{equation*}
If we consider the problem as a differential equation in ${\rm
T}^*\mathbb{R}^2$, the constants of the motion are first-integrals for the
vector field
$X+\partial/\partial {t}$ over $\mathbb{R}\times{\rm T}^*\mathbb{R}^2$. 
Then, as we have a distribution of rank one over a
five-dimensional manifold, there exist, at least locally, four functionally
independent first-integrals. Additionally, three of them can be chosen to be
$t$-independent ones (in terms of the variables $x_1',x_2',p_1',p_2'$). The
constants of the motion for the initial TDHO corresponding to some of such
first-integrals read
\begin{equation*}
I_i=\left(\frac{\omega_ix_i}{V(t)}\right)^2+\left(V(t)p_i-u_1x_i\right)^2,\qquad
i=1,2,
\end{equation*}
and 
\begin{equation*}
I_{12}=\frac{1}{\omega_1}\arcsin\left(\frac{x_1\omega_1}{V(t)\sqrt{I_1}}
\right)-\frac{1}{\omega_2}\arcsin\left(\frac{x_2\omega_2
}{\sqrt{V(t)I_2}}\right).
\end{equation*}
This first-integral is constant along the solutions. Nevertheless, in order for
the function to be correctly defined, $\omega_1/\omega_2$ needs to be rational.
Finally, with the aid of (\ref{integralnew}), we can obtain two $t$-dependent
constants of the motion of the form
\begin{equation*}
\bar
I_i=\frac{\omega_i}{V(t)u_1}+\arcsin\left(\frac{x_i'\omega_i}{\sqrt{I_i}}
\right)\qquad i=1,2.
\end{equation*}
As a consequence, we can explicitly obtain the $t$-evolution of the
system. Indeed, either from $\bar I_1$ or $\bar I_2$, we reach
      the following  solutions
for $x_1$ and $x_2$ 
\begin{equation*}
x_i(t)=\frac{V(t)\sqrt{I_i}}{\omega_i}\sin\left(\bar{I}_i-\frac{\omega_i}{
V(t)u_1}\right),\qquad i=1,2.
\end{equation*}
The properties of these 
solutions become clearer when we write them 
 in the following way
\begin{equation*}
\begin{aligned}
x_i(t)&=\frac{V(t)\sqrt{I_i}}{\omega_i}\sin\left(\bar{I}_i-\frac{\omega_i}{
u_1(u_1t+u_0)}\right),\qquad i=1,2,
\end{aligned}
\end{equation*}
and we realise that 
the quotient $x_1(t)/x_2(t)$ is a
 $t$-independent constant of the motion if 
$\omega_1/\omega_2$ is rational.

  These two equations can be considered as the parametric representation
of a curve on the configuration space $Q=\mathbb{R}^2$.
In the general case $x_1$ and $x_2$ evolve in an independent way and
the behaviour of the curve becomes blurred.
In the rational case, the evolutions of $x_1$ and $x_2$ are
correlated in such a way that the $t$--dependent coupling function
$I_{12}$ is preserved.
The particular form of this curve will depend on the relation
between $u_1$ and $u_0$.
If $u_1=0$ it will be a Lissajous curve.
If $u_1\ne 0$ it can be considered as a curve obtained by the addition of 
growing amplitudes to the oscillations of the corresponding
Lissajous curve.
We can refer to them as `$t$--dependent Lissajous' figures. Nevertheless, it is
not
totally clear whether this term is appropriate, since these new
 curves are `not closed'.

\chapter{Integrability in Quantum Mechanics}
Some papers have recently been devoted to applying the theory of Lie systems
\cite{CGM07,LS,PW} to Quantum Mechanics \cite{CLR08WN,CarRam05b}. As a result,
it has been proved that the theory of Lie systems can be used to treat some
types of Schr\"odinger equations, the so-called quantum Lie systems, to obtain
exact solutions, $t$-evolution operators, etc. One of the fundamental properties
found is that quantum Lie systems can be investigated by means of equations in a
Lie group. Through this equation we can analyse the properties of the associated
Schr\"odinger equation, e.g. the type of Lie group allows us to know if a
Schr\"odinger equation can be integrated \cite{CLR08WN}.

Lately, a lot of attention has also been dedicated to the study of integrability
of Lie systems and, in particular, of Riccati equations
\cite{CarRamGra,CRL07d,CLR07b}. In these papers, as in previous sections, it has
been shown that integrability conditions for Lie systems, in the case of Riccati
equations, appear related to some transformation properties of the associated
equations in $SL(2,\mathbb{R})$. Nevertheless, as we have pointed out in this
work and it was shown in \cite{CRL07d}, the same procedure used to investigate
Riccati equations can be applied to deal with any Lie system. 

Therefore, in the case of a quantum Lie system, there exists an equation on a
Lie group associated with it \cite{CLR08WN}. The transformation properties
investigated in the theory of integrability of Lie systems can be used to study
integrability conditions for quantum Lie systems. All results obtained in
Chapter \ref{IntRi}, can be generalised to apply to the quantum case and some
non-trivial integral models can be obtained. The aim of this chapter is to show
how to apply the theory of integrability of Lie systems so as to investigete
quantum Lie systems. All our results are illustrated by means of the analysis of
several types of spin Hamiltonians.

We must stress the practical importance of this method: It enables us to obtain
non-trivial exactly solvable $t$-dependent Schr\"odinger equations. This fact
allows us to investigate physical models by means of non-trivial exact
solutions. It also provides a procedure to avoid using numerical methods for
studying Schr\"odinger equations in many cases.

\section{Spin Hamiltonians}
In this section we investigate a particular quantum mechanical system whose
dynamics is given by Schr\"odinger--Pauli equation  \cite{CGM01}. We first prove
that 
this Hamiltonian corresponds to a quantum Lie system and we next
 apply the theory of integrability of Lie systems to such a system 
 to recover some exact known solutions and prove some new ones.

The system under study is described by the $t$-dependent Hamiltonian
$$
H(t)=B_x(t)S_x+B_y(t)S_y+B_z(t)S_z,
$$
with $S_x, S_y$ and $S_z$ being the spin operators. Let us denote $S_1=S_x$,
$S_2=S_y$ and $S_3=S_z$, then the $t$-dependent Hamiltonian $H(t)$ is a quantum
Lie system, because the spin operators are such that
\begin{equation}\label{ConmV}
[iS_j,iS_k]=-\sum_{l=1}^3\, \epsilon_{jkl}\,iS_l,\qquad j,k=1,2,3,
\end{equation}
with $\epsilon_{jkl}$ being the components of the fully skew-symmetric
 Levi-Civita tensor and where we have assumed $\hbar=1$. The Schr\"odinger
equation corresponding to this $t$-dependent Hamiltonian is
\begin{equation}\label{SE}
\frac{d\psi}{dt}=-iB_x(t)S_x(\psi)-iB_y(t)S_y(\psi)-iB_z(t)S_z(\psi),
\end{equation}
which can be seen as a differential equation determining the integral curves of
the $t$-dependent vector field in a (maybe infinite-dimensional) Hilbert space
$\mathcal{H}$ given by
$$
X_t=B_x(t)X^{SH}_1+B_y(t)X^{SH}_2+B_z(t)X^{SH}_3,
$$
with
$$
(X^{SH}_1)_\psi=-iS_x(\psi),\quad (X_2^{SH})_\psi=-iS_y(\psi),\quad
(X_3^{SH})_\psi=-iS_z(\psi).
$$
The $t$-dependent vector field $X$ can be written as a linear combination 
$$X_t={\displaystyle\sum_{k=1}^3} b_k(t)X^{SH}_k,
$$of the vector fields $X^{SH}_k$, with $b_1(t)=B_x(t)$, $b_2(t)=B_y(t)$ and
$b_3(t)=B_z(t),$
and therefore  our Schr\"odinger equation is a Lie system related to a quantum
Vessiot--Guldberg Lie algebra isomorphic to $\mathfrak{su}(2)$.

Take the basis for $\mathfrak{su}(2)$ given by the following
skew-self-adjoint $2\times 2$ matrices
$${\rm v}_1\equiv\frac 12
\left(\begin{array}{cc}
0 & i\\
i & 0\\
\end{array}\right),\qquad
{\rm v}_2\equiv\frac 12
\left(\begin{array}{cc}
0 & 1\\
-1 & 0\\
\end{array}\right),\qquad
{\rm v}_3\equiv\frac 12
\left(\begin{array}{cc}
i & 0\\
0 & -i\\
\end{array}\right).\qquad
$$
These matrices satisfy the commutation relations
$$
[{\rm v}_j,{\rm v}_k]=-\sum_{l=1}^3\epsilon_{jkl}{\rm v}_l,\qquad j,k=1,2,3,
$$
which are similar to (\ref{ConmV}). Hence, we can define an action
$\Phi^{SH}:SU(2)\times\mathcal{H}\rightarrow\mathcal{H}$ such that
$$
\Phi^{SH}(\exp(c_k{\rm v}_k),\psi)=\exp(c_kiH_k)(\psi), \qquad k=1,2,3,
$$
for any real constants $c_1, c_2$ and $c_3$. Moreover,
$$
\frac{d}{dt}\bigg|_{t=0}\Phi^{SH}(\exp(-it{\rm
v}_k,\psi)=\frac{d}{dt}\bigg|_{t=0}\exp(-itH_k)(\Phi)=-iH_k(\psi)=(X^{SH}
_k)_\psi,
$$
getting that each $X^{SH}_k$ is the fundamental vector field associated with
${\rm v}_k$. Thus, the equation on $SU(2)$ related, by means of $\Phi^{SH}$, to
the Schr\"odinger equation (\ref{SE}) is
\begin{equation}\label{EQG}
R_{g^{-1}*g}\dot g=-\sum_{k=1}^3b_k(t){\rm v}_k\equiv {\rm
a}(t)\in\mathfrak{su}(2),\qquad g(0)=e.
\end{equation}
It was shown in \cite{CLR08WN}, and previously in our work, that the group
$\mathcal{G}$ of curves in the group of a Lie system, in this case 
$\mathcal{G}={\rm Map}(\mathbb{R},SU(2))$, acts on the set of Lie systems
associated with an
 equation in the Lie group $G$ in such a way that, in a similar way to what
happened in \cite{CarRamGra}, a curve $\bar g\in\mathcal{G}$ transforms the
initial equation (\ref{EQG}) into the new one characterised by the curve
\begin{equation}\label{trans}
{\rm a}'(t)\equiv -\Ad(\bar g)\left(\sum_{k=1}^3{\it b_k(t)}{\rm v}_k\right)
+{\it R_{\bar g^{-1}*\bar g}}\frac{{\it d\bar g}}{{\it dt}}=-\sum_{k=1}^3{\it
b'_k(t)}{\rm v}_k,
\end{equation}
Once again, this new equation is related to a new Schr\"odinger equation in
$\mathcal{H}$ determined by a new Hamiltonian
$$
H'(t)=\sum_{k=1}^3b'_k(t)S_k\,.
$$

Additionally, the curve $\bar g(t)$ in $SU(2)$ induces a $t$-dependent unitary
transformation $\bar U(t)$ on $\mathcal{H}$ transforming the initial
$t$-dependent Hamiltonian $H(t)$ into $H'(t)$.

Summarising, the theory of Lie systems reduces the problem of
determining the solution of Schr\"odinger equations related to spin Hamiltonians
$H(t)$ to solving certain equations in the Lie group $SU(2)$. Then,
the transformation properties of the equations in $SU(2)$
describe the transformation properties of $H(t)$ by means of
certain $t$-dependent unitary transformations described by curves
in $SU(2)$.

Note that the theory here developed for spin Hamiltonians can be
straightforwardly employed to analyse any quantum Lie system. In this case, our
procedure remains essentially the same. It is only necessary to replace $SU(2)$
by the new Lie group $G$ associated with the quantum Lie system under study.

\section{Lie structure of an equation of transformation of Lie systems}

Our aim now is to prove that  the curves in $SU(2)$
relating the equations defined by two curves ${\rm a}(t)$ and ${\rm a}'(t)$ in
$T_ISU(2)$, respectively, can be found as solutions of 
 a Lie system of differential equations.

Recall that   the matrices of $SU(2)$ are of the form 
\begin{equation}\label{parametrizacion}
\bar g=\left(
\begin{array}{cc}
a & b\\
-b^* & a^*
\end{array}
\right),\qquad a,b\in \mathbb{C},
\end{equation}
with $|a|^2+|b|^2=1$ and that the elements of $\mathfrak{su}(2)$ are traceless
skew-Hermitian matrices, namely, real linear combinations of the matrices
$\{{\rm v}_k\mid k=1,2,3\}$. Then, the equation (\ref{trans}) becomes a matrix
equation that can be written
\begin{equation}\label{trasRule}
\frac{d\bar g}{dt}\bar g^{-1}=-\sum_{k=1}^3b'_k(t){\rm
v}_k+\sum_{k=1}^3b_k(t)\bar g{\rm v}_k\bar g^{-1}.
\end{equation}
Multiplying both sides of this equation by $\bar g$ on the right, we get
\begin{equation}\label{eqtran}
\frac{d\bar g}{dt}=-\sum_{k=1}^3b'_k(t){\rm v}_k\bar g+\sum_{k=1}^3b_k(t)\bar
g{\rm v}_k\,.
\end{equation}
If we consider a  reparametrisation of the $t$-dependent coefficients of $\bar
g$
\begin{equation*}
\begin{aligned}
a(t)&=x_1(t)+i\,y_1(t),\\
b(t)&=x_2(t)+i\,y_2(t),
\end{aligned}
\end{equation*}
for real functions $x_j$ and $y_j$, with $j=1,2$, a straightforward computations
shows  that (\ref{eqtran}) is a linear
system of differential equations in the new variables $x_1,x_2,y_1$ and $y_2$
that can be written as follows
\begin{equation}\label{FSys}
\matriz{c}{
\dot x_1\\
\dot x_2\\
\dot y_1\\
\dot y_2}
=\frac 12
\matriz{cccc}{
0&b'_2-b_2 &-b_3+b'_3&-b_1+b'_1\\
b_2-b'_2& 0 &-b_1-b'_1 &b_3+b'_3\\
b_3-b'_3& b'_1+b_1 & 0 & -b_2-b'_2\\
b_1-b'_1&-b_3-b'_3 &b_2+b'_2&0
}
\matriz{c}{
x_1\\
x_2\\
y_1\\
y_2}.
\end{equation}

Only the solutions of the above system obeying that $x_1^2+x_2^2+y_1^2+y_2^2=1$
describe curves in $SU(2)$ and, consequently, are related to solutions of system
(\ref{eqtran}). Nevertheless, we can forget such a restriction for the time
being, because it can
 be automatically implemented later in a more suitable way. Therefore, we
can deal with  the four variables
  in the preceding system of  
 differential equations (\ref{FSys}) as if they were independent. This linear
 system of  differential equations is a Lie system associated with a Lie algebra
of vector fields $\mathfrak{gl}(4,\mathbb{R})$, but the solutions with initial
condition related to a matrix in the subgroup $SU(2)$ always remain in such a
subgroup. In fact, consider the set of vector fields
\begin{eqnarray}\label{Vect}
N_1&=&\frac
12\left(-y_2\pd{}{x_1}-y_1\pd{}{x_2}+x_2\pd{}{y_1}+x_1\pd{}{y_2}\right),\cr
N_2&=&\frac
12\left(-x_2\pd{}{x_1}+x_1\pd{}{x_2}-y_2\pd{}{y_1}+y_1\pd{}{y_2}\right),\cr
N_3&=&\frac
12\left(-y_1\pd{}{x_1}+y_2\pd{}{x_2}+x_1\pd{}{y_1}-x_2\pd{}{y_2}\right),\cr
N'_1&=&\frac
12\left(y_2\pd{}{x_1}-y_1\pd{}{x_2}+x_2\pd{}{y_1}-x_1\pd{}{y_2}\right),\cr
N'_2&=&\frac
12\left(-x_2\pd{}{x_1}+x_1\pd{}{x_2}-y_2\pd{}{y_1}+y_1\pd{}{y_2}\right),\cr
N'_3&=&\frac
12\left(y_1\pd{}{x_1}+y_2\pd{}{x_2}-x_1\pd{}{y_1}-x_2\pd{}{y_2}\right),
\cr\nonumber
\end{eqnarray}
for which  the non-zero commutation relations are given by:
\begin{eqnarray}
\begin{aligned}
\left[N_1,N_2\right]=-N_3, \qquad &[N_2,N_3]=-N_1, \qquad &[N_3,N_1]=-N_2,\cr
[N'_1,N'_2]=-N'_3, \qquad &[N'_2, N'_3]=-N'_1,\qquad &[N'_3, 
N'_1]=-N'_2\,.\nonumber
\end{aligned}
\end{eqnarray}

Note that $[N_j,N'_k]=0$, for $j,k=1,2,3$, and therefore 
the system of linear differential equations (\ref{FSys}) is a Lie system on
$\mathbb{R}^4$ associated with a Lie algebra of vector
fields isomorphic to $\LG\equiv\mathfrak{su}(2)\oplus\mathfrak{su}(2)$, i.e. the
Lie
algebra decomposes into a direct sum of two Lie algebras isomorphic to
$\mathfrak{su}(2,\mathbb{R})$, the first one is generated by $\{N_1,N_2,N_3\}$
and the second
one by $\{N'_1,N'_2,N'_3\}$. 

If we  denote  $y\equiv\left(x_1,x_2,y_1,y_2\right)\in \mathbb{R}^4$,  the
system (\ref{FSys}) can be written as
 a system of differential equation in $\mathbb{R}^4$:
\begin{equation}\label{TRR4}
\frac{dy}{dt}=N(t,y),
\end{equation}
with $N_t$ being the $t$-dependent vector field given by
\begin{equation*}
N(t,y)=\sum_{k=1}^3 \left(b_k(t)N_k(y)+b'_k(t)N'_k(y)\right).
\end{equation*}

The vector fields $\{N_1,N_2,N_3,N'_1,N'_2,N'_3\}$ span a distribution of rank
three in almost any point of $\mathbb{R}^4$ and consequently  there exists, at
least locally, a first-integral for all
 the vector fields (\ref{Vect}). It can be verified that such a first-integral
is globally defined and reads $I(y)=x_1^2+x_2^2+y_1^2+y_2^2$. Hence, given a
solution $y(t)$ of system (\ref{TRR4}) with an initial condition
$I(y(0))=x_1^2+x_2^2+y_1^2+y_2^2=1$, 
then $ I(y(t))=1$  at any time $t$ and this solution describes a curve in
$SU(2)$. Therefore, we have found that 
the curves in $SU(2)$ relating two different equations on $SU(2)$ associated
with two Schr\"odinger equations of the form (\ref{SE}) can be described by
means of the solutions $y(t)$ of (\ref{TRR4}) with $I(y(0))=1$, and vice versa:

\begin{theorem}\label{ThMain} The curves in $SU(2)$  relating two equations on
the group
  $SU(2)$ characterised by the curves in $\mathfrak{su}(2)$ of the form
$$
{\rm a}'(t)=-\sum_{k=1}^3b'_k(t){\rm v}_k\,,
\qquad {\rm a}(t)=-\sum_{k=1}^3b_k(t){\rm v}_k$$
 are the solutions, $y(t)$, of the system
\begin{equation*}
\frac{dy}{dt}=N(t,y),
\end{equation*}
with
\begin{equation*}
N(t,y)=\sum_{k=1}^3\left(b_k(t)N_k(y)+b'_k(t)N'_k(y)\right),
\end{equation*}
satisfying that $I(y(0))=1$. This  is a Lie system related to a Lie algebra of
 vector fields isomorphic to  
$\mathfrak{su}(2)\oplus\mathfrak{su}(2)$.
\end{theorem}
\begin{corollary}  Given two Schr\"odinger equations corresponding to two spin
Hamiltonians, there always exists a curve in $SU(2)$ transforming one of them
into the other.
\end{corollary}

Although the above corollary ensures the existence of a $t$-dependent unitary
transformation mapping a given Spin Hamiltonian into any other one, obtaining
such a transformation involves solving system (\ref{TRR4}) explicitly. This Lie
system is related to a non-solvable Lie algebra and, consequently, it is
not easy to find its solutions in general. In view of this, it becomes
interesting to determine integrability conditions which allow us to solve this
system and obtain the corresponding transformation. This illustrates the
interest of the integrability conditions derived in next sections, which will be
used to derive exact solutions for some physical problems involving Spin
Hamiltonians. 

\section{Integrability conditions for {\it SU(2)} Schr\"odinger
equations}\label{GIC}

Let $\bar g(t)$ be a curve in $SU(2)$ transforming the equation on $SU(2)$
defined by the  curve ${\rm a}(t)$ into
 another characterised by ${\rm a}'(t)$ according to  the rule
 (\ref{trasRule}). If $g'(t)$  is the solution of the equation in $SU(2)$
characterised by ${\rm a}'(t)$, then
$g(t)=\bar g^{-1}(t)g'(t)$ is a solution for the equation in $SU(2)$
characterised by ${\rm a}(t)$. 

If ${\rm a}'(t)$ lies in a solvable Lie subalgebra of $\mathfrak{su}(2)$,
we can derive $g'(t)$ in many ways \cite{CarRamGra} and, once $g'(t)$ is
obtained, the  knowledge of the  curve $\bar g(t)$ transforming  the
 curve ${\rm a}(t)$ into ${\rm a}'(t)$ provides the curve $g(t)$ solution of the
equation on $SU(2)$ determined by ${\rm a}(t)$. 

Therefore, starting from a curve 
 ${\rm a}'(t)$ in a solvable Lie subalgebra of $\mathfrak{su}(2)$ and using
 (\ref{TRR4}), with curves in a restricted family of curves in $SU(2)$, we
can relate ${\rm a}'(t)$ to other possible curves ${\rm a}(t)$, finding, in this
way 
a family of equations on $SU(2)$, and thus spin Schr\"odinger equations on
$\mathcal{H}$, that can be exactly solved. 

Let us assume some restrictions on the family of solution curves of the system
(\ref{TRR4}), e.g. we choose
$b=0$. Consequently, there are instances of this system which
do not admit a solution under these restrictions, i.e. it is not 
possible to connect the  curves ${\rm a}(t)$ and ${\rm a}'(t)$ 
by a curve satisfying the assumed restrictions. This gives rise to some
compatibility conditions for the existence of one of these special solutions,
either
 algebraic and/or
 differential 
ones, between the $t$-dependent  coefficients of ${\rm a}'(t)$ and ${\rm
  a}(t)$ satisfied by explicitly solvable models 
found in the literature. Therefore, our approach is useful to provide exactly
integrable models found in the literature and, as we will see next, to derive
new ones.

The two main ingredients to be taken into account in the following sections are:

\begin{enumerate}
 \item {\it The equations which are characterised by a curve ${\rm a}'(t)$ for
which the solution can be obtained}. We here consider that ${\rm a}'(t)$ is
associated with a one-dimensional Lie subalgebra of
$\mathfrak{su}(2)$.

\item {\it The restriction on the set of curves considered as solutions of the
equation (\ref{TRR4})}. We next look for solutions of (\ref{TRR4}) related to
curves in a one-parameter subset of $SU(2)$. 
\end{enumerate}

Consider the below example: suppose that we want to connect a given ${\rm a}(t)$
with a final family of curves of the form ${\rm a}'(t)=-D(t)(c_1{\rm
  v}_1+c_2{\rm v}_2+c_3{\rm v}_3)$, with $c_1, c_2, c_3,$ being real numbers. In
this case, 
system (\ref{TRR4}), which describes the curves $\bar g(t)\subset SU(2)$ that
transform the equation described by ${\rm a}(t)$ into the equation determined by
${\rm a}'(t)$, reads
\begin{equation}\label{QMLie2}
\frac{dy}{dt}=\sum_{k=1}^3b_k(t)N_{k}(y)+D(t)
\sum_{k=1}^3c_k N'_k(y)=N(t,y).
\end{equation}
Note that the vector field
\begin{equation*}
N'=\sum_{k=1}^3c_k N'_k,
\end{equation*}
satisfies that
\begin{equation*}
\left[N_k,N'\right]=0,\quad\quad  k=1,2,3.
\end{equation*}
Hence, Lie system (\ref{QMLie2}) is related to a Lie algebra of vector fields
isomorphic 
to $\mathfrak{su}(2)\oplus \mathbb{R}$. As this Lie system is associated with a
non-solvable Vessiot-Guldberg Lie 
algebra, it is not integrable by quadratures and the solution 
cannot be easily found in the general case. Nevertheless, it is worth noting
that (\ref{QMLie2}) always has a solution.

In this way, we can consider some instances of (\ref{QMLie2}) for which the 
resulting system of differential equations can be integrated by quadratures. We
can consider that $x$ is related to a one-parameter family of elements of
$SU(2)$. Such a restriction implies that (\ref{QMLie2}) not
always has a solution, because sometimes it is not possible to connect ${\rm
a}(t)$ and ${\rm a}'(t)$ by means of the chosen
one-parameter family. This fact  imposes differential and algebraic restrictions
on the initial $t$-dependent functions $b_k$, with $k=1,2,3$. These restrictions
will describe known integrability conditions and other new ones. So,
we can develop the ideas of \cite{CLR07b,CRL07e} in the framework of Quantum
Mechanics. 
Moreover, from this point of view, we can find new integrability conditions
that 
 can be used to obtain exact solutions. 

\section{Application of integrability conditions in a {\it SU(2)} Schr\"odinger
equation}

In this section we restrict ourselves to the case ${\rm a}'(t)=-D(t){\rm v}_3$,
i.e. 
\begin{equation}
b'_1(t)=0,\qquad 
b'_2(t)=0,\qquad
b'_3(t)=D(t).
\end{equation}
Hence, the system of differential equations (\ref{FSys}) describing the
 curves $\bar g$ relating a Schr\"odinger equation to 
$H'(t)=D(t)S_z$
is
\begin{equation}\label{FSQM}
\matriz{c}{
\dot x_1\\
\dot x_2\\
\dot y_1\\
\dot y_2}
=\frac 12
\matriz{cccc}{
0&-b_2 &-b_3+D&-b_1\\
b_2& 0 &-b_1&b_3+D\\
b_3-D& b_1 & 0 & -b_2\\
b_1&-b_3-D &b_2&0
}
\matriz{c}{
x_1\\
x_2\\
y_1\\
y_2}.
\end{equation}
We see that, according to the result of  Theorem \ref{ThMain}, the $t$-dependent
vector 
field corresponding to such a  system of  differential equations can be written 
 as a linear combination with $t$-dependent coefficients of the vector fields
$N_1, N_2, N_3$ and $N_3'$:
$$
N(t,y)=\sum_{k=1}^3b_k(t)N_k(y)+D(t)N'_3(y).
$$
Thus, system (\ref{FSQM}) is associated with a Lie algebra of vector fields
 isomorphic to $\mathfrak{u}(1)\oplus
\mathfrak{su}(2)$. This Lie algebra is smaller than the initial one 
(\ref{FSys}),
but it is not solvable and the system is as difficult to solve as the
initial Schr\"odinger equation. Therefore, in order to get exact solvable cases,
we need to perform some kind of simplification once again, e.g. by means of the
imposition of some extra assumptions on the variables. This may result in a
system of differential equations whose solutions are incompatible with our
additional conditions. The necessary and sufficient conditions on the
$t$-dependent functions $b_1,b_2,b_3,b'_1,b'_2$ and $b'_3$ ensuring the
existence of a solution compatible with the assumed restrictions on the
variables give rise to integrability conditions for spin Hamiltonians. 

For instance, suppose that we impose on the solutions to be in the
one-parametric subset $A_\gamma\subset SU(2)$ given by
\begin{equation}\label{family}
A_\gamma=\left\{
\left(\begin{array}{cc}
         \cos\frac{\gamma}{2} & -e^{-bi}\sin\frac{\gamma}{2}\\
e^{bi}\sin\frac{\gamma}{2} &\cos\frac{\gamma}{2}
        \end{array}
 \right)\,\bigg |\,b\in [0,2\pi)
\right\}.
\end{equation}
where $\gamma$ is a fixed real constant such that $\gamma\neq 2\pi n$, with
$n\in \mathbb{Z}$, because in such a case $A_\gamma=\pm {\rm Id}$. In view of
the definition of the sets $A_\gamma$ and in terms of the parametrisation
(\ref{parametrizacion}), we have
\begin{equation}
x_1=\cos\frac{\gamma}{2},\quad
y_1=0,\quad
x_2=-\sin\frac \gamma 2\cos b,\quad
y_2=\sin\frac \gamma 2\sin b.
\end{equation}
The elements of $A_\gamma$ are matrices in $SU(2)$ and  the system of
differential equations we obtain reads
\begin{equation}\label{sysInt}
\matriz{c}{
0\\
\dot x_2\\
0\\
\dot y_2}
=\frac 12
\matriz{cccc}{
0&-b_2 &-b_3+D&-b_1\\
b_2& 0 &-b_1&b_3+D\\
b_3-D& b_1 & 0 & -b_2\\
b_1&-b_3-D &b_2&0
}
\matriz{c}{
x_1\\
x_2\\
0\\
y_2}.
\end{equation}
and then we get two integrability conditions
for the system (\ref{sysInt}): 
\begin{equation}\label{SpinIntCon}
\begin{aligned}
0&=-b_2x_2-b_1y_2,\\
0&=(b_3-D)x_1+b_1x_2-b_2y_2.
\end{aligned}
\end{equation}
We can write the components $(B_x(t),B_y(t),B_z(t))$ of the magnetic field in
polar coordinates,
$$
\begin{aligned}
B_x(t)&=B(t)\sin\theta(t)\cos\phi(t),\\
B_y(t)&=B(t)\sin\theta(t)\sin\phi(t),\\
B_z(t)&=B(t)\cos\theta(t),
\end{aligned}
$$
with $\theta\in [0,\pi)$ and $\phi\in[0,2\pi)$.

The first algebraic integrability condition reads, in polar coordinates, as
follows:
$$
B(t)\sin\theta(t)\sin\frac \gamma 2\left(\cos\phi(t)\sin b(t)-\sin\phi(t)\cos
b(t)\right)=0
$$
and thus,
$$
B(t)\sin\theta(t)\sin\frac\gamma 2\sin(b(t)-\phi(t))=0,
$$
from where we see  that $b(t)=\phi(t)$. In such a case, the second algebraic
integrability condition in (\ref{SpinIntCon}) reduces to
\begin{equation*}
(B_z-D)\cos\frac \gamma 2-B\sin\frac \gamma 2\sin\theta=0
\end{equation*}
and then, the $t$-dependent coefficient $D$ is 
\begin{equation}\label{DFactor}
D=\frac{B}{\cos\frac \gamma 2}\cos\left(\frac \gamma 2+\theta\right).
\end{equation}
Finally, we have to take into account the differential integrability condition
$$
\dot x_2=\frac 12\left(b_2\cos\frac \gamma 2+(b_3+D)\sin\frac \gamma 2\sin
b\right),
$$
which after some algebraic manipulation leads to
$$
\dot \phi=\frac{B}2\left(\frac{\sin(\theta+\frac \gamma 2)}{\sin\frac \gamma
2}+\frac{\cos(\frac \gamma 2+\theta)}{\cos\frac \gamma 2}\right),
$$
and then
\begin{equation}\label{CI2}
\dot \phi(t)=B(t)\,\frac{\sin(\theta(t)+\gamma)}{\sin\gamma},
\end{equation}
which is a far larger set of integrable Hamiltonians than the one of the
exactly solvable Hamiltonians of this type
 found  in the literature. As a particular example, when $\theta$ and $B$ are
 constant, we find 
\begin{equation}\label{CI}
\dot \phi=B\,\frac{\sin(\theta+\gamma)}{\sin\gamma}\equiv \omega
\end{equation}
and consequently,
$$
\phi=\omega t+\phi_0.
$$
In this way, we get that the $t$-dependent spin Hamiltonian $H(t)$ determined by
the magnetic vector field
\[
{\bf B}(t)=B(\sin\theta\cos(\omega t),\sin\theta\sin(\omega t),\cos\theta).
\]
is integrable.

Another interesting integrable case is that given by
$\theta=\frac{\pi}2$, that is, the magnetic field moves in the $XZ$
plane, see \cite{Bl08,KN09PI,KN09PII}. In such a case, in view of the
integrability conditions (\ref{CI}), the angular frequency
$\dot\phi$ reads
$$
\dot\phi=B\,{\rm cotan}\gamma=\omega.
$$

The last one of the most known integrable cases of Spin Hamiltonian is given by
a magnetic field in a fixed direction, i.e. ${\bf
B}(t)=B(t)(\sin\theta\cos\phi,\sin\theta\sin\phi,\cos\theta)$. Obviously, this
case satisfies integrability condition (\ref{CI}) for
$\gamma=-\theta$.

Apart from the previous cases, the integrability condition (\ref{CI2}) describes
more, as far as we know, new integrable cases.
For instance, consider the case with $\theta$ fixed and $B$ non-constant. In
this case, the corresponding $H(t)$ is integrable if
$$
\frac{\dot\phi}{B(t)}=\frac{\sin(\theta+\gamma)}{\sin\gamma},
$$
that is, if we fix $\gamma=\pi/2$ we have that
$$
\omega=\dot \phi=B(t)\cos\theta\Longrightarrow\phi(t)=\cos\theta\int^tB(t')dt'.
$$

Furthermore, we can consider $\theta(t)=t$ and $B$ constant. In this case, we
get that the $t$-dependent Hamiltonian $H(t)$ is integrable if the $\phi(t)$
holds the condition
$$
\dot\phi=B\cos\,t \Rightarrow \phi(t)=B\sin\,t.
$$
Indeed, note that in this case the integrability condition (\ref{CI2}) trivially
follows for $\gamma=-1/2$.

To sum up, we have shown that there exists a large family of $t$-dependent
integrable spin Hamiltonians that includes, as particular cases, many integrable
cases known up to now. Additionally, it is easy to check whether a $t$-dependent
spin Hamiltonian satisfies the integrability condition (\ref{IntCond}) and then,
it can be integrated.

\section{Applications to Physics}

Let us use the above results in order to solve a $t$-dependent spin Hamiltonian 
\[
H(t)={\bf B}(t)\cdot{\bf S},
\]
which broadly appears in Physics: the one characterised by a magnetic field
\begin{equation}\label{MF}
{\bf B}(t)=B(\sin\theta\cos(\omega t),\sin\theta\sin(\omega t),\cos\theta),
\end{equation}
that is, a magnetic field with a constant modulus rotating along the $OZ$  axis
with
a constant angular velocity $\omega$. Such Hamiltonians have been applied, for
instance, to analyse the spin precession in a transverse  $t$-dependent magnetic
field \cite{Sc37}, investigate the adiabatic approximation and the unitary of
the $t$-evolution operator through such an approximation \cite{MS04,PR08}, etc. 

In the previous section we showed that this $t$-dependent Hamiltonian is
integrable. Indeed, the integrability condition (\ref{CI}) can be written as
\begin{equation}\label{equation}
\tan \gamma=\frac{\sin\theta}{\frac{\dot\phi}{B}-\cos\theta},
\end{equation}
where we recall that $\gamma$ has to be a real constant. In the case of our
particular magnetic field (\ref{MF}) the angular frequency, $\omega=\dot\phi$,
the angle $\theta$ and  the modulus $B$ are constants. Therefore $\gamma$ is a
properly defined constant, the integrability condition (\ref{CI}) holds and the
value of $\gamma$ is given by equation (\ref{equation}) in terms of the
parameters $B$, $\theta$ and $\omega$, which characterise the magnetic vector
field (\ref{MF}). 

We have already shown that if $B(t)$ satisfies (\ref{CI}), then $H(t)$ is
integrable, because it can be transformed by means of a $t$-dependent change of
variables determined by a curve $g(t)$ in the set $A_\gamma$ into a
straightforwardly integrable Schr\"odinger equation determined by a
$t$-dependent Hamiltonian $H'(t)=D(t)S_z$.  For simplicity, let us parametrise
the elements of $A_\gamma$ in a new way. Consider ${\bf
\sigma}=(\sigma_1,\sigma_2,\sigma_3)$ and ${\bf n}\in\mathbb{R}^3$, where the
matrices $\sigma_i$ are the Pauli matrices, $\sigma_x,\sigma_y,\sigma_z$. We
have
$$
e^{i{\bf \sigma}\cdot {\bf n}\phi}={\rm Id}\cos\phi+i{\bf \sigma}\cdot{\bf
n}\sin\phi.
$$ 
So, for ${\bf n}=\frac{(\alpha_1,\alpha_2,0)}{\sqrt{\alpha_1^2+\alpha_2^2}}$
with real constants $\alpha_1$, $\alpha_2$ and taking into account that ${\rm
v}_1=\frac{i\sigma_x}2$, ${\rm v}_2=\frac{i\sigma_y}2$ and ${\rm
v}_3=\frac{i\sigma_z}2$, we get 
\begin{equation}\label{para}
\exp(\alpha_1{\rm v}_1+\alpha_2{\rm v}_2)=\exp\left(i\frac{\delta}{2}{\bf
\sigma}\cdot{\bf n}\right)=
\left(\begin{array}{cc}
{\rm cos}\frac{\delta}{2}&-e^{-i\varphi}\sin\frac{\delta}{2}\\
e^{i\varphi}\sin\frac{\delta}{2}&\cos \frac{\delta}{2}
\end{array}\right)
\end{equation}
with $\delta=\sqrt{\alpha_1^2+\alpha_2^2}$ and
$-e^{-i\varphi}=(i\alpha_1+\alpha_2)/\sqrt{\alpha_1^2+\alpha_2^2}$. In terms of
$\delta$ and $\varphi$ the variables $\alpha_1$ and $\alpha_2$ can be written
$\alpha_1=\delta\sin\varphi$ and $\alpha_2=-\delta\cos\varphi$. Hence, in view
of (\ref{para}), we see that we can describe the elements of $A_\gamma$ as 
\begin{equation}
\left(\begin{array}{cc}
         \cos\frac{\gamma}{2} & -e^{-bi}\sin\frac{\gamma}{2}\\
e^{bi}\sin\frac{\gamma}{2} &\cos\frac{\gamma}{2}
        \end{array}
 \right)=\exp(\gamma \sin b\,{\rm v}_1-\gamma \cos b\,{\rm v}_2),
\end{equation}
where $b$ and $\gamma$ are real constants. For magnetic vector fields
(\ref{MF}), the $t$-dependent change of variables transforming the initial
$H(t)$ into an integrable $H'(t)=D(t)S_z$ is determined by a curve in $A_\gamma$
with $\gamma$ determined by equation (\ref{CI}) and $b(t)=\phi(t)$. Thus, such a
curve in $A_\gamma$ takes the form
\begin{equation}\label{tra}
t\mapsto \exp(\gamma \sin (\omega t)\,{\rm v}_1-\gamma \cos (\omega t)\,{\rm
v}_2)
\end{equation}
We want to emphasise that the above $t$-dependent change of variables in $SU(2)$
transforms the equation in $SU(2)$ determined by the initial curve
$${\rm a}(t)=-B_x(t){\rm v}_1-B_y(t){\rm v}_2-B_z(t){\rm v}_3,
$$
into and a new equation in $SU(2)$ determined by a curve ${\rm
a}'(t)=-D(t){\rm v}_3$. Such a $t$-dependent transformation in
$SU(2)$ induces a $t$-dependent unitary change of variables in
$\mathcal{H}$ transforming the initial Schr\"odinger equation
determined by the $t$-dependent Hamiltonian $H(t)$, i.e.
$$
\frac{\partial\psi}{\partial t}=-iH(t)(\psi),
$$
into the new Schr\"odinger equation
\begin{equation}\label{finalSC}
\frac{\partial\psi'}{\partial t}=-iH'(t)(\psi')=-iD(t)S_z(\psi').
\end{equation}
The relation between $\psi$ and $\psi'$ is given by the corresponding
$t$-dependent change of variables in $\mathcal{H}$ induced by curve (\ref{tra}),
i.e.
\begin{equation}\label{QMChange}
\psi'=\exp(\gamma \sin (\omega t)\,iS_x-\gamma \cos (\omega t)\,iS_y)\psi.
\end{equation}

In view of expression (\ref{DFactor}), we see that
$$D=B(\cos\theta-{\tan}\frac\gamma 2\sin\theta),$$
and from (\ref{equation}) and the relations
$${\tan}\gamma=\frac{2{\tan}\frac \gamma 2}{1-{\tan}^2\frac\gamma 2}\Rightarrow
{\tan}\frac\gamma 2=\frac {-1\pm\sqrt{1+{\tan}^2\gamma}}{{\tan}\gamma},
$$
we obtain
$$
{\tan}\frac\gamma
2=\frac{1}{\sin\theta}\left(-\frac{\omega}{B}+\cos\theta\pm\sqrt{\frac{\omega^2}
{B^2}-2\frac\omega
B\cos\theta +1}\right).
$$
If we substitute the above expression in the latter expression for $D$, it turns
out that $$
D=\omega\pm\sqrt{\omega^2-2\omega B\cos\theta+B^2}.
$$
That is, $D$ becomes a constant. Thus, the general solution $\psi'_t$ for the
Schr\"odinger
equation (\ref{finalSC}) with initial condition $\psi'_0$ is
$$
\psi'(t)=\exp\left(-itDS_z \right)\psi'_0,
$$
and the solution for the initial Schr\"{o}dinger equation with initial
condition $\psi_0$ can
 be obtained undoing the $t$-dependent change of variables (\ref{QMChange}) to
get
$$
\psi_t=\exp\left(-i\gamma\sin \omega t \,S_x+i\gamma\cos \omega t\,
S_y\right)\exp\left(-iD tS_z \right)\psi_0.
$$

\chapter{The theory of quasi-Lie schemes and Lie families}

\section{Introduction}
Several important systems of first-order ordinary differential equations can be
studied through the theory of Lie systems. Moreover, this theory was recently
applied to study SODE Lie systems, quantum Lie systems, some partial
differential equations, etc. These last successes allow us to recover, from a
unifying point of view, several results disseminated  throughout the literature
and to prove multiple new properties of systems of differential equations
appearing in Physics and Mathematics. Apart from these successes, there are
still some reasons to go further in the generalisation of the theory of Lie
systems:
\begin{itemize}
 \item {\it Lie systems are important but rather exceptional}. The theory of Lie
systems investigates very interesting equations with many applications, e.g.
$t$-dependent frequency harmonic oscillators, Milne--Pinney equations, Riccati
equations, etc. Nevertheless, it fails to study many other (nonautonomous)
interesting systems, like nonlinear oscillators, Abel equations, or Emden
equations. 
\item {\it The theory of Lie systems does not allow us to investigate
superposition rules involving an explicit $t$-dependence} which appears in
various interesting systems, e.g. dissipative Milne--Pinney equation,
Emden--Fowler equations \cite{CLL09Emd}, second-order Riccati equations
\cite{CL09SRicc,In86}, whose properties are worth analysing. 

\item Lie systems have an associated group of $t$-dependent changes of variables
enabling us to transform each particular Lie system into a new one of the same
class, e.g. the group of curves in $SL(2,\mathbb{R})$ transforms a Riccati
equation into a new Riccati equation. A similar property frequently applies to
integrate differential equations, like Abel equations \cite{CR03}. A natural
question arises: Is there any kind of systems of differential equations more
general than Lie systems admitting an analogue property?
\end{itemize}

The theory of quasi-Lie schemes \cite{CGL08} and the Generalised Lie Theorem
\cite{CGL09}, which gives rise to the {\it Lie family} notion, provide an answer
to these problems. More specifically, quasi-Lie schemes, quasi-Lie systems and
Lie families are interesting because:

\begin{itemize}
\item {\it The theory of quasi-Lie schemes and the Generalised Lie Theorem
permit us to investigate a very large family of differential equations including
Lie systems}. More specifically, this family includes, for instance, the
following non-Lie systems: Emden--Fowler equations \cite{CGL08, CLL09Emd},
nonlinear oscillators \cite{CGL08}, dissipative Milne--Pinney equations
\cite{CGL08,CL08Diss}, second-order Riccati equations \cite{CL09SRicc}, Abel
equations \cite{CGL09}, etc. Moreover, not only   quasi-Lie schemes and Lie
families can be applied to investigate systems of first-order ordinary
differential equations, but they can also be employed, for instance, to
investigate second-order differential equations \cite{CLL09Emd,CL08Diss}.

\item {\it The theory of quasi-Lie schemes and the Generalised Lie Theorem
treat, in a natural way, systems admitting a $t$-dependent superposition rule}.
These theories show that many differential equations admit a $t$-dependent
superposition rule, e.g. Abel equations \cite{CGL09}, dissipative Milne--Pinney
equations \cite{CGL08}, Emden-Fowler equations \cite{CLL09Emd}, second-order
Riccati equations \cite{CL09SRicc}, etc.

\item {\it The quasi-Lie scheme concept permits us to transform a differential
equation within a fixed family, e.g. a first-order Abel equation into a new one
with different $t$-dependent coefficients}. This feature generalises the
transformation properties of Lie systems and enables us to derive integrability
conditions for differential equations from a unified point of view.

\end{itemize}

Consequently, the theory of quasi-Lie schemes and the Generalised Lie Theorem
represent powerful methods to study first- and higher-order differential
equations.

\section{Generalised flows and {\it t}-dependent vector fields}
Recall that a nonautonomous system of first-order ordinary differential
equations on $\mathbb{R}^n$ is represented in modern differential geometric
terms by a $t$-dependent vector field $X=X(t,x)$ on such a space. On a
non-compact manifold, the vector field
$X_t(x)=X(t,x)$, for a fixed $t$, is generally not defined globally, but it is
well defined on a neighbourhood
of every point $x_0\in \mathbb{R}^n$ for sufficiently small $t$. It is
convenient to add the variable $t$ to the manifold and
to consider the {\it autonomisation} of our system, i.e. the vector field
$$\overline{X}(t,x)=\pd{}{t}+X(t,x)\,,
$$
defined on a neighbourhood $U^X$ of $\{0\}\times \mathbb{R}^n$ in
$\mathbb{R}\times \mathbb{R}^n$. The vector field $X_t$ is then
defined on the open set of $\mathbb{R}^n$,
$$ U_t^X=\{x_0\in \mathbb{R}^n\mid (t,x_0)\in U^X\}\,,$$
 for all $t\in \mathbb{R}$. If $U_t^X=\mathbb{R}^n$ for all $t\in
\mathbb{R}$, we speak about a {\it global $t$-dependent vector field}. The 
system of differential equations
associated with the $t$-dependent vector field $X(t,x)$ is written in local
coordinates
$$\frac{dx^i}{dt}=X^i(t,x)\,,\qquad i=1,\ldots,n,
$$
where $X(t,x)=\sum_{i=1}^nX^i(t,x)\partial/\partial{x^i}$ is locally defined on
the manifold for sufficiently
small $t$.

A solution of this system is represented by a curve $s\mapsto \gamma(s)$ in
$\mathbb{R}^n$ (integral curve) whose tangent
vector $\dot \gamma$ at $t$, so at the point $\gamma(t)$ of the manifold, equals
$ X(t, \gamma(t))$. In other
words,
\begin{equation}\label{equati1}
\dot \gamma(t)= X(t, \gamma(t)).
\end{equation}
It is well-known that, at least for smooth $X$ we work with, for each $x_0$
there is a unique maximal solution $\gamma_X^{x_0}(t)$  of system
(\ref{equati1}) with the initial value $x_0$, i.e. satisfying
$\gamma_X^{x_0}(0)=x_0$. This solution is defined at least for $t$'s from a
neighbourhood of $0$. In case
$\gamma_X^{x_0}(t)$ is defined for all $t\in \mathbb{R}$, we speak about a {\it
global $t$-solution}.

The collection of all maximal solutions of the system (\ref{equati1}) gives rise
to a (local) generalised flow
$g^X$ on $\mathbb{R}^n$. By a {\it generalised flow} $g$ on $\mathbb{R}^n$ we
understand a smooth $t$-dependent family $g_t$ of
 local diffeomorphisms on $\mathbb{R}^n$, $g_t(x)=g(t,x)$, such that
$g_0=\text{id}_{\mathbb{R}^n}$. More precisely, $g$ is a smooth map from a
neighbourhood $U^g$ of $\{0\}\times \mathbb{R}^n$ in
$\mathbb{R}\times \mathbb{R}^n$ into $\mathbb{R}^n$, such that $g_t$ maps
diffeomorphically the open submanifold $U^g_t=\{x_0\in
\mathbb{R}^n\mid (t,x_0)\in U^g\}$ onto its image, and
$g_0=\text{id}_{\mathbb{R}^n}$. Again, for each $x_0\in \mathbb{R}^n$ there is a
neighbourhood $U_{x_0}$ of $x_0$ in $\mathbb{R}^n$ and $\epsilon>0$ such that
$g_t$ is defined on $U_{x_0}$ for
$t\in(-\epsilon,\epsilon)$ and maps $U_{x_0}$ diffeomorphically onto
$g_t(U_{x_0})$.

If $U_t^g=\mathbb{R}^n$ for all $t\in \mathbb{R}$, we speak about a {\it global
generalised flow}. In this case $g:t\in
\R\mapsto g_t\in{\rm Diff}(\mathbb{R}^n)$ may be viewed as a smooth curve in the
diffeomorphism group ${\rm Diff}(\mathbb{R}^n)$ with
$g_0=\text{id}_{\mathbb{R}^n}$.

Here it is also convenient to {\it autonomise} the generalised flow $g$
extending it to a single local
diffeomorphism
\begin{equation}\label{equ2} \overline{g}(t,x)=(t,g(t,x))
\end{equation}
defined on the neighbourhood $U^g$ of $\{0\}\times \mathbb{R}^n$ in
$\mathbb{R}\times \mathbb{R}^n$. The generalised flow $g^X$
induced by the $t$-dependent vector field $X$ is defined by
\begin{equation}\label{equ3}
g^X(t,x_0)=\gamma_X^{x_0}(t)\,.
\end{equation}
Note that, for $g=g^X$, equation (\ref{equ3}) can be rewritten in the form:
\begin{equation}\label{e4}
X_t= X(t,x)=\dot g_t\circ g_t^{-1}\,.
\end{equation}
In the above formula, we understood $ X_t$ and $\dot g_t$ as maps from
$\mathbb{R}^n$ into ${\rm T}\mathbb{R}^n$, where $\dot g_t(x)$ is
the vector tangent to the curve $s\mapsto g(s,x)$ at $g(t,x)$. Of course, the
composition $\dot g_t\circ
g_t^{-1}$, called sometimes the {\it right-logarithmic derivative} of $t\mapsto
g_t$, is only defined for
those points $x_0\in \mathbb{R}^n$ for which it makes sense. But this is always
the case for sufficiently small
$t$, at least locally.

Let us observe that equation (\ref{e4}) defines, in fact, a one-to-one
correspondence between generalised
flows and $t$-dependent vector fields modulo the observation that the domains of
$\dot g_t\circ g_t^{-1}$ and
$ X_t$ need not to coincide. In any case, however, $\dot g_t\circ g_t^{-1}$ and
$ X_t$ coincide in a
neighbourhood of any point for sufficiently small $t$. One can simply say that
the {\it germs} of $X$ and
$\dot g_t\circ g_t^{-1}$ coincide, where the germ in our context is understood
as the class of corresponding
objects that coincide on a neighbourhood of $\{0\}\times \mathbb{R}^n$ in
$\mathbb{R}\times \mathbb{R}^n$.

Indeed, for a given $g$, the corresponding $t$-dependent vector field is defined
by (\ref{e4}). Conversely,
for a given $X$, the equation (\ref{e4}) determines the germ of the generalised
flow $g(t,x)$ uniquely, as for
each $x=x_0$ and for small $t$ equation (\ref{e4}) implies that $t\mapsto
g(t,x_0)$ is the solution of the
system defined by $X$ with the initial value $x_0$. In this way we get the
following.
\begin{theorem}\label{t1} Equation (\ref{e4}) defines a
one-to-one correspondence between the germs of generalised flows and the germs
of $t$-dependent vector fields
on $\mathbb{R}^n$. 
\end{theorem}
\noindent Any two generalised flows $g$ and $h$ can be composed: by definition
$(g\circ h)_t=g_t\circ h_t$,
where, as usual, we view $g_t\circ h_t$ as a local diffeomorphism defined for
points for which the composition
is properly defined. It is important to emphasise that in a neighbourhood of any
point it really makes sense for sufficiently small
$t$. As generalised flows correspond to $t$-dependent vector fields, this gives
rise to an action  of a
generalised flow $h$ on a $t$-dependent vector field $X$, giving rise to $h_\di
X$, defined  by the equation
\begin{equation}
g^{h_\di X}=h\circ g^X\,. \label{e5}
\end{equation}
To obtain a more explicit form of this action, let us observe that
$$
(h_\di X)_t=\frac{d(h\circ g^X)_t}{dt}\circ (h\circ g^X)_t^{-1}=\left( \dot
h_t\circ g_t^X+Dh_t(\dot g_t^X)
\right)\circ (g^X)_t^{-1}\circ h_t^{-1},
$$
and therefore
$$(h_\di X)_t=\dot
h_t\circ h_t^{-1}+Dh_t\left (\dot g_t^X\circ (g^X)_t^{-1}\right)\circ h_t^{-1},
$$
i.e.
\begin{equation}\label{e6}
(h_\di X)_t=\dot h_t\circ h_t^{-1}+(h_t)_*(X_t)\,,
\end{equation}
where $(h_t)_*$ is the standard action of diffeomorphisms on vector fields. In a
slightly different form, this
can be written as an action of $t$-dependent vector fields on $t$-dependent
vector fields:
\begin{equation}\label{e7} (g^Y_\di X)_t=Y_t+(g_t^Y)_*(X_t)\,.
\end{equation}
For global $t$-dependent vector fields on compact manifolds, the latter defines
a group structure in global
$t$-dependent vector fields. This is an infinite-dimensional analogue of a group
structure on paths in a
finite-dimensional Lie algebra, which has been used as a source for a nice
construction of the corresponding
Lie group in \cite{DK}. Since every generalised flow has an inverse,
$(g^{-1})_t=(g_t)^{-1}$, so
generalised flows, or better to say, the corresponding germs, form a group and
the formula (\ref{e7}) allows
us to compute the $t$-dependent vector field (right-logarithmic derivative)
$X_t^{-1}$ associated with the
inverse. It is the $t$-dependent vector field
\begin{equation}
X_t^{-1}=-(g^X_t)^{-1}_*(X_t)\,.
\end{equation}
For $t$-independent vector fields $X_t=X_0$ for all $t$  we have $(g^X_t)_* X=X$
and also we get the
well-known formula
$$X^{-1}=-X\,.$$
Note that, by definition, the integral curves of $h_\di X$ are of the form
$h_t(\gamma(t))$, where $\gamma(t)$
are integral curves of $X$. We can summarise our observation as follows.
\begin{theorem}\label{t2} The equation (\ref{e6}) defines a natural action
of  generalised flows  on  $t$-dependent vector fields. This action is a group
action in the sense that
$$(g\circ h)_\di X=g_\di(h_\di X)\,.
$$
The integral curves of $h_\di X$ are of the form $h_t(\gamma(t))$, for
$\gamma(t)$ being an arbitrary integral
curve for $X$.
\end{theorem}
\noindent The above action of generalised flows on $t$-dependent vector fields
can also be defined in an
elegant way by means of the corresponding autonomisations. It is namely easy to
check the following.
\begin{theorem}\label{Sct3} For any generalised flow $h$ and any
$t$-dependent vector field $X$ on a manifold $\mathbb{R}^n$, the standard action
$\overline{h}_*\overline{X}$ of the
diffeomorphism $\overline{h}$, being the autonomisation of $h$, on the vector
field $\overline{X}$, being the
autonomisation of $X$, is the autonomisation of the $t$-dependent vector field
$h_\di X$:
$$\overline{h}_*\overline{X}=\overline{h_\di X}\,.$$
\end{theorem}

\section{Quasi-Lie systems and schemes}\label{Theory}

By a {\it quasi-Lie system} we understand a pair $(X,g)$ consisting of a
$t$-dependent vector field
$X$ on a manifold $\mathbb{R}^n$ (the {\it system}) and a generalised flow $g$
on $\mathbb{R}^n$ (the {\it control}) such that
$g_\di X$ is a Lie system.

Since for the Lie system $g_\di X$ we are able to obtain the general solution
out of a number of
known particular solutions, the knowledge of the control makes possible the
application of a similar procedure for our initial system
possible. Indeed, let $\Phi=\Phi(x_{1},\dots,x_{m};k_1,\dots,k_n)$ be a
superposition function for the Lie system
$g_\di X$, so that, knowing $m$ solutions $\bar x_{(1)},\ldots,\bar x_{(m)},$ of
$g_\di X$, we can derive the
general solution of the form
$$\bar x_{(0)}=\Phi(\bar x_{(1)},\ldots,\bar  x_{(m)};k_1,\ldots, k_n)\,.
$$
If we now know $m$ independent solutions, $x_{(1)},\ldots, x_{(m)},$ of $X$,
then, according to Theorem
\ref{Sct3}, $\bar x_{a}(t)=g_t(x_{a}(t))$ are solutions of $g_\di X$, producing
a general solution of $g_\di X$
in the form $\Phi(\bar x_{(1)},\ldots,\bar x_{(m)};k_1,\ldots, k_n)$. It is now
clear that
\begin{equation}\label{e8}
x_{(0)}(t)=g_t^{-1}\circ
\Phi(g_t(x_{(1)}(t)),\ldots,g_t(x_{(m)}(t));k_1,\ldots,k_n)
\end{equation}
is a general solution of $X$. In this way we have obtained a {\it $t$-dependent
superposition rule} for the
system $X$. We can summarise the above considerations as follows.
\begin{theorem}\label{t4} Any quasi-Lie system $(X,g)$ admits a
$t$-dependent superposition rule of the form (\ref{e8}), where $\Phi$ is a
superposition function for the Lie
system $g_\di X$.
\end{theorem}

Of course, the above $t$-dependent superposition rule is practically useless for
finding the
general solution of a system $X$ only if the generalised flow $g$ is explicitly
known. An alternative abstract
definition of a quasi-Lie system as a $t$-dependent vector field $X$ for which
there exists a generalised
flow $g$ such that $g_\di X$ is a Lie system does not have much sense, as every
$X$ would be a quasi-Lie
system
 in this
context. For instance, given a $t$-dependent vector field $X$, the pair $(X,
(g^X)^{-1})$ is a quasi-Lie
system because $(g^X)_t^{-1}\circ g_t^X={\rm id}_{\mathbb{R}^n}$, thus
$(g^X)^{-1}_\di X=0$, which is a Lie system
trivially. On the other hand, finding $(g^X)^{-1}$ is nothing but solving our
system $X$ completely, so we
just reduce to our original problem. In practice, it is therefore crucial that
the control $g$ comes from a
system which can be integrated effectively. There are, however, many cases when
our procedure works well and
provides a geometrical interpretation of many {\it ad hoc} methods of
integration.
Consider, for instance, the following scheme that can lead to `nice' quasi-Lie
systems.

Take a finite-dimensional real vector space $V$ of vector fields on
$\mathbb{R}^n$ and consider the family, $V(\mathbb{R})$, of all
$t$-dependent vector fields $X$ on $\mathbb{R}^n$ such that $X_t$ belongs to $V$
on its domain, i.e. $X_t\in
V_{|U^X_t}$ or, in short, $X\in V(\mathbb{R})$. We will say that these are
$t$-dependent vector fields taking values in $V$. The $t$-dependent vector
fields  of $V(\mathbb{R})$ depend on a finite family of control functions. For
example, take a basis
$\{X_1,\dots,X_r\}$ of $V$ and consider a general $t$-dependent system with
values in $V$ determined by
$b=b(t)=(b_1(t),\dots,b_r(t))$ as
 $$(X^{b})_t=\sum_{\alpha=1}^rb_\alpha(t)X_\alpha \,.$$
On the other hand, the nonautonomous systems of differential equations
associated with $X\in V|_{U_t^X}$ are
not Lie systems in general, if $V$ is not a Lie algebra itself. If we
additionally have a finitely
parametrised family of local diffeomorphism, say
$\underline{g}=\underline{g}(a_1,\dots,a_k)$, then any curve
$a=a(t)=(a_1(t),\dots,a_k(t))$ in the control parameters, defined for small $t$,
gives rise to a generalised
flow $g^{a}_t=\underline{g}(a(t))$. Let us additionally assume that there is a
Lie algebra $V_0$ of vector
fields contained in $V$. We can look for control functions $a(t)$ such that for
certain $b(t)$ we get that
$g^{a}_\di X^{b}$ has values in $V_0$ for each $t$. Let us denote this as
\begin{equation}\label{e9}g^{a}_\di X^{b}\in V_0(\mathbb{R}).
\end{equation}
Consequently, each pair $(X^{b},g^{a})$ becomes a quasi-Lie system and we can
get a $t$-dependent
superposition rule for the corresponding system $X^b$.

Let us observe that in the case when all the generalised flows $g^a$ preserve
$V$, i.e. for each
$t$-dependent vector field $X^b\in V(\mathbb{R})$ also $g^a_\di X^b\in
V(\mathbb{R})$, the inclusion (\ref{e9}) becomes a
differential equation for the control functions $a(t)$ in terms of the functions
$b(t)$. This situation is not
so rare, as it may seem at first sight. Suppose, for instance, that we find a
Lie algebra $W\subset
V$ such that $[W,V]\subset V$ and that the $t$-dependent vector fields with
values in $W$ can be effectively
integrated to generalised flows. In this case, any $t$-dependent vector field
$Y^{a}$ with values in $W$
gives rise to a generalised flow $g^{a}$ which, in view of transformation rule
(\ref{e7}), preserves the set
of $t$-dependent vector fields with values in $V$. For each $b=b(t)$ the
inclusion (\ref{e9}) becomes
therefore a differential equation for the control function $a=a(t)$ which often
can be effectively solved.
\begin{definition}\label{QLscheme}{\rm Let $W, V$ be finite-dimensional real
vector spaces of vector fields on $\mathbb{R}^n$.
We say that they form a {\it quasi-Lie scheme} $S(W,V)$ if the following
conditions are satisfied
\begin{enumerate}
\item $W$ is a vector subspace of $V$. \item $W$ is a Lie algebra of vector
fields, i.e. $[W,W]\subset W$.
\item $W$ normalises $V$, i.e. $[W,V]\subset V$.
\end{enumerate}
If $V$ is a Lie algebra of vector fields, we simply call the quasi-Lie scheme
$S(V,V)$ a {\it Lie scheme}
$S(V)$.}
\end{definition}

\begin{note}{\rm Although the normaliser of $V$ in $V$ is the largest Lie
algebra of vector fields that we can use as $W$, for practical purposes it is
sometimes useful to consider smaller Lie subalgebras.}
\end{note}
\begin{definition} {\rm We call the {\it group of the scheme} $S(W,V)$ the group
$\mathcal{G}(W)$
of generalised flows corresponding to the $t$-dependent vector fields with
values in $W$}.
\end{definition}
\begin{maintheorem}\label{Main} {\bf (Main property of a scheme)}
{\rm Given a quasi-Lie scheme $S(W,V)$, then $g_\di X\in V(\mathbb{R}^n)$ for
every $t$-dependent vector field $X\in V(\mathbb{R})$ and each generalised flow
$g\in
\mathcal{G}(W)$}.
\end{maintheorem}
The proof for this is obvious and follows straightforwardly from the fact that
if $g^Y$ is the generalised flow of a $t$-dependent vector field $Y\in
W(\mathbb{R})$ and $X$ takes values in $V$, then, according to the formula
(\ref{e7}), $g^Y_\di X$ takes values in $V$ as well, as
$[W,V]\subset V$ and $V$ is finite-dimensional.

In some applications, it turns out to be interesting to use a more general class
of transformations than those described by $\mathcal{G}(W)$. Nevertheless, such
transformations keep the main property of the generalised flows
$\mathcal{G}(W)$, namely, for a given scheme $S(W,V)$ they transform elements of
$V(\mathbb{R})$ into elements of this space. 

Recall that given a Lie algebra of vector fields
$W\subset\mathfrak{X}(\mathbb{R}^n)$, there always exists, at least locally in
$\mathbb{R}^n$, a group action $\Phi:G\times U\rightarrow U$, with $G$ a Lie
group with Lie algebra $\mathfrak{g}$, whose fundamental vector fields are those
of $W$ (cf. \cite{JP09} and Section \ref{ISOA}). For simplicity, we shall
suppose, as usual, that this action is globally defined on $\mathbb{R}^n$, and
we will write $\Phi:G\times \mathbb{R}^n\rightarrow\mathbb{R}^n$ and define the
restriction map $\Phi_g:x\in\mathbb{R}^n\mapsto
\Phi_g(x)=\Phi(g,x)\in\mathbb{R}^n$ for every $g\in G$. 
 
\begin{lemma}\label{The2} Given a scheme $S(W,V)$, an element
$g\in\exp(\mathfrak{g})$, and a vector field $X\in V(\mathbb{R})$, then
$\Phi_{g*}X\in V(\mathbb{R})$.
\end{lemma}
\begin{proof}As $g\in\exp(\mathfrak{g})$, there exists an element ${\rm
a}\in\mathfrak{g}$ such that $g=\exp({\rm a})$. Consider the curve
$h:s\in[0,1]\mapsto \exp(s\,{\rm a})\in G$. By means of the action $\Phi:G\times
\mathbb{R}^n\rightarrow \mathbb{R}^n$, whose fundamental vector fields are the
Lie algebra of vector fields $W$, the curve $h(s)$ induces the generalised flow
$h^Y_s:x\in \mathbb{R}^n\mapsto \Phi(\exp(s\,{\rm a}),x)\in \mathbb{R}^n$ of the
vector field 
$$
Y(x)=\dfrac{\partial}{\partial s}\bigg|_{s=0}h^Y_s(x)=\dfrac{\partial}{\partial
s}\bigg|_{s=0}\Phi(\exp(s\,{\rm a}),x)
$$
and, obviously, $Y\in W$. Taking into account the relation \cite[p.
91]{Foundations}
$$
\frac{\partial}{\partial s}h^Y_{-s*}X=h^Y_{-s*}[Y,X],
$$
we define, for each $s$, the vector field $Z^{(0)}_{-s}=h^Y_{-s*}X$ to get 
$$
(h_{-s*}^YX)_x=X_x+\int^s_0\dfrac{\partial}{\partial s'}Z^
{(0)}_{-s'}(x)ds'=X_x+\int^s_0(h^Y_{-s'*}[Y,X])_xds'.
$$
If we call $Z^{(1)}_{-s}=h^Y_{-s*}([Y,X])$ and apply the above expression to
$[Y,X]$, we get
$$
\begin{aligned}
(h^Y_{-s*}[Y,X])_x&=[Y,X]_x+\int^{s}_0\dfrac{\partial}{\partial s'}Z^
{(1)}_{-s'}(x)ds'\cr &=[Y,X]_x+\int^{s}_0(h^Y_{-s'*}[Y,[Y,X]])_xds'.
\end{aligned}
$$
Defining $Z_{-s}^{(k)}$ in an analogous way and applying all these results to
the initial formula for $h^Y_{-s*}X$ we obtain
\begin{equation*}
(h^Y_{-s*}X)_x=X_x+[Y,X]_xs+\frac{1}{2!}[Y,[Y,X]]_xs^2+\frac{1}{3!}[Y,[Y,[Y,X]]]
_xs^3+\ldots 
\end{equation*}
By means of the properties of the scheme, we obtain that each term belongs to
$V(\mathbb{R})$, i.e.
$$[Y,[Y,\ldots,[Y,X]\ldots]]\in V(\mathbb{R}),$$ 
and therefore
$$
\Phi_{g*}X=h^Y_{1*}X\in V(\mathbb{R}).
$$
\end{proof}

Note that every curve $g(t)$ in $G$ determines a diffeomorphism on
$\mathbb{R}\times\mathbb{R}^n$ of the form $\overline
\Phi_{g(t)}:(t,x)\in\mathbb{R}\times\mathbb{R}^n\mapsto
(t,\Phi_{g(t)}x)\in\mathbb{R}\times\mathbb{R}^n$. Therefore, given a
$t$-dependent vector field $X\in\mathfrak{X}_t(\mathbb{R}^n)$ and a curve
$g(t)$, this curve transforms $X$ into a new vector field $X'$ such that
$X'=\overline \Phi_{g(t)}\overline X$. For the sake of simplicity, we hereby
denote $X'=g_\di X$ and $g_t:x\in\mathbb{R}^n\mapsto
\Phi_{g(t)}x\in\mathbb{R}^n$. Obviously, in similarity with equation (\ref{e6}),
we have $(g_\di X)_t=\dot g_t\circ g^{-1}_{t}+g_{t*}(X)$ and the set of curves
in $G$ makes up an infinite-dimensional group acting on
$\mathfrak{X}_t(\mathbb{R}^n)$.

\begin{proposition} Given a scheme $S(W,V)$, a curve $g(t)$ in $G$, and a
$t$-dependent vector field $X\in V(\mathbb{R})$, then $g_\di X\in
V(\mathbb{R}).$
\end{proposition}
\begin{proof} As formula (\ref{e6}) remains valid for the action of curves
$g(t)$ included in $\exp(\mathfrak{g})$, proving that $g_\di X$ belongs to
$V(\mathbb{R})$ can be reduced to checking that the corresponding terms $\dot
g_t\circ g^{-1}_{t}$ and $g_{t*}X$ are in $V(\mathbb{R})$. On one hand,  $\dot
g_t\circ g^{-1}_{t}\in W(\mathbb{R})\subset V(\mathbb{R})$ and, by means of
Lemma \ref{The2}, we get that $g_{t*}X\in V(\mathbb{R})$ for each $t$.
Consequently, we see that $g_{\di}X\in V(\mathbb{R})$. Since every curve
$g(t)\subset G$ decomposes as a product $g=g_1\cdot\ldots\cdot g_p$ of curves
$g_j\subset\exp(\mathfrak{g})$, with $j=1,\ldots,p$, it follows that
$g_{\di}X\in V(\mathbb{R})$ for every curve $g(t)\subset G$.
\end{proof}

\begin{definition} Given a scheme $S(W,V)$, we call symmetry group of the
scheme, ${\rm Sym}(W)$, the set of $t$-dependent transformations $\Phi_{g(t)}$
induced by the curves $g(t)$ in $G$ and an action $\Phi$ associated with the Lie
algebra of vector fields $W$. 
\end{definition}

In order to simplify the notation, we hereby denote the $t$-dependent
transformation $\Phi_{g(t)}$ with the curve $g$.

\begin{definition} {\rm Given a quasi-Lie scheme $S(W,V)$ and a $t$-dependent
vector field $X\in V(\mathbb{R})$,
we say that $X$ is a {\it quasi-Lie system with respect to $S(W,V)$} if there
exists a $t$-dependent transformation 
$g\in {\rm Sym}(W)$ and a Lie algebra of vector fields $V_0 \subset V$ such that
$$
g_\di X\in V_0(\mathbb{R}).
$$}
\end{definition}
We emphasise that if $X$ is a quasi-Lie system with respect to the scheme
$S(W,V)$, it automatically admits a
$t$-dependent superposition rule in the form given by (\ref{e8}).

\section{{\it t}-dependent superposition rules}\label{TDS}
Minor modifications in the geometric approach to Lie systems detailed in Section
\ref{GLT} allow us to derive a new theory, based on the so-called {\it Lie
family} concept, in order to treat a much larger family of systems of
differential equations including Lie and quasi-Lie systems. Roughly speaking,
Lie families are sets of systems of differential equations admitting a common
superposition rule with $t$-dependence. This theory clearly generalises the
superposition rule notion and provides a characterisation, described by the
so-called {\it Generalised Lie Theorem}, of families of systems admitting such
a property. Next, we provide a brief description of this theory and summarise
its main results. For further details, see \cite{CGL09}.

Consider the family of nonautonomous systems of first-order ordinary
differential equations on $\mathbb{R}^ n$, parametrised by the elements $d$ of a
set $\Lambda$, of the form
\begin{equation}\label{eq1}
\frac{dx^i}{dt}=Y^i_d(t,x),\qquad i=1,\ldots,n,\qquad d\in\Lambda.
\end{equation}
describing the integral curves of the family of $t$-dependent vector fields
$\{Y_d\}_{d\in\Lambda}$ given by  
$$Y_d(t,x)=\sum_{i=1}^nY^i_d(t,x)\pd{}{x^i}.$$
Let us state the fundamental concept to be studied along this section:
\begin{definition} We say that the family of nonautonomous systems (\ref{eq1})
admits a {\sl common $t$-dependent superposition rule} if there exists a map
$\Phi:\mathbb{R}\times \mathbb{R}^{n(m+1)}\rightarrow\mathbb{R}^{n}$, i.e.
\begin{equation}\label{TSupRul}
x=\Phi(t,x_{(1)},\ldots,x_{(m)};k_1,\ldots,k_n),
\end{equation}
such that the general solution, $x(t),$ of any system $Y_d$ of the family
(\ref{eq1}) can be written, at least for sufficiently small $t$, as
$$
x(t)=\Phi(t,x_{(1)}(t),\ldots,x_{(m)}(t);k_1,\ldots,k_n),
$$
with $\{x_{(a)}(t)\,|\,a=1,\ldots,m\}$ being any generic family of particular
solutions of $Y_d$ and the set $\{k_1,\ldots,k_n\}$ being $n$ arbitrary
constants to be associated with each particular solution. A family of systems
(\ref{eq1}) admitting a common $t$-dependent superposition is called a {\it Lie
family}.
\end{definition}

\begin{definition}
Given a $t$-dependent vector field  $Y=\sum_{i=1}^nY^i(t,x)\partial/\partial
x^i$ on $\mathbb{R}^n$, we define its prolongation to
$\mathbb{R}\times\mathbb{R}^{n(m+1)}$ as the vector field on
$\mathbb{R}\times\mathbb{R}^{n(m+1)}$ given by
$$
Y^{\wedge}(t,x_{(0)},\ldots,x_{(m)})=\sum_{a=0}^{m}\sum_{i=1}^nY^
i(t,x_{(a)})\pd{}{x^ i_{(a)}},
$$
and its autonomisation, $\widetilde Y$, as the vector field on
$\mathbb{R}\times\mathbb{R}^{n(m+1)}$ of the form
$$
\widetilde Y(t,x_{(0)},\ldots,x_{(m)})=\pd{}{t}+\sum_{a=0}^{m}\sum_{i=1}^nY^
i(t,x_{(a)})\pd{}{x^ i_{(a)}}.
$$
\end{definition}

The Implicit Function Theorem states that, given a common $t$-dependent
superposition rule
$\Phi:\mathbb{R}\times\mathbb{R}^{n(m+1)}\rightarrow\mathbb{R}^n$ of a Lie
family $\{Y_d\}_{d\in\Lambda}$, the map
$\Phi(t,x_{(1)},\ldots,x_{(m)};):\mathbb{R}^n\longrightarrow\mathbb{R}^n$, which
reads  $x_{(0)}=\Phi(t,x_{(1)},\ldots,x_{(m)};k)$, can be inverted to give rise
to a map $\Psi:\mathbb{R}\times\mathbb{R}^{n(m+1)}\rightarrow \mathbb{R}^n$
given by
$$
k=\Psi(t,x_{(0)},\ldots, x_{(m)}),
$$
with $k=(k_1,\ldots,k_n)$ being the only point in $\mathbb{R}^n$ such that 
$$
x_{(0)}=\Phi(t,x_{(1)},\ldots,x_{(m)};k).
$$
As
the fundamental property of the map $\Psi$ says that 
$\Psi(t,x_{(0)}(t),\ldots, x_{(m)}(t))$
is constant for any $(m+1)$-tuple of particular solutions of any system of the
family (\ref{eq1}), the 
foliation determined by $\Psi$ is invariant under the permutation of its $(m+1)$
arguments $\{x_{(a)}\,|\,a=0,\ldots,m\}$ and
differentiating the preceding expression we get
\begin{equation}\label{Psifi}
\pd{\Psi^j}{t}+\sum_{a=0}^{m}\sum_{i=1}^nY_d^i(t,x_{(a)}(t))\pd{\Psi^j}{x^i_{(a)
}}=0,\qquad j=1,\ldots,n,\quad d\in\Lambda,
\end{equation}
with $\Psi=(\Psi^1,\ldots,\Psi^n)$.

The relation (\ref{Psifi}) shows that the functions of the set
$\{\Psi^i\,|\,i=1,\ldots,n\}$ are first-integrals for the vector fields
$\widetilde Y_d$, that is, $\widetilde Y_d\Psi^ i=0,$ with $i=1,\ldots,n.$
Therefore, they generically define an $n$-codimensional foliation $\mathfrak{F}$
on $\mathbb{R}\times\mathbb{R}^{n(m+1)}$ such that the vector fields $\widetilde
Y_d$, are tangent to the leaves $\mathfrak{F}_k$ of this foliation, with $k\in
\mathbb{R}^n$.

The foliation $\mathfrak{F}$ has another important property. Given the level set
$\mathfrak{F}_k$ of the map $\Psi$ corresponding to
$k=(k_1,\ldots,k_n)\in\mathbb{R}^n$ and a generic point
$(t,x_{(1)},\ldots,x_{(m)})$ of $\mathbb{R}\times\mathbb{R}^{mn}$, there is only
one point $x_{(0)}\in\mathbb{R}^n$ such that 
$(t,x_{(0)},x_{(1)},\ldots,x_{(m)})\in\mathfrak{F}_k$. Then, the projection onto
the last $m\cdot n$ coordinates and the time 
$$
\pi:(t,x_{(0)},\ldots,x_{(m)})\in\mathbb{R}\times\mathbb{R}^{n(m+1)}\mapsto
(t,x_{(1)},\ldots,x_{(m)})\in \mathbb{R}\times\mathbb{R}^{nm},$$
induces local diffeomorphisms on the leaves $\mathfrak{F}_k$ of $\mathfrak{F}$
into $\mathbb{R}\times\mathbb{R}^{nm}$.

This property can also be seen as the fact that the foliation $\mathfrak{F}$
corresponds to a zero curvature connection $\nabla$ on the bundle
$\pi:\mathbb{R}\times\mathbb{R}^{n(m+1)}\rightarrow\mathbb{R}\times\mathbb{R}^{
nm}$. Indeed, the restriction of the projection $\pi$ to a leaf gives a
one-to-one map. In this way, we get a linear map among vector fields on
$\mathbb{R}\times\mathbb{R}^{nm}$ and `horizontal' vector fields tangent to a
leaf.

Note that the knowledge of this connection (foliation) gives us the {\it common
$t$-dependent superposition rule} without referring to the map $\Psi$. If we fix
the point $x_{(0)}(0)$ and $m$ particular solutions,
$x_{(1)}(t),\ldots,x_{(m)}(t),$ for a system of the family,  then $x_{(0)}(t)$
is the unique curve in $\mathbb{R}^n$ such that
$$(t,x_{(0)}(t),x_{(1)}(t),\ldots, x_{(m)}(t))\in
\mathbb{R}\times\mathbb{R}^{nm}$$ belongs to the same leaf as the point
$(0,x_{(0)}(0),x_{(1)}(0),\ldots,x_{(m)}(0))$. Thus, it is only the foliation
$\mathfrak{F}$ what really matters when the {\it common $t$-dependent
superposition rule} is concerned.

On the other hand, if we have a zero curvature connection $\nabla$ on the bundle
$$\pi:\mathbb{R}\times\mathbb{R}^{n(m+1)}\rightarrow\mathbb{R}\times\mathbb{R}^{
nm},$$i.e. if we have an involutive horizontal distribution $\nabla$ on
$\mathbb{R}\times\mathbb{R}^{n(m+1)}$ that can be integrated to give a foliation
$\mathfrak{F}$ on $\mathbb{R}\times\mathbb{R}^{n(m+1)}$ and such that the vector
fields $\widetilde Y_d$ are tangent to the leaves of the foliation, then the
procedure described above determines a {\it common $t$-dependent superposition
rule} for the family of nonautonomous systems of first-order differential
equations (\ref{eq1}).

Indeed, let $k\in\mathbb{R}^n$ enumerate smoothly the leaves $\mathfrak{F}_k$ of
$\mathfrak{F}$, i.e. there exists a smooth map $\iota :\mathbb{R}^n\rightarrow
\mathbb{R}\times\mathbb{R}^{n(m+1)}$ such that $\iota(\mathbb{R}^n)$ intersects
every $\mathfrak{F}_k$ in a unique point. Then, if $x_{(0)}\in\mathbb{R}^n$ is
the unique point such that 
$$(t,x_{(0)},x_{(1)},\ldots,x_{(m)})\in\mathfrak{F}_k,$$ this fact 
gives rise to a {\it $t$-dependent superposition rule} 
$$
x_{(0)}=\Phi(t,x_{(1)},\ldots,x_{(m)};k)$$
for the family of nonautonomous systems of first-order ordinary differential
equations (\ref{eq1}). To see this, let us observe that the Implicit Function
Theorem shows that there exists a function
$\Psi:\mathbb{R}\times\mathbb{R}^{n(m+1)}\rightarrow\mathbb{R}$ such that 
$$\Psi(t,x_{(0)},\ldots,x_{(m)})=k,$$ 
which is equivalent to say that $(t,x_{(0)},\ldots,x_{(m)})\in \mathfrak{F}_k$.
If we fix a certain $k\in\mathbb{R}^n$ and take certain solutions,
$x_{(1)}(t),\ldots,x_{(m)}(t),$ of a particular instance of (\ref{eq1}), then
$x_{(0)}(t)$ defined by means of the condition
$\Psi(t,x_{(0)}(t),\ldots,x_{(m)}(t))=k$ also satisfies such an instance.
Indeed, let $x'_{(0)}(t)$ be the solution  with initial value
$x'_{(0)}(0)=x_{(0)}$. Since the vector fields $\widetilde Y_d$ are tangent to
$\mathfrak{F}$, the curve $$t\mapsto
(t,x_{(0)}(t),x_{(1)}(t),\ldots,x_{(m)}(t))$$ lies entirely in a leaf of
$\mathfrak{F}$, so in $\mathfrak{F}_k$. But the point of one leaf is entirely
determined by its projection $\pi$, so $x'_{(0)}(t)=x_{(0)}(t)$ and $x_{(0)}(t)$
is a solution.
\begin{proposition}
Giving a $t$-dependent superposition rule (\ref{TSupRul}) for a family of
systems of differential equations (\ref{eq1}) is equivalent to give a zero
curvature connection on the bundle
$\pi:\mathbb{R}\times\mathbb{R}^{(m+1)n}\rightarrow\mathbb{R}\times\mathbb{R}^{
nm}$ for which the $\widetilde Y_d$ are `horizontal' vector fields.
\end{proposition}

In general it is difficult to determine whether a family of differential
equations admits a common $t$-dependent superposition rule by means of the above
Proposition. It is therefore interesting to find a characterisation of Lie
families by means of a more convenient criterion, e.g. through an easily
verifiable condition based on the properties of the $t$-dependent vector fields
$\{Y_a\}_{a\in\Lambda}$. Finding such a criterion is the main result of the
theory of Lie families. It is formulated as Generalised Lie Theorem and based on
the following lemmas given below. The first two ones are straightforward, a
complete detailed proof for the third can be found in \cite{CGL09}.

\begin{lemma}\label{PureProl} Given two $t$-dependent vector fields $X$ and $Y$
on $\mathbb{R}^n$,
the commutator $[\widetilde X,\widetilde Y]$ on
$\mathbb{R}\times\mathbb{R}^{n(m+1)}$ is the prolongation
of a $t$-dependent vector field $Z$ on $\mathbb{R}^n$, $[\widetilde X,\widetilde
Y]=Z^\wedge$.
\end{lemma}
\begin{lemma}\label{Aut} Given a family of $t$-dependent vector fields,
$X_1,\ldots,X_r$,  on $\mathbb{R}^n$,
their autonomisations satisfy the relations
$$
[\bar X_j, \bar X_k](t,x)=\sum_{l=1}^rf_{jkl}(t)\bar X_l(t,x),\qquad
j,k=1,\ldots,r,
$$
for some $t$-dependent functions $f_{jkl}:\mathbb{R}\rightarrow\mathbb{R}$, if
and only if their
$t$-prolongations to $\mathbb{R}\times\mathbb{R}^{n(m+1)}$, $\widetilde
X_1,\ldots,\widetilde X_r$, obey
analogous relations
$$
[\widetilde X_j, \widetilde X_k](t,x)=\sum_{l=1}^rf_{jkl}(t)\widetilde
X_l(t,x),\qquad j, k=1,\ldots,r.
$$
Moreover, $\sum_{l=1}^rf_{jkl}(t)=0$ for all $j, k=1,\ldots,r$.
\end{lemma}

\begin{lemma}\label{Prol} Consider a family of $t$-dependent vector fields,
$Y_1,\ldots,Y_r$, with $t$-prolongations $\widetilde Y_1,\ldots,\widetilde Y_r$
to $\mathbb{R}\times\mathbb{R}^{n(m+1)}$ such that their projections
$\pi_*(\widetilde Y_j)$ are linearly independent at a generic point in
$\mathbb{R}\times\mathbb{R}^{nm}$.
Then, $\sum_{j=1}^r b_j\widetilde Y_j$, with $b_j\in
C^{\infty}(\mathbb{R}\times\mathbb{R}^{nm})$, is of the form $ Y^\wedge$ (resp.
$\widetilde Y$) for a $t$-dependent vector field $Y$ on
$\mathbb{R}^n$, if and only if the functions $b_j$ only depend on the variable
$t$, that is, $b_j=b_j(t)$, and
$\sum_{j=1}^rb_j=0$ (resp., $\sum_{j=1}^rb_j=1$).
\end{lemma}

\begin{maintheorem}{\bf (Generalised Lie Theorem)} \label{MT}
The family of systems (\ref{eq1}) admits a {\sl common $t$-dependent
superposition rule} if and only if the vector fields $\{\overline
Y_d\}_{d\in\Lambda}$ can be written in the form
$$
\overline{Y}_d(t,x)=\sum_{\alpha=1}^rb_{d
\alpha}(t)\overline{X}_\alpha(t,x),\qquad d\in\Lambda,
$$
where $b_{d \alpha}$ are functions of the single variable $t$ such that
$\sum_{\alpha=1}^r b_{d \alpha}=1$ and, $X_1,\ldots,X_r,$ are $t$-dependent
vector fields satisfying
\begin{equation}\label{condition}
[\overline{X}_\alpha,\overline{X}_\beta](t,x)=\sum_{\gamma=1}^rf_{
\alpha\beta\gamma}(t)\overline{X}_\gamma(t,x),\qquad \alpha,\beta=1,\ldots,r,
\end{equation}
for certain functions $f_{\alpha\beta\gamma}:\mathbb{R}\rightarrow\mathbb{R}$.
\end{maintheorem}

The denomination of the above theorem comes from the following proposition,
which shows that each Lie system can be embedded into a Lie family. In order to
formulate this result, let us denote by $S_g(W,V;V_0)$ the set of quasi-Lie
systems of the scheme $S(W,V)$ such that there exists a $g$ satisfying that
$g_\di X\in V_0(\mathbb{R})$ with $V_0$ a Lie algebra of vector fields included
in $V$. Again, complete proof of this proposition can be found in \cite{CGL09}.

\begin{proposition} The family of quasi-Lie systems $S_g(W,V;V_0)$ is a Lie
family admitting the common $t$-dependent superposition rule of the form
$$
\bar\Phi_g(t,x_{(1)},\ldots,x_{(m)},k)=g_t^{-1}\circ\Phi(g_t(x_{(1)},\ldots,
g_t)x_{(m)},k),
$$
for any $t$-independent superposition rule $\Phi$ associated with the Lie
algebra of vector fields $V_0$ by Lie Theorem.
\end{proposition}

\chapter{Applications of quasi-Lie schemes and Lie families}

The theory of quasi-Lie schemes, quasi-Lie systems \cite{CGL08} and the theory
of Lie families \cite{CGL09} can be used to investigate a very large set of
differential equations, namely, nonlinear oscillators \cite{CGL08}, dissipative
Milne--Pinney equations \cite{CGL08,CGL09,CL08Diss}, second-order Riccati
equations \cite{CL09SRicc}, Abel equations \cite{CGL09}, Emden equations
\cite{CGL08,CLL09Emd}, etc. As we showed in the previous section, these theories
enable us to obtain $t$-dependent superposition rules, constants of the motion,
exact solutions, integrability conditions, etc. The main aim in this chapter is
to show that the possibilities of application of these methods are very wide and
we can obtain a very large set of results from a unified point of view.

More exactly, in previous sections it was proved that Milne--Pinney could be
studied by means of the theory of Lie systems (see also \cite{CLuc08b}).
Nevertheless, there exist dissipative Milne--Pinney equations that cannot
straightforwardly be studied through this theory. In this section, we provide a
quasi-Lie scheme to treat these dissipative Milne--Pinney equations. Then, we
use this quasi-Lie scheme to relate these equations to usual Milne--Pinney
equations. By means of this relation, we obtain a $t$-dependent superposition
rule for dissipative Milne--Pinney equations.

Apart from dissipative Milne--Pinney equations, we also investigate
nonautonomous nonlinear oscillators. We show that some of these differential
equations can be transformed into autonomous nonlinear oscillators. This result
was already derived by Perelomov \cite{Pe78}, but here we recover it from a more
general point of view. More specifically, we obtain that the nonautonomous
nonlinear oscillators analysed by Perelomov can be seen as differential
equations obeying an integrability condition derived by means of a quasi-Lie
scheme.

As a last application of the quasi-Lie scheme notion, we extensively analyse
Emden equations. We provide a quasi-Lie scheme to obtain $t$-dependent constants
of the motion by means of particular solutions that obey an integrability
condition. The method developed also enables us to obtain Emden equations with a
fixed $t$-dependent integral of motion. Kummer--Liouville transformations are
also obtained by means of our scheme and many other properties are recovered. 

Finally, in the last two sections of this chapter, we apply common $t$-dependent
superposition rules to study some first- and second-order differential
equations. In this way, we will show how they can be used to analyse equations
which cannot be studied by means of the usual theory of Lie systems.
Additionally, some new results for the study of Abel and Milne--Pinney equations
are provided.

\section{Dissipative Milne--Pinney equations}
In this section, we study the so-called dissipative Milne--Pinney equations. We
show that the first-order
ordinary differential equations associated with these second-order ones in the
usual way, i.e. by considering
velocities as new variables, are not Lie systems. However, the theory of
quasi-Lie schemes can be used to deal
with such first-order systems. Here we provide a scheme which enables us to
transform a certain kind of
dissipative Milne--Pinney equations, considered as first-order systems, into
some first-order Milne--Pinney equations already
studied by means of the theory of Lie systems \cite{CLR07a}. As a result we get
a $t$-dependent superposition
rule for some of these dissipative Milne--Pinney equations.

Let us establish the problem under study. Consider the family of dissipative
Milne--Pinney equations of the form
\begin{equation}\label{eq1Appl}
\ddot x=a(t)\dot x+b(t) x+c(t)\frac{1}{x^3}\,.
\end{equation}

We are mainly interested in the case $c(t)\not =0$, so we assume that $c(t)$ has
a constant sign for the set of values of $t$ that we analyse. 

Usually, we associate such a second-order  differential equation with a system
of first-order differential
equations by introducing a new variable  $v$ and relating the differential
equation (\ref{eq1Appl}) to the system
of first-order differential equations
\begin{equation}\label{eq2}\left\{
\begin{array}{rcl}
\dot x&=&v,\\
\dot v&=&a(t)v+b(t) x+c(t)\dfrac{1}{x^3}.\\
\end{array}\right.
\end{equation}

Let us search for a quasi-Lie scheme to handle the above system. Remember that
we need to find linear spaces
$W_{{\rm DisM}}$ and $V_{{\rm DisM}}$ of vector fields such that
\begin{enumerate}
 \item $W_{{\rm DisM}}\subset V_{{\rm DisM}}$.
\item $[W_{{\rm DisM}},W_{{\rm DisM}}]\subset W_{{\rm DisM}}$. \item $[W_{{\rm
DisM}},V_{{\rm DisM}}]\subset
V_{{\rm DisM}}$.
\end{enumerate}
Also, in order to treat system (\ref{eq2}) through this scheme, we have to
ensure that the $t$-dependent
vector field
$$
X_t=v\pd{}{x}+\left(a(t)v+b(t)x+\frac{c(t)}{x^3}\right)\pd{}{v}\,,
$$
whose integral curves are solutions for the system (\ref{eq2}), is such  that
$X_t\in V_{\rm DisM}$ for every $t$ in an
open interval of $\mathbb{R}$.

Consider the vector space $V_{{\rm DisM}}$ spanned by the vector fields
\begin{equation*}
X_1=v\pd{}{v},\quad X_2=x\pd{}{v},\quad X_3=\frac{1}{x^3}\pd{}{v},\quad
X_4=v\pd{}{x}\,,\quad X_5=x\pd{}{x}
\end{equation*}
and the two-dimensional vector subspace $W_{{\rm DisM}}\subset V_{{\rm DisM}}$ 
generated by
\begin{equation*}
Y_1=X_1=v\pd{}{v},\qquad Y_2=X_2=x\pd{}{v}\,.
\end{equation*}
It can be seen that $W_{{\rm DisM}}$ is a Lie algebra,
\begin{equation*}
\left[Y_1,Y_2\right]=-Y_2\,,
\end{equation*}
and, additionally, as
\begin{equation*}
\begin{array}{lll}
\left[Y_1,X_3\right]=-X_3,&\quad
\left[Y_1,X_4\right]=X_4,&\quad\left[Y_1,X_5\right]=0\,,\cr
\left[Y_2,X_3\right]=0,&\quad
\left[Y_2,X_4\right]=X_5-X_1,&\quad\left[Y_2,X_5\right]=-X_2,
\end{array}
\end{equation*}
the linear space $V_{{\rm DisM}}$ is invariant under the action of the Lie
algebra $W_{{\rm DisM}}$ on
$V_{{\rm DisM}}$, i.e.  $\left[W_{{\rm DisM}},V_{{\rm DisM}}\right]\subset
V_{{\rm DisM}}$. Thus, the vector
spaces
$$
V_{{\rm DisM}}=\langle X_1,\ldots,X_5\rangle\qquad {\rm and} \qquad W_{{\rm
DisM}}=\langle Y_1,Y_2\rangle
$$
of vector fields form a quasi-Lie scheme $S(W_{{\rm DisM}},V_{{\rm DisM}})$. Let
us observe that
\begin{equation*}
X_t=a(t)X_1+b(t)X_2+c(t)X_3+X_4
\end{equation*}
and thus $X\in V_{{\rm DisM}}(\mathbb{R})$.

We stress that the vector space $V_{\rm DisM}$ is not a Lie algebra, because the
commutator $\left[X_3,X_4\right]$ does not belong to $V_{\rm DisM}$. Moreover,
$V''=\langle X_1,\ldots, X_4\rangle$ is not a Lie algebra of vector fields due
to a similar reason, i.e. $[X_3,X_4]\notin V''$. Additionally, there exists no
finite-dimensional real Lie algebra
$V'$ containing $V''$. Thus, system (\ref{eq2}) is not a Lie system, but we can
use the quasi-Lie scheme
$S(W_{{\rm DisM}},V_{{\rm DisM}})$ to investigate it.

The key tool provided by the scheme $S(W_{{\rm DisM}},V_{{\rm DisM}})$ is the
infinite-dimensional group
$\mathcal{G}(W_{{\rm DisM}})$ of generalised flows for the $t$-dependent vector
fields with values in $W$,
i.e. $\alpha_1(t)Y_1+\alpha_2(t)Y_2$, which leads to the group of $t$-dependent
changes of variables {\small
\begin{equation*} \mathcal{G}(W_{{\rm
DisM}})=\left\{g(\alpha(t),\beta(t))=\left\{
\begin{array}{rcl}
x&=&x'\\
v&=&\alpha(t)v'+\beta(t)x'
\end{array}\right.\bigg|\,\alpha(t)>0,\beta(0)=0,\alpha(0)=1\right\}.
\end{equation*}}
According to the general theory of quasi-Lie schemes, these previous
$t$-dependent changes of variables enable
us to transform system (\ref{eq2}) into a new one taking values in $V_{\rm
DisM}$,
\begin{equation}\label{fineq}
X'_t=a'(t)X_1+b'(t)X_2+c'(t)X_3+d'(t)X_4+e'(t)X_5\,.
\end{equation}
The new coefficients are
\begin{equation*}
\left\{
\begin{aligned}
a'(t)&=a(t)-\beta(t)-\frac{\dot \alpha(t)}{\alpha(t)},\\
b'(t)&=\frac{b(t)}{\alpha(t)}+a(t)\frac{\beta(t)}{\alpha(t)}-\frac{\beta^2(t)}
{\alpha(t)}-\frac{\dot\beta(t)}{\alpha(t)},\\
c'(t)&=\frac{c(t)}{\alpha(t)},\\
d'(t)&=\alpha(t),\\
e'(t)&=\beta(t).\\
\end{aligned}\right.
\end{equation*}
The integral curves for the $t$-dependent vector field (\ref{fineq}) are
solutions of the system
\begin{equation}\label{quasiErmsys}\left\{
\begin{array}{rcl}
\dfrac{dx'}{dt}&=&\beta(t)x'+\alpha(t)v',\cr
\dfrac{dv'}{dt}&=&\left(\dfrac{b(t)}{\alpha(t)}+a(t)\dfrac{\beta(t)}{\alpha(t)}
-\dfrac{\beta^2(t)}
{\alpha(t)}-\dfrac{\dot\beta(t)}{\alpha(t)}\right)
x'+\\&+&\left(a(t)-\beta(t)-\dfrac{\dot\alpha(t)}{\alpha(t)}\right)v'+\dfrac{
c(t)}{\alpha(t)}\dfrac{1}{x'^3}.
\end{array}\right.
\end{equation}
As it was said in Section \ref{Theory}, we use schemes to transform the
corresponding systems of first-order
differential equations into Lie ones. So, in this case, we must find a Lie
algebra of vector fields
$V_0\subset V_{\rm DisM}$ and a generalised flow $g\in\mathcal{G}(W_{\rm DisM})$
such that $g_\di X\in V_0(\mathbb{R})$. This leads to a
system of ordinary differential equations for the functions $\alpha(t)$,
$\beta(t)$ and some integrability
conditions on the initial functions $a(t),b(t)$ and $c(t)$ for such a
$t$-dependent change of variables to
exist.

In order to find a proper Lie algebra of vector fields $V_0\subset V$, note that
Milne--Pinney equations studied in
\cite{CLR07a} are Lie systems in the family of differential equations defined by
systems (\ref{eq2}) and
therefore it is natural to look for the conditions needed to  transform a given
system of (\ref{eq2}),
described by the $t$-dependent vector field $X_t$, into one of these first-order
Milne--Pinney equations of the form
\begin{equation}\label{Ermsys}\left\{
\begin{array}{rcl}
\dot x&=&f(t)v,\\
\dot v&=&-\omega(t) x+f(t)\dfrac{k}{x^3},\,
\end{array}\right.
\end{equation}
 where $k$ is a constant, i.e. a system describing the integral curves for a
$t$-dependent vector field with
 values in the Lie algebra of vector fields \cite{CLR07a}
$$
V_0=\langle X_4+k\,X_3, X_2, \frac{1}{2}(X_5-X_1)\rangle.
$$
As a result, we get that $\beta=0$, $\alpha=f$ and, furthermore, the functions 
$\alpha$, $a$ and $c$  must satisfy
\begin{equation}\label{alfaeq}
k \alpha^2=c,\qquad \qquad \dot\alpha-a{\alpha}=0,
\end{equation}
which yield that $c$ and $k$ have the same sign. The second condition is a
differential equation for
$\alpha$ and the first one determines $c$ in terms of $\alpha$. Therefore, both
conditions lead to a relation
between $c$ and $a$ providing the integrability condition
\begin{equation}\label{IntCondAppl}
c(t)=k\,{\rm exp}\left(2\int a(t)\,dt\right)
\end{equation}
and showing, in view of (\ref{quasiErmsys}), (\ref{Ermsys}) and (\ref{alfaeq}), 
that
$$
\alpha(t)={\rm exp}\left(\int a(t)\,dt\right) \qquad {\rm and}\qquad
\omega(t)=-b(t)\exp\left(-\int
a(t)dt\right),
$$
where we choose the constants of integration in order to get $\alpha(0)=1$ as
required.

Summarising the preceding results, under the integrability condition
(\ref{IntCondAppl}), the first-order Milne--Pinney 
equation
\begin{equation*}\left\{
\begin{array}{rcl}
\dot x&=&v,\\
\dot v&=&a(t)v+b(t) x+c(t)\dfrac{1}{x^3},\\
\end{array}\right.
\end{equation*}
can be transformed into the system
\begin{equation*}\left\{
\begin{array}{rcl}
\dfrac{dx'}{dt}&=&{\rm exp}\left(\int a(t)\,dt\right)v',\\
&&\\
\dfrac{dv'}{dt}&=&b(t){\rm exp}\left(-\int a(t)\,dt\right) x'+{\rm
exp}\left(\int a(t)\,dt\right)\dfrac{k}{x'^3},\,
\end{array}\right.
\end{equation*}
by means of the $t$-dependent change of variables
\begin{equation*}g\left({\rm exp}\left(\int a(t)\,dt\right),0\right)=\left\{
\begin{array}{rcl}
x'&=&x,\\
v'&=&{\rm exp}\left(\int a(t)\,dt\right)v.\,
\end{array}\right.
\end{equation*}

We stress the fact that the previous change of variables is a particular
instance of the so-called
Liouville transformation \cite{Milson}.

The final Milne--Pinney equation can be rewritten through the
$t$-reparametrisation
\begin{equation*}
\tau(t)=\int {\rm exp}\left(\int a(t)\,dt\right)dt,
\end{equation*}
as
\begin{equation*}\left\{
\begin{array}{rl}
\dfrac{dx'}{d\tau}&=v',\\
\dfrac{dv'}{d\tau}&={\rm exp}\left(-2\int a(t)dt
\right)b(t(\tau))x'+\dfrac{k}{x'^3}.
\end{array}\right.
\end{equation*}
These systems were analysed in \cite{CLR07b} and there it was shown through the
theory of Lie systems that
they admit the constant of the motion
\begin{equation*}
I=(\bar xv'-\bar vx')^2+k\left(\frac{\bar x}{x'}\right)^2\,,
\end{equation*}
where $(\bar x,\bar v)$ is a  solution of the system
\begin{equation*}\left\{
\begin{aligned}
\frac{d\bar x}{d\tau}&=\bar v,\\
\frac{d\bar v}{d\tau}&={\rm exp}\left(-2\int a(t)\, dt \right)b(t)\, \bar x\,,
\end{aligned}\right.
\end{equation*}
 which can be written  as a second-order differential equation
\begin{equation*}
\frac{d^2\bar x}{d\tau^2}={\rm exp}\left(-2\int a(t)\,dt \right)b(t)\,\bar x\,.
\end{equation*}
If we invert the $t$-reparametrisation,  we obtain the following differential
equation
\begin{equation}\label{SO}
\ddot{\bar x}-a(t)\dot{\bar  x}-b(t)\bar x=0,
\end{equation}
which is the linear differential equation associated with the initial
Milne--Pinney equation.

As it was shown in \cite{CLR07a}, we can obtain, by means of the theory of Lie
systems, the following
superposition rule
$$
x'=\frac {\sqrt 2} {|\bar x_1\bar v_2-\bar v_1\bar
x_2|}\left(I_2\bar{x}_1^2+I_1\bar{x}_2^2\pm\sqrt{4I_1I_2-k(\bar x_1\bar v_2-\bar
v_1\bar v_2)^2}\ \bar x_1\bar
x_2\right)^{1/2}\,,
$$
and as the $t$-dependent transformation performed does not change the variable
$x$, we get the $t$-dependent
superposition rule
$$
x=\frac {\sqrt 2\alpha(t)} {|\bar x_1\dot {\bar x}_2-\dot{ \bar x}_1\bar
x_2|}\left(I_2\bar{x}_1^2+I_1\bar{x}_2^2\pm\sqrt{4I_1I_2-\frac{k}{\alpha^2(t)}
(\bar x_1\dot {\bar x}_2-\dot{
\bar x}_1\bar x_2)^2}\ \bar x_1\bar x_2\right)^{1/2}\,,
$$
in terms of a set of solutions of the second-order linear system (\ref{SO}).

Summing up, the application of our scheme to the family of dissipative
Milne--Pinney equations
\begin{equation*}
\ddot x=a(t)\,\dot x+b(t)\,x+ {\rm exp}\left(2\int a(t)\,dt\right)\frac{k}{x^3}
\end{equation*}
shows that this family admits a $t$-dependent superposition principle:
$$
x=\frac {\sqrt 2\alpha(t)} {|y_1\dot y_2-y_2\dot
y_1|}\left(I_2y_1^2+I_1y_2^2\pm\sqrt{4I_1I_2-\frac{k}{\alpha^2(t)}(y_1\dot
y_2-y_2\dot y_1)^2}\,  y_1
y_2\right)^{1/2}\,,
$$
in terms of two independent solutions $y_1,y_2$ for the differential equation
\begin{equation*}
\ddot{y }-a(t)\,\dot{y}-b(t)\,y=0.
\end{equation*}

So, we have fully detailed a particular application of the theory of quasi-Lie
schemes to dissipative Milne--Pinney equations. As a result, we provide a
$t$-dependent superposition rule for a family of such systems.  Another
paper dealing with such an approach to dissipative Milne--Pinney equations and
explaining some of their properties can
be found in \cite{CL08Diss}.

\section{Non-linear oscillators}
As a second application of our theory, we use quasi-Lie schemes to deal with a
certain kind of nonlinear
oscillators. The main objective of this section is to explain sevaral properties
of a family of $t$-dependent
nonlinear oscillators studied by Perelomov in \cite{Pe78}. We also furnish a, as
far as we know, new
constant of the motion for these systems.

Consider the subset of the family of nonlinear oscillators investigated in
\cite{Pe78}:
\begin{equation*}
\ddot x=b(t)x+c(t)x^n,\qquad n\neq 0,1\,.
\end{equation*}
The cases $n=0,1$, are omitted because they can be handled with the usual theory
of Lie systems. As in the
section above, we link the above second-order ordinary differential equation to
the first-order system
\begin{equation}\label{NLO2}\left\{
\begin{aligned}
\dot x&=v,            \\
\dot v&=b(t)x+c(t)x^n.\cr
\end{aligned}\right.
\end{equation}

Let us provide a quasi-Lie scheme to deal with systems (\ref{NLO2}). Consider
the vector space $V_{NO}$
spanned by the linear combinations of the vector fields
\begin{equation*}
X_1=x\pd{}{v},\quad X_2=x^n\pd{}{v},\quad X_3=v\pd{}{x},\quad
X_4=v\pd{}{v},\quad X_5=x\pd{}{x}\,
\end{equation*}
on ${\rm T}\mathbb{R}$ and take the vector subspace $W_{NO}\subset V_{NO}$
generated by
$$
Y_1=X_4=v\pd{}{v},\quad Y_2=X_1=x\pd{}{v},\quad Y_3=X_5=x\pd{}{x}.
$$
Therefore, $W_{NO}$ is a solvable Lie algebra of vector fields,
$$[Y_1,Y_2]=-Y_2\,,\quad [Y_1,Y_3]=0\,,\quad [Y_2,Y_3]=-Y_2\,,
$$
and taking into account that
\begin{equation*}
\begin{array}{lll}
\left[Y_1,X_2\right]=-X_2,& \left[Y_1,X_3\right]=X_3, & \left[Y_2,X_2\right]=0,
\\ \left[Y_2,X_3\right]=X_5-X_4,
&\left[Y_3,X_2\right]=nX_2,&\left[Y_3,X_3\right]=-X_3,
\end{array}
\end{equation*}
we see that $V_{NO}$ is invariant under the action of $W_{NO}$, i.e.
$[W_{NO},V_{NO}]\subset V_{NO}$. In this
way we get the quasi-Lie scheme $S(W_{NO},V_{NO})$.

Now, we have to go over whether the solutions of system (\ref{NLO2}) are
integral curves for a $t$-dependent
vector field $X\in V_{NO})(\mathbb{R})$. In order to check this, we realise that
the system (\ref{NLO2})
describes the integral curves for the $t$-dependent vector field
$$
X_t=v\pd{}{x}+(b(t)x+c(t)x^n)\pd{}{v},
$$
which can be written as
\begin{equation}
X_t=b(t)X_1+c(t)X_2+X_3\,.\label{NLOPLS}
\end{equation}

Note also that  $[X_2,X_3]\notin V_{NO}$  and $V''=\langle X_1,X_2,X_3\rangle$
is not only a Lie algebra
of vector fields, but also there is no finite-dimensional Lie algebra $V'$
including $V''$. Thus, $X$ cannot be
considered as a Lie system and we conclude that the first-order nonlinear
oscillator
\begin{equation*}\left\{
\begin{array}{rcl}
\dot x&=&v,\\
\dot v&=&b(t) x+c(t)x^n.\\
\end{array}\right.
\end{equation*}
describing integral curves of the $t$-dependent vector field (which is not a Lie
system)
\begin{equation*}
X_t=b(t)X_1+c(t)X_2+X_3\,
\end{equation*}
can be described by means of the quasi-Lie scheme $S(W_{NO},V_{NO})$.

Now, the group of generalised flows $\mathcal{G}(W_{NO})$ associated with
$S(W_{NO},V_{NO})$ is made of the $t$-dependent transformations

\begin{equation*}
g(\alpha(t),\beta(t),\gamma(t))=\left\{
\begin{aligned}
x&=\gamma(t)x'\\
v&=\beta(t)v'+\alpha(t)x'\,
\end{aligned}\right.\,\,\beta(t),\gamma(t)>0,\beta(0)=\gamma(0)=1,\alpha(0)=0.
\end{equation*}

Let us restrict ourselves to the case $\alpha(t)=\dot\gamma(t)$ and
$\beta(t)=1/\gamma(t)$ and apply these transformations
to the system (\ref{NLO2}). The theory of quasi-Lie systems tells us that
$$g(\alpha(t),\beta(t),\gamma(t))_\di X\in V_{NO}(\mathbb{R}).$$ Indeed, these
$t$-dependent transformations
lead to the systems
\begin{equation}\label{transformed}
\left\{
\begin{aligned}
\frac{dx'}{dt}&=\frac{1}{\gamma^2(t)}v',\\
\frac{dv'}{dt}&=(\gamma^2(t)b(t)-\ddot\gamma(t)\gamma(t))x'+c(t)\gamma^{n+1}
(t)x'^n,\,
\end {aligned}
\right.
\end{equation}
which are related to the second-order differential equations
\begin{equation*}
\gamma^2(t) \ddot x' =-2\gamma(t)\dot\gamma(t)\dot
x'+(\gamma^2(t)b(t)-\ddot\gamma(t)\gamma(t))x'+c(t)\gamma^{n+1}(t)x'^n\,.
\end{equation*}
But the theory of quasi-Lie schemes is based on the search of  a  generalised
flow  $g\in\mathcal{G}(W_{NO})$ such
that ${g}_\di X$ becomes a Lie system, i.e. there exists a Lie algebra of vector
fields $V_0\subset V_{NO}$ such
that ${g}_\di X\in V_0(\mathbb{R})$. For instance, we can try to transform a
particular instance of the systems
(\ref{transformed}) into a first-order differential equation associated with a
nonlinear oscillator with a zero
$t$-dependent angular frequency, for example, into the first-order system
\begin{equation}\label{sys}
\left\{
\begin{aligned}
\frac{dx'}{dt}&=f(t)v',\\
\frac{dv'}{dt}&=f(t)c_0x'^n\,,
\end {aligned}
\right.
\end{equation}
related to the nonlinear oscillator
$$
\frac{d^2x'}{d\tau^2}=c_0x'^n,
$$
with $d\tau/dt=f(t)$.

The conditions ensuring such a transformation are
\begin{equation}\label{condit2}
\gamma(t)b(t)-\ddot\gamma(t)=0\,,\quad c(t)=c_0\gamma^{-(n+3)}(t),
\end{equation}
with $f(t)=\gamma^{-2}_1(t)$, where $\gamma_1$ is a  non-vanishing particular
solution for
$\gamma(t)b(t)-\ddot\gamma(t)=0$. We must emphasise that just particular
solutions  with $\gamma_1(0)=1$ and $\dot\gamma_1(0)=0$ are related to
generalised flows in $\mathcal{G}(W_{\rm NO})$. Nevertheless, any other
particular solution can also be used to transform a nonlinear oscillator into a
Lie system as we stated. The Lie system (\ref{sys}) is the system associated
with
the $t$-dependent vector field
\begin{equation*}
X_t=\frac1{\gamma^{2}_1(t)}\left(v'\pd{}{x'}+c_0x'^n\pd{}{v'}\right).
\end{equation*}

As a consequence of the standard methods developed for the theory of Lie systems
\cite{CLR08}, we  join two
copies of the above system in order to get the first-integrals
\begin{equation*}
I_i=\frac{1}{2}v_i'^2-\frac{c_0}{n+1}x_i'^{n+1}, \qquad i=1,2,
\end{equation*}
and
\begin{multline*}
I_3=\frac{x'_1}{\sqrt{I_1}}{\rm
Hyp}\left(\frac{1}{n+1},\frac{1}{2},1+\frac{1}{n+1},-\frac{c_0x_1'^{n+1}}{
I_1(n+1)}\right)\\-
\frac{x'_2}{\sqrt{I_2}}{\rm
Hyp}\left(\frac{1}{n+1},\frac{1}{2},1+\frac{1}{n+1},-\frac{c_0x_2'^{n+1}}{
I_2(n+1)}\right),
\end{multline*}
where ${\rm Hyp}(a,b,c,d)$ denotes the corresponding hypergeometric functions.
In terms of the
initial variables these first-integrals for ${g}_\di X$ read

\begin{equation}\label{integral1}
\begin{aligned}
I_i&=\frac{1}{2}(\gamma_1(t)\dot
x_i-\dot\gamma_1(t)x_i)^2-\frac{c_0}{\gamma_1^{n+1}(t)(n+1)}x_i^{n+1},\qquad
i=1,2,
\end{aligned}
\end{equation}
and 
\begin{multline}\label{integral2}
I_3=\frac{1}{\gamma_1(t)}\left(\frac{x_1}{\sqrt{I_1}}{\rm
Hyp}\left(\frac{1}{n+1},\frac{1}{2},1+\frac{1}{n+1},-\frac{c_0x_1^{n+1}}{
\gamma_1^{n+1}(t)I_1(n+1)}\right)\right.\\
\left. -\frac{x_2}{\sqrt{I_2}}{\rm
Hyp}\left(\frac{1}{n+1},\frac{1}{2},1+\frac{1}{n+1},-\frac{c_0x_2^{n+1}}{
\gamma_1^{n+1}(t)I_2(n+1)}\right)\right).
\end{multline}

As a particular application of conditions (\ref{condit2}), we can consider the
following example of \cite{Pe78}, where
the $t$-dependent Hamiltonian
\begin{equation*}
H(t)=\frac{1}{2}p^2+ \frac{\omega^2(t)}{2}x^2+c^2\gamma_1^{-(s+2)}(t)x^s\,,
\end{equation*}
with $\gamma_1$ being such that $\ddot\gamma_1(t)+\omega^2(t)\gamma_1(t)=0$, is
studied. The Hamilton
equations for the latter Hamiltonian are
\begin{equation}\label{IQLS}\left\{
\begin{aligned}
\dot x&=p,\\
\dot p&=-sc^2\gamma_1^{-(s+2)}(t)x^{s-1}-\omega^2(t)x,
\end{aligned}\right.
\end{equation}
which are associated with the second-order differential equation for the
variable $x$ given by
\begin{equation}\label{NLOS}
\ddot x=-sc^2\gamma_1^{-(s+2)}(t)x^{s-1}-\omega^2(t)x.
\end{equation}
Note that here the variable $p$ plays the same role as $v$ in our theoretical
development and the latter differential equation is a  particular case of our
Emden equations with
\begin{equation}\label{OurCase}
b(t)=-\omega^2(t)\,,\qquad c(t)=-sc^2\gamma_1^{-(s+2)}(t)\,,\quad n=s-1.
\end{equation}
Let us prove that the above coefficients satisfy the conditions (\ref{condit2}):
\begin{enumerate}
 \item By assumption, $\omega^2(t)\gamma_1(t)+\ddot\gamma_1(t)=0$. As
$\omega^2(t)=-b(t)$, then
 $\gamma_1(t)b(t)-\ddot\gamma_1(t)=0$.
\item If we fix $c_0=-sc^2$, in view of conditions (\ref{OurCase}), we obtain
$c(t)=c_0\gamma_1^{-(n+3)}(t)$.
\end{enumerate}
Therefore, we get that the $t$-dependent  frequency nonlinear oscillator
(\ref{NLOS}) can be transformed into
a new one with zero frequency, i.e.
$$
\frac{d^2x'}{d\tau^2}=-sc^2x'^{s-1},
$$
with
$$
\tau=\int \frac{dt}{\gamma^2_1(t)},
$$
reproducing the result given by Perelomov \cite{Pe78}. The choice of the
$t$-dependent frequencies is such that it
is possible to transform the initial $t$-dependent nonlinear oscillator into the
 final autonomous nonlinear
oscillator. Then, we recover here such frequencies as a result of an
integrability condition. Moreover, in
view of the expressions  (\ref{integral1}), (\ref{integral2}) and
(\ref{OurCase}), we get a, as far as we know, 
new $t$-dependent constants of the motion for these nonlinear oscillators.
\section{Dissipative Mathews--Lakshmanan oscillators}
In this section we provide a simple application of the theory of quasi-Lie
schemes to investigate the $t$-dependent dissipative Mathews-Lakshmanan
oscillator
\begin{equation}\label{DML}
(1+\lambda x^2)\ddot x-F(t)(1+\lambda x^2)\dot x-(\lambda x)\dot
x^2+\omega(t)x=0,\qquad \lambda>0.
\end{equation}
More specifically, we supply some integrability conditions to relate the above
dissipative oscillator to the Mathews--Lakshmanan one
\cite{CRS04,CRSS04,LR03,ML74}
\begin{equation}\label{ML}
(1+\lambda x^2)\ddot x-(\lambda x)\dot x^2+k x=0,\qquad \lambda>0,
\end{equation}
and by means of such a relation we get a, as far as we know, new $t$-dependent
constant of the motion.

Consider the system of first-order differential equation related to equation
(\ref{DML}) in the usual way, i.e.
\begin{equation}\label{FirstML}
\left\{
\begin{aligned}
\dot x&=v,\\
\dot v&=F(t)v+\frac{\lambda x v^2}{1+\lambda x^2}-\omega(t)\frac{x}{1+\lambda
x^2},
\end{aligned}\right.
\end{equation}
and determining the integral curves for the $t$-dependent vector field
$$
X_t=\left(F(t) v+\frac{\lambda x v^2}{1+\lambda x^2}-\omega(t)\frac{x}{1+\lambda
x^2}\right)\frac{\partial}{\partial v}+v\frac{\partial}{\partial x}.
$$
Let us provide a scheme to handle the system (\ref{FirstML}). Consider the
vector space $V$ spanned by the vector fields
\begin{equation}\label{MLbasis}
X_1=v\frac{\partial}{\partial x}+\frac{\lambda x v^2}{1+\lambda
x^2}\frac{\partial}{\partial v},
\quad X_2=\frac{x}{1+\lambda x^2}\frac{\partial}{\partial v},\quad
X_3=v\frac{\partial}{\partial v},
\end{equation}
and the linear space $W=\langle X_3\rangle$. The commutation relations
$$
[X_3,X_1]=X_1,\qquad [X_3,X_2]=-X_2,
$$
imply that the linear spaces $W,V$ make up a quasi-Lie scheme $S(W,V)$.  As the
$t$-dependent vector field $X_t$ reads in terms of the basis (\ref{MLbasis})
$$
X_t=F(t)X_3-\omega(t)X_2+X_1,
$$
we get that $X_t\in V(\mathbb{R})$.

The integration of $X_3$ shows that
{\footnotesize
\begin{equation*}
\mathcal{G}(W)=\left\{g(\alpha(t))=\left\{
\begin{aligned}
x&=x',\\
v&=\alpha(t)v'.\,
\end{aligned}\right.\bigg|\,\alpha(t)>0,\,\alpha(0)=1\right\},
\end{equation*}}
and the $t$-dependent changes of variables related to the controls of
$\mathcal{G}(W)$ transform the system (\ref{FirstML}) into
\begin{equation*}
\left\{
\begin{aligned}
\dot x'&=\alpha(t)v',\\
\dot v'&=\left(F(t)-\frac{\dot
\alpha(t)}{\alpha(t)}\right)v'-\frac{\omega(t)}{\alpha(t)}\frac{x'}{1+\lambda
x'^2}+\alpha(t)\frac{\lambda x'v'^2}{1+\lambda x'^2}.
\end{aligned}\right.
\end{equation*}
Suppose that we fix $\dot \alpha-F(t)\alpha=0$. Hence, the latter becomes
\begin{equation*}
\left\{
\begin{aligned}
\dot x'&=\alpha(t)v',\\
\dot v'&=-\frac{\omega(t)}{\alpha(t)}\frac{x'}{1+\lambda
x'^2}+\alpha(t)\frac{\lambda x'v'^2}{1+\lambda x'^2}.
\end{aligned}\right.
\end{equation*}
Let us try to search conditions for ensuring the above system to determine the
integral curves for a $t$-dependent vector field of the form $X(t,x)=f(t)\bar
X(x)$ with $\bar X\in V$, e.g.
\begin{equation*}
\left\{
\begin{aligned}
\dot x'&=f(t)v',\\
\dot v'&=f(t)\left(\frac{x'}{1+\lambda x'^2}+\frac{\lambda x'v'^2}{1+\lambda
x'^2}\right).
\end{aligned}\right.
\end{equation*}
In such a case, $\alpha(t)=f(t)$, $\omega(t)=-\alpha^2(t)$ and therefore
$\omega(t)=-\exp\left(2\int F(t)dt\right)$. The $t$-reparametrisation
$d\tau=f(t)dt$ transforms the previous system into the autonomous one
\begin{equation*}
\left\{
\begin{aligned}
\frac{dx'}{d\tau}&=v',\\
\frac{dv'}{d\tau}&=\frac{x'}{1+\lambda x'^2}+\frac{\lambda x'v'^2}{1+\lambda
x'^2}.
\end{aligned}\right.
\end{equation*}
determining the integral curves for the vector field $X=X_1+X_2$ and related to
a Mathews--Lakshmanan oscillator (\ref{ML}) with $k=1$.
The method of characteristics shows, after brief calculation, that this system
has a first-integral
$$
I(x',v')=\frac{1+\lambda x'^2}{1+\lambda v'^2},
$$
that reads in terms of the initial variables and the variable $t$ as a
$t$-dependent constant of the motion
$$
I(t,x,v)=\frac{\alpha^2(t)+\lambda \alpha^2(t)x^2}{\alpha^2(t)+\lambda v^2},
$$
for the $t$-dependent dissipative Mathews--Lakshmanan oscillator (\ref{DML})
getting a, as far as we know, new $t$-dependent constant of the motion.

\section{The Emden equation}
In this and following sections we analyse, from the perspective of the theory of
quasi-Lie schemes, the so-called Emden equations of the form
\begin{equation}\label{Emden}
\ddot x=a(t)\dot x+b(t)x^n,\quad n\ne 1.
\end{equation}
These equations can be  associated with the system of first-order differential
equations
\begin{equation}\label{Emdsys}\left\{
\begin{array}{rcl}
\dot x&=&v,\\
\dot v&=&a(t)v+b(t)x^n.
\end{array}\right. 
\end{equation}

This system was already studied in \cite{CGL08,CLL09Emd} by means of quasi-Lie
schemes. We hereafter summarise some of the results of these papers, which concern the determination of
$t$-dependent constants of the motion by means of particular solutions,
reducible particular cases of Emden equations, etc.
 
Consider the real vector space, $V_{{\rm Emd}}$, spanned by the vector fields
\begin{equation*}
X_1=x\frac{\partial}{\partial v},\quad X_2=x^n\frac{\partial}{\partial
{v}},\quad X_3=v\frac{\partial}{\partial {x}},\quad
X_4=v\frac{\partial}{\partial v},\quad X_5=x\frac{\partial}{\partial x}.
\end{equation*}
The $t$-dependent vector field determining the dynamics of system
(\ref{Emdsys}) can be written as a linear combination
$$X_t=a(t)X_4+X_3+b(t)X_2.
$$
Moreover, the linear space $W_{{\rm Emd}}\subset V_{{\rm Emd}}$ spanned by the
complete vector fields,
\begin{equation*}
Y_1=X_4=v\frac{\partial}{\partial {v}},\quad Y_2=X_1=x\frac{\partial}{\partial
{v}},\quad Y_3=X_5=x\frac{\partial}{\partial {x}},
\end{equation*}
  is a three-dimensional real Lie algebra of vector fields with respect to the
  ordinary Lie Bracket because these vector fields satisfy the relations 
\begin{equation*}
\begin{aligned}
\left[Y_1,Y_2\right]_{LB} &=-Y_2,\quad\left[Y_1,Y_3\right]_{LB} &=0,\quad
\left[Y_2,Y_3\right]_{LB} &=-Y_2.
\end{aligned}
\end{equation*}
Also $[W_{{\rm Emd}},V_{{\rm Emd}}]_{LB}\subset V_{{\rm Emd}}$ because
\begin{equation*}
\begin{array}{lll}
\left[Y_1,X_2\right]_{LB} =-X_2,& \left[Y_1,X_3\right]_{LB} =X_3,&
\left[Y_2,X_2\right]_{LB} =0,\cr
\left[Y_2,X_3\right]_{LB} =X_5-X_4,&  \left[Y_3,X_2\right]_{LB} =nX_2,&
\left[Y_3,X_3\right]_{LB} =-X_3.
\end{array}
\end{equation*}
So we get a quasi-Lie scheme $S(W_{{\rm Emd}},V_{{\rm Emd}})$ which can be used
to treat the Emden equations (\ref{Emdsys}). This suggests that if we perform
the $t$-dependent change of variables 
associated with this
quasi-Lie scheme, namely,
\begin{equation}\label{transfor3}
\left\{
\begin{array}{rcl}
x&=&\gamma(t)x',\\
v&=&\beta(t)v'+\alpha(t)x',
\end{array}\right. \quad \gamma(t) \,\beta(t)> 0\,,\forall t,
\end{equation}
the original system transforms  into
{\small
\begin{equation}\label{transformed2}\left\{
\begin{aligned}\frac{dx'}{dt}&=\left(\frac{\alpha(t)}{\gamma(t)}-\frac{\dot
\gamma(t)}{\gamma(t)}\right)x'+\frac{\beta(t)}{\gamma(t)}v',\\
\frac{dv'}{dt}&=\left(a(t)-\frac{\alpha(t)}{\gamma(t)}-\frac{\dot\beta(t)}{
\beta(t)}\right)v'+\frac{\alpha(t)}{\beta(t)}\left(a(t)-\frac{\alpha(t)}{
\gamma(t)}-\frac{\dot\alpha(t)}{\alpha(t)}+\frac{\dot\gamma(t)}{\gamma(t)}
\right)x'\\
& \quad +\frac{b(t)\gamma^n(t)}{\beta(t)}x'{}^n.
\end{aligned}\right. 
\end{equation}}

The key point of our method is to choose appropriate functions, $\alpha$,
$\beta $ and $\gamma$, in such a way that the system of differential equations
(\ref{transformed2})  becomes
a Lie system. A possible way to do so, consists in choosing $\alpha, \beta$ and
$\gamma$ so that the above system becomes determined by a $t$-dependent vector
field $X_t=f(t)\bar X$, where $\bar X$
is a true vector field and $f(t)$ is a non-vanishing function (on the interval
of $t$ under study). As it is shown in next section, this cannot always be  done
and some conditions must be imposed on the initial $t$-dependent
functions, $\alpha, \beta$ and $\gamma$, ensuring the existence of such a
transformation. These restrictions lead to integrability conditions.

Suppose that, for the time being, this is the case. Therefore, the system
(\ref{transformed2}) is
\begin{equation}\label{trans4}\left\{
\begin{aligned}
\frac{dx'}{dt}&=f(t)\left(c_{11}x'+c_{12}v'\right),\\
\frac{dv'}{dt}&=f(t)(c_{22}x'{}^n+c_xx'+c_{21}v')
\end{aligned}\right. 
\end{equation}
and it is determined by the $t$-dependent vector field
$$
X_t=f(t)\bar X,
$$
with
$$
\bar{X}=(c_{11}x'+c_{12}v')\frac{\partial}{\partial
{x'}}+(c_{22}x'^n+c_xx'+c_{21}v')\frac{\partial}{\partial {v'}}.
$$
Under the $t$-reparametrisation,
$$
\tau=\int^t f(t')dt',
$$
system (\ref{trans4}) is autonomous.
The new autonomous system of differential equations is determined by the vector
field $\bar X$ on ${\rm T}\mathbb{R}$ and therefore there exists a first
integral. This can be obtained by means of the method of characteristics,
which provides the characteristic curves where the first-integrals for
such a vector field
$\bar X$  are constant. These characteristic curves are determined by  
$$
\frac{dx'}{c_{11}x'+c_{12}v'}=\frac{dv'}{c_{21}v'+c_xx'+c_{22}x'{}^n},
$$
which can be written as
\begin{equation}\label{eq21}
(c_{21}v'+c_xx'+c_{22}x'{}^n)dx'-(c_{11}x'+c_{12}v')dv'=0.
\end{equation}
This expression can be straightforwardly integrated if             
\begin{equation}\label{CInt}
\pd{}{v'}(c_{21}v'+c_xx'+c_{22}x'{}^n)=-\pd{}{x'}(c_{11}x'+c_{12}
v')\Longrightarrow c_{21}=-c_{11}.
\end{equation}
Under this condition  we obtain the integral of the motion for (\ref{eq21}),
namely
\begin{equation}\label{IntOfMot}
I=-c_{12}\frac{v'{}^2}{2}+c_x\frac{x'{}^2}{2}+c_{21}v'x'+c_{22}\frac{x'{}^{n+1}}
{n+1}.
\end{equation}
Finally, if we write the latter expression in terms of the initial variables
$x,v$ and $t$, we get a constant  of the motion for the initial differential
equation.

If we do not wish to impose condition (\ref{CInt}), we can alternatively
integrate equation (\ref{eq21}) by means of an integrating
factor, i.e. we look for a function, $\mu(x',v')$, such that
\begin{equation*}
\pd{}{v'}\left(\mu(c_{21}v'+c_xx'+c_{22}x'^n)\right)=\pd{}{x'}(-\mu(c_{11}x'+c_{
12}v')).
\end{equation*}
Thus the integrating factor satisfies the partial differential equation
$$
\pd{\mu}{v'}(c_{21}v'+c_xx'+c_{22}x'^n)+\pd{\mu}{x'}(c_{11}x'+c_{12}v')=-\mu(c_{
11}+c_{21}).
$$
If $c_{11}+c_{21}=0$, the integral factor can be chosen to be $\mu=1$ and we get
the latter first-integral (\ref{IntOfMot}). On the other hand, if
$c_{11}+c_{21}\neq 0$, we can still look for a solution for the partial
differential equation for $\mu$ and obtain a new first-integral.

\section{{\it t}-dependent constants of the motion and particular solutions for
Emden equations}

The main purpose of this section is to show that the knowledge of a particular
solution of the Emden
equation allows us to transform it into a Lie system and to derive a
$t$-dependent constant of the motion. 
 
If we restrict ourselves to the case $\alpha(t)=0$ in the system of differential
equation
(\ref{transformed2}), it reduces to
\begin{equation}\label{eq8Appl}\left\{
\begin{aligned}
\frac{dx'}{dt}&=-\frac{\dot
\gamma(t)}{\gamma(t)}x'+\frac{\beta(t)}{\gamma(t)}v',\\
\frac{dv'}{dt}&=\left(a(t)-\frac{\dot\beta(t)}{\beta(t)}\right)v'+\frac{
b(t)\gamma^n(t)}{\beta(t)}x'{}^n.
\end{aligned}\right.
\end{equation}

In order to transform the original Emden--Fowler differential equation
 into a Lie
system by means of our quasi-Lie scheme, we try to write the transformed 
differential equation in the form
\begin{equation}\label{eq10}
\left\{
\begin{aligned}
\frac{dx'}{dt}&=f(t)\left(c_{11}x'+c_{12}v'\right),\\
\frac{dv'}{dt}&=f(t)\left(c_{22}x'{}^n+c_{21}v'\right),
\end{aligned}\right.
\end{equation}
where the $c_{ij}$ are constants. This system of differential equations can be 
reduced to an autonomous one as, under the $t$-dependent change of variables, 
$$
\tau=\int^t f(t')dt',
$$
the latter differential equation becomes
\begin{equation}\label{eq11}
\left\{
\begin{aligned}
\frac{dx'}{d\tau}&=c_{11}x'+c_{12}v',\\
\frac{dv'}{d\tau}&=c_{22}x'{}^n+c_{21}v'.
\end{aligned}\right.
\end{equation}
In order for system (\ref{eq8Appl}) to be similar to system (\ref{eq10}), we
look for functions $\alpha$, $\beta$ and 
$\gamma$ satisfying the conditions,
\begin{equation}\label{Relations}
\left\{
\begin{aligned}
f(t)\,c_{11}&=-\frac{\dot\gamma(t)}{\gamma(t)}, \qquad
&f(t)\,c_{12}&=\frac{\beta(t)}{\gamma(t)},\\
f(t)\,c_{22}&= b(t)\frac{\gamma^n(t)}{\beta(t)}, \qquad &
f(t)\,c_{21}&=a(t)-\frac{\dot\beta(t)}{\beta(t)}.
\end{aligned} \right.
\end{equation}
The conditions in the first line lead to
\begin{equation}\label{bet}
\beta(t)=-\frac{c_{12}}{c_{11}}\dot\gamma(t),
\end{equation}
and using this equation in the last relation we obtain
\begin{equation}\label{funo}
f(t)=\frac{a(t)}{c_{21}}-\frac{1}{c_{21}}\frac{\ddot\gamma(t)}{\dot\gamma(t)}.
\end{equation}
On the other hand from the three first relations in (\ref{Relations}) we get 
\begin{equation}\label{fdos}
f(t)=-\frac{b(t)c_{11}}{c_{22}c_{12}}\frac{\gamma^n(t)}{\dot\gamma(t)}.
\end{equation}

The equality of the right-hand sides of (\ref{funo}) and (\ref{fdos}) leads to
the following equation
 for the function $\gamma$:
\begin{equation*}
\ddot\gamma=a(t)\dot\gamma+\frac{c_{11}c_{21}}{c_{22}c_{12}}b(t)\gamma^n.
\end{equation*}
Suppose that we make the choice, with $c_{21}=-c_{11}$ as indicated in
(\ref{CInt}),
\begin{equation}\label{cschoice}
c_{22}=-1,\quad c_{11}=1,\quad c_{21}=-1,\quad c_{12}=1
\end{equation}
and thus $(c_{11}c_{22})/(c_{21}c_{12})=1$. Therefore we find that $\gamma$ must
be a
solution of the initial equation (\ref{Emden}). In other words, if we suppose 
 that a particular solution $x_p(t)$ of the Emden equation is known, we can
 choose $\gamma(t)=x_p(t)$. 
Then, according to  the expression (\ref{bet}) and our previous choice
(\ref{cschoice}), the corresponding function $\beta$ turns out to be 
$$\beta(t)=-\dot x_p(t).
$$ 
Finally, in view of conditions (\ref{Relations}), we get that
$$
\frac{-\dot\gamma(t)}{c_{11}\gamma(t)}=b(t)\frac{\gamma^ n(t)}{c_{22}\beta(t)}
$$
and taking into account our choice (\ref{cschoice}) and $\gamma(t)=x_p(t)$, we
obtain the condition satisfied by the particular solution:
\begin{equation}\label{IntCond2}
x^{n+1}_p(t)={\dot x}^2_p(t). 
\end{equation}
The system of differential equations (\ref{eq10}) for such a choice
(\ref{cschoice}) of the constants
$\{c_{ij}\,|\,i,j=1,2\}$ is the equation for the integrals curves for the
$t$-dependent vector field
$$
X_t=f(t)\left(\left(x'+\,v'\right)\frac{\partial}{\partial
x'}-\left(v'+\,{x'}^n\right)
\frac{\partial}{\partial v'}\right).
$$
The method of the characteristics can be used to find the following
 first-integral for this vector field and, in view of (\ref{IntOfMot}), we get
\begin{equation*}
\left\{\begin{aligned}
I(x',v')&=\frac{1}{n+1}x'{}^{n+1}+\frac{1}{2}v'{}^2+x'v',\qquad &n&\notin \{
-1,1\},\\
I(x',v')&={\rm log}\,x'+\frac{1}{2}v'{}^2+x'v',\qquad &n&= -1,
\end{aligned}\right.
\end{equation*}
and, if we express this integral of motion in terms of the initial variables and
$t$, we obtain a, as far as we know, new $t$-dependent constant of the motion
for the initial
Emden equation
\begin{equation}\label{Integral}
\left\{\begin{aligned}
I(t,x,v)&=\frac{x^{n+1}}{(n+1)x_p^{n+1}(t)}+\frac{v^2}{2\dot
x_p^2(t)}-\frac{xv}{x_p(t)\dot x_p(t)},\quad &n&\notin\{ -1,1\},\\
I(t,x,v)&={\rm log}\left(\frac{x}{x_p(t)}\right)+\frac{v^2}{2\dot
x_p^2(t)}-\frac{xv}{x_p(t)\dot x_p(t)},\quad &n&= -1.
\end{aligned}\right.
\end{equation}
So, the knowledge of a particular solution for the Emden equation enables us
first to
obtain a constant of the motion and then to reduce the initial Emden equation
into a Lie
system. Thus, all Emden equations are quasi-Lie
 systems with respect to the above mentioned  scheme. 

\section{Applications of particular solutions to study Emden equations}

This section is devoted to illustrating the usefulness of the previous theory
about Emden equations. More specifically, we detail several Emden
 equations for which one is able to find a particular solution satisfying an
integrability condition and  use 
is made of such a solution in order to derive $t$-dependent constants of the
motion.
 In this way we recover several results appearing in the literature about
Emden--Fowler equations from a unified point of view \cite{CLL09Emd}.

We start with a particular case of the Lane-Emden equation
\begin{equation}
\ddot x=-\frac{2}{t}\dot x-x^5.\label{Emdenpart}
\end{equation}
The more general
Lane-Emden equation is generally written as
$$
\ddot x=-\frac{2}{t}\dot x+f(x)
$$ 
and the example here  considered  corresponds to $f(x)=-x^n, \, n\ne 1$, which
is  one of the most interesting cases, together with that of $f(x)=-e^{-\beta
x}$. Equation
(\ref{Emdenpart}) appears in the study of the thermal behaviour of a spherical
cloud of gas \cite{KMM} and  also  
 in  astrophysical applications. A particular solution for (\ref {Emdenpart})
satisfying (\ref{IntCond2}) is
 $x_p(t)=(2t)^{-1/2}$. If we substitute this expression for $x_p(t)$  and 
the corresponding one for $\dot x_p(t)$ into the $t$-dependent constant of the
motion
(\ref{Integral}), we get that
$$
I'(t,x,v)=\frac{4t^3x^{6}}{3}+4t^ 3v^2+4t^ 2xv
$$
is a $t$-dependent constant of the motion proportional to (\ref{Integral}) and
also proportional to the $t$-dependent constants of the motion found in
\cite{SB80,CGL08,Lo77}.

We study from this new perspective other Emden equations investigated in
\cite{Le85}. Consider the particular instance
$$
\ddot x=-\frac{5}{t+K}\dot x-x^2.
$$
A particular solution for this Emden equation satisfying (\ref{IntCond2}) is 
$$x_p(t)=\frac{4}{(t+K)^2}.
$$
In this case a $t$-dependent constant of the motion is
$$
I'(t,x,v)=\frac 13{x^3(t+K)^6}+\frac{1}{2}v^2(t+K)^6+2\,x\,v(t+K)^5,
$$
which is proportional to the one found by Leach in \cite{Le85}. 

Now another Emden equation found in \cite{Le85},
$$
\ddot x=-\frac{3}{2(t+K)}\dot x-x^9,
$$
admits the particular solution
$$
x_p(t)=\frac{1}{\sqrt{2}(t+K)^{1/4}},$$
which satisfies (\ref{IntCond2}). The corresponding $t$-dependent constant of
the motion is given by
$$
I'(t,x,v)=(K+t)^{3/2}(10(K+t)v^2+5 vx+2(K+t)x^{10})
$$
which is proportional to that given in \cite{Le85}. 

Let us turn now to consider the Emden equation
$$
\ddot x=-\frac{5}{3(t+K)}\dot x-x^7,
$$
which admits a particular solution of the form 
$$
x_p(t)=\frac{1}{3^{1/3}(t+K)^{1/3}},$$
which obeys (\ref{IntCond2}) and leads to the $t$-dependent constant of the
motion 
$$
I'(t,x,v)=(K+t)^{5/3}(12(K+t)v^2+8 vx+3x^8(K+t)).
$$

Finally we apply our development to obtain a $t$-dependent constant of the
motion for the Emden equation
\begin{equation}\label{eq}
\ddot x=-\frac{1}{K_1+K_3t}\dot x-x^n
\end{equation}
with
$$
K_3= \frac{n-1}{n+3}.$$
We can find a particular solution of the form
$$
x_p(t)=\frac{K_2}{(K_1+K_3t)^\nu}, \quad \nu\ne 0.
$$
In order for $x_p(t)$ to be a particular solution we must have the following
relation
$$
\frac{(\nu+1)\nu K_2K_3^2}{(K_1+K_3t)^{\nu+2}}=\frac{\nu
K_2K_3}{(K_1+K_3t)^{\nu+2}}-\frac{K_2^n}{(K_1+K_3t)^{n\nu}}
$$
and thus
$$
\nu+2=n\nu\qquad {\rm and}\qquad \nu(\nu+1)K_3^2K_2=\nu K_2 K_3-K_2^n.
$$
From these equations we get 
$$
\nu=\frac{2}{n-1},\qquad K_2^{n-1}=\frac{2^2}{(n+3)^2}.
$$
Under these conditions it can be easily verified that $\dot x_p^2(t)=x_p^
{n+1}(t)$. Thus, a $t$-dependent constant of the motion is
\begin{multline}\label{IntFinal}
I'(t,x,v)=(K_1+K_3t)^{2(n+1)/(n-1)}\left(\frac{x^{n+1}}{n+1}+\frac{v^2}{2}
\right)+ \\+(K_1+K_3t)^{(n+3)/(n-1)}\frac{2vx}{n+3}, \qquad\qquad\end{multline}
which can also be found in \cite{Le85}.

Another advantage of our method is that it allows us to obtain Emden equations
admitting a previously fixed   $t$-dependent constant of the motion. 

Suppose that we want to construct an Emden equation admitting a previously
chosen
particular solution, $x_p(t)$, satisfying ${\dot x}^2_p(t)=x^{n+1}_p(t)$ for
certain $n\in\mathbb{Z}-\{1,-1\}$. We can integrate this equation to get all
possible particular solutions which can be used by means of our method, i.e.
$$
x_p(t)=\left(K+\frac{1-n}{2}t\right)^{-\frac{2}{n-1}}.
$$
We consider functions $a(t)$ and $b(t)$ such that
$$
\ddot x_p=a(t)\dot x_p+b(t)x_p^n.
$$
For the sake of simplicity, we can assume that $b(t)=-1$. Then we get
$$
a(t)=\frac{\ddot x_p+x_p^n}{\dot x_p}.
$$
If we substitute the chosen particular solution in the above expression,
we obtain 
$$
a(t)=\frac{3+n}{2(K+\frac{1-n}{2}t)}.
$$
which leads to an Emden equation equivalent to (\ref{eq}) and the $t$-dependent
constant of the motion for this equation is again (\ref{IntFinal}). In this way
we recover the cases studied in this section.

\section{The Kummer-Liouville transformation for a general Emden-Fowler
equation}

As far as we know, the most general form of the Emden--Fowler equation
considered nowadays is
\begin{equation}\label{GEFeq}
\ddot x+p(t)\dot x+q(t)x=r(t)x^n.
\end{equation}
This generalisation arises naturally as a consequence of our scheme. Indeed, the
above second-order differential equation is associated with the system of
first-order differential equations
\begin{equation}\label{FordGEFeq}
\left\{\begin{aligned}
\dot x&=v,\\
\dot v&=-p(t)v-q(t)x+r(t)x^n,
\end{aligned}\right.
\end{equation}
which determines the integral curves for the $t$-dependent vector field
$$
X_t=-p(t)X_4-q(t)X_1+r(t)X_2+X_3.
$$
This $t$-dependent vector field is a generalisation of the one studied in
previous sections.
Under the set of transformations (\ref{transfor3}),  the initial system
(\ref{FordGEFeq})
becomes the new system
{\small
\begin{equation*}\left\{
\begin{aligned}\frac{dx'}{dt}&=\left(\frac{\alpha(t)}{\gamma(t)}-\frac{\dot
\gamma(t)}{\gamma(t)}\right)x'+\frac{\beta(t)}{\gamma(t)}v',\\
\frac{dv'}{dt}&=\left(-p(t)-\frac{\alpha(t)}{\gamma(t)}-\frac{\dot\beta(t)}{
\beta(t)}\right)v'+
\frac{\alpha(t)}{\beta(t)}\left(-p(t)-\frac{\alpha(t)}{\gamma(t)}-\frac{
\dot\alpha(t)}{\alpha(t)}+\right.\\
&\left.+\frac{\dot\gamma(t)}{\gamma(t)}-q(t)\frac{\gamma(t)}{\alpha(t)}
\right)x'+\frac{r(t)\gamma^n(t)}{\beta(t)}x'{}^n.
\end{aligned}\right.
\end{equation*}}
If we choose $\alpha=\dot \gamma$,  the system reduces to
{\small
\begin{equation*}\left\{
\begin{aligned}\frac{dx'}{dt}&=\frac{\beta(t)}{\gamma(t)}v',\\
\frac{dv'}{dt}&=\left(-p(t)-\frac{\dot
\gamma(t)}{\gamma(t)}-\frac{\dot\beta(t)}{\beta(t)}\right)v'+\frac{\dot
\gamma(t)}{\beta(t)}\left(-p(t)-\frac{\ddot\gamma(t)}{\dot\gamma(t)}-q(t)\frac{
\gamma(t)}{\dot\gamma(t)}\right)x'\\&+\frac{r(t)\gamma^n(t)}{\beta(t)}x'{}^n.
\end{aligned}\right. 
\end{equation*}}
When the function $\gamma(t)$ is chosen in such a way that $\ddot
\gamma=-q(t)\gamma
-p(t)\dot\gamma$,
i.e. $\gamma$ is a solution of the associated linear equation, we obtain
{\small
\begin{equation}
\left\{
\begin{aligned}
\frac{dx'}{dt}&=\frac{\beta(t)}{\gamma(t)}v',\\
\frac{dv'}{dt}&=\left(-p(t)-\frac{\dot \gamma(t)}{\gamma(t)}-\frac{\dot
\beta(t)}{\beta(t)}\right)v'+\frac{r(t)\gamma^n(t)}{\beta(t)}x'{}^n.
\end{aligned}\right.
\end{equation}}
Finally, if the function $\beta(t)$ is such that
$$
-p(t)-\frac{\dot \gamma(t)}{\gamma(t)}-\frac{\dot \beta(t)}{\beta(t)}=0,
$$
we obtain
{\small
\begin{equation}
\left\{
\begin{aligned}
\frac{dx'}{dt}&=\frac{\beta(t)}{\gamma(t)}v',\\
\frac{dv'}{dt}&=\frac{r(t)\gamma^{n}(t)}{\beta(t)}x'{}^n,
\end{aligned}\right.
\end{equation}}which is related to the second-order differential equation
$$
\frac{\delta^2x'}{d\tau ^2}=r(t)\frac{\gamma^{n+1}(t)}{\beta^2(t)}x'{}^n,$$
with
$$
\tau{(t)}=\int^t\frac{\beta(t')}{\gamma(t')}dt'.
$$
 The new form of the differential equation is called the canonical form of the
 generalised Emden--Fowler equation.

This fact is obtained by means of an appropriate
 Kummer--Liouville transformation  in the previous literature, but we obtain it
here as a straightforward application of the properties of transformation of
quasi-Lie schemes thereby underscoring the theoretical explanation of such a
Kummer--Liouville transformation.

\section{Constants of the motion for sets of Emden-Fowler equations}
In this section we show that under certain assumptions on the $t$-dependent
coefficients $a(t)$ and $b(t)$  the original Emden equation can be reduced 
to a Lie system and then we can obtain a first-integral  which
provides  us with a $t$-dependent constant of the motion for the original
system. 

In fact consider the system of first-order differential equations
{\small
\begin{equation*}\left\{
\begin{aligned}\frac{dx'}{dt}&=\left(\frac{\alpha(t)}{\gamma(t)}-\frac{\dot
\gamma(t)}{\gamma(t)}\right)x'+\frac{\beta(t)}{\gamma(t)}v',\\
\frac{dv'}{dt}&=\left(a(t)-\frac{\alpha(t)}{\gamma(t)}-\frac{\dot\beta(t)}{
\beta(t)}\right)v'
+\frac{\alpha(t)}{\beta(t)}\left(a(t)-\frac{\alpha(t)}{\gamma(t)}-\frac{
\dot\alpha(t)}{\alpha(t)}+\frac{\dot\gamma(t)}{\gamma(t)}\right)x'\\
&+\frac{b(t)\gamma^n(t)}{\beta(t)}x'{}^n.
\end{aligned}\right. 
\end{equation*}}
This system describes all the systems of differential equations that can be
obtained by means of the set of $t$-dependent transformations we got through the
scheme $S(W_{Emd},V_{Emd})$. We recall that the $t$-dependent change of variable
which we use to relate the Emden equation (\ref{Emdsys}) with the latter system
of differential equation is 
\begin{equation*}
\left\{\begin{aligned}
x&=        \gamma(t)x',\\
v&=\beta(t)v'+\alpha(t)x'.
       \end{aligned}	\right.
\end{equation*}
As in previous papers on this topic, we try to relate the latter system of
differential equations to a Lie system determined by a $t$-dependent vector
field of the form $X'(t,x)=f(t)\bar X(x)$ and we suppose $f(t)$ to be
non-vanishing in the interval we study. So the system of differential equations
determining the integrals curves for this $t$-dependent vector field is a Lie
system and we can use the theory of Lie systems to analyse its properties.

As a first example we can consider that we just use the set of transformations
with $\gamma(t)=1$ and $\alpha(t)=0$. In this case system
(\ref{transfor3}) is
\begin{equation*}\left\{
\begin{aligned}\frac{dx'}{dt}&=\beta(t)v'\\
\frac{dv'}{dt}&=\left(a(t)-\frac{\dot\beta(t)}{\beta(t)}\right)v'+\frac{b(t)}{
\beta(t)}x'{}^n.
\end{aligned}\right.
\end{equation*}
We fix $\beta(t)$ to  be such that
$$
a(t)-\frac{\dot\beta(t)}{\beta(t)}=0,
$$
i.e. $\beta(t)$ is (proportional to) 
$$
\beta(t)=\exp\left(\int^ta(t')dt'\right).
$$
Therefore we get
\begin{equation*}\left\{
\begin{aligned}\frac{dx'}{dt}&=\exp\left(\int^ta(t')dt'\right)v',\\
\frac{dv'}{dt}&=b(t)\exp\left(-\int^ta(t')dt'\right)x'{}^n.
\end{aligned}\right. 
\end{equation*}
In order to get the last system of differential equations to describe the
integral curves for a $t$-dependent vector field, $X'(t,x)=f(t)\bar X(x)$, for
a given  function $a(t)$ a necessary and sufficient condition is
$$
b(t)\exp\left(-2\int^ta(t')dt'\right)=K,
$$
with $K$ being a  real constant. Under this assumption the last system becomes
\begin{equation*}\left\{
\begin{aligned}\frac{dx'}{dt}&=\exp\left(\int^ta(t')dt'\right)v',\\
\frac{dv'}{dt}&=\exp\left(\int^ta(t')dt'\right)K x'{}^n.
\end{aligned}\right. 
\end{equation*}
We introduce the $t$-reparametrisation
$$
\tau(t)=\int^t\exp\left(\int^{t'}a(t'')dt''\right)dt'
$$
and the latter system becomes
\begin{equation*}\left\{
\begin{aligned}\frac{dx'}{d\tau}&=v',\\
\frac{dv'}{d\tau}&=K x'{}^n,
\end{aligned}\right. 
\end{equation*}
which admits a first-integral
$$
I=\frac 12v'{}^2-K\frac{x'{}^{n+1}}{n+1}.
$$
In terms of the initial variables, the corresponding $t$-dependent constant  of
the motion is
$$
I=\exp\left(-2\int^ta(t')dt'\right)\left(\frac 12\dot
y^2-b(t)\frac{x^{n+1}}{n+1}\right),
$$
which is similar to that found in \cite{BV91}.

Suppose that we restrict the transformations (\ref{transfor3}) to the
case $\alpha(t)=0$. In this case the system of first-order differential
equations (\ref{transformed2}) becomes
\begin{equation*}\left\{
\begin{aligned}\frac{dx'}{dt}&=-\frac{\dot
\gamma(t)}{\gamma(t)}x'+\frac{\beta(t)}{\gamma(t)}v',\\
\frac{dv'}{dt}&=\left(a(t)-\frac{\dot\beta(t)}{\beta(t)}\right)v'+\frac{
b(t)\gamma^n(t)}{\beta(t)}x'{}^n.
\end{aligned}\right. 
\end{equation*}
In order for this system of differential equations to determine the integral
curves for a $t$-dependent vector field of the form $X'(t,x)=f(t)\bar X(x)$ we
need that
\begin{equation}\label{Rel}
\left\{\begin{aligned}
c_{11}f(t)&=-\frac{\dot \gamma(t)}{\gamma(t)},\quad
&c_{12}f(t)&=\frac{\beta(t)}{\gamma(t)},\\
c_{21}f(t)&=a(t)-\frac{\dot\beta(t)}{\beta(t)}, \quad
&c_{22}f(t)&=\frac{b(t)\gamma^n(t)}{\beta(t)}.\\
\end{aligned}\right.
\end{equation}

From these relations, or more exactly from those of the first row, we get $f(t)$
as
$$
f(t)=-\frac{1}{c_{11}}\frac{\dot
\gamma(t)}{\gamma(t)}=\frac{1}{c_{12}}\frac{\beta(t)}{\gamma(t)}
$$
and therefore
$$
\dot \gamma(t)=-\frac{c_{11}}{c_{12}}\beta(t).
$$
We choose  $c_{11}=-1$ and $c_{12}=1$ so that
\begin{equation}\label{beta}
\beta(t)=\dot \gamma(t).
\end{equation}
In view of this and using the third and second relations from (\ref{Rel}) we get
$$
\frac{c_{21}}{c_{12}}\frac{\beta(t)}{\gamma(t)}=a(t)-\frac{\dot\beta(t)}{
\beta(t)}
$$
and thus, as a consequence of  (\ref{beta}), the last differential equation
becomes
$$
\frac{c_{21}}{c_{12}}\frac{\dot\gamma(t)}{\gamma(t)}=a(t)-\frac{\ddot\gamma(t)}{
\dot\gamma(t)}
$$
and, as $c_{12}=1$ and fixing $c_{21}=1$, we obtain
$$
\frac{d}{dt}\log(\dot\gamma\gamma)=a(t),
$$
which can be rewritten as
$$
\frac{1}{2}\frac{d}{dt}\gamma^2(t)=\exp\left(\int^ta(t')dt'\right).
$$
Hence we have
$$
\gamma(t)=\sqrt{2\int^t\exp\left(\int^{t'}a(t'')dt''\right)dt'}
$$
and in view of (\ref{beta})  
$$
\beta(t)=\frac{1}{\sqrt{2\int^t\exp\left(\int^{t'}a(t'')dt''\right)dt'}}
\exp\left(\int^{t}a(t')dt'\right).
$$

So far we have only used three of the four relations we found. The fourth and
second
relations lead  to the integrability condition: there exist
a constant $c_{22}=K$ such that 
$$
K\frac{\beta(t)}{\gamma(t)}=\frac{b(t)\gamma^n(t)}{\beta(t)}.
$$
Therefore, using the above  expressions for $\gamma(t)$ and $\beta(t)$, we get
\begin{equation}\label{Inte}
b(t)\exp\left(-2\int^ta(t)dt'\right)\left(2\int^t\exp\left(\int^{t'}
a(t'')dt''\right)\right)^{(n+3)/2}=K.
\end{equation}

So under this assumption we have connected the initial Emden equation with the
Lie system,
\begin{equation*}\left\{
\begin{aligned}\frac{dx'}{dt}&=f(t)(-x'+v'),\\
\frac{dv'}{dt}&=f(t)(v'+Kx'{}^n),
\end{aligned}\right. 
\end{equation*}
and then  the method of characteristics shows that it  admits the first-integral
$$
I'=-\frac{1}{2}v'^2+\frac{K}{n+1}x'{}^{n+1}+v'x'.
$$
In terms of the initial variables the corresponding constant of the motion is
\begin{multline}
I=\left(\frac 12\dot
x^2-\frac{b(t)}{n+1}x^{n+1}
\right)\exp\left(-2\int^ta(t')dt'\right)\int^t\exp\left(\int^{t'}
a(t'')dt''\right)dt'\\
-\frac{1}{2}x\dot x \exp\left(-\int^t a(t')dt'\right)
\end{multline}
and in this way we recover the result found in \cite{BV91}. If we now consider
the particular case $n=-3$ we get that the integrability condition (\ref{Inte})
implies that there is a constant $K$ such that 
\begin{equation*}
b(t)\exp\left(-2\int^ta(t)dt'\right)=K,
\end{equation*}
and the corresponding $t$-dependent constant  of the motion is then given by
\begin{multline*}
I=\left(\frac 12\dot
x^2+\frac{b(t)}{2}x^{-2}
\right)\exp\left(-2\int^ta(t')dt'\right)\int^t\exp\left(\int^{t'}
a(t'')dt''\right)dt'\\
-\frac{1}{2}x\dot x \exp\left(-\int^t a(t')dt'\right),\qquad\qquad
\end{multline*}
which is equivalent to that one found in \cite{BV91}.
\section{A {\it t}-dependent superposition rule for Abel equations}
Let us now turn to illustrate the results of our theory of Lie families by
deriving a common $t$-dependent superposition rule for a Lie family of Abel
equations, whose elements do not admit a standard superposition rule except for
a few particular instances. In this way, we single out that our theory provides
new tools for investigating solutions of nonautonomous systems of differential
equations than cannot be investigated by means of the theory of Lie systems.

With this aim, we analyse the so-called Abel equations of the first-type
\cite{Bo05,CR03}, i.e. the
differential equations of the form
\begin{equation}\label{Abel}
 \frac{dx}{dt}=a_0(t)+a_1(t)x+a_2(t)x^2+a_3(t)x^3,
\end{equation}
with $a_3(t)\neq 0$. Abel equations appear in the analysis of several
cosmological models \cite{CK86,HM04,MML08} and other different fields in Physics
\cite{CLS04,CLP01,Es02,GE04,SP03,ZT09}. Additionally, the study of integrability
conditions for Abel equations is a research topic of current interest in
Mathematics and multiple studies have been carried out in order to analyse the
properties of the solutions of these equations
\cite{Al07,MCH01,CR03,Ch40,ReStr82}.

Note that, apart from its inherent mathematical interest, the knowledge of
particular solutions of Abel equations allows us to study the properties of
those physical systems that such equations describe. Thus, the expressions
enabling us to easily obtain new solutions of Abel equations by means of several
particular ones, like common $t$-dependent superposition rules, are interesting
to study the solutions of these equations and, therefore, their related physical
systems.

Unfortunately, all the expressions describing the general solution of Abel
equations presently known can only be applied to study autonomous instances and,
moreover, they depend on families of particular conditions satisfying certain
extra conditions, see \cite{Ch40,ReStr82}. Taking this into account, common
$t$-dependent superposition rules represent an improvement with respect to these
previous expressions, as they enable us to treat nonautonomous Abel equations
and they do not require the usage of particular solutions obeying additional
conditions.

Recall that, according to Theorem \ref{MT}, the existence of a common
$t$-dependent superposition rule for a
family of $t$-dependent vector fields $\{Y_d\}_{d\in\Lambda}$ requires the
existence of a system of generators, i.e.
a certain set of $t$-dependent vector fields, $X_1,\ldots, X_r$, satisfying
relations (\ref{condition}).
Conversely, given such a set, the family of $t$-dependent vector fields $Y$
whose autonomisations can be
written in the form
$$\bar Y_c(t,x)=\sum_{j=1}^rb_{cj}(t)\bar X_j(t,x),\qquad \sum_{j=1}^rb_{c
j}(t)=1,$$
admits a common $t$-dependent superposition rule and becomes a Lie family.

Consequently, a Lie family of Abel equations can be determined, for instance, by
finding two $t$-dependent
vector fields of the form
\begin{equation}\label{ansatz}
\begin{aligned}
X_1(t,x)&=(b_0(t)+b_1(t)x+b_2(t)x^2+b_3(t)x^3)\frac{\partial}{\partial x},\\
X_2(t,x)&=(b'_0(t)+b'_1(t)x+b'_2(t)x^2+b'_3(t)x^3)\frac{\partial}{\partial
x},\qquad b'_3(t)\neq 0,
\end{aligned}
\end{equation}
such that
\begin{equation}\label{cond}
[\bar X_1,\bar X_2]=2(\bar X_2-\bar X_1).
\end{equation}

Let us analyse the existence of such two $t$-dependent vector fields $X_1$ and
$X_2$ with commutation relations
(\ref{cond}). In coordinates, the Lie bracket $[\bar X_1,\bar X_2]$ reads
\begin{multline*}
[(b_3'b_2-b_2'b_3)x^4+(2(b_3'b_1-b_3b_1')-\dot b_3+\dot
b_3')x^3+(-3(b_0'b_3-b_0b_3')+(b_2'b_1-b_2b_1')\\-\dot
b_2+\dot b_2')x^2+(-2b_0'b_2+2b_0b'_2-\dot b_1+\dot b_1')x-b_0'b_1+b_0b_1'-\dot
b_0+\dot
b_0']\frac{\partial}{\partial x}.
\end{multline*}
Hence, in order to satisfy condition (\ref{cond}), $b_3'b_2-b_2'b_3=0$, e.g. we
may fix $b_2=b_3=0$. Additionally, for the sake of simplicity, we assume
$b_3'=1$. In this case, the previous expression takes the form
\begin{equation*}
[2b_1x^3+(3b_0+b_2'b_1+\dot b_2')x^2+(2b_0b'_2-\dot b_1+\dot
b_1')x-b_0'b_1+b_0b_1'-\dot b_0+\dot
b_0']\frac{\partial}{\partial x},
\end{equation*}
and, taking into account the values chosen for $b_2$, $b_3$ and $b_3'$,
assumption (\ref{cond}) yields $b_1=1$
and
\begin{equation*}\left\{
\begin{aligned}
 b_2'&=3b_0+\dot b_2',\\
 2(b_1'-1)&=2b_0b'_2+\dot b_1',\\
2(b_0'-b_0)&=-b_0'+b_0b_1'-\dot b_0+\dot b_0'.
\end{aligned}\right.
\end{equation*}
As this system has more variables than equations, we can try to fix some values
of the variables in order
to simplify it and obtain a particular solution. In this way, taking $b_0(t)=t$,
the above system reads
\begin{equation*}\left\{
\begin{aligned}
 \dot b_2'&=b_2'-3t,\\
 \dot b_1'&=2(b_1'-1)-2tb'_2,\\
\dot b_0'&=2(b_0'-t)+b_0'-tb_1'+1.
\end{aligned}\right.
\end{equation*}
This system is integrable by quadratures and one can check that it admits the
particular solution
$$b_2'(t)=3(1+t),\quad  b_1'(t)=3(1+t)^2+1,\quad  b_0'(t)=(1+t)^3+t.$$
Summing up, we have proved that the $t$-dependent vector fields
\begin{equation}\label{fami}
\left\{\begin{aligned} X_1(t,x)&=(t+x)\frac{\partial}{\partial x},\\
X_2(t,x)&=((1+t)^3+t+(3(1+t)^2+1)x+3(1+t)x^2+x^3)\frac{\partial}{\partial x},
\end{aligned}\right.
\end{equation}
satisfy (\ref{cond}) and, therefore, the family of $t$-dependent vector fields
$$Y_{b(t)}(t,x)=(1-b(t))X_1(x)+b(t)X_2(x)$$
is a Lie family. The corresponding family of Abel equations is
\begin{equation}\label{LieFamily}
\frac{dx}{dt}=(t+x)+b(t)(1+t+x)^3.
\end{equation}
According to the results proved in Section \ref{GLT}, in order to determine a
common $t$-dependent
superposition rule for the above Lie family, we have to determine a
first-integral for the vector fields of the
distribution $\mathcal{D}$ spanned by the $t$-prolongations $\widetilde X_1$ and
$\widetilde X_2$ on
$\mathbb{R}\times\mathbb{R}^{n(m+1)}$ for a certain $m$ so that the
$t$-prolongations of $X_1$ and $X_2$ to
$\mathbb{R}\times \mathbb{R}^{nm}$ are linearly independent at a generic point.
Taking into account
expressions (\ref{fami}), the prolongations of the vector fields $X_1$ and $X_2$
to
$\mathbb{R}\times\mathbb{R}^2$ are linearly independent at a generic point and,
in view of (\ref{cond}), the
$t$-prolongations $\widetilde X_1$ and $\widetilde X_2$ to
$\mathbb{R}\times\mathbb{R}^3$ span an involutive
generalised distribution $\mathcal{D}$ with two-dimensional leaves in a dense
subset of
$\mathbb{R}\times\mathbb{R}^3$. Finally, a first-integral for the vector fields
in the distribution
$\mathcal{D}$ will provide us a common $t$-dependent superposition rule for the
Lie family (\ref{LieFamily}).

Since, in view of (\ref{cond}), the vector fields $\widetilde X_1$ and
$\widetilde X_2$ span the distribution
$\mathcal{D}$, a function $G:\mathbb{R}\times\mathbb{R}^2\rightarrow \mathbb{R}$
is a first-integral of the
vector fields of the distribution $\mathcal{D}$ if and only if $G$ is a
first-integral of $\widetilde X_1$ and
$\widetilde X_1-\widetilde X_2$, i.e. $\widetilde X_1G=(\widetilde
X_2-\widetilde X_1)G=0$.

The condition $\widetilde X_1G=0$ reads
$$
\frac{\partial G}{\partial t}+(t+x_0)\frac{\partial G}{\partial
x_0}+(t+x_1)\frac{\partial G}{\partial x_1}=0,
$$
and, using the method of characteristics \cite{MetChar}, we note that the curves
on which $G$ is constant, the
so-called {\it characteristics}, are solutions of the system
$$
dt=\frac{dx_0}{t+x_0}=\frac{dx_1}{t+x_1}\Rightarrow \frac{dx_i}{dt}=t+x_i,\qquad
i=0,1,
$$
which read $x_i(t)=\xi_ie^{t}-t-1$, with $i=0,1$ and $\xi_0,\xi_1\in\mathbb{R}$.
Furthermore, these solutions are determined by the implicit equations
$\xi_0=e^{-t}(x_0+t+1)$ and $\xi_1=e^{-t}(x_1+t+1)$. Therefore, there exists
a function $G_2:\mathbb{R}^2\rightarrow \mathbb{R}$ such that
$G(t,x_0,x_1)=G_2(\xi_0,\xi_1)$. In other words,
each first-integral $G$ of $\widetilde X_1$ depends only on $\xi_0$ and $\xi_1$.

Taking into account the previous fact, we look for simultaneous first-integrals
of the vector field $\widetilde
X_2-\widetilde X_1$ and $\widetilde X_1$, that is, for solutions of the equation
$(\widetilde X_2-\widetilde X_1)G=0$ with $G$ depending on $\xi_0$ and $\xi_1$.
Using the expression of
$\widetilde X_2-\widetilde X_1$ in the system of coordinates $\{t,
\xi_0,\xi_1\}$, we get that
$$
(\widetilde X_2-\widetilde X_1)G=\xi_0^3\frac{\partial
G_2}{\partial \xi_0}+\xi_1^3\frac{\partial G_2}{\partial \xi_1}=0,
$$
and, applying again the method of characteristics, we obtain that there exists a
function
$G_3:\mathbb{R}\rightarrow\mathbb{R}$ such that
$G(t,x_0,x_1)=G_2(\xi_0,\xi_1)=G_3(\Delta)$, where
$\Delta=e^{2t}((x_0+t+1)^{-2}-(x_1+t+1)^{-2})$. Finally, using this
first-integral, we get that the common
$t$-dependent superposition rule for the Lie family (\ref{LieFamily}) reads
$$
k=e^{2t}((x_0+t+1)^{-2}-(x_1+t+1)^{-2}),
$$
with $k$ being a real constant. Therefore, given any particular solution
$x_1(t)$ of a particular instance of
the family of first-order Abel equations (\ref{LieFamily2}), the general
solution, $x(t)$, of this instance is
$$
x(t)=\left((x_1(t)+t+1)^{-2}+k e^{-2t}\right)^{-1/2}-t-1.
$$

Note that our previous procedure can be straightforwardly generalised to derive
common $t$-dependent superposition rules for generalised Abel equations
\cite{Mo03}, i.e. the differential equations of the form
$$
\frac{dx}{dt}=a_{0}(t)+a_1(t)x+a_2(t)x^2+\ldots+a_n(t)x^n, \qquad n\geq 3.
$$
Actually, their study can be approached by analysing the existence of two vector
fields of the form
\begin{equation*}
\begin{aligned}
Y_1(t,x)&=(b_0(t)+b_1(t)x+\ldots+b_n(t)x^n)\frac{\partial}{\partial x},\\
Y_2(t,x)&=(b'_0(t)+b'_1(t)x+\ldots+b'_n(t)x^n)\frac{\partial}{\partial x},\qquad
b'_n(t)\neq 0,
\end{aligned}
\end{equation*}
obeying the relation $[\bar Y_1,\bar Y_2]=2(\bar Y_2-\bar Y_1)$ and following a
procedure similar to the one developed above.

\section{Lie families and second-order differential equations}
Common $t$-dependent superposition rules describe solutions of nonautonomous
systems of first-order differential equations. Nevertheless, we shall now
illustrate how this new kind of superposition rules can also be applied to
analyse families of second-order differential equations. More specifically, we
shall derive a common $t$-dependent superposition rule in order to express the
general solution of any instance of a family of Milne--Pinney equations
\cite{Car08,Ch40,Re99,ReIr01} in terms of each generic pair of particular
solutions, two constants, and the variable $t$, i.e. the time. In this way, we
provide a generalization to the setting of dissipative Milne--Pinney equations
of the expression previously derived to analyse the solutions of Milne--Pinney
equations in \cite{CL08b}.

Consider the family of dissipative Milne--Pinney equations 
\cite{PK07,Re99,ReIr01,Th52} of the form
\begin{equation}\label{InfFam}
 \begin{aligned}
 \ddot x&=-\dot F\dot x+\omega^2x+e^{-2F}x^{-3},
 \end{aligned}
\end{equation}
with a fixed $t$-dependent function $F=F(t)$, and parametrised by an arbitrary
$t$-dependent function
$\omega=\omega(t)$. The physical motivation for the study of dissipative
Milne--Pinney equations comes from its appearance in dissipative quantum
mechanics \cite{NBA97,Ha75,Na86,Sr86}, where, for instance, their solutions are
used to obtain Gaussian solutions of non-conservative $t$-dependent quantum
oscillators \cite{Na86}. Moreover, the mathematical properties of the solutions
of dissipative Milne--Pinney equations have been studied by several authors from
different points of view as well as for different purposes
\cite{CGL08,CL08b,CL08Diss,RC82,Ha10,Re99,ReIr01,Wa68}. As relevant instances,
consider the works \cite{CL08Diss,Re99} which outline the state-of-the-art of
the investigation of dissipative and non-dissipative Milne--Pinney equations.
One of the main achievements on this topic (see \cite[Corollary 5]{Re99}) is
concerned with an expression describing the general solution of a particular
class of these equations in terms of a pair of generic particular solutions of a
second-order linear differential equations and two constants. Recently, the
theory of quasi-Lie schemes and the theory of Lie systems has enabled us to
recover this latter result and other new ones from a geometric point of view
\cite{CGL08,CLR08}.

Note that introducing a new variable $v\equiv \dot x$, we transform the family
(\ref{InfFam}) of second-order
differential equations into a family of first-order ones
\begin{equation}\label{LieFamily2}
\left\{\begin{aligned}
\dot x&=v,\\
\dot v&=-\dot Fv+\omega^2 x+e^{-2F}x^{-3},
\end{aligned}\right.
\end{equation}
whose dynamics is described by the family of $t$-dependent vector fields on
${\rm T}\mathbb{R}$
parametrised by $\omega$ of the form
\begin{equation*}
Y_\omega=\left(-\dot Fv+e^{-2F}x^{-3}+\omega^2x \right)\pd{}{v}+v\pd{}{x},
\qquad \omega\in
\Lambda=C^{\infty}(t).
\end{equation*}
Let us show that the above family is a Lie family whose common superposition
rule can be used to analyse the
solutions of the family (\ref{InfFam}).

In view of Theorem \ref{MT}, if the family of systems related to the above
family of $t$-dependent vector
fields is a Lie family, that is, it admits a common $t$-dependent superposition
rule in terms of $m$
particular solutions, then the family of vector fields on
$\mathbb{R}\times\mathbb{R}^{n(m+1)}$ given by
${\rm Lie}(\{Y_\omega\}_{\omega\in\Lambda})$ spans an involutive generalised
distribution with leaves of rank $r\leq n\cdot m+1$.

Note that the distribution spanned by all $\widetilde Y_\omega$ is generated by
the vector fields $\widetilde
Y_1$ and $\widetilde Y_2$, with
\begin{equation*}
Y_1=\left(-\dot Fv+e^{-2F}x^{-3}+x \right)\pd{}{v}+v\pd{}{x}, \quad
Y_2=\left(-\dot
Fv+e^{-2F}x^{-3}\right)\pd{}{v}+v\pd{}{x},
\end{equation*}
since $\widetilde Y_\omega= (1-\omega^2)\widetilde Y_2+\omega^2 \widetilde Y_1$.
 The
prolongation $[\widetilde Y_1,\widetilde Y_2]$ is not spanned by $\widetilde
Y_1$ and $\widetilde Y_2$ and, so
we have to include the prolongation $Y^\wedge_3=[\widetilde Y_1,\widetilde Y_2]$
to the picture, where
$$
Y_3=x\pd{}{x}-(v+x\dot F)\pd{}{v}.
$$
In the case $m=0$, the distribution spanned by the vector fields, $\widetilde
Y_1,\widetilde Y_2, Y^\wedge_3$, does not admit a non-trivial first-integral. 
In the case $m>0$, the vector fields, $\widetilde
Y_1,\widetilde Y_2, Y^\wedge_3,$ do not span the linear space ${\rm
Lie}(\{\widetilde Y_\omega\}_{\omega\in\Lambda})$ and we need to add a new
 prolongation  $Y^\wedge_4=[\widetilde Y_1,[\widetilde Y_1,\widetilde Y_2]]$ to
the previous set, with
$$
Y_4=(2v+x\dot F)\pd{}{x}+(2e^{-2F}x^{-3}-2x-\dot F(v+x\dot F)-x\ddot F)\pd{}{v}.
$$
The vector fields, $\widetilde Y_1,\widetilde Y_2, Y^\wedge_3, Y^\wedge_4,$
satisfy the commutation relations
$$
\begin{aligned}
\left[\widetilde Y_1,\widetilde Y_2\right]&=Y^\wedge_3,\cr \left[\widetilde
Y_1,Y^\wedge_3\right]&=Y^\wedge_4,\cr \left[\widetilde
Y_1,Y^\wedge_4\right]&=(4+\dot F^2+2\ddot F)Y^\wedge_3-(\dot
F\ddot F+\dddot F)(\widetilde Y_1-\widetilde Y_2),\cr \left[\widetilde
Y_2,Y^\wedge_3\right]&=2(\widetilde
Y_1-\widetilde Y_2)+ Y^\wedge_4,\cr \left[\widetilde
Y_2,Y^\wedge_4\right]&=(2+\dot F^2+2\ddot F)
Y^\wedge_3-(\dot F\ddot F+\dddot F)(\widetilde Y_1-\widetilde Y_2),\cr
\left[Y^\wedge_3,Y^\wedge_4\right]&=-2Y^\wedge_4-2(\widetilde Y_1-\widetilde
Y_2) (4+\dot F^2+2\ddot F).\\
\end{aligned}
$$
Consequently, the vector fields $\widetilde Y_1,\widetilde Y_2,Y^\wedge_3,
Y^\wedge_4$ span the linear space ${\rm Lie}(\{\widetilde
Y_\omega\}_{\omega\in\Lambda})$. Adding $\widetilde Y_1$ to each prolongation of
the previous set, that is,
by considering the vector fields $\widetilde X_1=\widetilde Y_1$, $\widetilde
X_2=\widetilde Y_2$, $\widetilde
X_3=\widetilde Y_1+Y^\wedge_3$, and $\widetilde X_4=\widetilde Y_1+Y^\wedge_4$,
we get the family of
$t$-prolongations, $\widetilde X_1,\widetilde X_2,\widetilde X_3,\widetilde
X_4$, which spans the vector
fields of the family ${\rm Lie}(\{\widetilde Y_\omega\}_{\omega\in\Lambda})$.
The commutation relations among them read
$$
\begin{aligned}
\left[\widetilde X_1,\widetilde X_2\right]&=\widetilde X_3-\widetilde X_1,\cr
\left[\widetilde X_1,\widetilde
X_3\right]&=\widetilde X_4-\widetilde X_1,\cr \left[\widetilde X_1,\widetilde
X_4\right]&=-(\dot F\ddot
F+\dddot F+4+\dot F^2+2\ddot F)\widetilde X_1+(\dot F\ddot F+\dddot F)\widetilde
X_2+(4+\dot F^2+2\ddot
F)\widetilde X_3,\cr \left[\widetilde X_2,\widetilde X_3\right]&=2\widetilde
X_1-2\widetilde X_2-\widetilde
X_3+\widetilde X_4,\cr \left[\widetilde X_2,\widetilde X_4\right]&=-(1+\dot
F^2+2\ddot F+\dot F\ddot F+\dddot
F)\widetilde X_1+ (\dot F\ddot F+\dddot F)\widetilde X_2+(1+\dot F^2+2\ddot
F)\widetilde X_3,\cr
\left[\widetilde X_3,\widetilde X_4\right]&=-3 \widetilde X_4+(4+\dot F^2+2\ddot
F)\widetilde X_3+(8+\dddot
F+\dot F\ddot F+2\dot F^ 2+4\ddot
F)\widetilde X_2+\\
+&(-9-3\dot{ F}^2-6\ddot F-\dot F\ddot F-\dddot F)\widetilde X_1.\cr
\end{aligned}
$$
As a consequence of Lemma \ref{Aut}, we get that the vector fields $\bar X_1$,
$\bar X_2$ $\bar X_3$ and
$\bar X_4$ satisfy the same commutation relations as the vector fields
$\widetilde X_1$, $\widetilde X_2$,
$\widetilde X_3$, $\widetilde X_4$. Hence, in view of Theorem \ref{MT}, the
family (\ref{LieFamily2}) is a Lie
family and the knowledge of non-trivial first-integrals of the vector fields of
the distribution $\mathcal{D}$
spanned by $\widetilde X_1$, $\widetilde X_2$, $\widetilde X_3$, $\widetilde
X_4$ provides us with a common
$t$-dependent superposition rule.

Let us now turn to determine the aforementioned common $t$-dependent
superposition rule. As the vector fields $\widetilde X_1$, $\widetilde
X_1-\widetilde X_2$ and their successive Lie brackets span the whole
distribution $\mathcal{D}$, a function $G:\mathbb{R}\times {\rm
T}\mathbb{R}^3\rightarrow
\mathbb{R}$ is a first-integral for the vector fields of such a distribution if
and only if it is a
first-integral for the vector fields $\widetilde X_1$ and $\widetilde
X_2-\widetilde X_1$. Therefore, we can
reduce the problem of finding first-integrals for the vector fields of the
distribution $\mathcal{D}$ to
finding common first-integrals $G$ for the vector fields $\widetilde X_1$ and
$\widetilde X_1-\widetilde X_2$.

Let us analyse the implications of $G$ being a first-integral of the vector
field
$$
\widetilde X_1-\widetilde X_2=\sum_{i=0}^2x_i\frac{\partial}{\partial v_i}.
$$
The characteristics of the above vector field are the solutions of the system
$$
\frac{dv_0}{x_0}=\frac{dv_1}{x_1}=\frac{dv_2}{x_2},\qquad dx_0=0,\quad
dx_1=0,\quad dx_2=0,\quad dt=0,
$$
that is, the solutions are curves in $\mathbb{R}\times{\rm T}\mathbb{R}^3$ of
the form
$s\mapsto(t,x_0,x_1,x_2,v_0(s),v_1(s),v_2(s))$, with 
$\xi_{02}=x_0v_2(s)-x_2v_0(s)$ and
$\xi_{12}=x_1v_2(s)-x_2v_1(s)$ for two real constants $\xi_{02}$ and $\xi_{12}$.
Thus, there exists a function
$G_2:\mathbb{R}^6\rightarrow\mathbb{R}$ such that
$G(p)=G_2(t,x_0,x_1,x_2,\xi_{02},\xi_{12})$, with
$p\in\mathbb{R}\times {\rm T}\mathbb{R}^3$, $\xi_{02}=x_0v_2-x_2v_0$, and
$\xi_{12}=x_1v_2-v_1x_2$. In other
words, $G$ is a function of $t,x_0,x_1,x_2,\xi_{02},\xi_{12}$.

The function $G$ also satisfies the condition $\widetilde X_1G=0$ which, in
terms of the coordinate system
$\{t,x_0,x_1,x_2,\xi_{02} \xi_{12},v_2\}$, reads
\begin{multline*}
\widetilde X_1 G=\frac{\partial G}{\partial
t}+\frac{(x_0v_2-\xi_{02})}{x_2}\frac{\partial G}{\partial
x_0}+\frac{(x_1v_2-\xi_{12})}{x_2}\frac{\partial G}{\partial
x_1}+v_2\frac{\partial G}{\partial x_2}-\\-
\left[\dot
F\xi_{12}+e^{-2F}\left(\frac{x_2}{x_1^3}-\frac{x_1}{x_2^3}\right)\right]\frac{
\partial
G}{\partial\xi_{12}}-\left[\dot
F\xi_{02}+e^{-2F}\left(\frac{x_2}{x_0^3}-\frac{x_0}{x_2^3}\right)\right]\frac{
\partial G}{\partial\xi_{02}}=0.
\end{multline*}
That is, defining the vector fields
$$
\begin{aligned}
\Xi_1&=\frac{\partial}{\partial t}-\frac{\xi_{12}}{x_2}\frac{\partial }{\partial
x_1}-\frac{\xi_{02}}{x_2}\frac{\partial}{\partial x_0}
+\left[-\dot
F\xi_{12}-e^{-2F}\left(\frac{x_2}{x_1^3}-\frac{x_1}{x_2^3}\right)\right]\frac{
\partial }{\partial\xi_{12}}\\
&+\left[-\dot
F\xi_{02}-e^{-2F}\left(\frac{x_2}{x_0^3}-\frac{x_0}{x_2^3}\right)\right]\frac{
\partial }{\partial\xi_{02}},\\
\Xi_2&=\frac{x_0}{x_2}\frac{\partial }{\partial
x_0}+\frac{x_1}{x_2}\frac{\partial }{\partial x_1}+\frac{\partial}{\partial
x_2},\\
\end{aligned}
$$
the condition $\widetilde X_1G=0$ implies that $\Xi_1G_2+v_2 \Xi_2G_2=0$ and, as
$G_2$ does not depend on
$v_2$, the function $G$ must simultaneously be a first-integral for $\Xi_1$ and
$\Xi_2$, i.e. $\Xi_1G=0$ and
$\Xi_2G=0$.

Applying the method of characteristics to the vector field $\Xi_2$, we get that
$F$ can just depend on
the variables $t,\xi_{02},\xi_{12},\Delta_{02}=x_0/x_2$ and
$\Delta_{12}=x_1/x_2$. In other words, there exists a
function $G_3:\mathbb{R}^5\rightarrow\mathbb{R}$ such that
$G(t,x_0,x_1,x_2,v_0,v_1,v_2)=G_2(t,x_0,x_1,x_2,\xi_{02},\xi_{12})=G_3(t,\xi_{02
},\xi_{12},\Delta_{02},\Delta_{12})$.

We are left to check the implications of the equation $\Xi_1G=0$. With the aid
of the coordinate system
$\{t,\xi_{02},\xi_{12},\Delta_{02},\Delta_{12},v_2,x_2\}$, the previous equation
can be
cast into the form $\Xi_1G=\frac{1}{x_2^2}\Upsilon_1G_3+\Upsilon_2G_3=0$, where
$$
\begin{aligned}
\Upsilon_1&=\sum_{i=0}^1\left(-\xi_{i2}\frac{\partial }{\partial
\Delta_{i2}}-e^{-2F}\left(\Delta_{i2}^{-3}-
\Delta_{i2}\right)\frac{\partial }{\partial\xi_{i2}}\right),\\
\Upsilon_2&= -\dot F\xi_{12}\frac{\partial}{\partial \xi_{12}}-\dot
F\xi_{02}\frac{\partial}{\partial
\xi_{02}}+\frac{\partial}{\partial t}.
\end{aligned}
$$
As $G_3$ only depends on the variables,
$t,\Delta_{02},\Delta_{12},\xi_{12},\xi_{02},$ we have that
$\Upsilon_1G=0$ and $\Upsilon_2G=0$. Repeating {\it mutatis mutandis} the
previous procedures in order to
determine the implications of being a first-integral of $\Upsilon_1$ and
$\Upsilon_2$, we finally get that the
first-integrals of the distribution $\mathcal{D}$ are functions of $I_1,I_2$ and
$I$, with
$$
I_i=e^{2F}(x_0
v_i-x_iv_0)^2+\left[\left(\frac{x_0}{x_i}\right)^2+\left(\frac{x_i}{x_0}
\right)^2\right],\qquad
i=1,2,
$$
and
$$
I=e^{2F}(x_1v_2-x_2v_1)^2+\left[\left(\frac{x_1}{x_2}\right)^2+\left(\frac{x_2}{
x_1}\right)^2\right].
$$
Defining $\bar v_2=e^{F}v_2, \bar v_1=e^{F}v_1$ and $\bar v_0=e^Fv_0$, the above
first-integrals read
$$
I_i=(x_0\bar v_i-x_i\bar
v_0)^2+\left[\left(\frac{x_0}{x_i}\right)^2+\left(\frac{x_i}{x_0}\right)^2\right
],\qquad i=1,2,
$$
and
$$
I=(x_1\bar v_2-x_2\bar
v_1)^2+\left[\left(\frac{x_1}{x_2}\right)^2+\left(\frac{x_2}{x_1}\right)^2\right
].
$$
Note that these first-integrals have the same form as the ones considered in
\cite{CLR08} for $k=1$.
Therefore, we can apply the procedure done there to obtain that
\begin{equation}\label{Super}
x_0=\sqrt{k_1x_1^2+k_2x_2^2+
2\sqrt{\lambda_{12}[-(x_1^4+x_2^4)+I\,x_1^2x_2^2\,]}}\,,\\
\end{equation}
with $\lambda_{12}$ being a function of the form
\begin{equation*}
\lambda_{12}(k_1,k_2,I)= \frac{k_1k_2 I+(-1+k_1^2+k_2^2)}{I^2-4},
\end{equation*}
and where the constants $k_1$ and $k_2$ satisfy special conditions in order to
ensure that $x_0$ is real
\cite{CL08b}.

Expression (\ref{Super}) permits us to determine the general solution, $x(t)$,
of any instance of family
(\ref{InfFam}) in the form
\begin{equation}\label{Super2}
x(t)=\sqrt{k_1x_1^2(t)+k_2x_2^2(t)+
2\sqrt{\lambda_{12}[-(x_1^4(t)+x_2^4(t))+I\,x_1^2(t)x_2^2(t)\,]}}\,,\\
\end{equation}
with
$$
I=e^{2F(t)}(x_1(t)\dot x_2(t)-x_2(t)\dot
x_1(t))^2+\left[\left(\frac{x_1(t)}{x_2(t)}\right)^2+\left(\frac{x_2(t)}{x_1(t)}
\right)^2\right],
$$
in terms of two of its particular solutions, $x_1(t)$, $x_2(t)$, its
derivatives, the constants $k_1$ and
$k_2$, and the variable $t$ (included in the constant of the motion $I$).

Note that the role of the constant $I$ in expression (\ref{Super2}) differs from
the roles played by
$k_1$ and $k_2$. Indeed, the value of $I$ is fixed by the particular solutions
$x_1(t)$, $x_2(t)$ and its
derivatives, while, for every pair of generic solutions $x_1(t)$ and $x_2(t)$,
the values of $k_1$ and $k_2$
range within certain intervals ensuring that $x(t)$ is real.

It is clear that the method illustrated here can also be applied to analyse
solutions of any other family of second-order differential equations related to
a Lie family by introducing the new variable $v=\dot x$. Additionally, it is
worth noting that in the case $F(t)=0$ the family of dissipative Milne--Pinney
equations (\ref{InfFam}) reduces to a family of Milne--Pinney equations broadly
appearing in the literature (see \cite{LA08} and references therein), and the
expression (\ref{Super2}) takes the form of the expression obtained in
\cite{CL08b} for these equations.

\chapter{Conclusions and outlook}

Apart from providing a quite self-contained introduction to the theory of Lie
systems, this essay describes most of the results concerning this theory and its
generalisations developed by the authors and other collaborators along very
recent years. In this way, our work presents a state-of-art of the subject and
establishes the foundations for our present research activity. Let us here
discuss some of the topics which we aim to analyse in a close future and their
relations to the contents of this essay.

The theory of superposition rules for second- and higher-order
differential equations has just been initiated
\cite{CL09SRicc,CL09SecSup,CLR08,CC87,WintSecond,Ve95} and many questions about
this topic must still be clarified. As an example, we can point out that there
exist several approaches to study systems of second-order differential equations
by means of the theory of Lie systems nowadays. For instance, one can use the
SODE Lie system notion \cite{CLR08}, which allows us to study a particular type
of systems of second-order differential equations. In addition, if a
second-order differential equation admits a regular Lagrangian, the
corresponding Hamiltonian formulation can lead to a system of first-order
differential equations which can also be a Lie system \cite{CLRan08}. Analysing
the relations between the results obtained through both approaches is still an
open problem. 

As a consequence of the above considerations, it became interesting to study a
class of Lie systems describing the Hamilton equations of a certain type of
$t$-dependent Hamiltonians. These systems are defined in a symplectic manifold
and this structure provides us with new tools for investigating such Lie
systems. In addition, these tools can be employed to study the integrability and
super-integrability of these particular Lie systems. Our aim is to analyse such
relations in depth in the future. 

After analysing the Lie systems defined in symplectic manifolds, a natural
question arises: What are the properties of those Lie systems describing the
solutions of a system in a Poisson manifold $(N,\{\cdot,\cdot\})$ of the form
$$
\frac{dx}{dt}=\{x,h_t\},\qquad x\in N,
$$
where, for every $t\in\mathbb{R}$, the function $h_t:N\rightarrow\mathbb{R}$
belongs to a finite-dimensional Lie algebra of functions (with respect to the
Poisson bracket). This challenging question has led to the analysis of the properties of
such Lie systems by means of the Poisson structure of the manifold, what represents an
interesting topic of research.

In \cite{BecGagHusWin87,BecGagHusWin90} Winternitz {\it et al.} proposed, for
the first time, a new type of superposition rules, the referred to as {\it
super-superposition rules}, that describe the general solution of a particular
family of systems of first-order differential equations in supermanifolds. These
articles gave rise to many interesting unanswered questions. Although it seems
that the geometric theory developed in 
\cite{CGM07} could easily be generalised to describe the properties of {\it
super-superposition rules}, multiple non-trivial technical problems arise. We
hope to solve such problems in the future and to develop a geometric theory of
Lie systems in graded manifolds.

In \cite[Remark 5]{CGM07}, it was proposed to accomplish the study of B\"acklund
transformations through a slight modification of the methods carried out to
analyse superposition rules geometrically, i.e., by means of a certain type of
flat connection. This topic deserves a further analysis in order to determine
more exactly its relevance and applications.

Since their first appearance in \cite{CGL08}, quasi-Lie schemes have been
employed to investigate multiple systems of differential equations: nonlinear
oscillators \cite{CGL08}, Mathews-Lakshmanan oscillators \cite{CGL08}, Emden
equations \cite{CLL09Emd}, Abel equations \cite{CLRAbel}, dissipative
Milne--Pinney equations \cite{CL08Diss}, etc. There are still many other
applications to be performed, e.g. we expect to apply this theory to study Abel
equations in depth. In addition, it would be interesting to continue the
analysis of the theory of quasi-Lie schemes and, for instance, to develop new
generalisations of this theory. Indeed, we are already investigating a
generalisation for the analysis of certain quantum systems, e.g. the quantum
Calogero-Moser system. In addition, it would be interesting to study the
generalisations of this theory to analyse stochastic Lie-Scheffers systems
\cite{JP09} or Control Lie systems \cite{Clem06}.

As we pointed out at the beginning of this essay, being a Lie system is rather
more an exception than a rule. In addition, just a few, but relevant, Lie
systems are known to have applications in Physics, Mathematics and other
branches of science. Consequently, one of our main purposes remains to find new
instances of Lie systems with remarkable applications. It seems to us that there
still exist  multiple applications of Lie systems and, in the future, we aim to
determine some of them. 

To finish, we hope to have succeeded in showing that the theory of Lie systems,
after more than a century of existence, is still an active and interesting field
of research.

\end{document}